\documentclass{article}
\usepackage{graphicx} 
\usepackage{bm,amsmath,amsfonts,amssymb,amsthm,mathtools}
\usepackage{subcaption}
\usepackage{authblk}
\usepackage{hyperref}
\usepackage[capitalise]{cleveref}
\usepackage[left=1.in,right=1.in,top=1.in,bottom=1.in]{geometry}
\usepackage{booktabs}
\usepackage{siunitx}
\usepackage[T1]{fontenc}
\usepackage[utf8]{inputenc}
\usepackage{csquotes}
\usepackage[english]{babel}
\usepackage[
backend=biber,
style=alphabetic,
sorting=nyt,
giveninits=true
]{biblatex}
\usepackage{accents}
\usepackage{xcolor}
\numberwithin{equation}{section}
\numberwithin{figure}{section}
\numberwithin{table}{section}
\usepackage{float}
\usepackage{changes}

\def\Xint#1{\mathchoice
{\XXint\displaystyle\textstyle{#1}}%
{\XXint\textstyle\scriptstyle{#1}}%
{\XXint\scriptstyle\scriptscriptstyle{#1}}%
{\XXint\scriptscriptstyle\scriptscriptstyle{#1}}%
\!\int}
\def\XXint#1#2#3{{\setbox0=\hbox{$#1{#2#3}{\int}$ }
\vcenter{\hbox{$#2#3$ }}\kern-.6\wd0}}

\def\dashint{\Xint-}

\newtheorem{lem}{Lemma}

\newtheorem{cmt}{Comment}

\newtheorem{propo}{Proposition}

\numberwithin{lem}{section}
\numberwithin{thm}{section}
\numberwithin{crl}{section}
\numberwithin{defn}{section}
\numberwithin{note}{section}
\numberwithin{propo}{section}

\newcommand{\ubar}[1]{\mkern1mu\underline{\mkern-1mu #1\mkern-1mu}\mkern1mu}

\newcommand{\Reynolds}{\operatorname{Re}}

\newcommand{\Prandtl}{\operatorname{Pr}}
\newcommand{\Nusselt}{\operatorname{Nu}}
\newcommand{\Strouhal}{\operatorname{St}}

\newcommand{\pOmega}{\partial\Omega}

\newcommand{\paavg}{\partial\mathrm{avg}}

\newcommand{\dm}[1]{\ubar{#1}}

\newcommand{\deldim}{\dm{\bm{\nabla}}}
\newcommand{\Deltadim}{\dm{\bm{\Delta}}}
\newcommand{\rhofdim}{\dm{\rho}_{\mathrm{f}}}
\newcommand{\rhosdim}{\dm{\rho}_{\mathrm{s}}}
\newcommand{\cfdim}{\dm{c}_{\mathrm{f}}}
\newcommand{\csdim}{\dm{c}_{\mathrm{s}}}
\newcommand{\rhoscsavgdim}{(\rhosdim\csdim)_{\mathrm{avg}}}
\newcommand{\ksinfdim}{\dm{k}_{\mathrm{s,inf}}}
\newcommand{\kfdim}{\dm{k}_{\mathrm{f}}}
\newcommand{\ksdim}{\dm{k}_{\mathrm{s}}}
\newcommand{\nufdim}{\dm{\nu}_{\mathrm{f}}}
\newcommand{\vinfdim}{\dm{v}_\infty}
\newcommand{\Tidim}{\dm{T}_\mathrm{i}}
\newcommand{\Tinfdim}{\dm{T}_\infty}

\newcommand{\tdim}{\dm{t}}
\newcommand{\hdim}{\dm{h}}
\newcommand{\Omegasdim}{\dm{\Omega}_{\mathrm{s}}}
\newcommand{\pOmegasdim}{\dm{\partial\Omega}_{\mathrm{s}}}
\newcommand{\pOmegafdim}{\dm{\partial\Omega}_{\mathrm{f}}}
\newcommand{\Omegafdim}{\dm{\Omega}_{\mathrm{f}}}
\newcommand{\Gammadim}{\dm{\Gamma}}
\newcommand{\tfdim}{\dm{t}_\mathrm{f}}
\newcommand{\tf}{{t}_\mathrm{f}}
\newcommand{\nulldim}{\dm{0}}
\newcommand{\nullbmdim}{\dm{\bm{0}}}
\newcommand{\onedim}{\dm{1}}
\newcommand{\elldim}{\dm{\ell}}
\newcommand{\Elldim}{\dm{\mathcal{L}}}
\newcommand{\xdim}{\dm{\bm{x}}}
\newcommand{\Tsavgdim}{\dm{T}_{\mathrm{s,avg}}^{\mathrm{CHT}}}
\newcommand{\usavg}{u_{\mathrm{s,avg}}}

\newcommand{\havg}{\dm{h}_{\paavg}}
\newcommand{\htavg}{\dm{\overline{h}}}
\newcommand{\hstavg}{\dm{\overline{h}}_{\paavg}}
\newcommand{\Nuavg}{\Nusselt_{\paavg}}
\newcommand{\Nutavg}{\overline{\Nusselt}}
\newcommand{\Nustavg}{\overline{\Nusselt}_{\paavg}}
\newcommand{\etabar}{\overline{\eta}}

\newcommand{\rhodim}{\dm{\rho}}
\newcommand{\cdim}{\dm{c}}
\newcommand{\rhocavgdim}{(\rhodim \cdim)_\mathrm{avg}}
\newcommand{\Tdim}{\dm{T}}
\newcommand{\kdim}{\dm{k}}

\newcommand{\Omegadim}{\dm{\Omega}}
\newcommand{\pOmegadim}{\dm{\partial\Omega}}
\newcommand{\Omegas}{\Omega_\mathrm{s}}

\newcommand{\uLump}{\widetilde{u}^{\mathrm{L}}}
\newcommand{\TLump}{\dm{\widetilde{T}}^{\mathrm{L}}}
\newcommand{\tauLdim}{\dm{\tau}^{\mathrm{L}}_{\mathrm{eq}}}
\newcommand{\tauL}{\tau^{\mathrm{L}}_{\mathrm{eq}}}
\newcommand{\taueq}{\tau_{\mathrm{eq}}}

\newcommand{\uavg}{u_{\mathrm{avg}}}
\newcommand{\Tavg}{\dm{T}_{\mathrm{avg}}}
\newcommand{\Tavgtilde}{\dm{\widetilde{T}}_{\mathrm{avg}}}
\newcommand{\uavgtilde}{\widetilde{u}_{\mathrm{avg}}}
\newcommand{\Ttilde}{\dm{\widetilde{T}}}
\newcommand{\utilde}{\widetilde{u}}

\newcommand{\usnd}{u_\mathrm{s}}

\newcommand{\delnd}{\bm{\nabla}}
\newcommand{\xnd}{\bm{x}}
\newcommand{\tnd}{t}
\newcommand{\Deltand}{\bm{\Delta}}
\newcommand{\Omegasnd}{\Omega_\mathrm{s}}
\newcommand{\Omegafnd}{\Omega_\mathrm{f}}

\newcommand{\tconvdim}{\dm{\tau}_{\mathrm{conv}}}
\newcommand{\tdiffdim}{\dm{\tau}_{\mathrm{diff}}}
\newcommand{\tconvnondim}{\tau_{\mathrm{conv}}}
\newcommand{\tdiffnondim}{\tau_{\mathrm{diff}}}
\newcommand{\uavgexp}{\usavg^{\mathrm{fit}}}
\newcommand{\usavgexp}{\usavg^{\mathrm{fit}}}

\newcommand{\deltaeta}{\delta^{\etabar}}
\newcommand{\deltasig}{\delta^{\sigma}}

\newcommand{\pOmegaindim}{\pOmegadim_{\mathrm{in}}}
\newcommand{\pOmegaoutdim}{\pOmegadim_{\mathrm{out}}}
\newcommand{\pOmegain}{\pOmega_{\mathrm{in}}}
\newcommand{\pOmegaout}{\pOmega_{\mathrm{out}}}

\newcommand{\ufiso}{u^{\mathrm{ISO}}_\mathrm{f}}
\newcommand{\Nuiso}{\Nusselt^{\mathrm{ISO}}}
\newcommand{\Nuavgiso}{\Nuavg^{\mathrm{ISO}}}

\newcommand{\Nustavgiso}{\Nustavg^{\mathrm{ISO}}}

\newcommand{\etaiso}{\eta^{\mathrm{ISO}}}
\newcommand{\etabariso}{\etabar^{\mathrm{ISO}}}

\newcommand{\medbullet}{\mathbin{\vcenter{\hbox{\raisebox{-0.85ex}{\scalebox{0.7}{$\bullet$}}}}}}

\DeclareMathOperator*{\argmin}{arg\,min}

\DeclareMathOperator*{\essinf}{ess\,inf}

\newcommand{\Biot}{B}
\newcommand{\BiDunk}{\operatorname{Bi}_{\mathrm{dunk}}}

\newcommand{\Biotiso}{\Biot^{\mathrm{ISO}}}

\newcommand{\hcht}{\hdim^{\mathrm{CHT}}}
\newcommand{\havgcht}{\havg^{\mathrm{CHT}}}
\newcommand{\htavgcht}{\htavg^{\mathrm{CHT}}}
\newcommand{\hstavgcht}{\hstavg^{\mathrm{CHT}}}
\newcommand{\Bcht}{\Biot^{\mathrm{CHT}}}
\newcommand{\etacht}{\eta^{\mathrm{CHT}}}
\newcommand{\Nucht}{\Nusselt^{\mathrm{CHT}}}
\newcommand{\Nuavgcht}{\Nuavg^{\mathrm{CHT}}}
\newcommand{\Nutavgcht}{\Nutavg^{\mathrm{CHT}}}
\newcommand{\Nustavgcht}{\Nustavg^{\mathrm{CHT}}}
\newcommand{\etabarcht}{\etabar^{\mathrm{CHT}}}

\newcommand{\vdimcht}{\dm{\bm{v}}^{\mathrm{CHT}}}
\newcommand{\pdimcht}{\dm{p}^{\mathrm{CHT}}}
\newcommand{\Tsdimcht}{\dm{T}_{\mathrm{s}}^{\mathrm{CHT}}}
\newcommand{\Tfdimcht}{\dm{T}_{\mathrm{f}}^{\mathrm{CHT}}}
\newcommand{\vndcht}{\bm{v}^{\mathrm{CHT}}}
\newcommand{\pndcht}{p^{\mathrm{CHT}}}
\newcommand{\usndcht}{u_\mathrm{s}^{\mathrm{CHT}}}
\newcommand{\ufndcht}{u_\mathrm{f}^{\mathrm{CHT}}}
\newcommand{\Tsavgdimcht}{\dm{T}_{\mathrm{s,avg}}^{\mathrm{CHT}}}
\newcommand{\usavgcht}{u_{\mathrm{s,avg}}^{\mathrm{CHT}}}
\newcommand{\Biapprox}{\undertilde{\Biot}}

\DeclareMathOperator\erf{erf}

\title{Mathematical Modeling and Error Estimation for the Thermal Dunking Problem: A Hierarchical Approach}

\author[1,*]{\bf{Theron Guo}}
\author[1]{\bf{Kento Kaneko}}
\author[2]{\bf{Claude Le Bris}}
\author[1]{\bf{Anthony T. Patera}}

\affil[1]{Massachusetts Institute of Technology, Cambridge, MA 02139, USA, \newline\texttt{\{t\_guo,kaneko,patera\}@mit.edu}\vspace{.1in}}

\affil[2]{\'Ecole des Ponts and INRIA, Champs-sur-Marne, 77455 Marne La Vall\'ee, France, \newline\texttt{claude.le-bris@enpc.fr}\vspace{.1in}}

\affil[*]{\small{Corresponding author}}

\begin{document}

\maketitle

\begin{abstract}
    We consider the thermal dunking problem, in which a solid body is suddenly immersed in a fluid of different temperature, and study both the temporal evolution of the solid and the associated Biot number—a non-dimensional heat transfer coefficient characterizing heat exchange across the solid-fluid interface. We focus on the small-Biot-number regime. The problem is accurately described by the conjugate heat transfer (CHT) formulation, which couples the Navier-Stokes and energy equations in the fluid with the heat equation in the solid through interfacial continuity conditions. Because full CHT simulations are computationally expensive, simplified models are often used in practice. Starting from the coupled equations, we systematically reduce the formulation to the lumped-capacitance model, a single ordinary differential equation with a closed-form solution, based on two assumptions: time scale separation and a spatially uniform solid temperature. The total modeling error is decomposed into time homogenization and lumping contributions. We derive an asymptotic error bound for the lumping error, valid for general heterogeneous solids and spatially varying heat transfer coefficients. Building on this theoretical result, we introduce a computable upper bound expressed in measurable quantities for practical evaluation. Time scale separation is analyzed theoretically and supported by physical arguments and simulations, showing that large separation yields small time homogenization errors. In practice, the Biot number must be estimated from so-called empirical correlations, which are typically limited to specific canonical geometries. We propose a data-driven framework that extends empirical correlations to a broader range of geometries through learned characteristic length scales. All results are validated by direct numerical simulations up to Reynolds numbers of 10,000.
    
    \textbf{Acknowledgments.} This work is supported by ONR Grant N000142312573. We thank Dr Reza Malek-Madani for his support and also for many helpful scientific discussions over the years.

    \vspace{1em}

\end{abstract}

\tableofcontents

\section{Introduction}\label{sec:intro}
Heat transfer plays a central role in many engineering applications, including heat treatment, electronic cooling, and materials processing~\cite{asm-handbook-volume-04,incropera1988convection}. The three fundamental physical mechanisms are conduction, convection (forced and natural), and radiation~\cite{cengel2014heat,incropera1990fundamentals,ahtt6e}. In many cases, heat exchange occurs between a solid object and a working fluid. Engineers are often concerned with the thermal response of the solid to support design decisions and predict system performance. Although the governing equations—a coupled system of partial differential equations—accurately describe the underlying physics, their solution is computationally expensive and requires detailed material data which may not be readily available. As a result, simplified models that require fewer inputs are often employed for estimation. However, such models are generally motivated heuristically and lack theoretical justification. 

In this work, we develop a rigorous mathematical framework for simplified heat-transfer models and derive corresponding error bounds which quantify their accuracy relative to the full conjugate heat-transfer formulation. To begin, we describe in this introduction the specific problem addressed, formulate the central research question, review related literature, and outline the main contributions along with the roadmap of the manuscript.

\subsection{Problem of Interest: Thermal Dunking Problem}
We consider the \emph{dunking problem}: a passive solid body at initial uniform temperature $\Tidim$ is abruptly placed at time $\tdim = \nulldim$ in a fluid environment at initial and far-field temperature $\Tinfdim$; we wish to find the temperature evolution of the solid body over the time interval $\nulldim < \tdim \leq \tfdim$. Note underlined variables refer to dimensional quantities, while non-dimensional variables are unadorned.

\subsubsection{A truth mathematical model}\label{subsubsec:truth_mathematical_model}
The dunking problem is best described by a two-domain mathematical model: the Navier-Stokes/Boussinesq equations and energy equation in the fluid domain $\Omegafdim$; the heat equation (time-dependent conduction) in the solid domain $\Omegasdim$; and interface conditions enforcing continuity of temperature and heat flux across the fluid-solid boundary $\Gammadim$. This formulation is commonly referred to as \emph{conjugate heat transfer (CHT) analysis} in the heat transfer literature~\cite{perelman1961conjugated}, and will hereafter be regarded as the \emph{truth model}. In this work, we neglect the effects of natural convection and thermal radiation\footnote{The truth model presented here can be extended to account for additional mechanisms of heat transfer. Natural convection is typically included via a buoyancy term in the Navier-Stokes equations, and thermal radiation can be incorporated by modifying the boundary conditions at the interface to include radiative heat exchange. Although our main results are derived under the assumption of forced convection, they may be extended to include these additional effects.}. We have the following governing equations for the velocity and pressure fields $\vdimcht$ and $\pdimcht$,
\begin{subequations}
\begin{alignat*}{3}
    \frac{\partial \vdimcht}{\partial \tdim} + (\vdimcht\cdot \deldim) \vdimcht &= -\frac{1}{\rhofdim} \deldim \pdimcht + \nufdim \Deltadim \vdimcht &\quad& \text{in } \Omegafdim,\\
    \deldim \cdot \vdimcht &= \nulldim &\quad& \text{in } \Omegafdim, \\
    \vdimcht &= \nulldim &\quad& \text{on } \Gammadim, \\
    \vdimcht(\tdim=\nulldim, \cdot) &= \vinfdim \bm{e}_x &\quad& \text{in } \Omegafdim, \\
    \vdimcht(\tdim, |\xdim|\to\infty) &= \vinfdim \bm{e}_x &\quad& \text{in } \Omegafdim,
\end{alignat*}
\end{subequations}
and for the temperature fields $\Tfdimcht$ and $\Tsdimcht$ in the fluid and solid domains,
\begin{subequations}
\begin{alignat*}{3}
    \rhofdim \cfdim \left(\frac{\partial \Tfdimcht}{\partial \tdim} + \vdimcht \cdot \deldim \Tfdimcht \right) &= \kfdim \Deltadim \Tfdimcht &\quad& \text{in } \Omegafdim, \\
    \rhosdim \csdim \frac{\partial \Tsdimcht}{\partial \tdim} &= \deldim \cdot (\ksdim \deldim \Tsdimcht) &\quad& \text{in } \Omegasdim,\\
    \Tfdimcht &= \Tsdimcht &\quad& \text{on } \Gammadim, \\
    \kfdim \deldim \Tfdimcht \cdot \bm{n} &= \ksdim \deldim \Tsdimcht \cdot \bm{n} &\quad& \text{on } \Gammadim, \\
    \Tfdimcht(\tdim=\nulldim, \cdot) &= \Tinfdim &\quad& \text{in } \Omegafdim, \\
    \Tsdimcht(\tdim=\nulldim, \cdot) &= \Tidim &\quad& \text{in } \Omegasdim, \\
    \Tfdimcht(\tdim, |\xdim|\to\infty) &= \Tinfdim &\quad& \text{in } \Omegafdim.
\end{alignat*}
\end{subequations}
The far-field velocity $\vinfdim$ is assumed constant and, without loss of generality, taken to be directed along the $x$-axis. The unit normal vector $\bm{n}$ on $\Gammadim$ is defined to point outward from the solid into the fluid. To ensure a unique solution for the pressure field, the pressure level is fixed by enforcing a zero-mean condition over the fluid domain at all times.

The CHT model requires several parameters: the thermophysical properties of the solid, namely thermal conductivity $\ksdim$, density $\rhosdim$, and specific heat $\csdim$; the properties of the fluid, including density $\rhofdim$, specific heat $\cfdim$, thermal conductivity $\kfdim$, and kinematic viscosity $\nufdim$; and the far-field velocity $\vinfdim$ as well as the initial and far-field temperatures, $\Tidim$ and $\Tinfdim$. The fluid properties are assumed constant, while the solid properties may vary spatially. Of particular interest are multi-phase heterogeneous solids, where the material properties are constant within each phase and may exhibit discontinuities across phase interfaces; while such discontinuities are not directly compatible with the strong form, they are naturally accommodated in the weak formulation, which we adopt for the theoretical analysis and numerical solution. We do not consider temperature-dependent thermophysical properties.

After solving the coupled system, we can compute the \emph{heat transfer coefficient} $\hcht$, defined on the fluid-solid interface $\Gammadim$, as 
\begin{align*}
    \hcht \coloneqq - \frac{\kfdim \deldim \Tfdimcht \cdot \bm{n}}{\Tfdimcht - \Tinfdim}.
\end{align*} This coefficient is of central importance in engineering applications, as it characterizes the rate of heat exchange between the fluid and solid normalized by the difference in interface and far-field temperatures. It generally varies in both space and time. Practitioners are typically interested in the space-time average, $\hstavgcht$. We additionally define the time-dependent spatial average, $\havgcht$, and the spatially varying temporal average, $\htavgcht$.

We now describe the non-dimensionalization of the CHT equations. We associate to the solid domain two length scales: a chosen first length scale, denoted extrinsic, $\elldim$; and a second length scale, denoted intrinsic, $\Elldim \coloneqq |\Omegasdim|/|\pOmegasdim|$, where $|\Omegasdim|$ and $|\pOmegasdim|$ denote the volume and surface area of the solid, respectively. We may then introduce the non-dimensional variables that will be used throughout this work. Lengths are scaled with the extrinsic length scale $\elldim$, and hence the non-dimensional spatial coordinate is given by $\xnd \coloneqq \xdim/\elldim$. Times are scaled with the solid diffusive time scale $\rhoscsavgdim \elldim^2/\ksinfdim$, where $\rhoscsavgdim$ denotes the spatial average of $\rhosdim\csdim$ over $\Omegasdim$ and $\ksinfdim$ is the (essential) infimum of $\ksdim$ over $\Omegasdim$; the resulting non-dimensional time variable (Fourier number) is $\tnd \coloneqq \tdim \ksinfdim / (\rhoscsavgdim \elldim^2)$. Velocity is non-dimensionalized with respect to the free-stream velocity $\vinfdim$, so that $\vndcht \coloneqq \vdimcht/\vinfdim$; pressure is scaled by $\nufdim \rhofdim \vinfdim / \elldim$, and thus the non-dimensional pressure field reads $\pndcht \coloneqq \pdimcht \elldim / (\nufdim \rhofdim \vinfdim)$. Temperature is non-dimensionalized relative to the initial and far-field temperatures: in the solid, $\usndcht \coloneqq (\Tsdimcht - \Tinfdim)/(\Tidim - \Tinfdim)$, and in the fluid, $\ufndcht \coloneqq (\Tfdimcht - \Tinfdim)/(\Tidim - \Tinfdim)$.

The non-dimensional CHT model depends on six key input parameters. These are the volumetric specific heat ratio and the thermal conductivity ratio between fluid and solid,
\begin{align}
    r_1 \coloneqq \frac{\rhofdim \cfdim}{\rhoscsavgdim}, \quad r_2 \coloneqq \frac{\kfdim}{\ksinfdim}; \label{eq:intro_r1r2}
\end{align} the non-dimensional thermophysical properties of the solid, 
\begin{align}
    \sigma \coloneqq \frac{\rhosdim \csdim}{\rhoscsavgdim}, \quad \kappa \coloneqq \frac{\ksdim}{\ksinfdim}; \label{eq:intro_sigma_kappa}
\end{align} and the Reynolds and Prandtl numbers,
\begin{align}
    \Reynolds \coloneqq \frac{\vinfdim \elldim}{\nufdim}, \quad \Prandtl \coloneqq \frac{\nufdim}{\kfdim / (\rhofdim \cfdim)}. \label{eq:intro_re_pr}
\end{align}
The non-dimensional CHT equations are given in~\cref{subsubsec:cht_nondim}.

Finally, we introduce the non-dimensional counterpart of the heat transfer coefficient, the \emph{Nusselt number}, $$\Nucht \coloneqq \hcht \elldim / \kfdim,$$ along with its spatial, temporal, and space-time averages: respectively $\Nuavgcht$, $\Nutavgcht$, and $\Nustavgcht$. To quantify the spatial and temporal variations of the local Nusselt number relative to its mean, we define the variation function $$\etacht \coloneqq \Nucht / \Nustavgcht,$$ and its temporal average, $\etabarcht$. We also define the \emph{Biot number}, $$\Bcht \coloneqq \hstavgcht \elldim / \ksinfdim,$$ which is a constant scalar, independent of space and time. It is directly related to the space-time averaged Nusselt number through $$\Bcht = r_2 \Nustavgcht.$$ While the Nusselt number reflects the intensity of convective heat transfer in the fluid, the Biot number captures how effectively the solid conducts heat relative to the external convective loading. We have the following relations between the heat transfer coefficient, Nusselt number, $\etacht$, and Biot number:
\begin{align}
    \hcht &= \frac{\kfdim}{\elldim} \Nucht = \frac{\ksinfdim}{\elldim} r_2 \Nustavgcht \etacht = \frac{\ksinfdim}{\elldim} \Bcht \etacht, \\ \hstavgcht &= \frac{\kfdim}{\elldim} \Nustavgcht = \frac{\ksinfdim}{\elldim} \Bcht. \label{eq:intro_havgchtNuBi}
\end{align}

We focus on problems in the \emph{small-Biot limit}, $\Bcht\rightarrow 0$. In practice, the small-Biot limit frequently arises, in particular for “everyday” artifacts at modest temperatures subject to natural convection and radiation heat transfer, as well as for smaller artifacts (for example, sensors) subject to modest forced convection heat transfer. In contrast, larger systems at very elevated temperatures or subject to brisk forced convection or change-of-phase heat transfer typically yield larger Biot number. An engineering example of a small-Biot application is the process of annealing by natural convection—for instance, to relieve residual stresses.

We are primarily interested in one quantity of interest (QoI): the weighted average temperature of the solid as a function of time, denoted by $\Tsavgdimcht(\tdim)$ (dimensional) and $\usavgcht(\tnd)$ (non-dimensional); their precise definitions are given in~\cref{sec:cht}.

The truth mathematical model, which we have just described above, is rarely invoked in engineering practice, in particular in the context of engineering estimation. The computational cost is prohibitively high, and the data requirements are restrictive: many materials are inherently heterogeneous, and the exact spatial distribution are often unknown without destructive inspection.

\subsubsection{Robin heat equations}\label{subsubsec:intro_rhe}
The CHT model can be recast as a single-domain heat equation defined solely on the solid domain, in which the fluid domain is eliminated and its influence captured through a Robin boundary condition on the interface. We refer to this formulation as the \emph{equivalent Robin heat equation} (RHE).

As we now only consider the solid domain, we omit the subscript ``s'' and set $\pOmegadim\equiv\Gammadim$. The dimensional equations are given by
\begin{subequations}
\begin{alignat*}{3}
     \rhodim \cdim \frac{\partial \Tdim^{\mathrm{CHT}}}{\partial \tdim} &= \deldim \cdot(\kdim \deldim \Tdim^{\mathrm{CHT}}) &\quad& \text{in } \Omegadim, \\
    \kdim \deldim \Tdim^{\mathrm{CHT}} \cdot \bm{n} + \hdim^{\mathrm{CHT}}(\tdim,\xdim) (\Tdim^{\mathrm{CHT}}-\Tinfdim) &= \nulldim &\quad& \text{on } \pOmegadim, \\
    \Tdim^{\mathrm{CHT}}(\tdim=\nulldim, \cdot) &= \onedim &\quad& \text{in } \Omegadim,
\end{alignat*}
\end{subequations}
with the average temperature as QoI, as before, denoted by $\Tavg^{\mathrm{CHT}}(\tdim)$ (but now without the subscript ``s''). The heat transfer coefficient $\hdim^{\mathrm{CHT}}$ now appears as the Robin coefficient in the boundary condition.

The solution to the RHE exactly reproduces the solid temperature field of the full CHT model, i.e., $\Tdim^{\mathrm{CHT}} \equiv \Tsdimcht$ and $\Tavg^{\mathrm{CHT}} \equiv \Tsavgdimcht$. However, since $\hcht$ is not known \emph{a priori} and must be extracted from the coupled fluid-solid solution, the RHE in this form is not directly useful for practical prediction and estimation. Nonetheless, the equations provide a valuable starting point for systematic model simplification. More details on the equivalence are provided in~\cref{subsec:rhe}.

In the non-dimensional setting, the Robin coefficient is given by $\Bcht \etacht(\tnd,\xnd)$, and the equations are
\begin{subequations}
\begin{alignat*}{3}
     \sigma\frac{\partial u^{\mathrm{CHT}}}{\partial \tnd} &= \delnd \cdot(\kappa \delnd u^{\mathrm{CHT}}) &\quad& \text{in } \Omega, \\
    \kappa \delnd u^{\mathrm{CHT}} \cdot \bm{n} + \Bcht \etacht(\tnd,\xnd) u^{\mathrm{CHT}} &= 0 &\quad& \text{on } \pOmega, \\
    u^{\mathrm{CHT}}(\tnd=0, \cdot) &= 1 &\quad& \text{in } \Omega.
\end{alignat*}
\end{subequations}
We denote the corresponding non-dimensional average temperature by $\uavg^{\mathrm{CHT}}(\tnd)$. Recall that both $\Bcht$ and $\etacht$ are derived from the local Nusselt number $\Nusselt^{\mathrm{CHT}}(\tnd,\xnd)$ through the relations $\Bcht = r_2 \Nustavg^{\mathrm{CHT}}$ and $\etacht(\tnd,\xnd) = \Nusselt^{\mathrm{CHT}}(\tnd,\xnd) / \Nustavg^{\mathrm{CHT}}$, with $\Bcht$ being a constant scalar independent of space and time. We can now motivate the small-Biot limit: for $\Bcht\rightarrow 0$, the Robin boundary condition asymptotically approaches a homogeneous Neumann boundary condition, from which it follows that temperature gradients within the solid are small, and the solid temperature remains approximately uniform. We focus on this limit in the present work.

Our theoretical framework is informed by the general behavior of $\Nusselt^{\mathrm{CHT}}(\tnd,\xnd)$. It is observed that $\Nusselt^{\mathrm{CHT}}$ fluctuates in time but quickly settles around a steady mean value. This behavior is illustrated in~\cref{fig:introduction}, which shows the temporal evolution of the Nusselt number at various locations on a circular cylinder in cross-flow, along with the non-dimensional average solid temperature, for one representative case in the small-Biot regime ($\Bcht=0.0680$ with diameter as length scale). This result was obtained numerically using Nek5000~\cite{nek5000-web-page}, a spectral-element solver for the CHT formulation; additional details are provided in~\cref{subsec:example_cht}. Following an initial sharp drop, the spatially averaged Nusselt number stabilizes by $\tnd \approx 0.3$, with only minor fluctuations thereafter. While a transition occurs near $\tnd \approx 1.5$, the magnitude of $\Nuavgcht$ remains largely unchanged after $\tnd \approx 0.3$. During this period, the solid temperature changes very little, indicating a clear \emph{separation of time scales}: the flow equilibrates much faster than the thermal response of the solid.

To quantify this, we define the convective time scale $\tconvdim \coloneqq \elldim/\vinfdim$, which becomes $\tconvnondim = r_1 / (r_2 \Reynolds \Prandtl)$ under our non-dimensionalization. For the case shown, this yields $\tconvnondim \approx 0.012$. Meanwhile, the average temperature $\usavgcht(\tnd)$ evolves approximately as a simple exponential decay. Defining the equilibrium time constant $\taueq$ such that $\usavgcht(\taueq) = \exp(-1)$, we find $\taueq \approx 3.77$—over two orders of magnitude larger than $\tconvnondim$, indicating a time scale separation. The approximate exponential, 
\begin{align}
    \uavgexp\coloneqq \exp(-\tnd/\taueq) \label{eq:intro_uavg_exp}
\end{align}
is also shown in~\cref{fig:introduction:uavg} and yields an accurate approximation of $\usavgcht$. More details on time scales will be provided in~\cref{subsec:time_homogenization,subsec:lcm_time_scale}.

This time scale separation motivates the definition of the \emph{autonomous Robin heat equation} (RHE$_a$), in which the Robin coefficient---which characterizes the fluid---is assumed to be time-independent but spatially varying. The governing equations are identical to those of the RHE, except that we restrict the heat transfer coefficient to be fixed in time and only vary in space, here denoted $\htavg$:
\begin{subequations}\label{eq:intro_rhe_a}
\begin{alignat}{3}
     \rhodim \cdim \frac{\partial \Ttilde}{\partial \tdim} &= \deldim \cdot(\kdim \deldim \Ttilde) &\quad& \text{in } \Omegadim, \\
    \kdim \deldim \Ttilde \cdot \bm{n} + \htavg(\xdim) (\Ttilde-\Tinfdim) &= \nulldim &\quad& \text{on } \pOmegadim,\label{eq:intro_rhe_bc} \\
    \Ttilde(\tdim=\nulldim, \cdot) &= \onedim &\quad& \text{in } \Omegadim.
\end{alignat}
\end{subequations}
In the non-dimensional formulation, the Robin coefficient becomes $\Biot \etabar(\xdim)$, where $\etabar$ is assumed to vary only in space. 

The Robin coefficient serves as an input to the system and, in principle, can be prescribed arbitrarily, provided that it is positive, i.e., $\htavg > 0$ or equivalently $\Biot\etabar > 0$. However, since no choice of the Robin coefficient can make the RHE$_a$ exactly equivalent to the full CHT formulation, we compute in~\cref{eq:intro_rhe_a} an approximation to the true temperature field (obtained with the CHT formulation), which we denote by $\Ttilde$ in the dimensional setting and $\utilde$ in the non-dimensional setting. The corresponding QoIs—i.e., the average temperatures—are denoted by $\Tavgtilde(\tdim)$ and $\uavgtilde(\tnd)$, respectively.

The RHE$_a$ represents a significant simplification by eliminating the time dependence of the Robin coefficient. However, it requires a spatially varying heat transfer coefficient $\htavg$ (or equivalently, $\Biot\etabar$) as input. Since this quantity varies along the boundary, see~\cref{fig:introduction:Nu}, estimating it directly in engineering practice is generally impractical. As a result, the RHE$_a$ formulation, while useful for analysis, is rarely used for estimation.
\begin{figure}[th]
    \centering
    \begin{subfigure}[b]{0.48\textwidth}
        \centering
        \includegraphics[width=\textwidth]{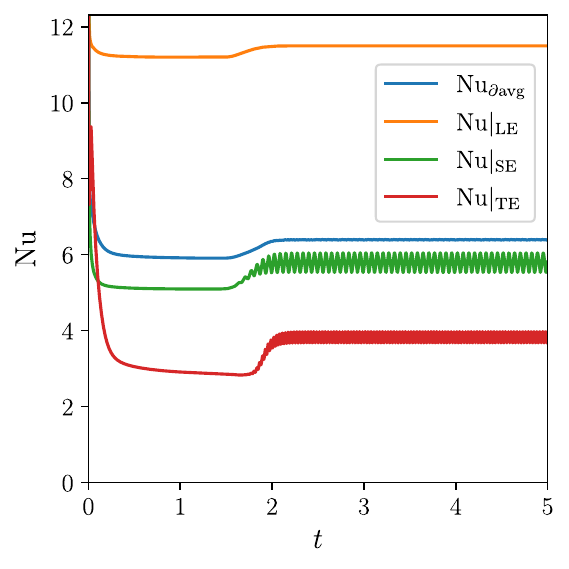}
        \caption{Local and average Nusselt number.}
        \label{fig:introduction:Nu}
    \end{subfigure}
    \hfill
    \begin{subfigure}[b]{0.48\textwidth}
        \centering
        \includegraphics[width=\textwidth]{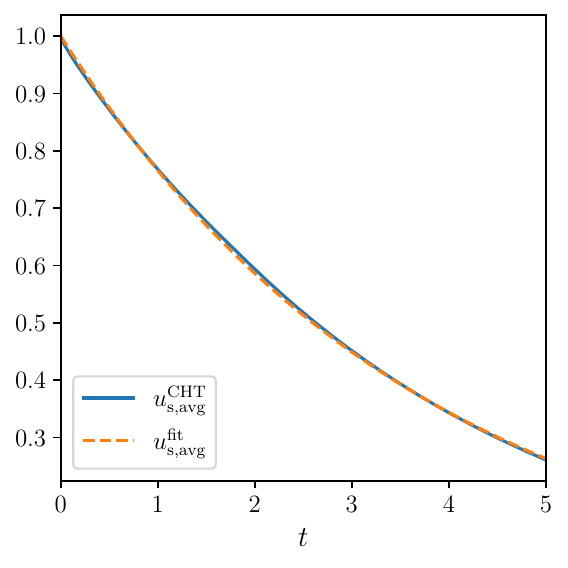}
        \caption{Non-dimensional average temperature over time.}
        \label{fig:introduction:uavg}
    \end{subfigure}
    \caption{(a) Nusselt number at the leading edge (LE), trailing edge (TE), side edge (SE), and the spatial average over the surface of a circular cylinder in cross-flow. The superscripts ``CHT'' are omitted. (b) Evolution of the non-dimensional average solid temperature, which closely follows a simple exponential decay.}
    \label{fig:introduction}
\end{figure}

\subsubsection{Lumped capacitance model}
In engineering estimation, a further reduction—now relative to the autonomous Robin heat equation—is desirable. This leads to the widely used \emph{lumped capacitance model} (LCM), which can be derived from the RHE$_a$ by assuming that the solid maintains a spatially uniform temperature at all times. Denoting this uniform temperature by $\TLump(\tdim)$, we integrate the governing equation in~\cref{eq:intro_rhe_a} over the solid domain and apply the divergence theorem to obtain:
\begin{align*}
    \frac{d\TLump}{d\tdim} &= -\frac{\hstavg |\pOmegadim|}{\rhocavgdim |\Omegadim|} (\TLump - \Tinfdim).
\end{align*}
Hence, the LCM is a single ordinary differential equation (ODE) for $\TLump(\tdim)$, which admits an analytical solution,
\begin{align*}
    (\TLump(\tdim) - \Tinfdim) / (\Tidim - \Tinfdim) = \exp(-\tdim / \tauLdim),
\end{align*}
where $\tauLdim\coloneqq \rhocavgdim\Elldim/\hstavg$ is the associated dimensional time constant. Recall that $\Elldim = |\Omegadim|/|\pOmegadim|$ is the intrinsic length scale. In the non-dimensional setting, the LCM yields a time constant $\tauL \coloneqq 1 / (\Biot \gamma)$, where $\gamma \coloneqq \elldim / \Elldim$ is the ratio between extrinsic and intrinsic length scale, and 
\begin{align*}
    \uLump(\tnd)\coloneqq \exp(-\tnd/\tauL).
\end{align*}
Comparing this expression with the approximate exponential decay of the average temperature in~\cref{eq:intro_uavg_exp}, we see that the two expressions differ only in their time constants, $\taueq$ and $\tauL$. 

The time constant $\tauL$ is inversely proportional to the Biot number $\Biot$. We can compare this time constant with the diffusive time scale in the solid, $\tdiffdim \coloneqq \rhoscsavgdim \elldim^2 / \ksinfdim$, which characterizes how quickly diffusion acts within the solid, and which, under our nondimensionalization, reduces to $\tdiffnondim = 1$. The ratio of the two time scales is given by $\tauL / \tdiffnondim = 1/(\Biot \gamma)$, and shows that for small Biot numbers the cooling of the solid is much slower than the internal diffusion. Hence surface and bulk temperatures equilibrate rapidly relative to the overall thermal response, justifying the uniform temperature assumption for small Biot number. This also means that the LCM is a good approximation only in the small-Biot regime, i.e., as $\Biot \rightarrow 0$. We will study this error in~\cref{subsec:lcm_lumping_error}.

A key advantage of the LCM is its minimal input requirements: the average volumetric specific heat of the solid, $\rhocavgdim$; the intrinsic length scale $\Elldim$, which only depends on body geometry; and the average heat transfer coefficient $\hstavg$. The former two can, in principle, be determined from the solid geometry and composition. The latter is now only a scalar value, much more amenable for estimation.

\subsubsection{Nusselt approximation: empirical correlations}
The biggest challenge in using the LCM is finding a good estimate for the average heat transfer coefficient $\hstavg$ (or $\Biot)$ for a given dunking scenario, \emph{without} solving the CHT. In practice, that means finding a good approximation of the true average Nusselt number $\Nustavgcht$, see~\cref{eq:intro_havgchtNuBi}.

Historically, such approximations have been derived from so-called \emph{empirical correlations}~\cite{incropera1990fundamentals,ahtt6e}, which relate the averaged Nusselt number to the Reynolds and Prandtl numbers (see~\cref{eq:intro_re_pr} for their definitions):
\begin{align}\label{eq:intro_nu_correlation}
    \Nustavgcht \approx f^{\mathcal{S}}(\Reynolds,\Prandtl),
\end{align}
where $f^{\mathcal{S}}$ is defined for a family of similar shapes $\mathcal{S}$, such as spheres or circular cylinders of varying diameter. These expressions typically assume a homogeneous solid with constant properties, i.e., $\kappa = \sigma = 1$ (see the definitions in~\cref{eq:intro_sigma_kappa}), and neglect the thermophysical property ratios $r_1$ and $r_2$ (defined in~\cref{eq:intro_r1r2}). Neglecting the effect on $r_1$ and $r_2$ is justified by the observation that, in many common engineering applications, these ratios are small, and the value of $\Nustavgcht$ shows negligible sensitivity to them.

The functions $f^{\mathcal{S}}$ are typically of power-law form, motivated by boundary layer theory~\cite{schlichting2016boundary}, with coefficients fitted to experimental data. More recently, high-fidelity numerical simulations have also been used to construct such correlations. However, new correlations must be developed for each geometry and flow configuration, making the process labor-intensive. Note that in this work, we use the term correlation interchangeably with empirical correlation, without any statistical interpretation.

\subsection{Research Question}\label{subsec:intro_research_question}
Our goal is to develop an error estimate between the LCM and the full CHT model for \emph{general geometries}, \emph{general thermophysical property distributions}, and \emph{general flow conditions}, in the small-Biot regime. First recall that the non-dimensional RHE QoI, $\uavg^{\mathrm{CHT}}(\tnd;\kappa,\sigma,\Bcht,\etacht(\tnd,\xnd))$, is equivalent to the QoI computed with the CHT formulation. Our objective is to estimate the error incurred by using the LCM with an approximate Biot number $\Biapprox^{\mathrm{CHT}}$ \emph{without} solving the full CHT model, i.e., $\Bcht$ and $\etacht(\tnd,\xnd)$ are unknown. To achieve this, we decompose the total error, given by $|\uavg^{\mathrm{CHT}}(\tnd;\kappa,\sigma,\Bcht,\etacht(\tnd,\xnd)) - \uLump(\tnd;\Biapprox^{\mathrm{CHT}})|$, into three contributions:
\begin{align}
    |\uavg^{\mathrm{CHT}}(\tnd;\kappa,\sigma,\Bcht,\etacht(\tnd,\xnd)) &- \uLump(\tnd;\Biapprox^{\mathrm{CHT}})| \nonumber \\ 
    &\leq |\uavg^{\mathrm{CHT}}(\tnd;\kappa,\sigma,\Bcht,\etacht(\tnd,\xnd)) - \uavgtilde(\tnd;\kappa,\sigma,\Bcht,\etabar(\xnd))| \label{eq:intro_temp_bound}\\
    &+ |\uavgtilde(\tnd;\kappa,\sigma,\Bcht,\etabar(\xnd)) - \uLump(\tnd;\Bcht)| \label{eq:intro_lcm_bound}\\
    &+ |\uLump(\tnd;\Bcht) - \uLump(\tnd;\Biapprox^{\mathrm{CHT}})|. \label{eq:intro_biot_bound}
\end{align}
The first error corresponds to a \emph{temporal approximation error}, incurred by replacing the time-varying Robin coefficient $\Bcht\etacht(\tnd,\xnd)$ with a time-independent approximation $\Bcht\etabar(\xnd)$; the second is the \emph{lumped approximation error}, introduced by replacing the spatially varying RHE$_a$ solution with the LCM; and the third is the \emph{Biot approximation error}, resulting from substituting the Biot number $\Bcht$ with an approximate value $\Biapprox^{\mathrm{CHT}}$.

\subsection{Related Work}
Despite the simplicity and utility of the small-Biot lumped approximation, there is relatively little analysis of the associated approximation error. Most engineering textbooks~\cite{cengel2014heat,incropera1990fundamentals,ahtt6e} indicate only that $\BiDunk$ must be \emph{small}. Here, $\BiDunk$ refers to the Biot number based on the intrinsic length scale $\Elldim$, and relates to our Biot number through $\BiDunk \coloneqq \Biot/\gamma$, where $\gamma=\elldim/\Elldim$. In some cases, a rough threshold is suggested, such as $\BiDunk \leq 0.1$, but typically with vague qualifiers such as “usually”.

To the best of our knowledge, the only study that systematically examines the full CHT-to-LCM error is the work of Virag et al.~\cite{virag2011cooling}, in which the authors perform a numerical comparison between the LCM and full CHT model for a homogeneous sphere subjected to natural convection, where the flow is hence characterized by the Rayleigh number rather than the Reynolds number. Their parametric study spans four non-dimensional inputs: the initial Rayleigh number, initial Biot number, Prandtl number, and the fluid-to-solid thermal diffusivity ratio.
Their results demonstrate that, for many combinations of these parameters, the CHT-to-LCM error is small. Nevertheless, their study is limited to one specific setting, and does not provide any more generalizable theoretical insights.

Other existing works focus solely on the lumped approximation error. In the work of Gockenbach and Schmidtke~\cite{gockenbach2009newton}, a rigorous asymptotic error estimate is derived for the particular case of a sphere with homogeneous thermophysical properties and spatially uniform Robin coefficient. In our notation (see~\cref{eq:intro_lcm_bound}), the result reads: 
\begin{align*}
    |\uavgtilde(\tnd;\kappa=1,\sigma=1,\Bcht,\etabar=1) - \uLump(\tnd;\Bcht)| \leq (3/5) \Bcht / (\gamma\exp(1)) + O((\Biot^{\mathrm{CHT}})^2), \quad \forall \tnd\geq0.
\end{align*}

Recently, we have generalized the bound of Gockenbach and Schmidtke in~\cite{kaneko2024error} to arbitrary geometries with heterogeneous thermophysical properties, while still assuming a spatially uniform Robin coefficient. The bound we derived in~\cite{kaneko2024error}, expressed in our notation here, reads: 
\begin{align}
    |\uavgtilde(\tnd;\kappa,\sigma,\Bcht,\etabar=1) - \uLump(\tnd;\Bcht)| \leq \phi(\kappa,\sigma) \Bcht / (\gamma\exp(1)) + O((\Biot^{\mathrm{CHT}})^2), \quad \forall \tnd\geq0, \label{eq:intro_gockenbach_generalized}
\end{align}
where $\phi(\kappa,\sigma)$ is a dimensionless coefficient determined as the solution to an auxiliary linear elliptic partial differential equation. In the special case of a homogeneous sphere, this bound reduces to that of Gockenbach and Schmidtke with $\phi = 3/5$. An in-depth analysis of the coefficient $\phi$ is also provided, including its mathematical properties, bounds, and numerical evaluations for various geometries. However, because the analysis in~\cite{kaneko2024error} assumes a spatially uniform Robin coefficient—except for some exploratory work on non-uniform Robin coefficients in an Appendix D—the result cannot be directly used to bound the error in the LCM approximation relative to the CHT solution.

Many studies in the literature have focused on estimating the average Nusselt number for particular geometries and flow conditions, usually by proposing empirical correlations derived from experimental data. Classic examples in forced convection include the semi-analytical flat-plate correlation for axial flow~\cite{schlichting2016boundary}, as well as empirical fits for cylinders~\cite{churchill1977correlating} and spheres~\cite{whitaker1972forced} under cross-flow. To extend beyond these specific shapes, Culham et al.~\cite{culham2001simplified} proposed a unified correlation for arbitrary cuboids, based on flow path lengths. Similarly, Yovanovich~\cite{yovanovich1988general} developed a general correlation for spheroids of varying aspect ratios in cross-flow. For natural convection, Hassani et al.~\cite{Hassani1989} investigated a wide variety of shapes and introduced a simple correlation using geometric shape factors to account for different geometries. Lienhard~\cite{lienhard1973commonality} observed that, with an appropriate choice of length scale, the ratio of Nusselt number to Rayleigh number remains nearly constant across many geometries; however, he did not provide a general method for selecting such a length scale.

More recently, high-fidelity numerical simulations have been employed to investigate convective heat transfer in complex geometries. For example, Ke et al.~\cite{ke2018drag} studied the drag force and average Nusselt number for scalene ellipsoids under forced convection. Hirasawa et al.~\cite{hirasawa2012numerical} analyzed heat transfer in packed beds of spheres. Nada et al.~\cite{nada2007heat} combined experimental and numerical approaches to examine semi-circular tubes. A convective heat transfer model for disk-shaped particles ranging from flat disks to cylinders was proposed in~\cite{wang2025new}, while Richter and Nikrityuk~\cite{richter2012drag,richter2013new} considered a broad range of shapes—including cubes, cylinders, spheroids, and cuboids at various angles of attack. This list is not exhaustive but highlights the growing role of simulation in heat transfer studies.

Despite all these contributions in Nusselt number estimation, we are not aware of any systematic studies that examine the influence of the thermophysical property ratios $r_1$ and $r_2$ (defined in~\cref{eq:intro_r1r2}). Most existing works either assume these ratios to be very small—an assumption valid for many engineering applications, and under which the Nusselt number has been empirically shown to be largely insensitive to them—or consider only a limited set of specific cases where they are prescribed explicitly. Moreover, nearly all works are restricted to homogeneous solids; even though many engineering applications—such as pebble-bed nuclear reactors and composite systems—involve heterogeneous materials. In such settings, the fluid-solid property ratios may not be small, and layered or multi-material structures might substantially alter the heat transfer behavior.

\subsection{Contributions}
We make the following contributions:
\begin{enumerate}
    \item In~\cref{subsubsec:cht_lcm_error}, we derive a sharp asymptotic upper bound for the lumped approximation error in~\cref{eq:intro_lcm_bound}:
    \begin{align*}
        |\uavgtilde(\tnd;\kappa,\sigma,\Bcht,\etabar) - \uLump(\tnd;\Bcht)| \leq \phi(\kappa,\sigma,\etabar) \Bcht / (\gamma\exp(1)) + O((\Biot^{\mathrm{CHT}})^2), \quad \forall \tnd\geq0.
    \end{align*}
    This bound extends our previous results in~\cite{kaneko2024error} to the general case of a spatially non-uniform Robin coefficient. The coefficient $\phi(\kappa,\sigma,\etabar)$ generalizes the corresponding quantity in~\cref{eq:intro_gockenbach_generalized}, which was derived under the assumption $\etabar = 1$, and is likewise obtained as the solution to an auxiliary linear elliptic partial differential equation, specified in~\cref{eq:phi_definition,eq:sensitivity_problem}. By accounting for spatial variation in $\etabar$, the present result establishes a direct connection to the full CHT model.
    \item To compute $\phi(\kappa,\sigma,\etabar)$ from the auxiliary problem requires full knowledge of the spatial fields $\kappa$, $\sigma$, and $\etabar$, see~\cref{eq:sensitivity_problem}. Since $\etabar$ depends on the flow and is generally unknown without solving the CHT problem, the derived error bound is not practical for predictive use. Moreover, even $\kappa$ and $\sigma$ may be inaccessible in practice—for example, we often do not know the internal structure of the solid without destructive inspection. To overcome this limitation, we derive an upper bound for $\phi(\kappa,\sigma,\etabar)$ in~\cref{subsubsec:phi_bound},
    \begin{align*}
        \phi(\kappa,\sigma,\etabar) \leq (\sqrt{\phi(1,1,1)} + \sqrt{\deltasig} + \sqrt{\deltaeta})^2,
    \end{align*}
    where $\phi(1,1,1)$ is a purely geometric quantity that can be computed from the auxiliary problem with uniform inputs. The correction terms $\deltasig$ and $\deltaeta$, defined in~\cref{eq:deltasig,eq:deltaeta}, depend on two stability constants—geometric quantities—and on the variances of $\sigma$ and $\etabar$, given by $\int_\Omega (\sigma - 1)^2 / |\Omega|$ and $\int_{\pOmega} (\etabar - 1)^2 / |\pOmega|$ respectively. Note that in this work the term “variance” is not used in the statistical sense, but rather as a measure of spatial heterogeneity. The stability constants are determined by solving two eigenvalue problems, defined in~\cref{eq:stability_1,eq:stability_2}. The numerical treatment of the eigenvalue problems is outlined in Appendix~\ref{sec:steklov_numerical}.

    For a heterogeneous solid with piecewise constant material properties, the variance of $\sigma$ can be evaluated \emph{without} detailed knowledge of the spatial microstructure; only the volume fraction and thermophysical property of each phase is required. The variance of $\etabar$, in contrast, depends on the flow and is therefore more difficult to estimate. Nevertheless, since it is a scalar quantity, it may be realistic to estimate it. From a physical standpoint, it is unlikely that the variance would exceed 1.0, as this would imply regions of extremely high heat transfer adjacent to regions of negligible transfer. Based on both archival data and our own simulation results, we observe that for several convex geometries (e.g., spheres, cylinders), the variance of $\etabar$ does not exceed 0.5. Details on this study are provided in Appendix~\ref{sec:eta_study}. Moreover, this variance tends to decrease with increasing Reynolds number, possibly due to enhanced mixing and turbulence of the flow. In practice, this suggests that one could estimate the variance for a given geometry under low Reynolds number conditions—where simulations are computationally much easier—and safely reuse this bound for higher Reynolds numbers.
    
    \item We investigate the temporal approximation error, defined in~\cref{eq:intro_temp_bound}, through mathematical analysis and numerical simulations in~\cref{subsec:lcm_time_scale}, and find that it remains small when there is a clear separation of time scales between fluid and solid dynamics—that is, when the flow equilibrates much faster than the thermal response of the solid. Motivated by this behavior—and consistent with engineering intuition—we select $\etabar$ as the temporal average of the time-dependent $\etacht$. This can be interpreted as a time homogenization step that averages over fast fluctuations. In Appendix~\ref{sec:time_homogenization_rhe}, we provide a mathematical proof that shows for a simplified setting—a Robin heat equation with a spatially uniform time-periodic Robin coefficient—that the optimal approximation is indeed given by the temporal average.

    \item A key practical question is how to accurately estimate the average Nusselt number $\Nustavgcht$ (and thus the Biot number $\Bcht$) for arbitrary geometries and flow conditions, thereby minimizing the Biot estimation error in~\cref{eq:intro_biot_bound}. In~\cref{subsubsec:universal_correlations}, we propose a novel data-driven framework for estimating the average Nusselt number $\Nustavgcht$ across a wide range of geometries and flow conditions using a single empirical correlation in the form of~\cref{eq:intro_nu_correlation}. For each geometry, we learn an appropriate length scale $\elldim$, inspired by Lienhard's work in natural convection~\cite{lienhard1973commonality}, to be used for the evaluation of the correlation. This has two main advantages: (1)~improved interpretability, and (2)~the ability to extend correlations originally developed for specific shape families to more general geometries. We demonstrate the effectiveness of this approach using a large dataset of spheroidal geometries in~\cref{subsec:numerical_spheroids}.
    
    \item We present numerical results in~\cref{subsubsec:iso_versus_cht} that highlight the limitations of existing empirical correlations for the Nusselt number; as they typically do not account for solid heterogeneity ($\kappa, \sigma$) nor the thermophysical property ratios between solid and fluid ($r_1, r_2$). In this work, we restrict ourselves to the homogeneous case ($\kappa = \sigma = 1$) to systematically investigate the influence of $r_1$ and $r_2$. When $r_1$ and $r_2$ are small, the average Nusselt number $\Nustavgcht$ exhibits negligible sensitivity to these parameters, and empirical correlations remain reasonably accurate. However, this insensitivity breaks down for larger values of $r_1$ and $r_2$, and neglecting their effect leads to significant predictive error. We identify threshold values of $r_1$ and $r_2$ at which the error increases sharply for several representative cases. We also provide material properties of fluid-solid combinations commonly encountered in engineering applications.
\end{enumerate}

\subsection{Roadmap}
In~\cref{sec:models,sec:lcm}, we present the full set of mathematical models considered in this work: the conjugate heat transfer (CHT) model, the Robin heat equation (RHE), the autonomous Robin heat equation (RHE$_a$), and the lumped capacitance model (LCM). We also introduce the isothermal wall formulation (ISO), which was not discussed in this introductory section. We discuss their derivation, non-dimensionalization, and define all quantities of interest. An illustrative example is considered to compare the models. In~\cref{sec:error}, we turn to our main theoretical contribution: the model error analysis between the CHT model and the LCM. We provide detailed proofs and complement the analysis with numerical experiments. \Cref{sec:nusselt} focuses on the Nusselt number. We develop a data-driven framework for estimating the average Nusselt number of general convex geometries by extending the applicability of empirical correlations through learned characteristic length scales, and demonstrate the approach on a family of spheroids. Finally, in~\cref{sec:application}, we summarize our results and offer perspectives on broader implications. Additional technical details on the numerical implementation and simulation studies are provided in the Appendices. Some proofs, previously derived in~\cite{kaneko2024error} and used in this work, are not repeated here.
\section{Foundational Mathematical Models}\label{sec:models}
We present three mathematical models that describe the physics of the dunking problem: the conjugate heat transfer (CHT) formulation, the autonomous Robin heat equation (RHE$_a$), and the isothermal wall (ISO) formulation. The CHT model, previously introduced in~\cref{subsubsec:truth_mathematical_model}, serves as the high-fidelity or “truth” model. The RHE$_a$, discussed in~\cref{subsubsec:intro_rhe}, and the ISO model, which we now introduce, are simplifications of the full model.

Regarding notation, dimensional quantities are indicated with an underline, while non-dimensional quantities are left unadorned. Vector-valued variables are written in boldface. We also introduce the normalized integral
\begin{align*}
    \dashint_{S} (\medbullet) \coloneqq \frac{1}{|S|} \int_{S} (\medbullet),
\end{align*}
where $|S|$ denotes the measure of the set $S$.

\subsection{Conjugate Heat Transfer Formulation (CHT)}\label{sec:cht}
The CHT model describes a coupled fluid-solid system, governed by the incompressible Navier-Stokes equations and advection-diffusion equation in the fluid, the heat equation in the solid, and interface conditions enforcing continuity in the heat flux and temperature. We make the following assumptions:
\begin{itemize}
    \item Effects from natural convection and radiation are negligible.
    \item Pressure variations in the flow do not affect thermodynamic properties (density and temperature).
    \item Fluid and solid properties do not vary substantially over the temperature range dictated by the far-field and the initial solid temperatures, and are therefore assumed constant with respect to temperature.
    \item Viscous dissipation in the fluid is negligible.
    \item The fluid is assumed incompressible, i.e., the velocity field is divergence-free.
    \item The solid domain is stationary.
    \item No heat sources are present.
\end{itemize}

\subsubsection{Dimensional equations}
Let $\Omegafdim$ denote the fluid domain with boundary $\pOmegafdim$ and $\Omegasdim$ the solid domain with boundary $\pOmegasdim$. Both domains are separated by the fluid-solid interface $\Gammadim = \pOmegafdim \cap \pOmegasdim$. The time $\tdim$ is defined from $[\nulldim, \tfdim]$ where $\tfdim$ is the final time. The typical setup that we consider is shown schematically in~\cref{fig:setup}. As we consider an external flow, the outer boundaries of the fluid domain are assumed to be sufficiently far away from the solid domain such that they do not affect the flow around the solid. The inflow boundary is denoted by $\pOmegaindim$, where the fluid enters the domain with far-field velocity $\vinfdim$ and far-field temperature $\Tinfdim$. The outflow boundaries are defined as $\pOmegaoutdim \coloneqq \pOmegafdim \setminus (\pOmegaindim \cup \Gammadim)$.
\begin{figure}
    \centering
    \includegraphics[width=0.5\textwidth]{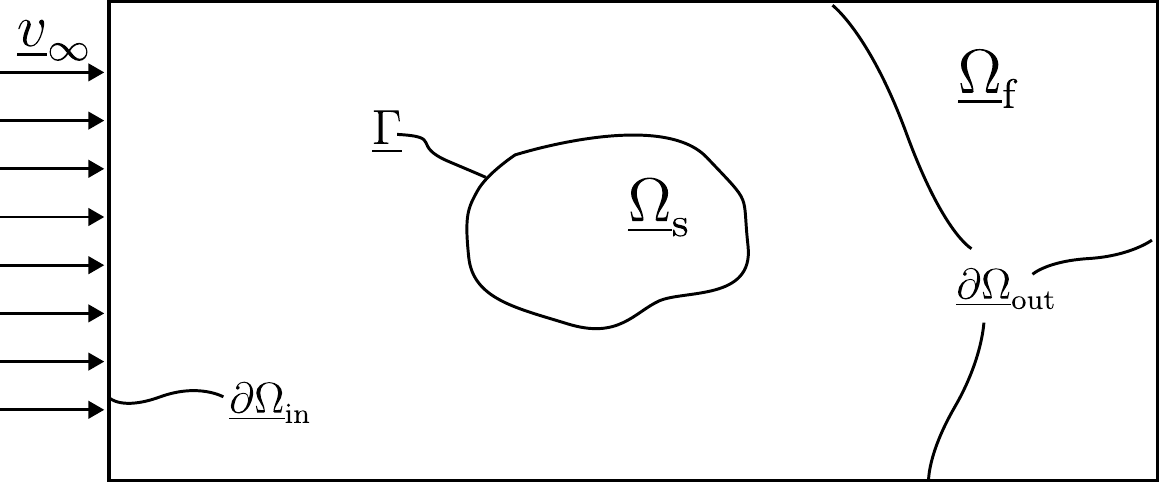}
    \caption{Schematic of the fluid-solid system. The solid domain is denoted by \(\Omegasdim\) and the fluid domain by \(\Omegafdim\). The interface between the two domains is denoted by \(\Gammadim\).}
    \label{fig:setup}
\end{figure}
The fluid has constant material properties: density $\rhofdim$, kinematic viscosity $\nufdim$, heat capacity $\cfdim$, and thermal conductivity $\kfdim$. The solid domain is heterogeneous, characterized by a spatially varying volumetric heat capacity $\rhosdim\csdim(\xdim)$ and thermal conductivity $\ksdim(\xdim)$. The fluid is moving in the far-field  upstream with constant velocity $\vinfdim$ in the $\bm{e}_x$ direction. The initial temperatures of the fluid and solid are given by $\Tinfdim$ and $\Tidim$. 

The velocity and pressure fields in the fluid are governed by the incompressible Navier-Stokes equations, while the temperature fields in both fluid and solid are described by a coupled advection-diffusion equation (in the fluid) and heat conduction equation (in the solid). The incompressible Navier-Stokes equations are given by
\begin{subequations}\label{eq:cht_dim_vp}
\begin{alignat}{3}
    \frac{\partial \vdimcht}{\partial \tdim} + (\vdimcht \cdot \deldim) \vdimcht &= -\frac{1}{\rhofdim} \deldim \pdimcht + \nufdim \Deltadim \vdimcht &\quad& \text{in } \Omegafdim,\\
    \deldim \cdot \vdimcht &= \nulldim &\quad& \text{in } \Omegafdim, \\
    \vdimcht &= \vinfdim\bm{e}_x &\quad& \text{on } \pOmegaindim, \\
    \vdimcht &= \nullbmdim &\quad& \text{on } \Gammadim, \\
    -\pdimcht \bm{n} + \nufdim\rhofdim(\deldim\vdimcht+(\deldim\vdimcht)^T)\cdot \bm{n} &= \nullbmdim &\quad& \text{on } \pOmegaoutdim, \\
    \vdimcht(\tdim=\nulldim, \cdot) &= \vinfdim\bm{e}_x &\quad& \text{in } \Omegafdim.
\end{alignat}
\end{subequations}
To ensure a unique solution for the pressure field, the pressure level is fixed by enforcing a zero-mean condition over the fluid domain at all times. The heat equations within the fluid and the solid are respectively given by the following advection-diffusion and heat equations:
\begin{subequations}\label{eq:cht_dim_T}
\begin{alignat}{3}
    \rhofdim \cfdim \left(\frac{\partial \Tfdimcht}{\partial \tdim} + \vdimcht \cdot \deldim \Tfdimcht \right) &= \kfdim \Deltadim \Tfdimcht &\quad& \text{in } \Omegafdim, \label{eq:cht_dim_v_1}\\
    \rhosdim \csdim \frac{\partial \Tsdimcht}{\partial \tdim} &= \deldim \cdot (\ksdim \deldim \Tsdimcht) &\quad& \text{in } \Omegasdim, \label{eq:cht_dim_T_1}\\
    \Tfdimcht &= \Tinfdim &\quad& \text{on } \pOmegaindim, \\
    \Tfdimcht &= \Tsdimcht &\quad& \text{on } \Gammadim, \label{eq:cht_dim_T_2}\\
    \kfdim \deldim \Tfdimcht \cdot \bm{n} &= \ksdim \deldim \Tsdimcht \cdot \bm{n} &\quad& \text{on } \Gammadim, \label{eq:cht_dim_T_3}\\
    \Tfdimcht(\tdim=\nulldim, \cdot) &= \Tinfdim &\quad& \text{in } \Omegafdim, \\
    \Tsdimcht(\tdim=\nulldim, \cdot) &= \Tidim &\quad& \text{in } \Omegasdim. \label{eq:cht_dim_T_4}
\end{alignat}
\end{subequations}
The unit normal vector $\bm{n}$ on $\Gammadim$ is defined to point from the solid domain into the fluid domain. The Navier-Stokes equations in~\cref{eq:cht_dim_vp} do not depend on the temperature and can be solved independently for the velocity and pressure field. After inserting the velocity field into~\cref{eq:cht_dim_v_1},~\cref{eq:cht_dim_T} can be solved for the temperatures in fluid and solid. The quantity of interest (QoI) is the average temperature in the solid domain, defined as
\begin{align*}
    \Tsavgdim(\tdim) \coloneqq \frac{1}{\int_{\Omegasdim} \rhosdim(\xdim) \csdim(\xdim)} \int_{\Omegasdim} \rhosdim(\xdim) \csdim(\xdim) \Tsdimcht(\tdim,\xdim).
\end{align*}
The weighting factor $\rhosdim(\xdim) \csdim(\xdim)$ acts as a thermal mass density and $\Tsavgdim$ is thus related to the thermal energy in the solid.

\subsubsection{Non-dimensional equations}\label{subsubsec:cht_nondim}
The dimensional formulation involves nine input parameters:
\begin{itemize}
    \item Fluid properties: density $\rhofdim$, kinematic viscosity $\nufdim$, specific heat $\cfdim$, thermal conductivity $\kfdim$,
    \item Solid properties: volumetric specific heat $(\rhosdim\csdim)(\xdim)$, thermal conductivity $\ksdim(\xdim)$,
    \item Flow and thermal conditions: far-field velocity $\vinfdim$ and temperature $\Tinfdim$, initial solid temperature $\Tidim$.
\end{itemize}
To reduce this number and to generalize the solution across a wider class of problems, it is common practice to introduce a non-dimensional formulation. We first define two characteristic length scales associated with the solid domain:
\begin{itemize}
    \item Intrinsic length scale: $\Elldim = |\Omegasdim| / |\pOmegasdim|$, where $|\Omegasdim|$ and $|\pOmegasdim|$ denote the volume and surface area of the solid domain, respectively.
    \item Extrinsic length scale: $\elldim$, which can be chosen arbitrarily. Common choices include the diameter, or volume-equivalent sphere diameter.
\end{itemize}
We then define the solid diffusive time scale with
\begin{align}
    \tdiffdim \coloneqq \frac{\elldim^2 \rhoscsavgdim}{\ksinfdim}, \label{eq:tdiff}
\end{align}
and
\begin{align*}
    \ksinfdim \coloneqq \essinf_{\xdim \in \Omegasdim} \ksdim, \quad \rhoscsavgdim\coloneqq \dashint_{\Omegasdim} \rhosdim \csdim.
\end{align*}
Subsequently, we write the non-dimensional time and space variables as
\begin{align*}
    \xnd \coloneqq \frac{\xdim}{\elldim}, \quad
    \delnd \coloneqq \frac{\deldim}{\elldim^{-1}}, \quad
    \tnd \coloneqq \frac{\tdim}{\tdiffdim}, \quad
    \tf \coloneqq \frac{\tfdim}{\tdiffdim},
\end{align*}
as well as the non-dimensional domains and boundaries
\begin{align*}
    \Omegafnd \coloneqq \left\{ \xnd \ \middle|  \ \xnd = \frac{\xdim}{\elldim}, \ \xdim \in \Omegafdim \right\}, \quad
    \Omegasnd \coloneqq \left\{ \xnd \ \middle| \ \xnd = \frac{\xdim}{\elldim}, \ \xdim \in \Omegasdim \right\}, \quad
    \Gamma \coloneqq \left\{ \xnd \ \middle| \ \xnd = \frac{\xdim}{\elldim}, \ \xdim \in \Gammadim \right\}, \\
    \pOmegain \coloneqq \left\{ \xnd \ \middle| \ \xnd = \frac{\xdim}{\elldim}, \ \xdim \in \pOmegaindim \right\}, \quad 
    \pOmegaout \coloneqq \left\{ \xnd \ \middle| \ \xnd = \frac{\xdim}{\elldim}, \ \xdim \in \pOmegaoutdim \right\}. 
\end{align*}
Moreover, we define the non-dimensional solid material properties
\begin{align}
\kappa(\xnd) \coloneqq \frac{\ksdim(\xdim)}{\ksinfdim}, \quad
\sigma(\xnd) \coloneqq \frac{(\rhosdim \csdim)(\xdim)}{\rhoscsavgdim}, \label{eq:kappa_sigma_1}
\end{align}
which fulfill
\begin{align}
    \inf_{\xnd \in \Omegasnd} \kappa(\xnd) = 1, \quad
    \dashint_{\Omegasnd} \sigma(\xnd) = 1; \label{eq:kappa_sigma_2}
\end{align}
the fluid-to-solid thermophysical property ratios
\begin{align}
    r_1 \coloneqq \frac{\rhofdim \cfdim}{\rhoscsavgdim}, \quad
    r_2 \coloneqq \frac{\kfdim}{\ksinfdim}; \label{eq:r1_r2}
\end{align}
and the Reynolds and Prandtl numbers
\begin{align}
    \Reynolds\coloneqq \frac{\vinfdim\elldim}{\nufdim}, \quad 
    \Prandtl\coloneqq \frac{\nufdim}{\kfdim / (\rhofdim\cfdim)}. \label{eq:Re_Pr}
\end{align}
Note that the Prandtl number is a material parameter of the fluid, while the Reynolds number depends on the chosen length scale.

Finally, with the non-dimensional velocity, pressure, and temperature fields defined as
\begin{align}
    \vndcht(\tnd,\xnd) \coloneqq \frac{\vdimcht(\tdim,\xdim)}{\vinfdim}&, \quad
    \pndcht(\tnd,\xnd) \coloneqq \frac{\pdimcht(\tdim,\xdim)}{\nufdim \rhofdim \vinfdim / \elldim},\\
    \ufndcht(\tnd,\xnd) \coloneqq \frac{\Tfdimcht(\tdim,\xdim) - \Tinfdim}{\Tidim - \Tinfdim}&, \quad
    \usndcht(\tnd,\xnd) \coloneqq \frac{\Tsdimcht(\tdim,\xdim) - \Tinfdim}{\Tidim - \Tinfdim},
\end{align}
we can write the non-dimensional incompressible Navier-Stokes equations,
\begin{subequations}\label{eq:cht_nondim_vp}
\begin{alignat}{3}
    r_1(r_2)^{-1}\frac{\partial \vndcht}{\partial \tnd} + \Reynolds\Prandtl(\vndcht\cdot \delnd) \vndcht &= -\Prandtl\delnd \pndcht + \Prandtl\Deltand \vndcht &\quad& \text{in } \Omegafnd,\label{eq:ns_nondim_1}\\
    \delnd \cdot \vndcht &= 0 &\quad& \text{in } \Omegafnd, \label{eq:ns_nondim_2}\\
    \vndcht &= \bm{e}_x &\quad& \text{on } \pOmegain, \label{eq:ns_nondim_3}\\
    \vndcht &= \bm{0} &\quad& \text{on } \Gamma, \label{eq:ns_nondim_4}\\
    -\pndcht\bm{n} + (\delnd \vndcht+(\delnd \vndcht)^T)\cdot\bm{n} &=0 &\quad& \text{on } \pOmegaout, \label{eq:ns_nondim_5}\\ 
    \vndcht(\tnd=0, \cdot) &= \bm{e}_x &\quad& \text{in } \Omegafnd, \label{eq:ns_nondim_6}
\end{alignat}
\end{subequations}
and the non-dimensional advection-diffusion equation and heat equation,
\begin{subequations}\label{eq:cht_nondim_T}
\begin{alignat}{3}
    r_1 \frac{\partial \ufndcht}{\partial \tnd} + \Reynolds\Prandtl r_2\vndcht \cdot \delnd \ufndcht &= r_2\bm{\Delta} \ufndcht &\quad& \text{in } \Omegafnd, \label{eq:cht_nondim_1}\\
    \sigma \frac{\partial \usndcht}{\partial \tnd} &= \delnd\cdot(\kappa \delnd \usndcht) &\quad& \text{in } \Omegasnd, \label{eq:cht_nondim_2}\\
    \ufndcht &= 0 &\quad& \text{on } \pOmegain, \label{eq:cht_nondim_3}\\
    \ufndcht &= \usndcht &\quad& \text{on } \Gamma, \label{eq:cht_nondim_4}\\
    r_2 \delnd \ufndcht \cdot \bm{n} &= \kappa \delnd \usndcht \cdot \bm{n} &\quad& \text{on } \Gamma, \label{eq:cht_nondim_5} \\
    \ufndcht(\tnd=0, \cdot) &= 0 &\quad& \text{in } \Omegafnd, \label{eq:cht_nondim_6}\\
    \usndcht(\tnd=0, \cdot) &= 1 &\quad& \text{on } \Omegasnd. \label{eq:cht_nondim_7}
\end{alignat}
\end{subequations}
The QoI is the non-dimensional weighted average temperature in the solid domain, defined as
\begin{align}
    \usavgcht(\tnd) \coloneqq \frac{1}{\int_{\Omegasnd} \sigma(\xnd)}\int_{\Omegasnd} \sigma(\xnd) {\usndcht}(\tnd,\xnd) = \dashint_{\Omegasnd} \sigma(\xnd) {\usndcht}(\tnd,\xnd), \label{eq:uavg}
\end{align}
where we used~\cref{eq:kappa_sigma_2} in the last equality.

The non-dimensional formulation has six parameters: 
\begin{itemize}
    \item Two dimensionless groups characterizing the flow: $\Reynolds$ and $\Prandtl$,
    \item Two fluid-to-solid thermophysical property ratios: $r_1$ and $r_2$,
    \item Two non-dimensional solid properties: $\kappa(\xnd)$ and $\sigma(\xnd)$.
\end{itemize}

\subsubsection{Heat transfer coefficient, Nusselt number, and Biot number}
We can associate to our CHT solution a \textit{heat transfer coefficient} to characterize the heat exchange between fluid and solid. It is defined on the interface $\Gammadim$ as the heat flux normalized by the difference in far-field temperature and interface temperature,
\begin{align}
\hcht(\tdim,\xdim) \coloneqq - \frac{\kfdim\deldim \Tfdimcht \cdot \bm{n}}{\Tfdimcht - \Tinfdim}.
\label{eq:htc_dim}
\end{align}
The larger the heat transfer coefficient, the more heat is transferred, per unit temperature difference, between fluid and solid. Inserting the non-dimensional temperature field, we can rewrite the heat transfer coefficient in~\cref{eq:htc_dim} as
\begin{align*}
    \hcht(\tdim,\xdim) = -\frac{\kfdim}{\elldim} \frac{\delnd \ufndcht \cdot \bm{n}}{\ufndcht},
\end{align*}
and define its non-dimensional counterpart, the \textit{Nusselt number},
\begin{align}
    \Nucht(\tnd,\xnd) \coloneqq \frac{\hcht(\tnd,\xnd) \elldim}{\kfdim} = -\frac{\delnd\ufndcht\cdot\bm{n}}{\ufndcht}, \label{eq:local_nusselt}
\end{align}
which is of central importance and generally depends on all six non-dimensional parameters, i.e.,
\begin{align*}
    \Nucht(\tnd,\xnd) = \Nucht(\tnd,\xnd; \Reynolds, \Prandtl, r_1, r_2, \kappa(\xnd), \sigma(\xnd)).
\end{align*}
For many engineering applications, the space-time averaged heat transfer coefficient or Nusselt number is of interest,
\begin{align}
    \hstavgcht &\coloneqq \dashint_{\Gammadim} \dashint_{0}^{\tfdim} \hcht(\tdim,\xdim), \label{eq:cht_htc} \\
    \Nustavgcht &\coloneqq \frac{\hstavgcht\elldim}{\kfdim} = \dashint_{\Gamma} \dashint_{0}^{\tf} \Nucht(\tnd,\xnd) = \dashint_{\Gamma} \dashint_{0}^{\tf} \left(- \frac{\delnd \ufndcht \cdot \bm{n}}{\ufndcht}\right). \label{eq:cht_nusselt}
\end{align}
We further define the spatial and temporal average of the Nusselt number with,
\begin{align*}
    \Nutavgcht(\xnd) &\coloneqq \dashint_{0}^{\tf} \Nucht(\tnd,\xnd) \\
    \Nuavgcht(\tnd) &\coloneqq \dashint_{\Gamma} \Nucht(\tnd,\xnd).
\end{align*}
Another important engineering quantity is the non-dimensional \textit{Biot number}, which is a constant scalar and closely related to the space-time averaged Nusselt number, defined as
\begin{align}
    \Bcht &\coloneqq \frac{\hstavgcht\elldim}{\ksdim} = r_2 \Nustavgcht. \label{eq:cht_biot}
\end{align}
Of particular interest is the limit of small Biot number, which has several implications for the heat transfer behavior of the solid. This will be discussed in more detail in~\cref{subsec:rhea}. Lastly, we can define the \textit{variation function}
\begin{align}
    \etacht(\tnd,\xnd) \coloneqq \frac{\Nucht(\tnd,\xnd)}{\Nustavgcht}, \label{eq:cht_uniformity}
\end{align}
which is a measure for the variation of the heat transfer on the interface, with
\begin{align}
    \dashint_\Gamma \dashint_0^{\tf} \etacht = 1,
\end{align}
and $\etacht=1$ if $\Nucht$ does not depend on space and time; we also define the temporal average as
\begin{align*}
    \etabarcht(\xnd) \coloneqq \dashint_0^{\tf} \etacht(\tnd,\xnd).
\end{align*}
Note, that the variation function does not depend on the choice of length scale $\elldim$.

\subsubsection{Example: Cross-flow around a circular cylinder (CHT)}\label{subsec:example_cht}
We conclude this section with a numerical example: a cross-flow over a circular cylinder of diameter $\dm{D}$ in two dimensions with geometry shown in Figure~\ref{fig:cross_flow}. This example has already been briefly discussed in~\cref{subsubsec:intro_rhe}.
We choose the diameter of the cylinder as extrinsic length scale, i.e., $\elldim = \dm{D}$, and solve the non-dimensional system. The parameters considered are given in~\cref{tab:flow_params}, and the properties of the solid are assumed to be uniform, i.e., $\sigma(\xnd) = \kappa(\xnd) = 1$ for all $\xnd \in \Omegasnd$. The chosen values of $r_1$ and $r_2$ do not correspond to any specific fluid-solid pair; they are selected so that the CHT problem exhibits the qualitative behavior typical of common fluid-solid heat-transfer problems. The Prandtl number corresponds to air at standard conditions. We compute the CHT solution with the spectral element solver Nek5000~\cite{nek5000-web-page}, using 7th-order polynomials for spatial discretization and BDF3 for time stepping until $\tf=5.0$. We use an adaptive time step that maintains a CFL condition of around 0.4. Verification studies for the numerical results are provided in Appendix~\ref{subsec:convergence_study_cyl}.
\begin{figure}[ht]
\centering
\includegraphics[width=0.6\textwidth]{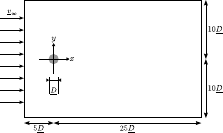}
\caption{Cross-flow over a circular cylinder of diameter $\dm{D}$ in two dimensions (not to scale). The flow enters the domain from the left, as illustrated in~\cref{fig:setup}.}
\label{fig:cross_flow}
\end{figure}
\begin{table}[ht]
\centering
\caption{Flow and material parameters used in the CHT example.}
\label{tab:flow_params}
\begin{tabular}{c c c c c}
\toprule
$\Prandtl$ & $\Reynolds$ & $r_1$ & $r_2$ \\
\midrule
0.71 & 143 & 0.0132 & 0.01085 \\
\bottomrule
\end{tabular}
\end{table}

A few snapshots of the temperature field around the cylinder are shown in~\cref{fig:disk_temp}. After initial boundary layer formation, the wake becomes instable and vortices start shedding. To analyze the results, three spatial locations on the cylinder surface $\Gamma$ are defined, see~\cref{fig:disk_0}: leading edge (LE), trailing edge (TE), and side edge (SE). The leading edge is the point where the flow first impinges on the cylinder, the trailing edge is the point where the flow leaves the cylinder, and the side edge is the point on the cylinder surface that is perpendicular to the flow direction.

\begin{figure}[ht]
\centering
\begin{subfigure}{0.19\textwidth}
    \centering
    \includegraphics[width=\linewidth]{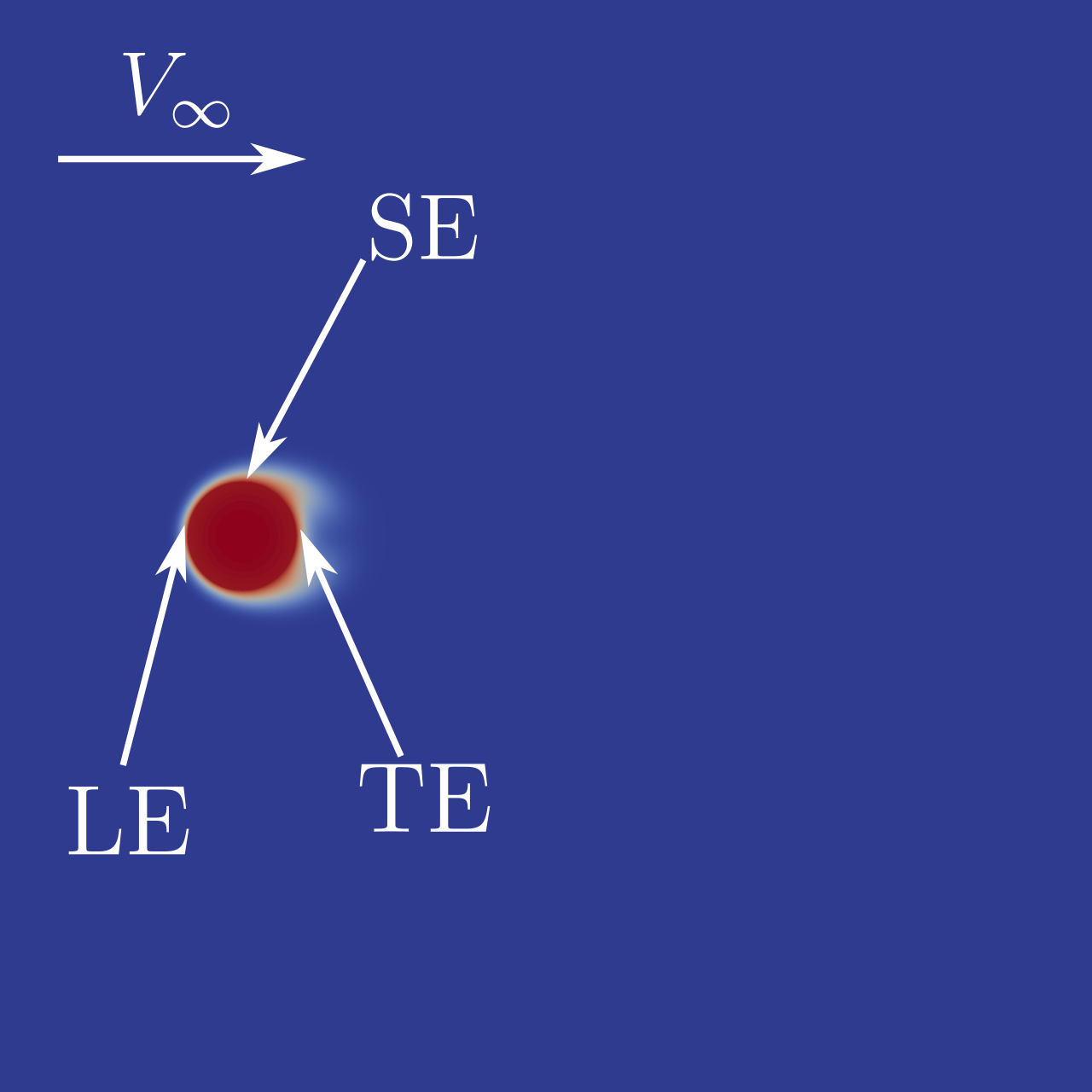}
    \caption{$\tnd=0.013$}
    \label{fig:disk_0}
\end{subfigure}
\hfill
\begin{subfigure}{0.19\textwidth}
    \centering
    \includegraphics[width=\linewidth]{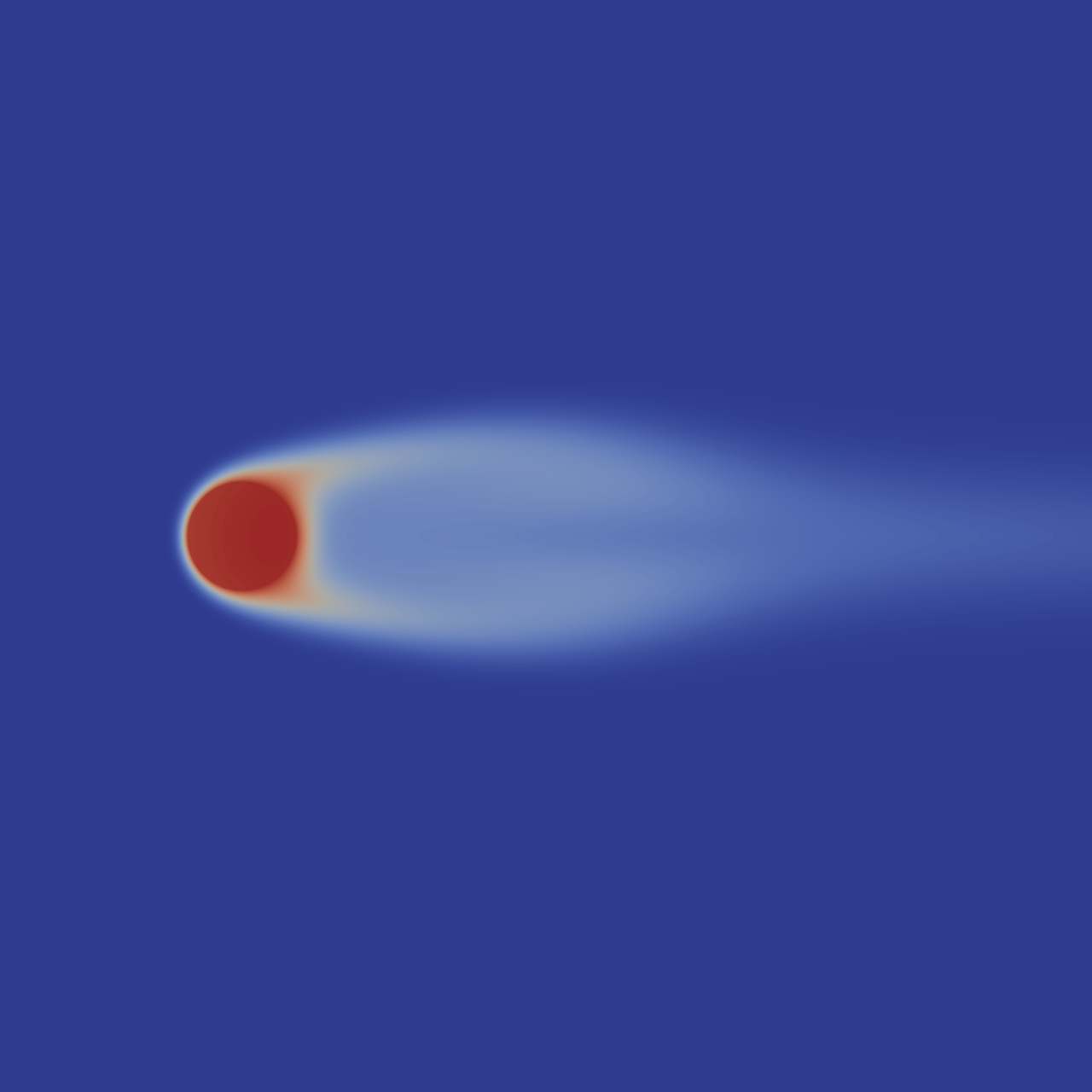}
    \caption{$\tnd=0.2$}
    \label{fig:disk_1}
\end{subfigure}
\hfill
\begin{subfigure}{0.19\textwidth}
    \centering
    \includegraphics[width=\linewidth]{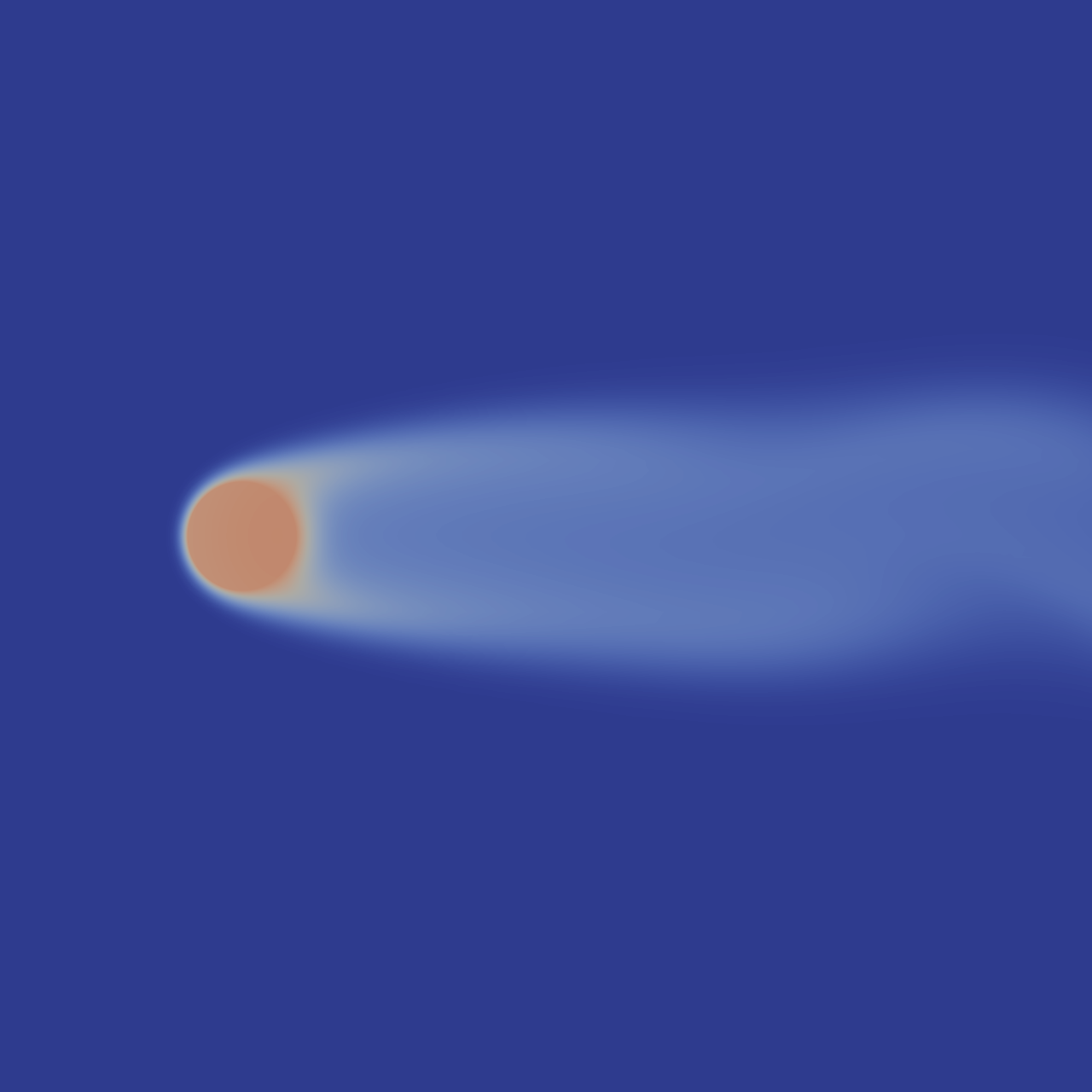}
    \caption{$\tnd=1.5$}
    \label{fig:disk_2}
\end{subfigure}
\hfill
\begin{subfigure}{0.19\textwidth}
    \centering
    \includegraphics[width=\linewidth]{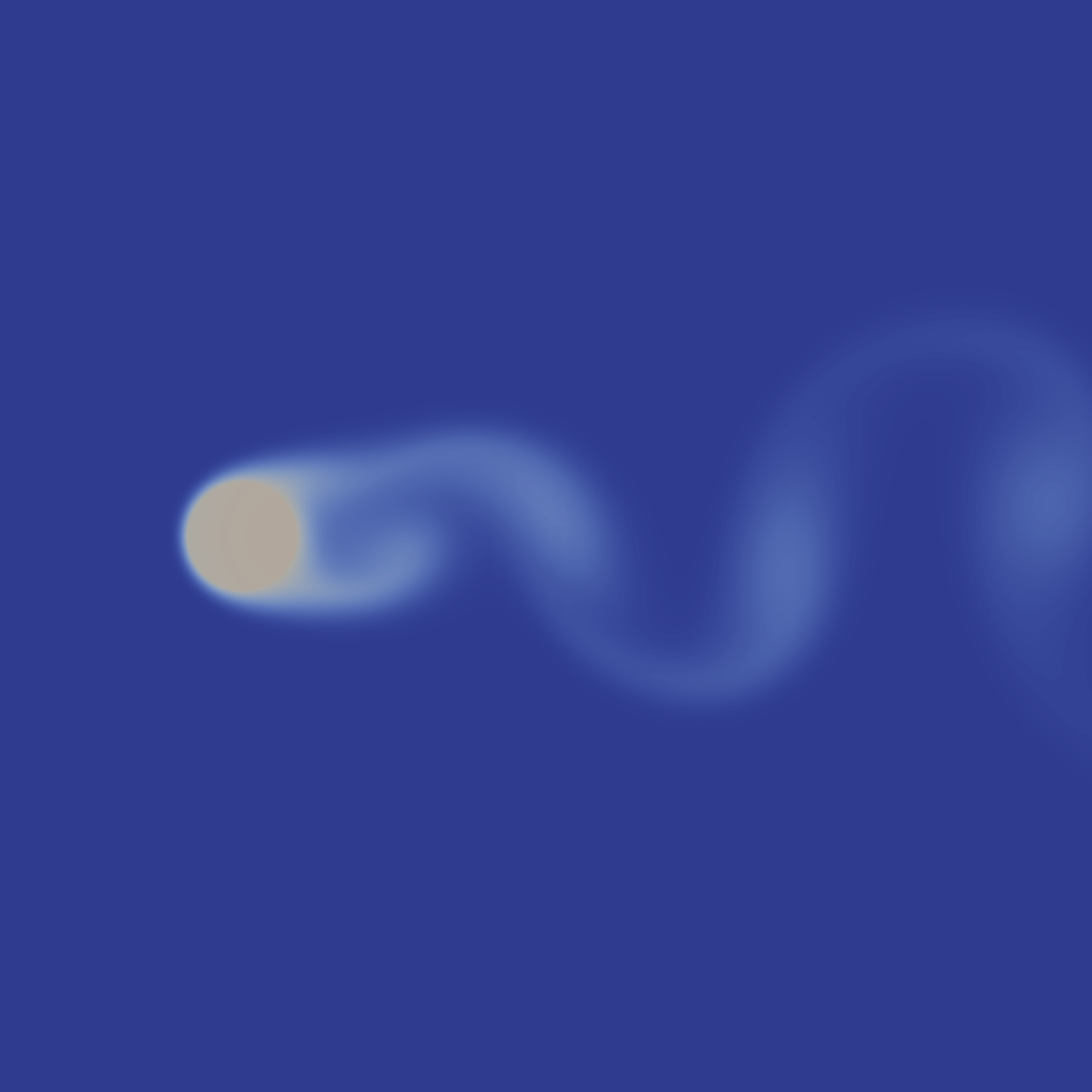}
    \caption{$\tnd=2.5$}
    \label{fig:disk_3}
\end{subfigure}
\hfill
\begin{subfigure}{0.07\textwidth}
    \centering
    \includegraphics[width=\linewidth]{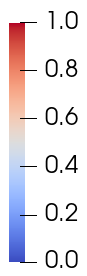}
    \label{fig:disk_colorbar}
\end{subfigure}
\caption{Temperature fields from the CHT solution at different time steps.}
\label{fig:disk_temp}
\end{figure}

The Nusselt numbers at LE, TE, SE are plotted over time in~\cref{fig:cht_nusselt_points}, together with the spatially averaged Nusselt number $\Nuavgcht$. We can make several observations:
\begin{itemize}
    \item \textbf{Temporal behavior:} The Nusselt number initially drops sharply for $\tnd \leq 0.1$, falling from essentially infinity to a finite value. This initial phase is dominated by pure conduction before boundary layers have developed and will be discussed in more detail in~\cref{subsubsec:time_stability}. For $0.1 \leq \tnd \leq 1.5$, thermal and velocity boundary layers develop, and a wake begins to form behind the solid. The Nusselt number stabilizes at around $\tnd \approx 0.3$. Beyond $\tnd \geq 1.5$, the wake becomes unstable and vortices start to shed, leading to an oscillatory Nusselt number, and a transition in the Nusselt number. However, despite the transition and the fluctuations, the magnitude of the spatially averaged Nusselt number, $\Nuavgcht$, remains nearly unchanged after $\tnd \approx 0.3$.

    \item \textbf{Spatial behavior:} The Nusselt number is highest at the leading edge (LE), where the flow first contacts the cylinder. The values at the side edge (SE) and trailing edge (TE) are significantly lower and exhibit fluctuations due to the alternating nature of vortex shedding.
\end{itemize}
\cref{fig:cht_nusselt_space} shows the Nusselt number as a function of the angular coordinate $\theta \in [-\pi, \pi)$, measured around the cylinder starting from the trailing edge (TE), at selected time steps. Initially, the Nusselt number is nearly uniform along the surface. As time progresses, spatial variations emerge due to boundary layer development. The temporally averaged Nusselt number $\Nutavgcht$ over the surface is also shown, highlighting the high spatial variation of the Nusselt number. The same data are shown as a polar plot in~\cref{fig:cht_Nu_polar}, which clearly highlights the spatial non-uniformity around the cylinder surface.
\begin{figure}[ht]
    \centering
    \begin{subfigure}[b]{0.48\textwidth}
        \centering
        \includegraphics[width=\textwidth]{figs_models/Nu_points_over_time_CHT.pdf}
        \caption{}
        \label{fig:cht_nusselt_points}
    \end{subfigure}
    \hfill
    \begin{subfigure}[b]{0.48\textwidth}
        \centering
        \includegraphics[width=\textwidth]{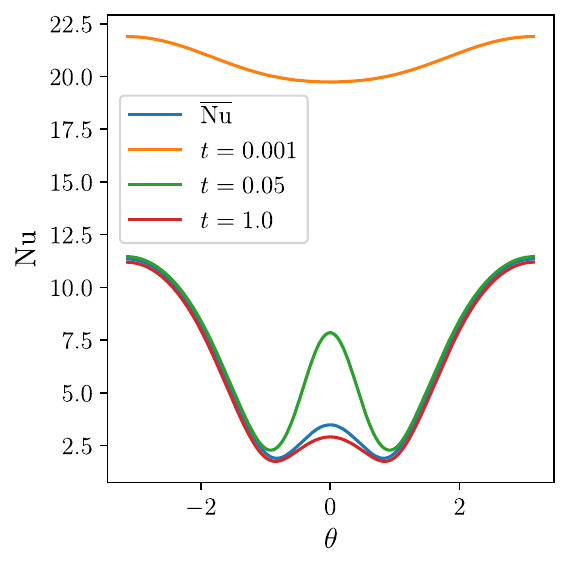}
        \caption{}
        \label{fig:cht_nusselt_space}
    \end{subfigure}
    \caption{(a) Nusselt number at LE, TE, SE and spatial average over time. (b) Nusselt number measured in angular coordinate $\theta$ at different time steps and temporal average. The superscripts ``CHT'' are omitted for brevity.}
    \label{fig:cht_Nu_uavg}
\end{figure}
\begin{figure}
    \centering
    \includegraphics[width=0.6\textwidth]{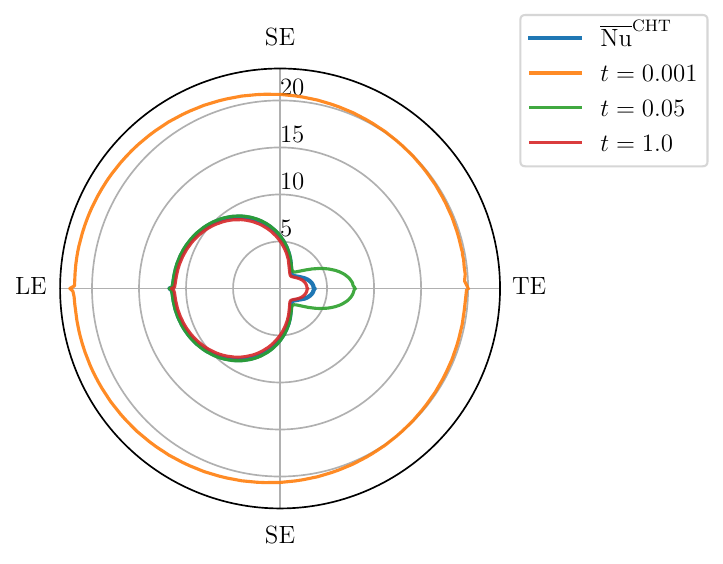}
    \caption{Nusselt number over time, and temporal average, presented as a polar plot.}
    \label{fig:cht_Nu_polar}
\end{figure}
\begin{figure}[ht]
    \centering
    \begin{subfigure}[b]{0.48\textwidth}
        \centering
        \includegraphics[width=\textwidth]{figs_models/usavg_time.pdf}
        \caption{Average temperature}
        \label{fig:disk_uavg}
    \end{subfigure}
    \hfill
    \begin{subfigure}[b]{0.48\textwidth}
        \centering
        \includegraphics[width=\textwidth]{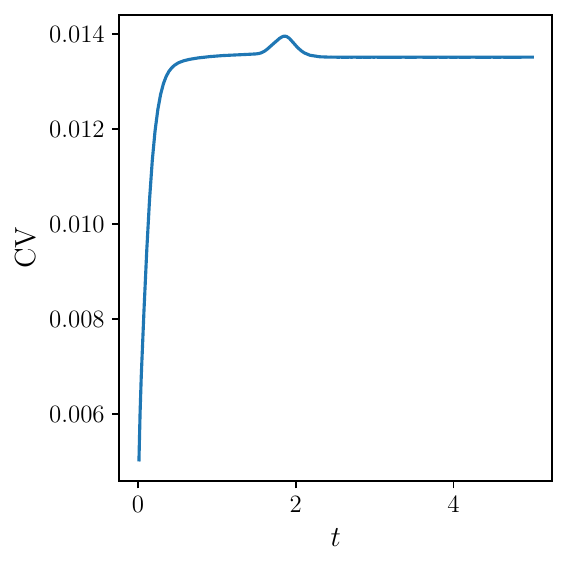}
        \caption{Coefficient of variation}
        \label{fig:disk_temp_uniformity}
    \end{subfigure}
    \caption{(a) Average temperature over time with CHT and exponential fit. (b) Coefficient of variation of the temperature over time.}
    \label{fig:disk_temperature}
\end{figure}

The average temperature, $\usavgcht(\tnd)$, is shown in~\cref{fig:disk_uavg}. Despite the unsteady and complex nature of the flow, the temperature decreases smoothly and monotonically over time. Its temporal evolution closely resembles an exponential decay. For interpretation, we can define an approximate solution using
\begin{align}
    \usavgexp(\tnd) \coloneqq \exp\left(-\frac{\tnd}{\taueq}\right), \label{eq:exp_fit}
\end{align}
where $\taueq$ represents the equilibration time constant, defined implicitly by $\usavgcht(\taueq) = \exp(-1)$. For the present case, we obtain $\taueq = 3.77$. The corresponding exponential approximation from~\cref{eq:exp_fit} is also shown in~\cref{fig:disk_uavg}, and it closely matches the CHT QoI.

During the time required for the Nusselt number to develop and stabilize (up to $\tnd \approx 0.3$), the average temperature remains nearly unchanged. This reveals a \emph{time scale separation} between the fluid and solid dynamics. The fluid is characterized by a convective time scale, defined in dimensional terms as $\tconvdim = \elldim / \vinfdim$, which under our non-dimensionalization (see~\cref{eq:tdiff}) becomes 
\begin{align}
    \tconvnondim \coloneqq \frac{\tconvdim}{\tdiffdim} = r_1 / (r_2 \Reynolds \Prandtl). \label{eq:tconv_nondim}
\end{align}
For the parameters considered, we obtain $\tconvnondim = 0.012$. In comparison, the equilibration time scale $\taueq$ is about 300 times larger, confirming a pronounced separation of fluid and solid time scales. More details on this will be provided in~\cref{subsec:lcm_time_scale}.

It is an interesting fact that the Nusselt number becomes stationary even though the temperature field itself is clearly time-dependent. The Nusselt number was defined in~\cref{eq:local_nusselt} as the heat flux at the interface normalized by the temperature difference between the solid surface and the far-field fluid. While the temperature field $\ufndcht$ continues to evolve over time (affecting the denominator), the normal heat flux $\delnd \ufndcht \cdot \bm{n}$ (the numerator) varies in a nearly proportional manner. As a result, the ratio remains approximately constant after the initial transient phase, leading to an effectively stationary Nusselt number.

\begin{figure}[t]
\centering
\begin{subfigure}{0.19\textwidth}
    \centering
    \includegraphics[width=\linewidth]{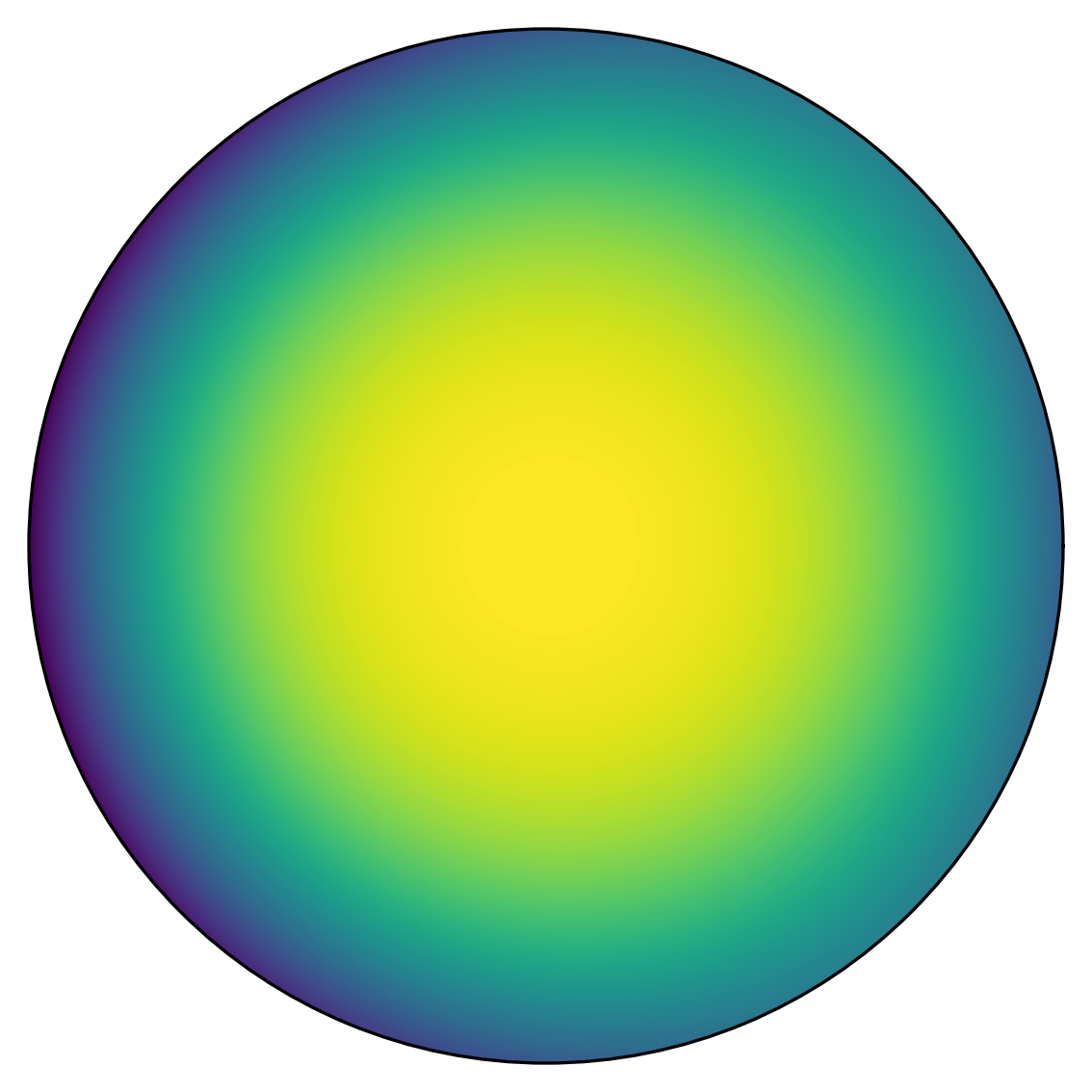}
    \caption{$\tnd=0.013$}
    \label{fig:disk_temp_0}
\end{subfigure}
\hfill
\begin{subfigure}{0.19\textwidth}
    \centering
    \includegraphics[width=\linewidth]{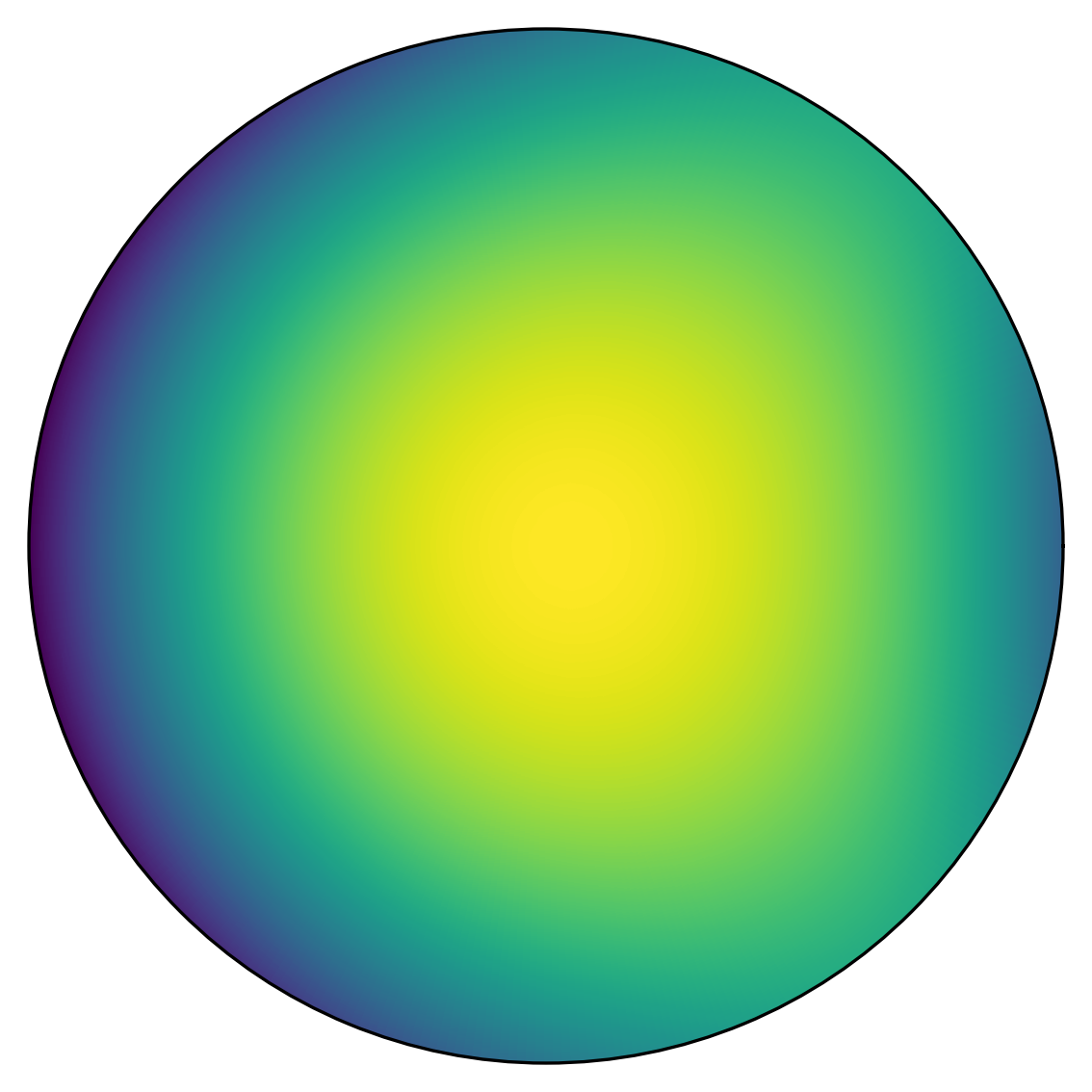}
    \caption{$\tnd=0.2$}
    \label{fig:disk_temp_1}
\end{subfigure}
\hfill
\begin{subfigure}{0.19\textwidth}
    \centering
    \includegraphics[width=\linewidth]{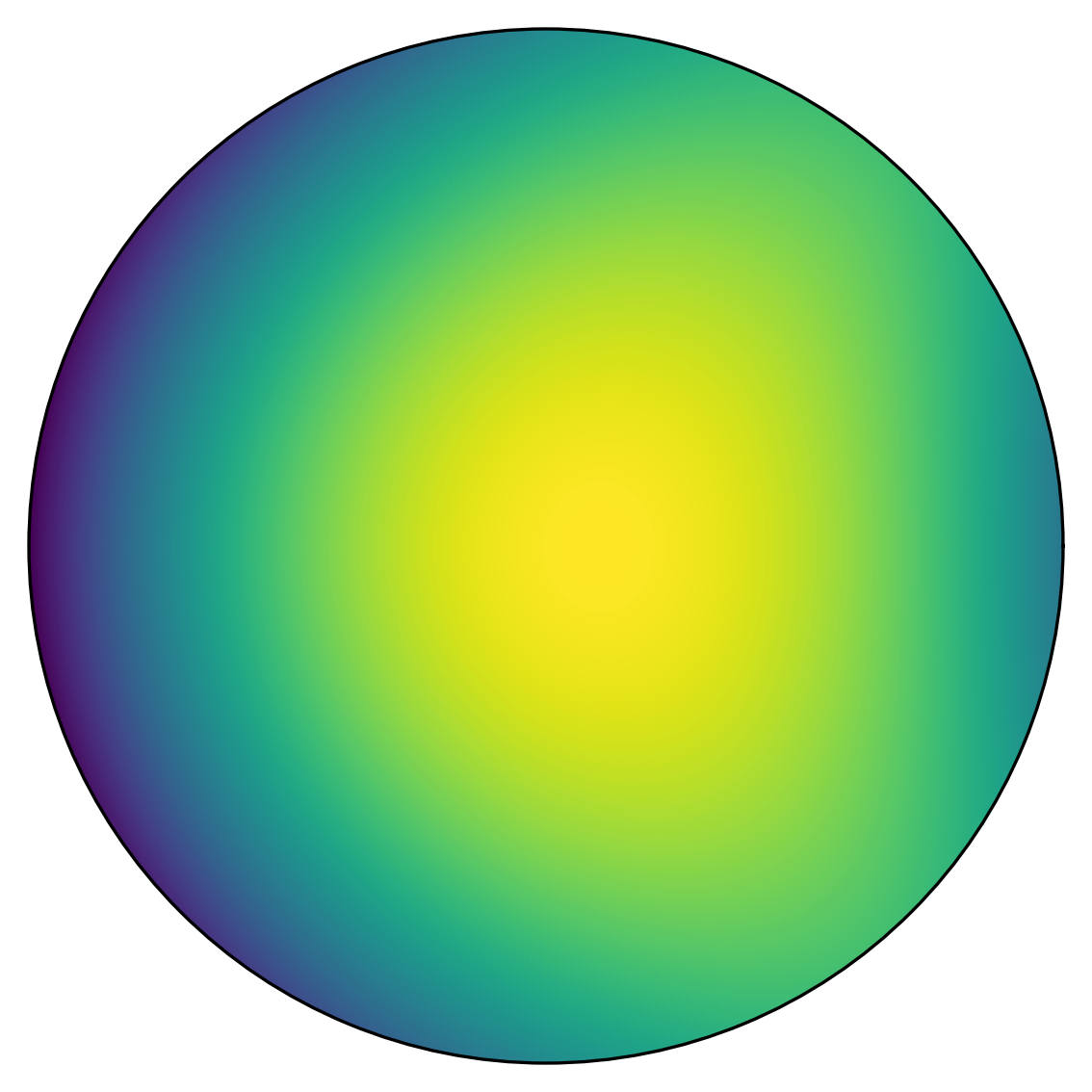}
    \caption{$\tnd=1.5$}
    \label{fig:disk_temp_2}
\end{subfigure}
\hfill
\begin{subfigure}{0.19\textwidth}
    \centering
    \includegraphics[width=\linewidth]{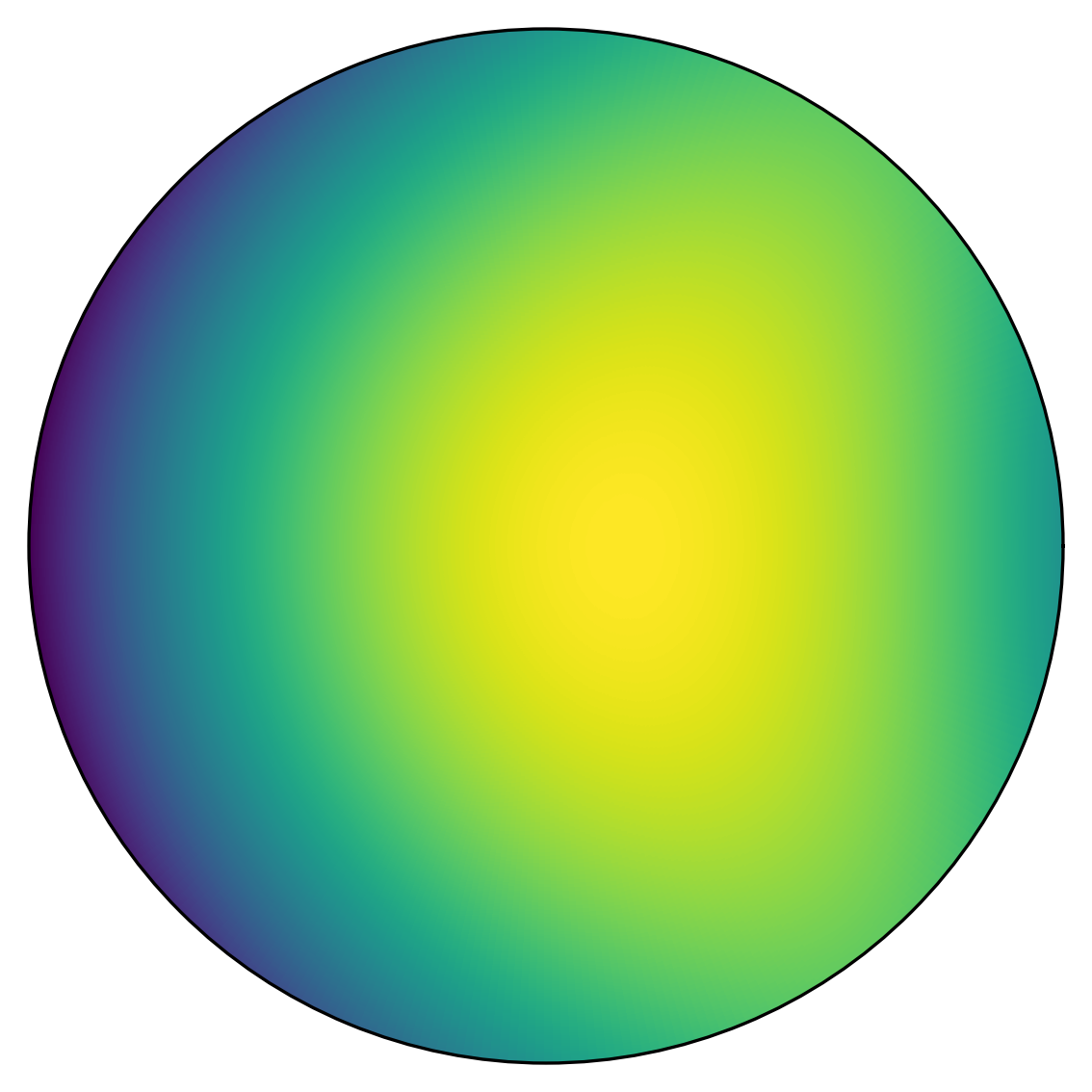}
    \caption{$\tnd=2.5$}
    \label{fig:disk_temp_3}
\end{subfigure}
\hfill
\begin{subfigure}{0.095\textwidth}
    \centering
    \includegraphics[width=\linewidth]{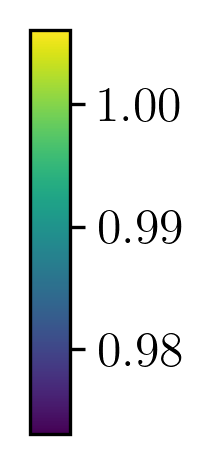}
    \label{fig:disk_temp_4}
\end{subfigure}
\caption{Temperature fields normalized by mean value at different time steps.}
\label{fig:disk_temp_solid}
\end{figure}

To study the spatial uniformity of the temperature field within the solid, we define the coefficient of variation (CV) as
\begin{align*}
    \text{CV}(\tnd) \coloneqq \frac{\sqrt{\displaystyle \dashint_{\Omegasnd} (\usndcht(\tnd,\xnd) - \usavgcht(\tnd))^2}}{\usavgcht(\tnd)}.
\end{align*}
Its value is plotted in~\cref{fig:disk_temp_uniformity}. We can see that the coefficient increases as the boundary layer develops, but remains relatively low for all times (below 0.015), indicating a fairly uniform temperature distribution within the solid. We can see a small peak around $\tnd=1.75$, which corresponds to the onset of vortex shedding. We also visualize the temperature field normalized by its mean value at a few time steps in~\cref{fig:disk_temp_solid}. We see that all values remain reasonably close to 1 (between 0.97 and 1.01), indicating only slight variations from the mean.

We can also quantify the uniformity of the temperature field by considering the ratio of the equilibration time scale to the diffusive time scale (equal to $1$ in our normalization). In this case, $\taueq / \tdiffnondim = 3.77$, indicating that the average temperature of the solid decays much more slowly than diffusion acts within it. As a result, internal gradients are rapidly smoothed out, and the temperature field remains essentially uniform.

Lastly, we evaluate the Biot number. We find $\Nustavgcht = 6.27$ and $\Bcht = r_2 \Nustavgcht = 0.0680$, which is reasonably small. As will be discussed in~\cref{subsec:rhea}, a \emph{small Biot number} implies a roughly uniform temperature distribution within the solid, which is consistent with our observations.

\begin{cmt}
    The periodic behavior of the Nusselt number observed in~\cref{fig:cht_nusselt_points} is a consequence of the highly idealized setup: a cross-flow over a two-dimensional circular cylinder. In three-dimensional configurations, the Nusselt number generally exhibits random and chaotic fluctuations for higher Reynolds numbers. Nevertheless, its spatially averaged value remains relatively stable over time. Three-dimensional illustrative examples are provided in~\cref{sec:nusselt}.
\end{cmt}

\subsection{\texorpdfstring{Autonomous Robin Heat Equation (RHE$_a$)}{Autonomous Robin Heat Equation (RHEa)}}\label{subsec:rhea}
Due to the high computational cost of solving the full CHT model and the extensive data requirements, particularly the spatial distributions of heterogeneous material properties, simplified formulations are often necessary. One such model is the \textit{autonomous Robin heat equation} (RHE$_a$), which assumes a time-independent heat transfer coefficient—or equivalently, a time-independent Nusselt number—hence the term “autonomous”. In this section, we first show that the CHT model can be equivalently recast as a Robin boundary value problem posed solely on the solid domain. This equivalence forms the foundation for both the subsequent model simplifications and the error analysis presented in~\cref{sec:error}. We then simplify the problem to obtain the RHE$_a$, derive its weak formulation, and revisit the numerical example from the previous section.

\subsubsection{Robin heat equation}\label{subsec:rhe}
To motivate a Robin heat equation for the solid, we begin by rearranging the boundary condition~\cref{eq:cht_nondim_5},
\begin{align}
    \kappa \delnd \usndcht \cdot \bm{n} - r_2 \delnd \ufndcht \cdot \bm{n} &= 0 \quad \text{on } \Gamma. \label{eq:cht_nondim_5_rearranged}
\end{align}
We focus on the left-hand side and rewrite it step by step:
\begin{align*}
    \kappa \delnd \usndcht \cdot \bm{n} - r_2 \delnd \ufndcht \cdot \bm{n} &= \kappa \delnd \usndcht \cdot \bm{n} - r_2 \delnd \ufndcht \cdot \bm{n} \frac{\usndcht}{\usndcht} \\
    &= \kappa \delnd \usndcht \cdot \bm{n} + r_2 \underbrace{\left(-\frac{\delnd \ufndcht \cdot \bm{n}}{\ufndcht}\right)}_{=\Nucht}\usndcht \\
    &= \kappa \delnd \usndcht \cdot \bm{n} + r_2 \Nucht \usndcht,
\end{align*}
where we used~\cref{eq:cht_nondim_4} and the definition of the local Nusselt number in~\cref{eq:local_nusselt}.

Using now~\cref{eq:cht_nusselt,eq:cht_biot,eq:cht_uniformity}, we can express the the result in terms of the Biot number and the variation function:
\begin{alignat*}{3}
    \kappa \delnd \usndcht \cdot \bm{n} + r_2 \Nucht \usndcht &= \kappa \delnd \usndcht \cdot \bm{n} + r_2 \Nucht \usndcht \frac{\Nustavgcht}{\Nustavgcht} \\
    &= \kappa \delnd \usndcht \cdot \bm{n} + \underbrace{r_2 \Nustavgcht}_{=\Bcht} \underbrace{\frac{\Nucht}{\Nustavgcht}}_{=\etacht(\tnd,\xnd)}\usndcht \\
    &= \kappa \delnd \usndcht \cdot \bm{n} + \Bcht\etacht(\tnd,\xnd)\usndcht.
\end{alignat*}
Combining this expression with~\cref{eq:cht_nondim_5_rearranged}, we finally obtain
\begin{align}
    \kappa \delnd \usndcht \cdot \bm{n} + \Bcht\etacht(\tnd,\xnd)\usndcht = 0 \quad \text{on } \Gamma. \label{eq:cht_robinbc}
\end{align}
\cref{eq:cht_robinbc} represents a Robin boundary condition on the interface $\Gamma$ with spatially and temporally varying Robin coefficient $\Bcht\etacht(\tnd,\xnd)$. Along with~\cref{eq:cht_nondim_2,eq:cht_nondim_7}, we can write the following boundary value problem for the solid domain only:
\begin{subequations}\label{eq:rhe_nondim_T_cht}
\begin{alignat}{3}
    \sigma \frac{\partial \usnd}{\partial \tnd} &= \delnd\cdot(\kappa \delnd \usnd) &\quad& \text{in } \Omegas, \label{eq:rhe_nondim_T_cht_1}\\
    \kappa \delnd \usnd \cdot \bm{n} + \Bcht\etacht(\tnd,\xnd)\usnd &= 0 &\quad& \text{on } \Gamma, \label{eq:rhe_nondim_T_cht_2}\\
    \usnd(\tnd=0, \cdot) &= 1 &\quad& \text{in } \Omegas.\label{eq:rhe_nondim_T_cht_3}
\end{alignat}
\end{subequations}
The solution to~\cref{eq:rhe_nondim_T_cht} is equivalent to the temperature solution of the CHT model restricted to the solid domain, i.e., $\usnd\equiv\usndcht$. The effect of the fluid is fully captured by the Biot number $\Bcht$ and the variation function $\etacht(\tnd,\xnd)$.

We can now pose the abstract non-dimensional \textit{Robin heat equation} (RHE) for the solid:
\begin{subequations}\label{eq:rhe_nondim_T}
\begin{alignat}{3}
    \sigma \frac{\partial u}{\partial \tnd} &= \delnd\cdot(\kappa \delnd u) &\quad& \text{in } \Omega, \label{eq:rhe_nondim_T_1}\\
    \kappa \delnd u \cdot \bm{n} + \Biot\eta(\tnd,\xnd)u &= 0 &\quad& \text{on } \pOmega, \label{eq:rhe_nondim_T_2}\\
    u(\tnd=0, \cdot) &= 1 &\quad& \text{in } \Omega.\label{eq:rhe_nondim_T_3}
\end{alignat}
\end{subequations}
As the formulation no longer involves a fluid domain, we do not use any subscripts, and we use $\pOmega$ to denote the fluid-solid interface. The effect of the fluid is modeled through the Biot number $\Biot$ and the variation function $\eta(\tnd,\xnd)$, which are now inputs to the RHE, and can be chosen arbitrarily. As before, we choose the weighted average temperature as QoI,
\begin{align}
    \uavg(\tnd) &\coloneqq \dashint_{\Omega} \sigma u(\tnd,\xnd). \label{eq:urhea_avg}
\end{align}
Comparing~\cref{eq:rhe_nondim_T} with~\cref{eq:rhe_nondim_T_cht}, we see that their solutions become identical for $\Biot = \Bcht$ and $\eta = \etacht$, i.e., $u \equiv \usnd\equiv\usndcht$. Consequently, the average temperatures are also identical, $ \uavg\equiv\usavgcht$.

Looking at the Robin boundary condition in~\cref{eq:rhe_nondim_T_2}, we can now understand the \emph{small-Biot} number limit: as $\Biot \to 0$, the Robin boundary condition reduces to a homogeneous Neumann condition ($\kappa\delnd u \cdot \bm{n} = 0$), corresponding to vanishing interfacial heat flux. Consequently, temperature gradients are small, and the field is nearly uniform.

\subsubsection{Time scale considerations and simplification}\label{subsec:time_homogenization}
In the example presented in~\cref{subsec:example_cht}, we observed a \emph{time scale separation} between fluid and solid: while the Nusselt number was essentially fully developed—in the sense that its temporal variations had become small—the solid temperature had not changed significantly. It was also observed that the Nusselt number was spatially non-uniform, see~\cref{fig:cht_Nu_polar}. This motivates the definition of the \emph{autonomous Robin heat equation} (RHE$_a$), which assumes a time-independent but spatially varying Robin coefficient:
\begin{subequations}\label{eq:rhea_nondim_T}
\begin{alignat}{3}
    \sigma \frac{\partial \utilde}{\partial \tnd} &= \delnd\cdot(\kappa \delnd \utilde) &\quad& \text{in } \Omega, \label{eq:rhea_nondim_T_1}\\
    \kappa \delnd \utilde \cdot \bm{n} + \Biot\etabar(\xnd)\utilde &= 0 &\quad& \text{on } \pOmega, \label{eq:rhea_nondim_T_2}\\
    \utilde(\tnd=0, \cdot) &= 1 &\quad& \text{in } \Omega,\label{eq:rhea_nondim_T_3}
\end{alignat}
\end{subequations}
where $\etabar(\xnd)$ is now a time-independent variation function, and the QoI is defined as
\begin{align*}
    \uavgtilde(\tnd) &\coloneqq \dashint_{\Omega} \sigma \utilde(\tnd,\xnd).
\end{align*}
As there is no longer an equivalence between CHT and RHE$_a$, we adorn the temperature field with a tilde to distinguish it from the RHE and CHT solutions, and to emphasize that it is an approximation to the true temperature field (even if $\Biot=\Bcht$).

If the true $\Bcht$ and $\etacht(\tnd,\xnd)$ (and, hence, $\Nucht$) were known, the natural choice would be to take $\Biot = \Bcht$ and $\etabar(\xnd)$ as the time average of the true variation function, i.e., 
\begin{align}
    \etabar(\xnd) = \etabarcht(\xnd) = \dashint_0^{\tf} \etacht(\tnd,\xnd), \quad \dashint_{\pOmega} \etabar(\xnd) = 1. \label{eq:time_avg_eta}
\end{align}
This way, a good approximation of the true Nusselt number, $\Nucht$, is typically achieved, i.e., $\Biot \etabar \approx \Bcht \etacht = r_2 \Nucht$. This is illustrated in~\cref{fig:cht_nusselt_points_tavg}, where we replot~\cref{fig:cht_nusselt_points} (cross-flow over a cylinder in 2D) with the temporal average at each point shown as dashed lines. Notice the stiff behavior during short times: in~\cref{subsubsec:time_stability}, we provide an analysis on its effect on the average temperature. The temporal averaging can be interpreted as a time homogenization procedure that averages over fast fluctuations. In Appendix~\ref{sec:time_homogenization_rhe}, we demonstrate for a model Robin problem---no initial transient, and a sinusoidal Robin coefficient in time and independent of space---that the optimal approximation is, not unexpectedly, given by the time average of the Robin coefficient.
\begin{figure}[ht]
    \centering
    \includegraphics[width=0.5\textwidth]{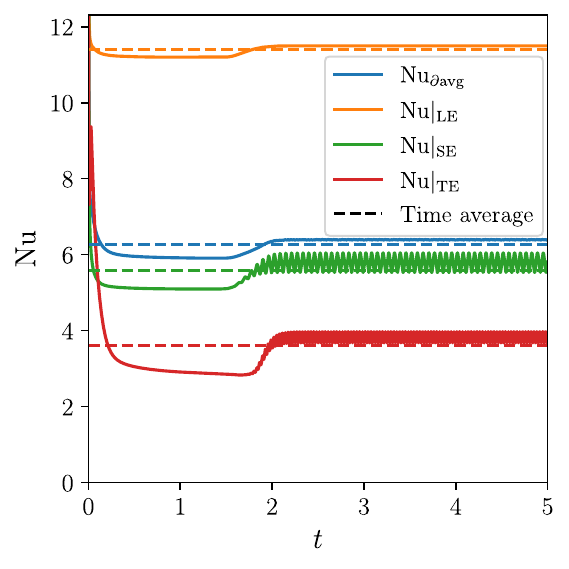}
    \caption{A replot of~\cref{fig:cht_nusselt_points}, but now with temporal averages in dashed lines.}
    \label{fig:cht_nusselt_points_tavg}
\end{figure}

The autonomous formulation offers several advantages, including a corresponding time-independent eigenvalue problem, presented and utilized for the derivation of error estimates in~\cref{subsubsec:cht_lcm_error}, and a more efficient numerical solution. Compared to the full RHE, the RHE$_a$ is more practical, as it requires only the temporally averaged Nusselt number $\Nutavg(\xnd)$ (from which $\Biot$ and $\etabar$ are derived; see~\cref{eq:cht_biot,eq:cht_uniformity,eq:time_avg_eta}), which is easier to estimate than $\Nusselt(\tnd,\xnd)$. However, even $\Nutavg(\xnd)$ remains challenging to determine in practice due to its spatial variation, and the thermophysical properties $\kappa$ and $\sigma$ are often unknown without destructive testing.

\subsubsection{Weak formulation}\label{subsec:rhea_weak}
We next formulate the weak form of the RHE$_a$, which serves both as the basis (i) for theoretical purposes, and for rigorous treatment of heterogeneous materials, and (ii) for the foundation for finite element spatial discretization.

The weak formulation of~\cref{eq:rhea_nondim_T} is: find $\utilde(\tnd) \in H^1(\Omega)$ such that
\begin{align*}
    \int_{\Omega} \sigma \, \frac{\partial \utilde}{\partial \tnd} \, v
    + \int_{\Omega} \kappa \, \delnd \utilde \cdot \delnd v
    + \Biot \int_{\pOmega} \etabar \, \utilde \, v
    = 0, \quad \forall v \in H^1(\Omega),
\end{align*}
subject to initial condition $\utilde(\tnd=0, \cdot) = 1$. It is assumed that $\tf, \Biot > 0$, and
\begin{alignat}{3}
    \sigma &\in L^\infty(\Omega), &&\quad \sigma(\xnd) \geq \sigma_{\min} > 0 \text{ for } \xnd \in \Omega, &&\quad \dashint_\Omega \sigma = 1; \label{eq:sigma_prop}\\
    \kappa &\in L^\infty(\Omega), &&\quad \essinf_{\xnd \in \Omega} \kappa(\xnd) = 1; \label{eq:kappa_prop}\\
    \etabar &\in L^r(\pOmega), \; r \geq 1, &&\quad \etabar(\xnd) \geq 0 \text{ for } \xnd \in \pOmega, &&\quad \dashint_{\pOmega} \etabar(\xnd) = 1. \label{eq:eta_prop}
\end{alignat}
For $w, v \in H^1(\Omega)$, we can introduce the following bilinear forms:
\begin{align}
    a_0(w, v; \kappa) &\coloneqq \int_{\Omega} \kappa \, \delnd w \cdot \delnd v, \label{eq:a0}\\
    a_1(w, v; \etabar) &\coloneqq \int_{\pOmega} \etabar \, w \, v, \label{eq:a1}\\
    a(w, v; \kappa, \Biot, \etabar) &\coloneqq a_0(w, v; \kappa) + \Biot a_1(w, v; \etabar), \label{eq:a}\\
    m(w, v; \sigma) &\coloneqq \int_{\Omega} \sigma \, w \, v. \label{eq:m}
\end{align}
The weak form can then be restated as: given the problem data $\{\Omega, \tf, \kappa, \sigma, \Biot, \etabar \}$, find $\utilde(\cdot,\tnd) \in H^1(\Omega)$ for $\tnd\in(0,\tf]$ such that
\begin{align}
    m\left(\frac{\partial \utilde}{\partial \tnd}, v; \sigma\right) + a(\utilde, v; \kappa, \Biot, \etabar) &= 0, \quad \forall v \in H^1(\Omega), \label{eq:rhe_weak}
\end{align}
subject to initial condition $\utilde(\tnd=0, \cdot) = 1$. The QoI can be expressed as
\begin{align}
    \uavgtilde(\tnd) &= m(\utilde(\cdot,\tnd), 1; \sigma). \label{eq:rhe_qoi}
\end{align}

\subsubsection{\texorpdfstring{Example: Cross-flow around a circular cylinder ($\mathrm{RHE}_a$)}{Example: Cross-flow around a circular cylinder (RHEa)}} \label{subsec:example_rhea}
We revisit the example introduced in~\cref{subsec:example_cht}, and solve it with RHE$_a$. The Biot number $\Biot = 0.068$ and the temporally averaged variation function $\etabar$ are extracted from the CHT solution, with $\etabar$ shown in~\cref{fig:disk_eta}. The simulation is performed using DOLFINx~\cite{baratta2023dolfinx}, with a BDF2 time-stepping scheme, a fixed timestep of $\Delta \tnd = 0.0025$, and a total of 2000 timesteps. The resulting average temperature obtained from the RHE$_a$ model is shown in~\cref{fig:disk_uavg_rhe}, alongside the CHT result. The two curves show excellent agreement, confirming the time scale separation.

\begin{figure}[ht]
    \centering
    \begin{subfigure}[b]{0.48\textwidth}
        \centering
        \includegraphics[width=\textwidth]{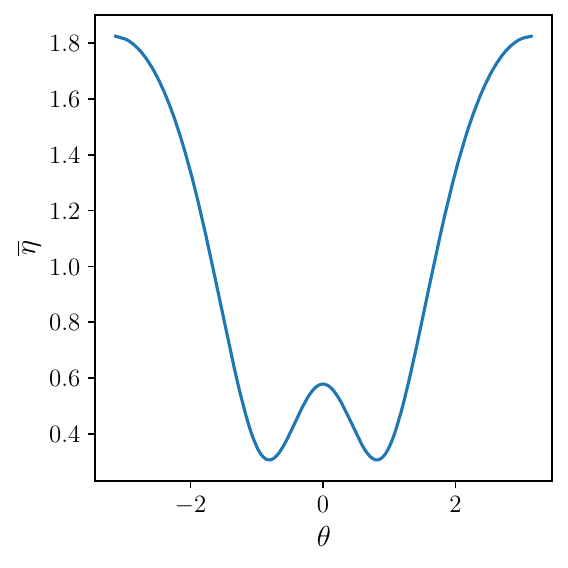}
        \caption{Time-averaged variation}
        \label{fig:disk_eta}
    \end{subfigure}
    \hfill
    \begin{subfigure}[b]{0.48\textwidth}
        \centering
        \includegraphics[width=\textwidth]{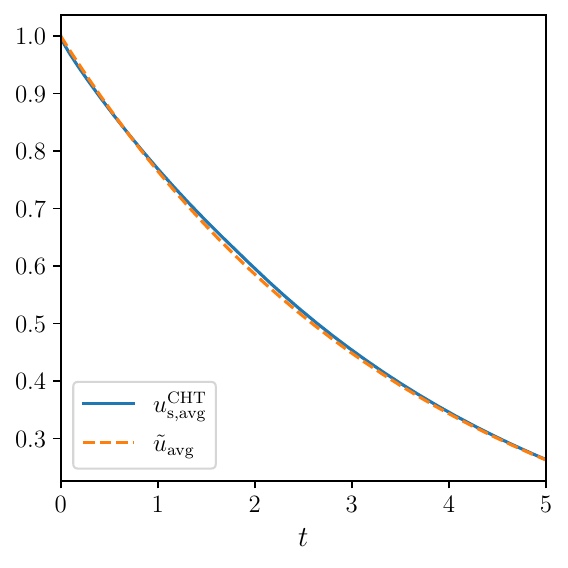}
        \caption{Average temperature}
        \label{fig:disk_uavg_rhe}
    \end{subfigure}
    \caption{(a) Time-averaged variation $\etabarcht$ computed with CHT. (b) Average temperature over time with CHT and RHE$_a$.}
    \label{fig:disk_rhe}
\end{figure}

\subsection{Isothermal Wall (ISO) Formulation}\label{subsec:iso}
To accurately approximate the CHT result with the RHE$_a$, an accurate estimate of the $\Biot$ and $\etabar$ is crucial. One way to obtain these quantities is through the isothermal wall (ISO) formulation. While the RHE$_a$ model represents a solid-side reduction of the CHT problem, the ISO formulation can be viewed as its fluid-side counterpart: it models the solid as an isothermal boundary and solves only the fluid equations. 

This approximation is valid only when there is a clear \emph{time scale separation} between fluid and solid dynamics, and the \emph{Biot number is small}. In the small-Biot limit, the Robin boundary condition in~\cref{eq:cht_robinbc} approaches a Neumann condition, implying that the temperature within the solid remains nearly uniform. This uniformity was empirically demonstrated in~\cref{subsec:example_cht}, particularly in~\cref{fig:disk_temp_uniformity}, where the coefficient of variation remained below 1.5\% throughout the simulation. In the same example, the solid temperature evolved much more slowly than the fluid field. This motivates the assumption of a uniform and quasi-stationary solid temperature adopted in the ISO formulation.

The ISO formulation remains computationally expensive, as it still requires solving the Navier-Stokes equations. However, ISO avoids the disparate time scales present in CHT. Moreover, ISO serves as the foundation for subsequent approximations of the Nusselt number—such as empirical correlations, discussed later in~\cref{sec:nusselt}—which are considerably less expensive to evaluate.

\subsubsection{Governing equations}
In the ISO formulation, the solid is no longer explicitly modeled but represented as an isothermal boundary condition for the fluid. The incompressible Navier-Stokes equations in~\cref{eq:cht_nondim_vp} remain unchanged and are used to compute the velocity and pressure fields in the fluid. The thermal problem, however, is simplified: an advection-diffusion equation is solved only in the fluid domain, subject to an isothermal boundary condition representing the solid. The governing equations for the temperature field are given in non-dimensional form below:
\begin{subequations}\label{eq:iso}
\begin{alignat}{3}
    r_1 \frac{\partial \ufiso}{\partial \tnd} + \Reynolds\Prandtl r_2\vndcht \cdot \delnd \ufiso &= r_2\Deltand \ufiso &\quad& \text{in } \Omegafnd, \label{eq:iso_1}\\
    \ufiso &= 0 &\quad& \text{on } \pOmegain, \label{eq:iso_2}\\
    \ufiso &= 1 &\quad& \text{on } \Gamma, \label{eq:iso_3}\\
    \ufiso(\tnd=0, \cdot) &= 0 &\quad& \text{in } \Omegafnd. \label{eq:iso_4}
\end{alignat}
\end{subequations}
We retain the same non-dimensionalization as in the CHT formulation to enable direct comparison. The resulting equations involve the parameters $\Reynolds$, $\Prandtl$, $r_1$, and $r_2$, while the solid-specific parameters $\kappa$ and $\sigma$ no longer appear. Since the solid domain is entirely absent in the ISO formulation, $r_1$ and $r_2$ act purely as scaling factors inherited from the original non-dimensionalization. If the equations were non-dimensionalized based on fluid properties, i.e., by redefining the diffusive time scale in~\cref{eq:tdiff} with $\kfdim$ and $\rhofdim\cfdim$, both $r_1$ and $r_2$ would equal one and disappear from the formulation. To make this clear, we define $\hat{t} = t\, r_2/r_1$ and $\hat{u}^{\mathrm{ISO}}_\mathrm{f}(\hat{t}, \cdot) = \ufiso(t, \cdot)$ and insert into~\cref{eq:iso_1} to obtain
\begin{align*}
    r_1 \frac{\partial \hat{u}^{\mathrm{ISO}}_\mathrm{f}}{\partial \hat{t}} \frac{\partial\hat{t}}{\partial t} + \Reynolds\Prandtl r_2\vndcht \cdot \delnd \hat{u}^{\mathrm{ISO}}_\mathrm{f} = r_2 \Deltand \hat{u}^{\mathrm{ISO}}_\mathrm{f} \quad \text{in } \Omegafnd.
\end{align*}
As $\partial \hat{t}/\partial t = r_2/r_1$, the $r_1$ cancels out and we can divide through by $r_2$ to obtain
\begin{align*}
    \frac{\partial \hat{u}^{\mathrm{ISO}}_\mathrm{f}}{\partial \hat{t}} + \Reynolds\Prandtl \vndcht \cdot \delnd \hat{u}^{\mathrm{ISO}}_\mathrm{f} = \Deltand \hat{u}^{\mathrm{ISO}}_\mathrm{f} \quad \text{in } \Omegafnd.
\end{align*}

\subsubsection{Nusselt number and Biot number}
Given a geometry, Reynolds and Prandtl numbers, the ISO problem in~\cref{eq:iso} can be solved to compute an approximation of the true Nusselt number, Biot number, and variation function. Analogously to the CHT formulation, these quantities can be defined on the interface $\Gamma$ within the ISO framework as follows:
\begin{align}
    \Nuiso(\tnd,\xnd) &\coloneqq -\frac{\delnd \ufiso \cdot \bm{n}}{\ufiso}, \label{eq:nusselt_iso}\\
    \Nustavgiso &\coloneqq \dashint_{\Gamma} \dashint_{0}^{\tf} \Nuiso(\tnd,\xnd), \label{eq:nusselt_iso_stavg}\\
    \Biotiso &\coloneqq r_2 \Nustavgiso, \label{eq:biot_iso}\\
    \etaiso(\tnd,\xnd) &\coloneqq \frac{\Nuiso(\tnd,\xnd)}{\Nustavgiso}, \label{eq:uniformity_iso}\\
    \etabariso(\xnd) &\coloneqq \dashint_0^{\tf} \etaiso(\tnd,\xnd). \label{eq:uniformity_iso_tavg}
\end{align}
As there is no solid domain in the ISO formulation, the Nusselt number (and its derived quantities) depend only on the Reynolds and Prandtl numbers, i.e.,
\begin{align*}
    \Nuiso(\tnd,\xnd) &= \Nuiso(\tnd,\xnd; \Reynolds, \Prandtl).
\end{align*}
The fluid-to-solid thermophysical ratios $r_1$ and $r_2$ and the solid properties $\kappa$ and $\sigma$ have no influence (the final time $\tf$ needs to be adapted accordingly with $r_1$ and $r_2$).
For small $r_1$ and $r_2$, the ISO formulation provides a satisfactory approximation of the CHT problem, i.e., $\Nuiso \approx \Nucht$, as has been empirically observed in many engineering applications. This will be examined in more detail in~\cref{subsubsec:iso_versus_cht}.

\subsubsection{Triangular system}
The ISO model provides a practical way to approximate the Nusselt number for use in the RHE$_a$ model, which can then be solved to estimate the solid temperature field. The procedure is as follows:
\begin{enumerate}
    \item Given a geometry, Reynolds and Prandtl numbers, solve the ISO problem in~\cref{eq:iso} for $\ufiso$.
    \item With $r_2$, compute $\Biotiso$ and $\etabariso$ using~\cref{eq:nusselt_iso,eq:nusselt_iso_stavg,eq:biot_iso,eq:uniformity_iso,eq:uniformity_iso_tavg}.
    \item Set $\Biot = \Biotiso$ and $\etabar(\xnd) = \etabariso(\xnd)$.
    \item With $\kappa$ and $\sigma$, solve the RHE$_a$ problem in~\cref{eq:rhe_weak} to obtain the solid temperature field $\utilde(\tnd,\xnd)$ and compute the average temperature $\uavgtilde(\tnd)$ with~\cref{eq:rhe_qoi}.
\end{enumerate}
We refer to this approach as the \emph{triangular system}, as it involves solving two decoupled problems in sequence: first ISO and then RHE$_a$. The triangular system is expected to be accurate when the assumptions of \emph{small-Biot} and \emph{time scale separation} hold.

\begin{cmt}
    Instead of solving RHE$_a$, one could also consider solving the full RHE with a time-dependent Robin coefficient $\Biotiso \etaiso(\tnd,\xnd)$. However, the gain in accuracy would be marginal, since time scale separation is assumed, implying $\etaiso(\tnd,\xnd) \approx \etabariso(\xnd)$. If the underlying CHT problem does not exhibit such separation, the triangular system cannot be expected to yield accurate results even with $\etaiso(\tnd,\xnd)$. Moreover, the computational cost would be significantly higher, as a time-dependent Robin coefficient leads to a varying left-hand side that must be reassembled at each timestep.
\end{cmt}

\subsubsection{Example: Cross-flow around a circular cylinder (ISO)}\label{subsec:example_iso}
We apply the triangular system to the cross-flow around a circular cylinder example introduced in~\cref{subsec:example_cht}. The ISO model is solved using Nek5000~\cite{nek5000-web-page} with the same parameters and numerical settings as for the CHT case in~\cref{subsec:example_cht}. The spatially averaged Nusselt number obtained with ISO is shown together with the CHT result in~\cref{fig:iso_nusselt_points}. In addition, $\etabarcht$ is plotted as a function of the angular coordinate $\theta$ (defined in~\cref{subsec:example_cht}) in~\cref{fig:iso_nusselt_space}, alongside the CHT reference. Both plots demonstrate that, although the ISO model exhibits slight deviations, it provides an excellent approximation of the true Nusselt number for this case. For the Nusselt number, we obtain $\Nustavgiso = 6.25$, in close agreement with the CHT value $\Nustavgcht = 6.27$. For the Biot number, we find $\Biotiso = 0.0678$, which is likewise very close to the CHT value $\Bcht = 0.0680$.
Subsequently, the values of $\Biotiso$ and $\etabariso$ are used as inputs to the RHE$_a$ model, which is solved with DOLFINx~\cite{baratta2023dolfinx} under the same settings as in~\cref{subsec:example_rhea}. The resulting average temperature, shown in~\cref{fig:disk_uavg_iso} alongside the CHT reference, exhibits excellent agreement between the two curves.
\begin{figure}[ht]
    \centering
    \begin{subfigure}[b]{0.48\textwidth}
        \centering
        \includegraphics[width=\textwidth]{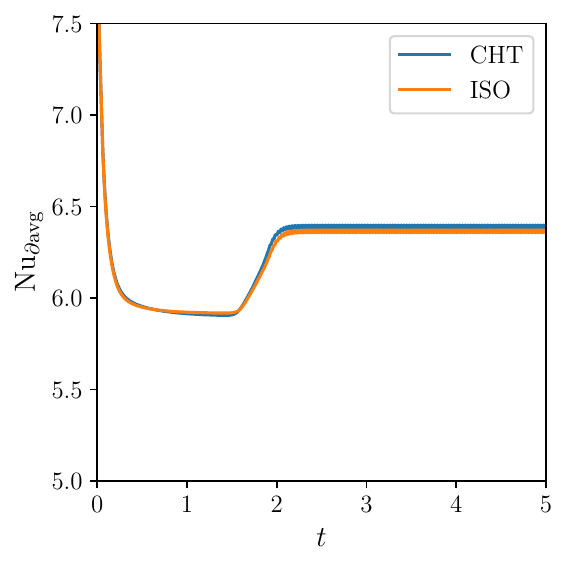}
        \caption{}
        \label{fig:iso_nusselt_points}
    \end{subfigure}
    \hfill
    \begin{subfigure}[b]{0.48\textwidth}
        \centering
        \includegraphics[width=\textwidth]{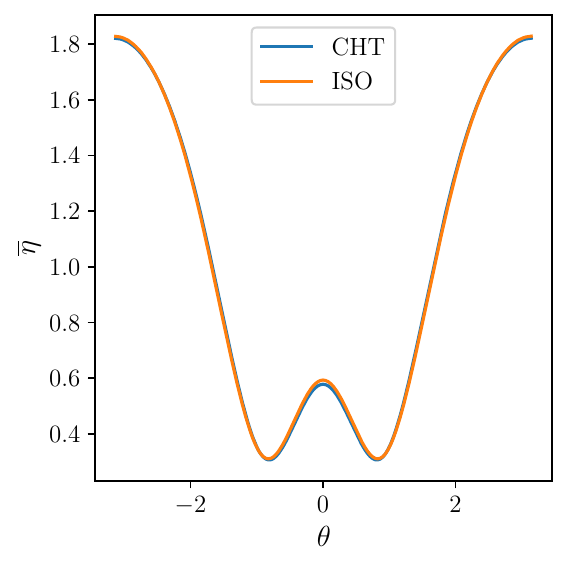}
        \caption{}
        \label{fig:iso_nusselt_space}
    \end{subfigure}
    \caption{(a) Spatial averaged Nusselt number computed with CHT and ISO over time. (b) Temporally averaged variation measured in angular coordinate $\theta$ for CHT and ISO.}
    \label{fig:iso_Nu_uavg}
\end{figure}
\begin{figure}[ht]
    \centering
    \includegraphics[width=0.5\textwidth]{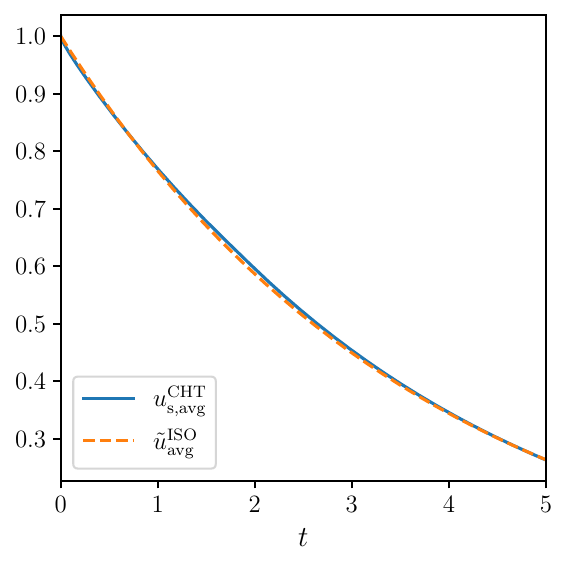}
    \caption{Average temperature computed with CHT and triangular system (ISO + RHE$_a$).}
    \label{fig:disk_uavg_iso}
\end{figure}

\subsection{Summary}
In~\cref{sec:models}, several models were presented for the dunking problem. Below we summarize each of them in terms of their inputs and underlying assumptions.

\paragraph{Conjugate heat transfer (CHT) model:}
The CHT model serves as the reference, or truth, model for the dunking problem. In its non-dimensional form, it involves six input parameters: the Reynolds number $\Reynolds$, the Prandtl number $\Prandtl$, the fluid-to-solid volumetric heat capacity ratio $r_1$, the fluid-to-solid thermal conductivity ratio $r_2$, and the solid thermophysical properties $\kappa$ and $\sigma$. While it is the most complete and accurate formulation, it is computationally expensive and requires detailed knowledge of the thermophysical properties of the solid, which may not be readily available.

\paragraph{Autonomous Robin heat equation (RHE$_a$):}
The RHE$_a$ model is a simplification of the CHT model, based on the assumption of \emph{time scale separation} between fluid and solid dynamics. It requires four input parameters: the solid thermophysical properties $\kappa$ and $\sigma$, the fluid-to-solid thermal conductivity ratio $r_2$, and a time-independent but spatially varying Nusselt number $\Nutavg$ that encapsulates the fluid effects. The RHE$_a$ solution is much less expensive to compute than the CHT. However, like the CHT model, it still requires knowledge of $\kappa$ and $\sigma$, and in addition the time-independent Nusselt number must be provided.

\paragraph{Isothermal wall (ISO) formulation:}
The ISO model is another simplification of the CHT model that requires only two input parameters: the Reynolds number and the Prandtl number. The formulation rests on two assumptions: a \emph{small-Biot} number, implying a nearly spatially uniform temperature field within the solid, and \emph{time scale separation} between fluid and solid. For problems satisfying these assumptions, the ISO model offers a more efficient way to estimate the Nusselt number required by the RHE$_a$ model. However, as the formulation does not take into account the solid properties $\kappa$ and $\sigma$ nor the thermophysical property ratios $r_1$ and $r_2$, it is unclear for which property ranges the  ISO-predicted Nusselt number accurately approximates the CHT value; this question is examined in~\cref{subsubsec:iso_versus_cht}. Owing to the relatively high computational cost of ISO simulations, practitioners often rely on \emph{empirical correlations}—simplified approximations of the ISO-derived Nusselt number—which are discussed in detail in~\cref{sec:nusselt}.

\section{Lumped Capacitance Model (LCM)}\label{sec:lcm}
The lumped capacitance model (LCM) is a simple and widely used approach for estimating the thermal response of a solid body immersed in a fluid. It can be derived from the RHE$_a$ model by assuming, in addition to time scale separation, an exactly spatially uniform temperature field. It is therefore only valid for problems in the small Biot limit. In this section, we derive the non-dimensional form of the LCM, present its dimensional counterpart, and revisit the example of cross-flow around a circular cylinder.

\subsection{Non-Dimensional Form}
To enforce the spatial uniformity of the temperature within the weak form in~\cref{eq:rhe_weak}, we can select a test and trial space that contains only the constant value, i.e., $X^{\text{L}}\coloneqq \operatorname{span}\{1_\Omega\}$. The problem becomes: find $\uLump(\tnd) \in X^{\text{L}}$ such that
\begin{align*}
    m\left(\frac{\partial \uLump}{\partial \tnd}, v; \sigma\right) + a(\uLump, v; \kappa, \Biot, \etabar) &= 0, \quad \forall v \in X^{\text{L}}
\end{align*}
with initial condition $\uLump(0) = 1$, and the bilinear forms were defined in~\cref{eq:m,eq:a,eq:a0,eq:a1}. As $\uLump(\tnd)$ and $v$ do not depend on space, they can be pulled out of the space integrals, and we can directly integrate to obtain
\begin{align*}
    \frac{d\uLump}{d \tnd}\int_\Omega \sigma + \Biot \uLump \int_{\pOmega} \etabar &= 0.
\end{align*}
With the properties of $\sigma$ and $\etabar$ in~\cref{eq:sigma_prop,eq:eta_prop}, we find
\begin{align*}
    \frac{d\uLump}{d \tnd}|\Omega| + \Biot |\pOmega| \uLump &= 0.
\end{align*}
We can then write the analytical solution of the ODE as
\begin{align}
    \uLump(\tnd) = \exp\left(-\Biot \gamma \tnd\right), \label{eq:lumped_capacitance_solution_intermediate}
\end{align}
where the geometric factor $\gamma$ is defined as
\begin{align}
    \gamma \coloneqq \frac{|\pOmega|}{|\Omega|}. \label{eq:gamma}
\end{align}
We can also represent $\gamma$ in terms of dimensional quantities as
\begin{align}
    \gamma = \frac{|\pOmegadim| \elldim}{|\Omegadim|} = \frac{\elldim}{\Elldim}. \label{eq:gamma_dim}
\end{align}
where $\Elldim$ and $\elldim$ are the intrinsic and extrinsic length scale, respectively. We can rewrite~\cref{eq:lumped_capacitance_solution_intermediate} as
\begin{align}
    \uLump(\tnd) = \exp\left(-\frac{\tnd}{\tauL}\right), \label{eq:lumped_capacitance_solution}
\end{align}
where
\begin{align}
    \tauL = \frac{1}{\Biot \gamma} \label{eq:teq_nondim}
\end{align}
is the nondimensional equilibration time constant. It is clear from \cref{eq:lumped_capacitance_solution,eq:teq_nondim} that the LCM solution only requires knowledge of the body volume $|\Omega|$, the surface area $|\pOmega|$, and the Biot number $\Biot$. As compared to the RHE$_a$ problem, the LCM does not require knowledge of the time-independent variation $\etabar(\xnd)$ nor the detailed geometry of the solid domain $\Omega$. Moreover, material inhomogeneities in the solid, $\kappa$ and $\sigma$, do not have any effect on the solution.

\subsection{Dimensional Form}
The LCM solution in~\cref{eq:lumped_capacitance_solution} can be expressed in dimensional form as
\begin{align}
    \frac{\TLump(\tdim) - \Tinfdim}{\Tidim - \Tinfdim} &= \exp\left(-\frac{\tdim}{\tauLdim}\right), \label{eq:uLCM_dim}
\end{align}
where 
\begin{align}
    \tauLdim \coloneqq \frac{|\Omegadim|}{|\pOmegadim|}\frac{\rhocavgdim}{\hstavg} \label{eq:teq_dim}
\end{align}
is the dimensional equilibration time constant. The dimensional solution in~\cref{eq:uLCM_dim,eq:teq_dim} requires only the volume $|\Omegadim|$, surface area $|\pOmegadim|$, the volume-averaged heat capacity $\rhocavgdim$, the averaged HTC $\hstavg$, and the far-field and initial solid temperatures, $\Tinfdim$ and $\Tidim$. The volume, surface area and volume-averaged heat capacity, far-field and initial solid temperatures are typically known, while the heat transfer coefficient can be estimated, as discussed in~\cref{sec:nusselt}. 

\subsection{Example Revisited: Cross-Flow Around a Circular Cylinder (LCM)}\label{subsec:example_lcm}
We revisit the cross-flow around a circular cylinder considered in~\cref{subsec:example_cht}. The problem was non-dimensionalized using $\dm{D}$ as the characteristic length. The geometric factor $\gamma$ can be computed with~\cref{eq:gamma_dim},
\begin{align}\label{eq:gamma_cylD}
    \gamma &= \frac{|\pOmega|}{|\Omega|} 
           = \frac{|\pOmegadim| \elldim}{|\Omegadim|} 
           = \frac{\pi \dm{D}^2}{\tfrac{1}{4}\pi \dm{D}^2} 
           = 4.
\end{align}
Using the Biot number obtained with the ISO model in~\cref{subsec:example_iso}, $\Biotiso = 0.0678$, we find a time constant
\begin{align*}
    \tauL = \frac{1}{\Biotiso \gamma} = 3.69,
\end{align*}
which is in very close agreement with the exponential fit of the CHT solution in~\cref{eq:exp_fit}, where $\taueq = 3.77$ was found.

It is remarkable that the LCM, which requires only the volume, surface area, and an estimated Biot number (here obtained via ISO), can approximate the full CHT solution—a coupled system of PDEs with six inputs—so closely. Its accuracy rests on two key assumptions: approximate spatial uniformity of the solid temperature, valid when the Biot number is small (here $\Biotiso = 0.0678$), and separation of fluid and solid time scales.

However, these assumptions are often only qualitatively stated in engineering practice. For many fluid-solid combinations, it is unclear whether the Biot number is sufficiently small or whether a true time scale separation exists. In the next section, we examine these assumptions in detail and make them more precise through mathematical analysis and numerical simulations.

\section{Model Error Estimation}\label{sec:error}
We are interested in developing an error bound for the CHT-to-LCM error, given by 
\begin{align}
    |\usavgcht(\tnd;\kappa,\sigma,r_1,r_2,\Reynolds,\Prandtl) - \uLump(\tnd;\Biapprox)|. \label{eq:error_triangular}
\end{align}
Here, $\usavgcht$ denotes the CHT QoI defined in~\cref{eq:uavg}, and $\uLump$ the LCM QoI defined in~\cref{eq:lumped_capacitance_solution_intermediate}. The CHT model depends on six dimensionless input parameters: the nondimensional solid properties $\kappa$ and $\sigma$ from~\cref{eq:kappa_sigma_1}; the fluid-to-solid property ratios $r_1$ and $r_2$ from~\cref{eq:r1_r2}; and the Reynolds and Prandtl numbers $\Reynolds$ and $\Prandtl$ from~\cref{eq:Re_Pr}. The LCM model depends only on the Biot number. Since the true Biot number $\Bcht$ is typically unknown in practice, we introduce an approximate value $\Biapprox$, representing an estimate of $\Bcht$ (e.g., obtained via the ISO model or empirical correlations).

In~\cref{subsec:rhea}, it was shown that the CHT problem can be equivalently reformulated as a RHE problem, and that the corresponding QoIs are identical:
\begin{align*}
    \uavg(\tnd;\kappa,\sigma,\Bcht,\etacht(\tnd,\xnd)) \equiv \usavgcht(\tnd;\kappa,\sigma,r_1,r_2,\Reynolds,\Prandtl),
\end{align*}
where $\Bcht$ and $\etacht(\tnd,\xnd)$ denote the truth Biot number and variation function defined in~\cref{eq:cht_biot,eq:cht_uniformity}. Accordingly, the error in~\cref{eq:error_triangular} can be rewritten as
\begin{align*}
    |\uavg(\tnd;\kappa,\sigma,\Bcht,\etacht(\tnd,\xnd)) - \uLump(\tnd;\Biapprox)|.
\end{align*}
For brevity, we omit the superscripts “CHT” for the remainder of this section and simply write $\Biot \equiv \Bcht$ and $\eta \equiv \etacht$. With this convention, we now consider the maximum-in-time error and apply the triangle inequality to obtain
\begin{align}
    \max_{\tnd\in[0,\tf]}|\uavg(\tnd;\kappa,\sigma,\Biot,\eta(\tnd,\xnd)) - \uLump(\tnd;\Biapprox)| &\leq \max_{\tnd\in[0,\tf]}|\uavg(\tnd;\kappa,\sigma,\Biot,\eta(\tnd,\xnd)) - \uavgtilde(\tnd;\kappa,\sigma,\Biot,\etabar(\xnd))| \label{eq:error_triangular_time} \\
    &+ \max_{\tnd\in[0,\tf]}|\uavgtilde(\tnd;\kappa,\sigma,\Biot,\etabar(\xnd)) - \uLump(\tnd;\Biot)| \label{eq:error_triangular_lcm} \\
    &+ \max_{\tnd\in[0,\tf]}|\uLump(\tnd;\Biot) - \uLump(\tnd;\Biapprox)|. \label{eq:error_triangular_biot}
\end{align}
Here, the three terms correspond to:
\begin{itemize}
    \item (\ref{eq:error_triangular_time}): the error due to \emph{temporal approximation} (replacing $\eta(\tnd,\xnd)$ by $\etabar(\xnd)$),
    \item (\ref{eq:error_triangular_lcm}): the error due to \emph{lumping} (replacing RHE$_a$ by LCM), and
    \item (\ref{eq:error_triangular_biot}): the error due to \emph{Biot approximation} (replacing $\Biot$ by an estimate $\Biapprox$).
\end{itemize}
We first review the current engineering practice for assessing the applicability of the LCM. We then, one by one, discuss strategies to quantify and bound the errors in~\cref{eq:error_triangular_time,eq:error_triangular_lcm,eq:error_triangular_biot}, and provide supporting numerical results.

For some of the proofs, we will use the following estimate for exponential functions: For all $z \geq 0$ and $0 \leq \epsilon < 1$, the inequality
\begin{align}
    \big|\exp(-z(1-\epsilon)) - \exp(-z)\big| \leq \frac{|\epsilon|}{\exp(1)} + \mathcal{O}(\epsilon^2) \label{eq:exp_inequality}
\end{align}
holds. A proof of this elementary estimate is provided for instance in~\cite[Appendix C.17]{kaneko2024error}.

\subsection{Current Engineering Practice}\label{subsec:current_practice}
Because of its simplicity and minimal input requirements, the LCM is widely used in engineering practice as a quick estimate. This naturally raises the question: How accurate is the LCM, and under what conditions is it applicable? Standard engineering textbooks such as~\cite{ahtt6e,incropera1990fundamentals} typically propose the following criterion:
\begin{align}
    \BiDunk \coloneqq \frac{\hstavg \Elldim}{\ksinfdim} = \frac{\hstavg\elldim}{\ksinfdim} \frac{\Elldim}{\elldim} = \frac{B}{\gamma} \leq 0.1, \label{eq:biot_small}
\end{align}
where $\BiDunk$ denotes the Biot number defined with respect to the intrinsic length scale $\Elldim$. The heuristic condition in~\cref{eq:biot_small} suggests that the LCM will be a reliable approximation when the Biot number is small. Although supported by empirical evidence in many cases, it lacks rigorous theoretical justification, and indeed, in some cases,~\cref{eq:biot_small} can be very misleading. We intend to address this gap in the following sections.

\subsection{Error Due to Temporal Approximation}\label{subsec:lcm_time_scale}
In this section, we analyze the temporal approximation error in~\cref{eq:error_triangular_time},
$$ \max_{\tnd\in[0,\tf]}|\uavg(\tnd;\kappa,\sigma,\Biot,\eta(\tnd,\xnd)) - \uavgtilde(\tnd;\kappa,\sigma,\Biot,\etabar(\xnd))|,$$
arising from replacing the time-dependent $\eta(\tnd,\xnd)$ with a time-independent $\etabar(\xnd)$. We first examine the short-time stiff behavior of the RHE and study its effect on the averaged temperature, before discussing the separation of fluid and solid time scales. Finally, we present numerical results supporting the theoretical findings.

\subsubsection{Initial transient and stability}\label{subsubsec:time_stability}
Looking back at~\cref{fig:cht_nusselt_points}, we observe that the Nusselt number exhibits a very stiff behavior in short-time: it essentially starts at infinity and decays rapidly to a finite value. This raises the question of whether such transient behavior can contribute large errors. We first present a proposition that characterizes the initial stiff behavior using analytical solutions. We then establish uniform bounds for the weighted averaged temperature of the RHE and derive an energy bound that suggests stability.

\begin{propo}\label{propo:short_time_stiff}
    For very short times
    \[
        \tnd \ll C_x\tconvnondim,
    \]
    where $C_x$ depends on the distance from the leading edge as well as on the local regularity of the domain (e.g., curvature and corners) and the convective time was defined in~\cref{eq:tconv_nondim}, the following short-time asymptotic behavior is expected:
    \begin{align}
        \left.\usndcht\right|_{\Gamma} &\sim \frac{1}{1 + \sqrt{r_1 r_2}}, \label{eq:short_time_stiff}\\
        \Nusselt(\tnd,\xnd) &\sim \frac{1}{\sqrt{\pi \tnd}} \, \sqrt{\frac{r_1}{r_2}}. \label{eq:short_time_nusselt}
    \end{align}
\end{propo}

\begin{proof}
    The result is obtained by treating both the solid and the fluid as semi-infinite solids in perfect thermal contact and solving the corresponding one-dimensional transient conduction problem. For simplicity, the solid is assumed homogeneous, i.e., $\kappa = 1$ and $\sigma = 1$. The analysis is valid only for sufficiently short times, before convective boundary layers have developed, i.e., $\tnd \ll \tconvnondim$.

    We first analyze the fluid side. Consider transient conduction into a semi-infinite medium occupying $\psi > 0$, where $\psi$ denotes the coordinate normal to the interface $\Gamma$. At $\psi = 0$, a surface temperature $u_\Gamma$ is suddenly imposed at $t = 0$. The governing equation for the fluid (in non-dimensional form; compare with~\cref{eq:cht_nondim_1} for $\vndcht = 0$) is
    \begin{subequations}\label{eq:semi_inf_fluid}
    \begin{align}
        r_1 \frac{\partial \ufndcht}{\partial \tnd}
        &= r_2 \frac{\partial^2 \ufndcht}{\partial \psi^2}
        \quad \text{in } \Omegafnd, \label{eq:semi_inf_fluid_1}
    \end{align}
    \end{subequations}
    with initial and boundary conditions
    \begin{align}
        \ufndcht(0,\psi) = 0, \qquad
        \ufndcht(\tnd,0) = u_\Gamma, \qquad
        \ufndcht(\tnd,\psi \to \infty) = 0.
    \end{align}
    Introducing the similarity variable
    $\xi = \psi / (2\sqrt{(r_2/r_1)\tnd})$
    and setting $\ufndcht(\tnd,\psi) = f(\xi)$ reduces the PDE to an ODE, whose solution---with the boundary conditions applied---is
    \begin{align}
        \ufndcht(\tnd,\psi)
        = u_\Gamma \!\left[ 1 - \erf\!\left(\frac{\psi}{2\sqrt{(r_2/r_1)\tnd}}\right) \right],
    \end{align}
    where $\erf$ denotes the error function.
    The corresponding non-dimensional heat flux at the surface $\psi = 0$ is
    \begin{align}
        q_f(\tnd,0)
        = -r_2 \left. \frac{\partial \ufndcht}{\partial \psi} \right|_{\psi=0}
        = u_\Gamma \sqrt{\frac{r_1 r_2}{\pi \tnd}}.
    \end{align}
    Hence, using~\cref{eq:local_nusselt}, the local Nusselt number becomes
    \begin{align}
        \Nusselt(\tnd,\xnd)
        = \frac{\left. \frac{\partial \ufndcht}{\partial \psi} \right|_{\psi=0}}{u_\Gamma}
        = \frac{1}{\sqrt{\pi \tnd}} \sqrt{\frac{r_1}{r_2}},
    \end{align}
    which proves~\cref{eq:short_time_nusselt}.

    For the solid, the conduction problem (compare with~\cref{eq:cht_nondim_2} with $\kappa=\sigma=1$) reads
    \begin{subequations}\label{eq:semi_inf_solid}
    \begin{align}
        \frac{\partial \usndcht}{\partial \tnd}
        &= \frac{\partial^2 \usndcht}{\partial \psi^2}
        \quad \text{in } \Omegasnd, \label{eq:semi_inf_solid_1}
    \end{align}
    \end{subequations}
    with initial and boundary conditions
    \begin{align}
        \usndcht(0,\psi) = 1, \qquad
        \usndcht(\tnd,0) = u_\Gamma, \qquad
        \usndcht(\tnd,\psi \to \infty) = 1.
    \end{align}
    Using the similarity variable $\xi = \psi / (2\sqrt{\tnd})$ and setting $\usndcht(\tnd,\psi) = g(\xi)$ gives
    \begin{align}
        \usndcht(\tnd,\psi)
        = 1 + (1 - u_\Gamma)
          \erf\!\left(\frac{\psi}{2\sqrt{\tnd}}\right).
    \end{align}
    The corresponding heat flux at $\psi = 0$ is
    \begin{align}
        q_s(\tnd,0)
        = -\left. \frac{\partial \usndcht}{\partial \psi} \right|_{\psi=0}
        = (1 - u_\Gamma) \frac{1}{\sqrt{\pi \tnd}}.
    \end{align}
    Enforcing continuity of heat flux between fluid and solid, $q_f(\tnd,0) = q_s(\tnd,0)$, yields
    \begin{align}
        u_\Gamma = \frac{1}{1 + \sqrt{r_1 r_2}},
    \end{align}
    which proves~\cref{eq:short_time_stiff}.
\end{proof}
What~\cref{propo:short_time_stiff} shows is that, for very short times, the Nusselt number starts spatially uniform at a very large value (that depends on $r_1$ and $r_2$) and decays with the inverse square root of time. This explains the stiff behavior observed in~\cref{fig:cht_nusselt_points}, and also explains the spatial uniformity of the Nusselt number at very short times, see~\cref{fig:cht_Nu_polar}. To study the effects of the initial behavior on the solution, we now establish uniform bounds for the weighted averaged temperature of the RHE.

Let $u_1$ and $u_2$ be solutions of the RHE in~\cref{eq:rhe_nondim_T} with parameters 
$\{\Omega,\tf,\kappa,\sigma,\Biot_1,\eta_1(\tnd,\xnd)\}$ and 
$\{\Omega,\tf,\kappa,\sigma,\Biot_2,\eta_2(\tnd,\xnd)\}$, respectively, and $\Biot_1$, $\Biot_2$, $\eta_1$, $\eta_2$ non-negative. Define $w \coloneqq u_1 - u_2$. Then the following bounds hold for all $(\tnd,\xnd)\in[0,\tf]\times\Omega$:
\begin{align}
    0 \leq u_1(\tnd,\xnd) \leq 1, 
    \quad 0 \leq u_2(\tnd,\xnd) \leq 1, 
    \quad |w(\tnd,\xnd)| \leq 1, \label{eq:solution_bounds}
\end{align}
which follow from the maximum principle for parabolic equations; see, e.g.,~\cite{evans2022partial}.
As~\cref{eq:rhe_nondim_T} is linear, subtracting the equations satisfied by $u_1$ and $u_2$, multiplying by $w$, and integrating over space and time yield
\begin{align}
    \frac{1}{2}\int_\Omega \sigma w^2(\cdot,\tnd) 
    + \int_0^\tnd \left( \int_\Omega \kappa|\delnd w|^2 
    + \int_{\pOmega} \Biot_1 \eta_1 w^2 \right) dt' 
    = \int_0^\tnd \int_{\pOmega} (\Biot_2 \eta_2 - \Biot_1 \eta_1) u_2 w \, dt',
    \label{eq:energy_identity}
\end{align}
for all $\tnd \in [0,\tf]$.

As we are looking to bound the weighted averaged temperature $\uavg$ defined in~\cref{eq:urhea_avg}, we relate it to the $L^2$-norm appearing in~\cref{eq:energy_identity}. For any $w \in H^1(\Omega)$, the weighted average $w_{\mathrm{avg}}$ as defined in~\cref{eq:urhea_avg} satisfies
\begin{align}
    w_{\mathrm{avg}}^2 = \left( \frac{1}{|\Omega|}\int_\Omega \sigma w \right)^2 
    \leq \frac{1}{|\Omega|^2}\int_\Omega \sigma \cdot \int_\Omega \sigma w^2 
    = \frac{1}{|\Omega|} \int_\Omega \sigma w^2, \label{eq:uavg_bound}
\end{align}
where~\cref{eq:sigma_prop} was used in the last step. By combining the bounds in~\cref{eq:solution_bounds} together with~\cref{eq:energy_identity,eq:uavg_bound} we obtain
\begin{align}\label{eq:L1L1-stability}
    \frac{|\Omega|}{2} w_{\mathrm{avg}}(\tnd)^2 + \int_0^\tnd \left( \int_\Omega \kappa|\delnd w|^2
    + \int_{\pOmega} \Biot_1 \eta_1 w^2 \right) dt' 
    &\leq \left\Vert\Biot_2\eta_2 - \Biot_1\eta_1\right\Vert_{L^1((0,\tnd);L^1(\pOmega))},
\end{align}
where $w_{\mathrm{avg}}(\tnd)$ is the difference between the weighted averaged temperatures corresponding to $u_1$ and $u_2$.
 
By setting $(\Biot_1, \eta_1) = (\Biot, \eta(\tnd,\xnd))$ and 
$(\Biot_2, \eta_2) = (\Biot, \etabar(\xnd))$, 
we obtain a direct estimate for the temporal approximation error in~\cref{eq:error_triangular_time}, 
Since the integral terms on the left-hand side of~\cref{eq:L1L1-stability} are nonnegative, 
we can immediately deduce the bound
\begin{align}
\begin{aligned}
    &\max_{\tnd\in[0,\tf]}|\uavg(\tnd;\kappa,\sigma,\Biot,\eta(\tnd,\xnd)) - \uavgtilde(\tnd;\kappa,\sigma,\Biot,\etabar(\xnd))| \\
    &\quad\leq \sqrt{\frac{2\Biot}{|\Omega|} \left\Vert \eta(\tnd,\xnd) -\etabar(\xnd) \right\Vert_{L^1((0,\tf);L^1(\pOmega))}}.
\end{aligned}
\label{eq:temporal_approx_error_bound}
\end{align}
The estimate shows that perturbations of the Robin coefficient in 
$L^1((0,\tf);L^1(\pOmega))$ influence the solid temperature in a controlled manner, 
as measured by the $L^2$-norm of the averaged temperature. 
At short times, the norm is dominated by the initial transient, 
but over longer time horizons, this effect is averaged out. 
Consistent with this analysis, numerical results confirm that for large time scale separation, 
the initial behavior has only a minor influence on the averaged temperature. This will be studied quantitatively in~\cref{subsubsec:time_numerics}.

\subsubsection{Time scale separation between fluid and solid and time averaging}\label{subsubsec:time_scale_separation}
Having discussed the initial transient behavior and its effects, we now develop physical arguments to assess the degree of time scale separation between the fluid and solid. In the case of time scale separation, the fluid flow and Nusselt number develop quickly and become stationary long before the solid temperature changes significantly. From the perspective of the solid, the fluid thus appears quasi-stationary, and it is reasonable to approximate $\eta(\tnd,\xnd)$ by its time average $\etabar(\xnd)$,
\begin{align}
    \etabar(\xnd) = \dashint_0^{\tf} \eta(\tnd,\xnd). \label{eq:eta_time_avg}
\end{align}
To quantify the degree of time scale separation, we compare the fluid convective and solid equilibration time scales. The non-dimensional convective time scale was given in~\cref{eq:tconv_nondim}.
In the small-Biot regime, the solid temperature evolution is well described by the LCM, whose characteristic equilibration time---the time for the solid to cool to $1/\exp(1)$ of its initial temperature---is
\begin{align*}
    \tauL = \frac{1}{B\gamma},
\end{align*}
as previously defined in~\cref{eq:teq_nondim}.  
The ratio of the two time scales is then
\begin{align}\label{eq:time_ratio_2}
    \frac{\tauL}{\tconvnondim}
    = \frac{r_2}{r_1}\frac{\Reynolds\Prandtl}{B\gamma}
    = \frac{1}{r_1}\frac{\Reynolds\Prandtl}{\Nustavg\gamma},
\end{align}
where~\cref{eq:cht_biot} was used in the last step. A large value of this ratio indicates that the solid dynamics are much slower than the fluid dynamics, implying strong time scale separation.

In the small-Biot limit and for fixed Reynolds and Prandtl numbers, time scale separation depends primarily on $r_1$, while $r_2$ affects it only indirectly through its influence on $\Nustavg$. The effects of $r_1$, $r_2$, and $\Reynolds$ on $\Nustavg$ will be analyzed in~\cref{subsubsec:iso_versus_cht}.  
It remains unclear how large the ratio in~\cref{eq:time_ratio_2} must be for the temporal approximation error to become negligible. We present numerical examples in~\cref{subsubsec:time_numerics}, where we vary $r_1$, to study this empirically.

\begin{cmt}
Strictly speaking, the argument based on~\cref{eq:time_ratio_2} assumes that the LCM accurately predicts the equilibration time of the CHT model. Since the LCM itself presupposes time scale separation, this assumption is not exact. Nevertheless, the ratio in~\cref{eq:time_ratio_2} provides a useful measure of time scale separation in the small-Biot regime, as will be confirmed empirically in~\cref{subsubsec:time_numerics}.
\end{cmt}

\begin{cmt}
It should be noted that the temporal average in~\cref{eq:eta_time_avg} used for $\etabar$ is not necessarily optimal. In Appendix~\ref{sec:time_homogenization_rhe}, we present a proof for a simplified setting—neglecting the initial transient, and prescribing a spatially uniform time-periodic $\eta$—which demonstrates that the average is indeed optimal. The general case, however, remains unresolved and is left for future work.
\end{cmt}

\subsubsection{Numerical results}\label{subsubsec:time_numerics}
We study the cylinder cross-flow case from~\cref{subsec:example_cht} for varying $r_1$ to illustrate the effect of time scale separation on the temporal approximation error~\cref{eq:error_triangular_time}, which we denote here as
\begin{align}
    E_t \coloneqq \max_{\tnd\in[0,\tf]}|\uavg(\tnd;\kappa,\sigma,\Biot,\eta(\tnd,\xnd)) - \uavgtilde(\tnd;\kappa,\sigma,\Biot,\etabar(\xnd))|.
\end{align}
The setup is identical to~\cref{subsec:example_cht}, except that $r_1$ is varied. For each $r_1$, we compute the Nusselt number and average temperature with CHT using Nek5000; the results are shown in~\cref{fig:time_r1_study}. For larger $r_1$ values, vortex shedding occurs later in time or not at all within the chosen time horizon, and the Nusselt number continues to decrease without stabilizing. The corresponding average temperature decays more rapidly at early times—since the Nusselt number has higher values initially—but more slowly at later times, reflecting the lower Nusselt values. This behavior can partially be explained with the short time behavior, see~\cref{propo:short_time_stiff}: for larger $r_1$, the initial Nusselt number is larger (see~\cref{eq:short_time_nusselt}), leading to a more rapid initial decrease in the average temperature.
\begin{figure}[ht]
    \centering
    \begin{subfigure}[b]{0.48\textwidth}
        \centering
        \includegraphics[width=\textwidth]{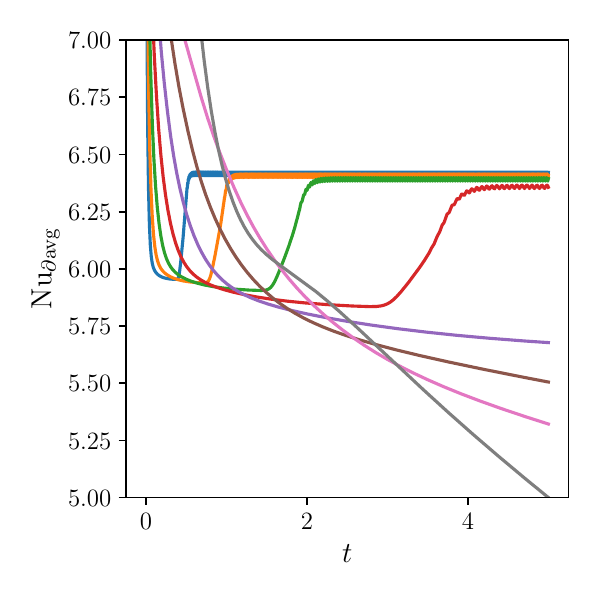}
        \caption{Spatially averaged Nusselt number}
        \label{fig:time_r1_study_Nu}
    \end{subfigure}
    \hfill
    \begin{subfigure}[b]{0.48\textwidth}
        \centering
        \includegraphics[width=\textwidth]{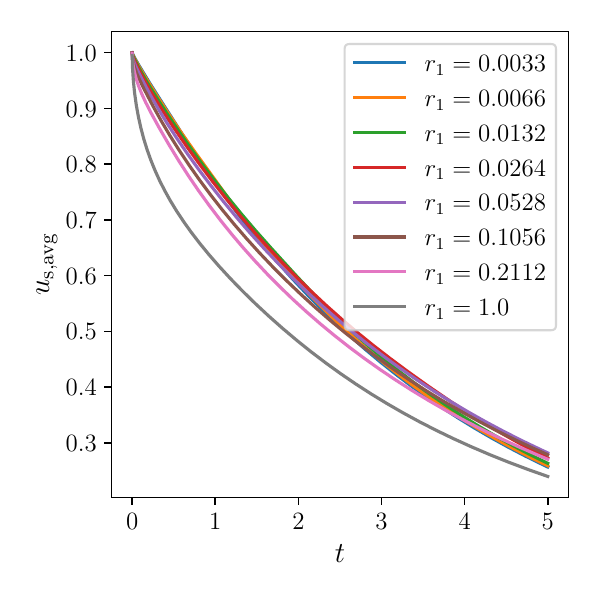}
        \caption{Average solid temperature}
        \label{fig:time_r1_study_uavg}
    \end{subfigure}
    \caption{Effect of $r_1$ on (a) the spatially averaged Nusselt number and (b) the average solid temperature.}
    \label{fig:time_r1_study}
\end{figure}

For each case we also compute $\Biot$ and $\etabar(\xnd)$, and solve the RHE$_a$ in DOLFINx~\cite{baratta2023dolfinx}. Interestingly, the Biot number remains small and varies only weakly with $r_1$ across the entire range. This highlights that using the LCM, which neglects the influence of $r_1$, can lead to large errors even when the Biot number is small. 

The temporal error $E_t$ is plotted against $r_1$ in~\cref{fig:time_r1_study_err}. As expected, the error increases with $r_1$, i.e., as time scale separation becomes less pronounced, and the initial transient effects become more significant. We highlight the case $r_1=0.0132$ in red color, which corresponds to the example in~\cref{subsec:example_cht}. The error for this case is $E_t=\num{0.010307143115081074}$. We also provide all computed values of $\Biot$, $\tauL/\tconvnondim$ and $E_t$ in~\cref{tab:time_scale_ratios}. We can see that for this problem, the temporal approximation error is below 1\% for $\tauL/\tconvnondim \gtrsim 300$.
\begin{figure}[t]
    \centering
    \includegraphics[width=0.5\textwidth]{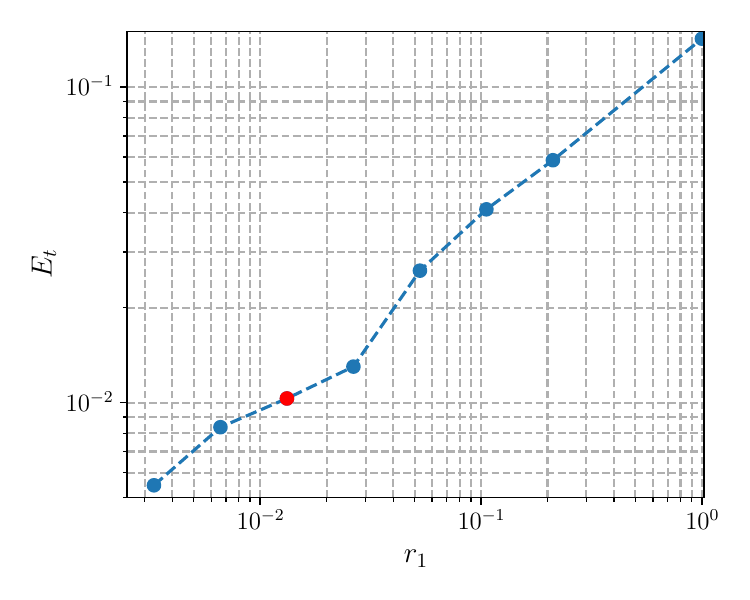}
    \caption{Maximum temporal approximation error $E_t$ as a function of $r_1$. The red point corresponds to the example in~\cref{subsec:example_cht}.}
    \label{fig:time_r1_study_err}
\end{figure}
\begin{table}[t]
\sisetup{scientific-notation = false}
\centering
\caption{Values of $r_1$ and the corresponding computed Biot number $\Biot$, time scale ratio $\tauL/\tconvnondim$ and maximum temporal approximation error $E_t$. The other parameters are fixed as $r_2=0.01085$, $\Reynolds=143$, $\Prandtl=0.71$.}
\label{tab:time_scale_ratios}
\begin{tabular}{c|cccccccc}
$r_1$                & 0.0033                     & 0.0066                     & 0.0132                     & 0.0264                     & 0.0528                     & 0.1056                    & 0.2112         & 1.0000           \\ \hline
$\Biot$              & 0.0693                     & 0.0689                     & 0.0680                     & 0.0662                     & 0.0645                     & 0.0651                    & 0.0666        & 0.0728            \\
$\tauL/\tconvnondim$ & 1205                       & 606                        & 307                        & 158                        & 80                         & 40                        & 20          & 4              \\
$E_t$                & 0.00547 & 0.00836 & \num{0.010307143115081074} & \num{0.013005113903375776} & \num{0.026193240600391188} & \num{0.04097377776132938} & \num{0.05863628036813262} & \num{0.14231574168950878}
\end{tabular}
\end{table}

\subsection{Error Due to Lumping}\label{subsec:lcm_lumping_error}
In this section, we derive an upper bound for the error in~\cref{eq:error_triangular_lcm} due to lumping (uniform temperature assumption),
$$\max_{\tnd\in[0,\tf]}|\uavgtilde(\tnd;\kappa,\sigma,\Biot,\etabar) - \uLump(\tnd;\Biot)|.$$
The key ingredients of the derivation are a generalized eigenvalue problem associated with the autonomous RHE$_a$ model in~\cref{eq:rhea_nondim_T}, and an asymptotic expansion of its first eigenvalue.

\subsubsection{Error bound}\label{subsubsec:cht_lcm_error}
We first introduce the generalized eigenvalue problem (we henceforth omit the adjective ``generalized'' in our exposition) associated to the autonomous RHE$_a$ model in~\cref{eq:rhea_nondim_T} and the corresponding Rayleigh quotient. For given $(\kappa,\sigma,B,\etabar)$, find $(\psi_j,\lambda_j)$ for $j=1,2,\ldots$ solution to
\begin{subequations}\label{eq:sensitivity_eigenvalue_problem}
\begin{alignat}{4}
- \delnd \cdot (\kappa\,\delnd \psi_j) &= \lambda_j \,\sigma\, \psi_j\, &\quad& \text{in }&\Omega, \\
\kappa\, \delnd \psi_j \cdot \bm{n} + \Biot\, \etabar\, \psi_j &= 0 &\quad& \text{on }&\pOmega, \\
\int_\Omega \sigma\, \psi_j^2 &= 1. &\quad& & &
\end{alignat}
\end{subequations}
We can write the weak form of~\cref{eq:sensitivity_eigenvalue_problem} as: find $(\psi_j,\lambda_j) \in H^1(\Omega) \times \mathbb{R}$ such that
\begin{align}
    \int_\Omega \kappa \,\delnd\psi_j \cdot \delnd v + \Biot \int_{\pOmega} \etabar\, \psi_j\, v = \lambda_j \int_\Omega \sigma \,\psi_j\, v \quad \forall v\in H^1(\Omega), \label{eq:eigenvalue_problem_weak}
\end{align}
with normalization
\begin{align*}
    \int_\Omega \sigma\, \psi_j^2 &= 1.
\end{align*}
Using the bilinear forms $a_0$ and $a_1$ defined in~\cref{eq:a0,eq:a1}, we can write the weak form in~\cref{eq:eigenvalue_problem_weak} more compactly as: find $(\psi_j,\lambda_j) \in H^1(\Omega) \times \mathbb{R}$ such that
\begin{align*}
    a_0(\psi_j,v;\kappa) + \Biot a_1(\psi_j,v;\etabar) = \lambda_j m(\psi_j,v;\sigma) \quad \forall v\in H^1(\Omega)
\end{align*}
with normalization
\begin{align}
    m(\psi_j,\psi_j;\sigma) = 1. \label{eq:eigenvalue_problem_weak_normalization}
\end{align}
The corresponding Rayleigh quotient is
\begin{align}
    R(v;\kappa,\sigma,\Biot,\etabar) \coloneqq \frac{a_0(v,v;\kappa) + \Biot a_1(v,v;\etabar)}{m(v,v;\sigma)}. \label{eq:rayleigh_quotient}
\end{align}
From the Rayleigh quotient, it follows that the eigenvalues are non-negative for $\Biot\geq0$ and the properties of $\sigma,\kappa,\etabar$ defined in~\cref{eq:eta_prop,eq:kappa_prop,eq:sigma_prop}. We sort them in increasing order, hence $0 \leq \lambda_1 \leq \lambda_2 \leq \ldots$, and the corresponding eigenfunctions form an orthonormal basis of $L^2(\Omega)$ with respect to the inner product induced by $m(\cdot,\cdot;\sigma)$. With the eigenpairs, we can expand the solution to the autonomous RHE$_a$ problem in~\cref{eq:rhea_nondim_T} as
\begin{align}
    \utilde(\tnd,\xnd;\kappa,\sigma,\Biot,\etabar) = \sum_{j=1}^\infty |\Omega|\psi_j(\xnd) m(1,\psi_j;\sigma) \exp(-\lambda_j \tnd), \label{eq:rhea_solution_eigen}
\end{align}
and the QoI as
\begin{align}
    \uavgtilde(\tnd;\kappa,\sigma,\Biot,\etabar) = \sum_{j=1}^\infty |\Omega| m(1,\psi_j;\sigma)^2 \exp(-\lambda_j \tnd), \label{eq:rhea_qoi_eigen}
\end{align}
see~\cite[Proposition 2.2 and 2.3]{kaneko2024error} for more details. 
Furthermore, we have the following lemma for the coefficients of the first term in~\cref{eq:rhea_qoi_eigen}.
\begin{lem}
    The coefficient of the first term in~\cref{eq:rhea_qoi_eigen} can be expanded as
    \begin{align}
        |\Omega| m(1,\psi_1;\sigma)^2 = 1 - \Biot^2 \Upsilon + \mathcal{O}(\Biot^3), \label{eq:first_eigenfunction_expansion}
    \end{align}
    where
    \begin{align*}
        \Upsilon \coloneqq \int_\Omega \sigma \psi'^0_1 \psi'^0_1.
    \end{align*}
\end{lem}
The proof is given in~\cite[Appendix C.5]{kaneko2024error}.
In the small-Biot limit, the first eigenvalue $\lambda_1$ dominates the long-time solution in~\cref{eq:rhea_qoi_eigen}. To analyze this regime, i.e., as $\Biot \to 0$, we introduce an asymptotic expansion of the first eigenvalue in powers of $\Biot$ around $\Biot = 0$:
\begin{align}
    \lambda_1(\Biot) = \lambda_1^0 + \lambda_1'^0 \Biot + \frac{1}{2}\lambda_1''^0 \Biot^2 + \mathcal{O}(\Biot^3). \label{eq:eigenvalue_expansion}
\end{align}
We use the superscript $0$ to denote the case $\Biot=0$, and the prime denotes differentiation with respect to $\Biot$. In the following, we find expressions for $\lambda_1^0$, $\lambda_1'^0$, and $\lambda_1''^0$.

From the Rayleigh quotient in~\cref{eq:rayleigh_quotient}, and setting $\Biot=0$, we directly find the first eigenvalue and eigenfunction,
\begin{align}
    \lambda_1^0 = 0, \quad \psi_1^0 = |\Omega|^{-1/2}. \label{eq:eigenvalue_problem_biot0}
\end{align}
\begin{lem}
    The first and second derivatives of the first eigenvalue at $\Biot=0$ are given by
    \begin{align}
        \lambda_1'^0 &= \gamma, \label{eq:eigenvalue_derivative_1} \\
        \lambda_1''^0 &= -2 a_0(\psi'^0_1,\psi'^0_1;\kappa). \label{eq:eigenvalue_derivative_2}
    \end{align}
\end{lem}
\begin{proof}
    Starting from the weak form in~\cref{eq:eigenvalue_problem_weak}, we consider the sensitivity equation obtained by differentiating~\cref{eq:eigenvalue_problem_weak} with respect to $\Biot$ and evaluating at $\Biot=0$ for the first eigenpair $(\psi_1,\lambda_1) \in H^1(\Omega) \times \mathbb{R}$:
    \begin{align}
        a_0(\psi'^0_1,v;\kappa) + a_1(\psi_1^0,v;\etabar) = \lambda_1'^0 m(\psi_1^0,v;\sigma) \quad \forall v\in H^1(\Omega). \label{eq:sensitivity_equation}
    \end{align}
    We can similarly repeat for the normalization condition in~\cref{eq:eigenvalue_problem_weak_normalization} to obtain
    \begin{align}
        m(\psi_1^0,\psi'^0_1;\sigma) = 0, \label{eq:sensitivity_normalization}
    \end{align}
    which is equal to
    \begin{align}
        m(1,\psi'^0_1;\sigma) = \int_\Omega \sigma \psi'^0_1 = 0, \label{eq:sensitivity_normalization_explicit}
    \end{align}
    because of~\cref{eq:eigenvalue_problem_biot0}.
    By choosing the test function as $v=\psi_1^0$ in~\cref{eq:sensitivity_equation} and using~\cref{eq:eigenvalue_problem_weak_normalization,eq:eta_prop,eq:eigenvalue_problem_biot0}, we find
    \begin{align}
        \lambda_1'^0 = a_1(\psi_1^0,\psi_1^0;\etabar) = |\Omega|^{-1} \int_{\pOmega} \etabar = \gamma, \label{eq:eigenvalue_derivative_1_proof}
    \end{align}
    which proves~\cref{eq:eigenvalue_derivative_1}. Next, we can insert $v=\psi'^0_1$ into~\cref{eq:sensitivity_equation} to find the relation
    \begin{align}
        a_0(\psi'^0_1,\psi'^0_1;\kappa) + a_1(\psi_1^0,\psi'^0_1;\etabar) = 0, \label{eq:sensitivity_equation_2}
    \end{align}
    where we used~\cref{eq:sensitivity_normalization}. To find the second derivative, we differentiate~\cref{eq:eigenvalue_problem_weak} twice with respect to $\Biot$ and evaluate at $\Biot=0$ with $\lambda_1^0=0$ and $\lambda_1'^0=\gamma$ to obtain
    \begin{align*}
        a_0(\psi''^0_1,v;\kappa) + 2 a_1(\psi'^0_1,v;\etabar) = \lambda_1''^0 m(\psi_1^0,v;\sigma) + 2 \gamma m(\psi'^0_1,v;\sigma) \quad \forall v\in H^1(\Omega).
    \end{align*}
    Now with $v=\psi_1^0$ and using~\cref{eq:eigenvalue_problem_weak_normalization,eq:eigenvalue_problem_biot0,eq:sensitivity_equation_2,eq:sensitivity_normalization}, we find
    \begin{align}
        \lambda_1''^0 = 2 a_1(\psi'^0_1,\psi_1^0;\etabar) = -2 a_0(\psi'^0_1,\psi'^0_1;\kappa), \label{eq:eigenvalue_derivative_2_proof}
    \end{align}
    which concludes the proof.
\end{proof}
We now define the coefficient $\phi$ as
\begin{align}
    \phi(\kappa,\sigma,\etabar) \coloneqq a_0(\psi'^0_1,\psi'^0_1;\kappa) = \int_\Omega \kappa |\delnd \psi'^0_1|^2, \label{eq:phi_definition}
\end{align}
which can be related to $\lambda_1''^0$ with~\cref{eq:eigenvalue_derivative_2_proof},
\begin{align}
    \phi(\kappa,\sigma,\etabar) = -\frac{\lambda_1''^0}{2},
\end{align}
and we can write the asymptotic expansion of the first eigenvalue in~\cref{eq:eigenvalue_expansion} as
\begin{align}
    \lambda_1(\Biot) = \gamma \Biot - \phi(\kappa,\sigma,\etabar) \Biot^2 + O(\Biot^3). \label{eq:eigenvalue_expansion_final}
\end{align}
The coefficient $\phi$ is of central importance for the error analysis, as seen below. To evaluate $\phi$, we need to solve the sensitivity equation, derived in~\cref{eq:sensitivity_equation,eq:sensitivity_normalization_explicit}, and rewritten here: find $\psi'^0_1 \in H^1(\Omega)$ such that
\begin{align*}
    a_0(\psi'^0_1,v;\kappa) = |\Omega|^{-1/2} L(v;\sigma,\etabar), \quad \forall v\in H^1(\Omega),
\end{align*}
with normalization condition
\begin{align*}
    m(\psi'^0_1,1;\sigma) = 0,
\end{align*}
and where the linear form $L$ is defined as
\begin{align}
\begin{aligned}
    L(v;\sigma,\etabar) &\coloneqq \gamma m(1,v;\sigma) - a_1(1,v;\etabar)\\ 
    &=\gamma \int_\Omega \sigma v - \int_{\pOmega} \etabar v.
\end{aligned}\label{eq:Lform}
\end{align}
We now define the function space
\begin{align*}
    Z_0(b) \coloneqq \left\{\, w \in H^1(\Omega) \ \big| \ \int_\Omega b w = 0 \,\right\},
\end{align*}
with which we can rewrite the sensitivity problem as: find $\psi'^0_1 \in Z_0(\sigma)$ such that
\begin{align}
    a_0(\psi'^0_1,v;\kappa) = |\Omega|^{-1/2} L(v;\sigma,\etabar) \quad \forall v\in Z_0(\sigma), \label{eq:sensitivity_problem_rewritten}
\end{align}
or expanded explicitly,
\begin{align*}
    \int_\Omega \kappa \delnd \psi'^0_1 \cdot \delnd v = -|\Omega|^{-1/2} \int_{\pOmega} \etabar v \quad \forall v\in Z_0(\sigma),
\end{align*}
where the first term on the right hand side vanishes due to the definition of $Z_0(\sigma)$. The associated strong form is given by
\begin{alignat}{4}
\begin{aligned}
 - \delnd \cdot (\kappa\delnd \psi'^0_1) &= 0 &\quad& \text{in }&\Omega, \\
 \kappa\delnd \psi'^0_1 \cdot \bm{n} &= -|\Omega|^{-1/2} \etabar &\quad& \text{on }&\pOmega,\\
    \int_\Omega \sigma \psi'^0_1 &= 0. &
\end{aligned} \label{eq:sensitivity_problem}
\end{alignat}
The numerical solution of this problem was discussed in~\cite{kaneko2024error} (for $\etabar=1$).
We now get to the central result of this section.
\begin{propo}[\textbf{Lumping error}]\label{propo:lcm_error_bound}
The error due to lumping is bounded asymptotically from above with
\begin{align}
    \max_{\tnd\in[0,\tf]}|\utilde_{\rm{avg}}(\tnd;\kappa,\sigma,\Biot,\etabar) - \uLump(\tnd;\Biot)| \leq \frac{\phi(\kappa,\sigma,\etabar)}{\gamma\exp(1)} \Biot + O(\Biot^2), \label{eq:lcm_error_bound}
\end{align}
where the coefficient $\phi > 0$ was defined in~\cref{eq:phi_definition}.
\end{propo}
\begin{proof}
    We start by writing the difference between the RHE$_a$ in~\cref{eq:rhea_qoi_eigen} and LCM solution in~\cref{eq:lumped_capacitance_solution},
    \begin{align*}
        \utilde_{\rm{avg}}(\tnd;\kappa,\sigma,\Biot,\etabar) - \uLump(\tnd;\Biot) &= \sum_{j=1}^\infty |\Omega| m(1,\psi_j;\sigma)^2 \exp(-\lambda_j \tnd) - \exp(-\gamma \Biot \tnd) \\
        &= \exp(-\lambda_1 \tnd) - \exp(-\gamma \Biot \tnd) + |\Omega|\sum_{j=2}^\infty m(1,\psi_j;\sigma)^2 \exp(-\lambda_j \tnd),
    \end{align*}
    where we inserted~\cref{eq:first_eigenfunction_expansion} for the first term of the series and truncated terms of order $\Biot^2$.
    As $\lambda_j\geq0$, we can bound the last term with
    \begin{align*}
        |\Omega|\sum_{j=2}^\infty m(1,\psi_j;\sigma)^2 \exp(-\lambda_j \tnd) &\leq |\Omega|\sum_{j=2}^\infty m(1,\psi_j;\sigma)^2 \\
        &= |\Omega|\sum_{j=1}^\infty m(1,\psi_j;\sigma)^2 - |\Omega| m(1,\psi_1;\sigma)^2 \\
        &= |\Omega|\sum_{j=1}^\infty m(1,\psi_j;\sigma)^2 - (1-\Biot^2 \Upsilon + \mathcal{O}(\Biot^3)),
    \end{align*}
    where we again used~\cref{eq:first_eigenfunction_expansion}. From Parseval's identity, we know
    \begin{align*}
        |\Omega|\sum_{j=1}^\infty m(1,\psi_j;\sigma)^2 = 1,
    \end{align*}
    and it follows that
    \begin{align*}
        |\Omega|\sum_{j=2}^\infty m(1,\psi_j;\sigma)^2 \exp(-\lambda_j \tnd) &\leq \Biot^2 \Upsilon + \mathcal{O}(\Biot^3).
    \end{align*}
    Now, by inserting the asymptotic expansion of the first eigenvalue in~\cref{eq:eigenvalue_expansion_final} and truncating terms of order $\Biot^2$, we obtain 
    \begin{align*}
        \utilde_{\rm{avg}}(\tnd;\kappa,\sigma,\Biot,\etabar) - \uLump(\tnd;\Biot) &\leq \exp(-\lambda_1 \tnd) - \exp(-\gamma \Biot \tnd) + \mathcal{O}(\Biot^2) \\
        &= \exp(-(\gamma\Biot-\phi\Biot^2) \tnd) - \exp(-\gamma \Biot \tnd) + \mathcal{O}(\Biot^2) \\
        &= \exp\left(-\gamma\Biot \tnd\left(1-\frac{\phi\Biot}{\gamma}\right)\right) - \exp(-\gamma \Biot \tnd) + \mathcal{O}(\Biot^2).
    \end{align*}
    Finally, we use the exponential bound in~\cref{eq:exp_inequality} with $z=\gamma \Biot \tnd$ and $\epsilon = \phi\Biot/\gamma$ to find the bound in~\cref{eq:lcm_error_bound}.
\end{proof}

\subsubsection{\texorpdfstring{Upper bound for $\phi$}{Upper bound for phi}}\label{subsubsec:phi_bound}

We can use~\cref{eq:lcm_error_bound} to bound the lumping error given knowledge of the coefficient $\phi(\kappa,\sigma,\etabar)$. However, to compute $\phi(\kappa,\sigma,\etabar)$,~\cref{eq:sensitivity_problem} needs to be solved, for which the spatial distributions of the material properties $\kappa$ and $\sigma$ as well as the time-independent variation $\etabar$ are required---quantities that might be unknown in applications. In this section, we derive an upper bound for $\phi$ that is less restrictive: it only requires knowledge of the geometry $\Omega$ and the variances of $\sigma$ and $\etabar$, given by
\begin{align}
    \dashint_{\Omega} (\sigma - 1)^2, \label{eq:var_sigma}\\
    \dashint_{\pOmega} (\etabar - 1)^2 \label{eq:var_eta}
\end{align}
which are scalar derived quantities.

Before we state the main result, we first derive an upper bound for $\phi$ that removes the dependency of $\phi$ on $\kappa$.

\begin{lem}\label{lem:phi_bound_kappa}
    The coefficient $\phi(\kappa,\sigma,\etabar)$ is bounded from above with
    \begin{align}
        \phi(\kappa,\sigma,\etabar) \leq \phi(1,\sigma,\etabar). \label{eq:phi_bound_kappa}
    \end{align}
\end{lem}
\begin{proof}
    First, define $\xi \in Z_0(\sigma)$ and $\xi^1 \in Z_0(\sigma)$ as the solutions to~\cref{eq:sensitivity_problem_rewritten} for a general $\kappa$ and $\kappa=1$ respectively, i.e.,
    \begin{align}
        a_0(\xi,v;\kappa) &= |\Omega|^{-1/2} L(v;\sigma,\etabar) \quad \forall v\in Z_0(\sigma) \label{eq:kappa_prf_1}, \\
        a_0(\xi^1,v;1) &= |\Omega|^{-1/2} L(v;\sigma,\etabar) \quad \forall v\in Z_0(\sigma) \label{eq:kappa_prf_2}.
    \end{align}
    Moreover, we have with~\cref{eq:phi_definition} that
    \begin{align}
        \phi(\kappa,\sigma,\etabar) = a_0(\xi,\xi;\kappa), \label{eq:kappa_prf_3} \\
        \phi(1,\sigma,\etabar) = a_0(\xi^1,\xi^1;1). \label{eq:kappa_prf_4}
    \end{align}
    By inserting $v=\xi$ into~\cref{eq:kappa_prf_1} and $v=\xi^1$ into~\cref{eq:kappa_prf_2}, we obtain the following relations:
    \begin{align}
        |\Omega|^{-1/2} L(\xi;\sigma,\etabar) = \phi(\kappa,\sigma,\etabar), \label{eq:kappa_prf_5} \\
        |\Omega|^{-1/2} L(\xi^1;\sigma,\etabar) = \phi(1,\sigma,\etabar). \label{eq:kappa_prf_6}
    \end{align}
    We can now introduce the energy functional associated with~\cref{eq:sensitivity_problem_rewritten},
    \begin{align}
        J(v;\kappa,\sigma,\etabar) \coloneqq \frac{1}{2} a_0(v,v;\kappa) - |\Omega|^{-1/2} L(v;\sigma,\etabar). \label{eq:kappa_prf_7}
    \end{align}
    As the solution $\xi^1$ minimizes $J(v;1,\sigma,\etabar)$,we can write
    \begin{align*}
        J(\xi^1;1,\sigma,\etabar) &\leq J(\xi;1,\sigma,\etabar).
    \end{align*}
    The left term can be rewritten with~\cref{eq:kappa_prf_4,eq:kappa_prf_6} as
    \begin{align*}
        J(\xi^1;1,\sigma,\etabar) &= \frac{1}{2} a_0(\xi^1,\xi^1;1) - |\Omega|^{-1/2} L(\xi^1;\sigma,\etabar) \\ 
        &= \frac{1}{2} \phi(1,\sigma,\etabar) - \phi(1,\sigma,\etabar) \\
        &= -\frac{1}{2} \phi(1,\sigma,\etabar).
    \end{align*}
    Since $\kappa\geq 1$, see~\cref{eq:kappa_prop}, the right term can be rewritten with~\cref{eq:kappa_prf_3,eq:kappa_prf_5} as
    \begin{align*}
        J(\xi;1,\sigma,\etabar) &\leq J(\xi;\kappa,\sigma,\etabar) \\
        &= \frac{1}{2} a_0(\xi,\xi;\kappa) - |\Omega|^{-1/2} L(\xi;\sigma,\etabar) \\
        &= -\frac{1}{2} \phi(\kappa,\sigma,\etabar).
    \end{align*}
    Combining both results and dividing by $-1/2$ on both sides yields~\cref{eq:phi_bound_kappa}.
\end{proof}

\begin{propo}
    The coefficient $\phi(\kappa,\sigma,\etabar)$ is bounded from above with
    \begin{align}
        \phi(\kappa,\sigma,\etabar) \leq (\sqrt{\phi(1,1,1)} + \sqrt{\deltasig} + \sqrt{\deltaeta})^2, \label{eq:phi_bound_result}
    \end{align}
    where $\phi(1,1,1)$ depends only on the geometry $\Omega$. The corrections, $\deltasig$ and $\deltaeta$, are defined as
    \begin{align}
        \deltasig &\coloneqq \frac{\gamma^2}{\mu}\dashint_\Omega(\sigma-1)^2, \label{eq:deltasig} \\
        \deltaeta &\coloneqq \frac{\gamma}{\Lambda} \dashint_{\pOmega} (\etabar - 1)^2. \label{eq:deltaeta}
    \end{align}
    Here, $\mu$ and $\Lambda$ correspond to eigenvalues defined as
    \begin{align}
        \mu \coloneqq \inf_{w\in Z_0(1)} \frac{a_0(w,w)}{\int_\Omega w^2}, \label{eq:stability_1} \\
        \Lambda \coloneqq \inf_{w\in Z_0(1)} \frac{a_0(w,w)}{\int_{\pOmega} w^2}, \label{eq:stability_2}
    \end{align}
    which are also purely geometric quantities. We note that the stability constants $\gamma^2/\mu$ and $\gamma/\Lambda$ are scale invariant. The numerical treatment of the eigenvalue problems is outlined in Appendix~\ref{sec:steklov_numerical}.
\end{propo}
\begin{proof}
    We start the proof by using~\cref{lem:phi_bound_kappa} to remove the dependence on $\kappa$,
    \begin{align*}
        \phi(\kappa,\sigma,\etabar) \leq \phi(1,\sigma,\etabar).
    \end{align*}
    For the remainder of this proof, we thus set $\kappa=1$ and write, for brevity, hereafter
    \begin{align*}
        a_0(v,w) &\equiv a_0(v,w;1), \\
        \phi(\sigma,\etabar) &\equiv \phi(1,\sigma,\etabar).
    \end{align*}
    We start by defining
    \begin{align*}
        \phi(\sigma,\etabar) \coloneqq a_0(\psi'^0_1, \psi'^0_1),
    \end{align*}
    where $\psi'^0_1 \in Z_0(\sigma)$ is the solution to~\cref{eq:sensitivity_problem_rewritten},
    \begin{align*}
        a_0(\psi'^0_1,v) = |\Omega|^{-1/2} L(v;\sigma,\etabar), \quad \forall v\in Z_0(\sigma).
    \end{align*}
    We also define $\xi\in Z_0(1)$ which solves the same problem but in a different space,
    \begin{align}
        a_0(\xi,v) = |\Omega|^{-1/2} L(v;\sigma,\etabar), \quad \forall v\in Z_0(1). \label{eq:phi_prf_2}
    \end{align}
    Both $\psi'^0_1$ and $\xi$ are related via a constant shift,
    \begin{align*}
        \psi'^0_1 = \xi - \int_{\Omega} \sigma \xi,
    \end{align*}
    and, thus, it holds,
    \begin{align}
        \phi(\sigma,\etabar) = a_0(\psi'^0_1, \psi'^0_1) = a_0(\xi, \xi), \label{eq:phi_prf_3}
    \end{align}
    since $a_0$ only depends on gradients. We can now split the right hand side in~\cref{eq:phi_prf_2} with~\cref{eq:Lform} into three parts,
    \begin{align*}
        a_0(\xi,v) = |\Omega|^{-1/2} \left(\gamma \int_\Omega(\sigma-1)v + \gamma\int_\Omega v - \int_{\pOmega} (\etabar - 1) v - \int_{\pOmega} v \right), \quad \forall v\in Z_0(1),
    \end{align*}
    as well as $\xi = \overline{\xi} + \widehat{\xi} + \widetilde{\xi}$, where $\overline{\xi}, \widehat{\xi}, \widetilde{\xi} \in Z_0(1)$, to obtain three separate problems:
    \begin{align}
        a_0(\overline{\xi},v) &= |\Omega|^{-1/2} \left(\gamma\int_\Omega v - \int_{\pOmega} v \right) = |\Omega|^{-1/2} L(v;1,1), \quad \forall v\in Z_0(1), \label{eq:phi_prf_4}\\
        a_0(\widehat{\xi},v) &= |\Omega|^{-1/2} \gamma \int_\Omega(\sigma-1)v, \quad \forall v\in Z_0(1), \label{eq:phi_prf_5} \\
        a_0(\widetilde{\xi},v) &= - |\Omega|^{-1/2}  \int_{\pOmega} (\etabar - 1) v, \quad \forall v\in Z_0(1). \label{eq:phi_prf_6}
    \end{align}
    The first equation in~\cref{eq:phi_prf_4} does not depend on $\sigma$ nor $\etabar$ and can be directly solved for $\overline{\xi}$ given a geometry $\Omega$, which yields
    \begin{align}
        \phi(1,1) \coloneqq a_0(\overline{\xi},\overline{\xi}). \label{eq:phi_prf_7}
    \end{align}
    Choosing $v=\widehat{\xi}$ in the second equation in~\cref{eq:phi_prf_5}, successively applying the Cauchy-Schwarz inequality and inserting the definition of $\mu$ in~\cref{eq:stability_1} gives
    \begin{align}
    \begin{aligned}
        a_0(\widehat{\xi},\widehat{\xi}) &= |\Omega|^{-1/2} \gamma \int_\Omega(\sigma-1)\widehat{\xi} \\
        &\leq |\Omega|^{-1/2} \gamma \sqrt{\int_\Omega(\sigma-1)^2} \sqrt{\int_\Omega \widehat{\xi}^2} \\
        &\leq |\Omega|^{-1/2} \gamma \sqrt{\int_\Omega(\sigma-1)^2} \sqrt{a_0(\widehat{\xi},\widehat{\xi}) \sup_{w\in Z_0(1)} \frac{\int_\Omega w^2}{a_0(w,w)}} \\
        &= |\Omega|^{-1/2} \gamma \sqrt{\int_\Omega(\sigma-1)^2} \sqrt{a_0(\widehat{\xi},\widehat{\xi}) \mu^{-1}}.
    \end{aligned}\label{eq:phi_prf_8}
    \end{align}
    By dividing both sides of~\cref{eq:phi_prf_8} by $\sqrt{a_0(\widehat{\xi},\widehat{\xi})}$, squaring and inserting~\cref{eq:deltasig}, we obtain
    \begin{align}
        a_0(\widehat{\xi},\widehat{\xi}) \leq \deltasig. \label{eq:phi_prf_9}
    \end{align}
    Similarly, choosing $v=\widetilde{\xi}$ in the third equation in~\cref{eq:phi_prf_6}, successively applying the Cauchy-Schwarz inequality, and inserting the definition of $\Lambda$ in~\cref{eq:stability_2} yields
    \begin{align}
    \begin{aligned}
        a_0(\widetilde{\xi},\widetilde{\xi}) &= - |\Omega|^{-1/2}  \int_{\pOmega} (\etabar - 1) \widetilde{\xi} \\
        &\leq |\Omega|^{-1/2} \sqrt{\int_{\pOmega} (\etabar - 1)^2} \sqrt{\int_{\pOmega} \widetilde{\xi}^2} \\
        &\leq |\Omega|^{-1/2} \sqrt{\int_{\pOmega} (\etabar - 1)^2} \sqrt{a_0(\widetilde{\xi},\widetilde{\xi}) \sup_{w\in Z_0(1)} \frac{\int_{\pOmega} w^2}{a_0(w,w)}} \\
        &= |\Omega|^{-1/2} \sqrt{\int_{\pOmega} (\etabar - 1)^2} \sqrt{a_0(\widetilde{\xi},\widetilde{\xi}) \Lambda^{-1}},
    \end{aligned}\label{eq:phi_prf_10}
    \end{align}
    and dividing both sides of~\cref{eq:phi_prf_10} by $\sqrt{a_0(\widetilde{\xi},\widetilde{\xi})}$, squaring, and inserting~\cref{eq:deltaeta}, yields
    \begin{align}
        a_0(\widetilde{\xi},\widetilde{\xi}) \leq \deltaeta. \label{eq:phi_prf_11}
    \end{align}
    Starting now from~\cref{eq:phi_prf_3}, and inserting $\xi = \overline{\xi} + \widehat{\xi} + \widetilde{\xi}$, we have
    \begin{align}
    \begin{aligned}
        \phi(\sigma,\etabar) &= a_0(\xi, \xi) \\
        &= a_0(\overline{\xi} + \widehat{\xi} + \widetilde{\xi}, \overline{\xi} + \widehat{\xi} + \widetilde{\xi}) \\
        &= a_0(\overline{\xi},\overline{\xi})+a_0(\widehat{\xi},\widehat{\xi})+a_0(\widetilde{\xi},\widetilde{\xi})+2a_0(\overline{\xi},\widehat{\xi})+2a_0(\overline{\xi},\widetilde{\xi})+2a_0(\widehat{\xi},\widetilde{\xi}).
    \end{aligned}\label{eq:phi_prf_12}
    \end{align}
    The mixed terms in~\cref{eq:phi_prf_12} can be bounded with the Cauchy-Schwarz inequality:
    \begin{align}
        a_0(\overline{\xi},\widehat{\xi}) &\leq \sqrt{a_0(\overline{\xi},\overline{\xi})}\sqrt{a_0(\widehat{\xi},\widehat{\xi})}, \label{eq:phi_prf_13}\\
        a_0(\overline{\xi},\widetilde{\xi}) &\leq \sqrt{a_0(\overline{\xi},\overline{\xi})}\sqrt{a_0(\widetilde{\xi},\widetilde{\xi})}, \label{eq:phi_prf_14}\\
        a_0(\widetilde{\xi},\widehat{\xi}) &\leq \sqrt{a_0(\widetilde{\xi},\widetilde{\xi})}\sqrt{a_0(\widehat{\xi},\widehat{\xi})}. \label{eq:phi_prf_15}
    \end{align}
    Inserting~\cref{eq:phi_prf_13,eq:phi_prf_14,eq:phi_prf_15} together with~\cref{eq:phi_prf_7,eq:phi_prf_9,eq:phi_prf_11} into~\cref{eq:phi_prf_12} concludes the proof:
    \begin{align*}
    \begin{aligned}
        \phi(\sigma,\etabar) &= a_0(\overline{\xi},\overline{\xi})+a_0(\widehat{\xi},\widehat{\xi})+a_0(\widetilde{\xi},\widetilde{\xi})+2a_0(\overline{\xi},\widehat{\xi})+2a_0(\overline{\xi},\widetilde{\xi})+2a_0(\widehat{\xi},\widetilde{\xi}) \\
        &\leq \left(\sqrt{a_0(\overline{\xi},\overline{\xi})} + \sqrt{a_0(\widehat{\xi},\widehat{\xi})} + \sqrt{a_0(\widetilde{\xi},\widetilde{\xi})} \right)^2 \\
        &\leq \left(\sqrt{\phi(1,1)} + \sqrt{\delta^\sigma} + \sqrt{\delta^{\etabar}}\right)^2.
    \end{aligned}
    \end{align*}
\end{proof}
\begin{cmt}
    Instead of solving the eigenvalue problems in~\cref{eq:stability_1,eq:stability_2} to obtain $\mu$ and $\Lambda$, it is possible to derive lower bounds for these quantities. For example, for convex domains, the Payne-Weinberger inequality~\cite{payne1960optimal} can be used to bound $\mu$ from below with
    \begin{align*}
        \mu \geq \frac{\pi^2}{\mathcal{D}_{\Omega}^2},
    \end{align*}
    where $\mathcal{D}_{\Omega}$ is the diameter of the domain $\Omega$, defined as the largest distance between any two points in $\Omega$. Similar lower bounds for $\Lambda$ can be found in~\cite{girouard2017spectral}.
\end{cmt}

To evaluate the bound in~\cref{eq:phi_bound_result}, only the geometry $\Omega$ (from which $\phi(1,1,1)$, $\mu$ and $\Lambda$ can be computed) and the variances in~\cref{eq:var_sigma,eq:var_eta} are required. Importantly, no detailed knowledge of the spatial distributions of $\kappa$, $\sigma$, or $\etabar$ is needed; it suffices to know only the derived quantities.
For the special case of composite materials, where the solid is composed of $n$ materials of known volume fraction $v_i$ for $i=1,\dots,n$ with $\sum_{i=1}^n v_i=1$ and known constant density $\rhodim_i$ and constant heat capacity $\cdim_i$ for $i=1,\dots,n$, the variance of $\sigma$ can be directly evaluated,
\begin{align}
    \dashint_{\Omega} (\sigma-1)^2 = \sum_{i=1}^n v_i \left(\frac{(\rhodim \cdim)_i}{\sum_{j=1}^n v_j (\rhodim \cdim)_j} - 1\right)^2. \label{eq:var_sigma_composite}
\end{align}
However, evaluating the variance of $\eta$ is significantly more challenging, as it depends on the external flow conditions.

In the following two sections, we study the tightness of the bounds in~\cref{eq:phi_bound_result,eq:lcm_error_bound} with numerical examples for \emph{given} domain, $\kappa$, $\sigma$, and $\etabar$. In Appendix~\ref{sec:eta_study}, we provide a numerical study on the values of $\etabar$ for various geometries under realistic flow conditions coming from CHT solutions and experimental data, and provide guidance on estimating the $\eta$-variance in practice.

\subsubsection{\texorpdfstring{Numerical examples for homogeneous solids}{Numerical examples for homogeneous solids}}
In this section, we consider only homogeneous solids, $\sigma = 1$ and $\kappa = 1$, to study how the shape of $\etabar$ affects the lumping error. The influence of geometry on $\phi$ for $\etabar=1$ has been analyzed in detail in~\cite{kaneko2024error}. In the examples presented here, we study the influence of $\etabar\not=1$ on $\phi$, as well as the upper bound derived for different domains and $\etabar$.

We consider the four two-dimensional geometries shown in~\cref{fig:eta_geometries}—a disk, a square, an equilateral triangle, and a cross. For each geometry, we compute the relevant geometric quantities: the coefficient $\phi(1,1,1)$, and the scale-independent stability constants $\gamma^2/\mu$ and $\gamma/\Lambda$. The computations are performed with DOLFINx~\cite{baratta2023dolfinx}, as described in Appendix~\ref{sec:steklov_numerical}, and the results are summarized in~\cref{tab:geometry_constants}. For the disk, the analytical values $\phi(1,1,1) = 1/2$ and $\gamma/\Lambda = 2$ are known~\cite{kaneko2024error,girouard2017spectral}, and they agree with our numerical results.
\begin{figure}[ht]
    \centering
    \includegraphics[width=0.6\textwidth]{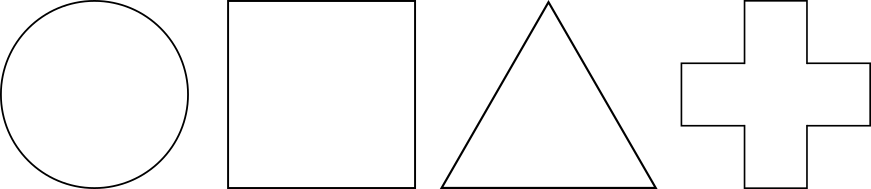}
    \caption{Geometries used in the numerical study of the LCM error.}
    \label{fig:eta_geometries}
\end{figure}
\begin{table}[ht]
    \centering
    \sisetup{scientific-notation=fixed, round-mode=places, round-precision=3}
    \caption{Comparison of $\phi$ constants for various two-dimensional geometries: the coefficient $\phi(1,1,1)$ defined in~\cref{eq:phi_definition}, the stability constant $\gamma^2/\mu$ in~\cref{eq:deltasig}, and the stability constant $\gamma/\Lambda$ in~\cref{eq:deltaeta}.}
    \label{tab:geometry_constants}
    \begin{tabular}{l|cccc}
                    & Disk        & Square      & Triangle       & Cross                   \\ \hline
    $\phi(1,1,1)$    & \num{0.4999999999998502}   & \num{0.6666666666666539}   & \num{1.0}      & \num{0.8353507069334796}           \\
    $\gamma^2/\mu$ & \num{1.1799557177878428}   & \num{1.6211389031933416} & \num{2.735671954819209} & \num{4.385215317997851} \\
    $\gamma/\Lambda$   & \num{1.999999999999442} & \num{2.905909227062346} & \num{5.3672862015104}    & \num{6.677162272258544}
    \end{tabular}
\end{table}

We consider four functions for $\etabar$, defined in non-dimensionalized coordinates ($x, y$) scaled by the domain diameter, $\operatorname{diam}(\Omega)\coloneqq\max_{\xnd,\bm{y}\in\Omega}\|\xnd-\bm{y}\|_2$, with the origin at the centroid:
\begin{itemize}
    \item[(i)] Constant function: $\etabar=c$.
    \item[(ii)] Linear function: $\etabar(x,y)=c(1+x+y)$.
    \item[(iii)] Sinusoidal function: $\etabar(x,y)=c(1+\sin(10\pi x)\sin(5\pi y))$.
    \item[(iv)] Step function: 
        \begin{align}
            \etabar(x,y)=
            \begin{cases}
            0 & \text{if } y < 0, \\
            2c & \text{if } y \geq 0.
            \end{cases}
            \label{eq:step_function}
        \end{align}
\end{itemize}
For each domain and function, $c>0$ is chosen such that $\dashint_{\pOmega}\etabar = 1$.

The variance $\dashint_{\pOmega}(\etabar - 1)^2$, together with the corresponding $\deltaeta$, the upper bound $\phi^{\text{ub}}$, and the actual value of $\phi$, are reported in~\cref{tab:disk_eta,tab:square_eta,tab:triangle_eta,tab:cross_eta} for each domain and choice of $\etabar$. In all cases, $\phi^{\text{ub}}$ is about 2-4 times larger than the actual value of $\phi$, indicating that the bound is reasonably tight for our purposes.

Since the variance of $\etabar$ may be unknown in practice, we also report $\phi^{\text{ub,est}}$, obtained by assuming a variance of $\dashint_{\pOmega}(\etabar - 1)^2=1$. In this case, $\phi^{\text{ub,est}}$ can be up to ten times larger than the actual $\phi$. A variance of 1, however, is likely an overestimate for realistic applications, as it implies that the standard deviation of $\etabar$ is comparable to its mean—corresponding to regions with very high heat transfer alongside regions with none. The step function defined in~\cref{eq:step_function} is one example with such a variance on some domains. In Appendix~\ref{sec:eta_study}, we analyze convex geometries under realistic flow conditions and find that the variance is typically much smaller than 1, which brings $\phi^{\text{ub,est}}$ significantly closer both to $\phi^{\text{ub}}$ and the actual value of $\phi$.
\begin{table}[p]
    \centering
    \sisetup{scientific-notation=fixed, round-mode=places, round-precision=3}
    \caption{Comparison of actual values of $\phi$ for a disk with the upper bounds $\phi^{\text{ub}}$ and $\phi^{\text{ub,est}}$ for different $\etabar$. Also reported are the variance $\dashint_{\pOmega}(\etabar - 1)^2$ and the correction $\deltaeta$. To compute $\phi^{\text{ub,est}}$, we assume $\dashint_{\pOmega}(\etabar - 1)^2 = 1.0$.}
    \label{tab:disk_eta}
    \begin{tabular}{l|cccc}
        & Constant function & Linear function & Sinusoidal function & Step function \\ \hline
        $\phi$ & \num{0.4999999999995885} & \num{0.9999999989845357} & \num{0.6721977240334304} & \num{2.20521657192347} \\
    \hline
        $\phi(1,1,1)$ & \num{0.4999999999999203} & \num{0.4999999999999203} & \num{0.4999999999999203} & \num{0.49999999999992034} \\
        $\phi^{\text{ub}}$ & \num{0.49999999999992345} & \num{1.9999999979703322} & \num{2.277541156152137} & \num{4.4959995082459985} \\
        $\phi^{\text{ub,est}}$ & \num{4.499999999999129} & \num{4.499999999999129} & \num{4.499999999999129} & \num{4.499999999999129} \\
    \hline
        $\deltaeta$ & \num{4.498862513591685e-30} & \num{0.49999999898524583} & \num{0.6432772754888557} & \num{1.997333301998454} \\
        $\dashint_{\pOmega}(\etabar-1)^2$ & \num{2.2494312567963163e-30} & \num{0.2499999994926756} & \num{0.32163863774449564} & \num{0.9986666509994374} \\
    \end{tabular}
\end{table}
\begin{table}[p]
    \centering
    \sisetup{scientific-notation=fixed, round-mode=places, round-precision=3}
    \caption{Comparison of actual values of $\phi$ for a square with the upper bounds $\phi^{\text{ub}}$ and $\phi^{\text{ub,est}}$ for different $\etabar$. Also reported are the variance $\dashint_{\pOmega}(\etabar - 1)^2$ and the correction $\deltaeta$. To compute $\phi^{\text{ub,est}}$, we assume $\dashint_{\pOmega}(\etabar - 1)^2 = 1.0$.}
    \label{tab:square_eta}
    \begin{tabular}{l|cccc}
        & Constant function & Linear function & Sinusoidal function & Step function \\ \hline
        $\phi$ & \num{0.6666666666666415} & \num{1.1405770116592213} & \num{0.8003697176729266} & \num{3.256150570685944} \\
    \hline
        $\phi(1,1,1)$ & \num{0.6666666666666545} & \num{0.6666666666666539} & \num{0.6666666666666539} & \num{0.6666666666666539} \\
        $\phi^{\text{ub}}$ & \num{0.6666666666666632} & \num{2.2874333794816253} & \num{2.825561370403062} & \num{6.3505638807264075} \\
        $\phi^{\text{ub,est}}$ & \num{6.3562948580222915} & \num{6.35629485802229} & \num{6.35629485802229} & \num{6.35629485802229} \\
    \hline
        $\deltaeta$ & \num{2.731842492069457e-29} & \num{0.4843182045103952} & \num{0.747263497940897} & \num{2.90203468142628} \\
        $\dashint_{\pOmega}(\etabar-1)^2$ & \num{9.400990459812617e-30} & \num{0.1666666666666681} & \num{0.25715307655921643} & \num{0.9986666666666725} \\
    \end{tabular}
\end{table}
\begin{table}[p]
    \centering
    \sisetup{scientific-notation=fixed, round-mode=places, round-precision=3}
    \caption{Comparison of actual values of $\phi$ for a triangle with the upper bounds $\phi^{\text{ub}}$ and $\phi^{\text{ub,est}}$ for different $\etabar$. Also reported are the variance $\dashint_{\pOmega}(\etabar - 1)^2$ and the correction $\deltaeta$. To compute $\phi^{\text{ub,est}}$, we assume $\dashint_{\pOmega}(\etabar - 1)^2 = 1.0$.}
    \label{tab:triangle_eta}
    \begin{tabular}{l|cccc}
        & Constant function & Linear function & Sinusoidal function & Step function \\ \hline
        $\phi$ & \num{0.9999999999999273} & \num{1.8333333333331259} & \num{1.7539006986661871} & \num{5.543611371132418} \\
    \hline
        $\phi(1,1,1)$ & \num{1.00000000000005} & \num{1.0000000000000497} & \num{1.0000000000000497} & \num{1.0000000000000497} \\
        $\phi^{\text{ub}}$ & \num{1.0000000000000564} & \num{3.786158336986286} & \num{6.334251443759022} & \num{12.839407821457366} \\
        $\phi^{\text{ub,est}}$ & \num{11.00076705353117} & \num{11.00076705353117} & \num{11.00076705353117} & \num{11.00076705353117} \\
    \hline
        $\deltaeta$ & \num{9.774715601270339e-30} & \num{0.894547700251733} & \num{2.3006636800424842} & \num{6.672983937410922} \\
        $\dashint_{\pOmega}(\etabar-1)^2$ & \num{1.8211653402271803e-30} & \num{0.1666666666666666} & \num{0.4286456122639888} & \num{1.243269631407598} \\
    \end{tabular}
\end{table}
\begin{table}[p]
    \centering
    \sisetup{scientific-notation=fixed, round-mode=places, round-precision=3}
    \caption{Comparison of actual values of $\phi$ for a cross with the upper bounds $\phi^{\text{ub}}$ and $\phi^{\text{ub,est}}$ for different $\etabar$. Also reported are the variance $\dashint_{\pOmega}(\etabar - 1)^2$ and the correction $\deltaeta$. To compute $\phi^{\text{ub,est}}$, we assume $\dashint_{\pOmega}(\etabar - 1)^2 = 1.0$.}
    \label{tab:cross_eta}
    \begin{tabular}{l|cccc}
        & Constant function & Linear function & Sinusoidal function & Step function \\ \hline
        $\phi$ & \num{0.835350706936011} & \num{1.8749607468210994} & \num{1.3606059369213064} & \num{4.974787264265159} \\
    \hline
        $\phi(1,1,1)$ & \num{0.8353507069334798} & \num{0.8353507069334796} & \num{0.8353507069334796} & \num{0.8353507069334798} \\
        $\phi^{\text{ub}}$ & \num{0.8353507069335263} & \num{3.8765561834651225} & \num{6.783388943807153} & \num{12.231957114402185} \\
        $\phi^{\text{ub,est}}$ & \num{12.235974516853036} & \num{12.235974516853036} & \num{12.235974516853036} & \num{12.235974516853036} \\
    \hline
        $\deltaeta$ & \num{6.378355518546946e-28} & \num{1.11286037870974} & \num{2.8578537039386918} & \num{6.674194644581832} \\
        $\dashint_{\pOmega}(\etabar-1)^2$ & \num{9.552494395780908e-29} & \num{0.16666666666666408} & \num{0.42800423105068935} & \num{0.9995555555555328} \\
    \end{tabular}
\end{table}

We solve the corresponding RHE$_a$ for each $\etabar$ and domain using DOLFINx for various $\Biot/\gamma$ values, and compare the results with the LCM solution to find the actual lumping error in~\cref{eq:error_triangular_lcm} which we denote here by
$$
E_l \coloneqq \max_{\tnd\in[0,\tf]}|\uavgtilde(\tnd;\kappa,\sigma,\Biot,\etabar(\xnd)) - \uLump(\tnd;\Biot)|
$$
For RHE$_a$, we employ BDF2 time stepping with 2000 steps, integrating up to $\tf = 3\tauL$ where $\tauL$ was defined in~\cref{eq:teq_nondim}. The error $E_l$, together with the error bounds based on $\phi$, $\phi^{\text{ub}}$, and $\phi^{\text{ub,est}}$, is shown in~\cref{fig:homogeneous_eta_results} for two example cases. We observe that the bound based on $\phi$ is asymptotically tight for small $\Biot$. The results are qualitatively the same for the other geometries and $\etabar$ functions, and are therefore not shown here.

\begin{figure}[ht]
    \centering
    \begin{subfigure}[b]{0.48\textwidth}
        \centering
        \includegraphics[width=\textwidth]{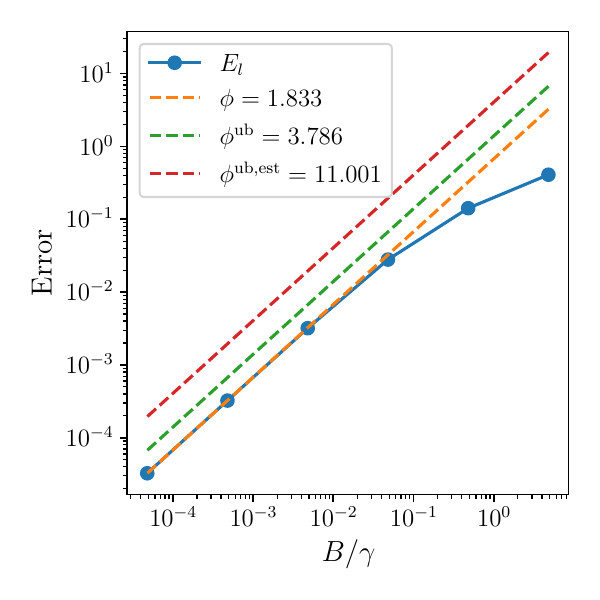}
        \caption{Triangle, linear $\etabar$}
        \label{fig:homogeneous_eta_results_1}
    \end{subfigure}
    \hfill
    \begin{subfigure}[b]{0.48\textwidth}
        \centering
        \includegraphics[width=\textwidth]{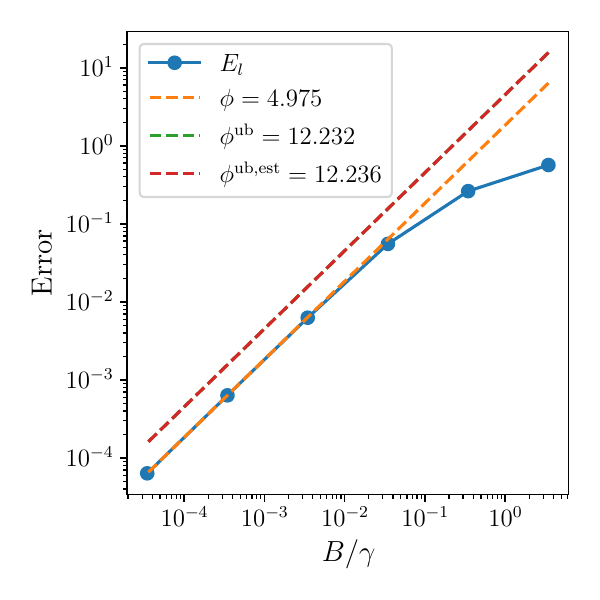}
        \caption{Cross, step $\etabar$}
        \label{fig:homogeneous_eta_results_2}
    \end{subfigure}
    \caption{Comparison of $\phi$, $\phi^{\text{ub}}$, $\phi^{\text{ub,est}}$ error bounds and actual error $E_l$ for two cases. (a)~Triangle with linear $\etabar$. (b)~Cross with step $\etabar$. The solid line is the actual error. For the cross, the green curve ($\phi^{\text{ub}}$) and the red curve ($\phi^{\text{ub,est}}$) overlap, since the variance of the step-function $\etabar$ is close to 1.}
    \label{fig:homogeneous_eta_results}
\end{figure}

\subsubsection{Numerical example for a layered solid}
In the previous examples, we considered homogeneous solids, i.e., $\sigma=\kappa=1$. The error bound with $\phi$ in~\cref{eq:lcm_error_bound} and the upper bound for $\phi$ derived in~\cref{eq:phi_bound_result} also extend to heterogeneous solids. In this section, we present one example with a heterogeneous solid to demonstrate that the error bound in~\cref{eq:lcm_error_bound} remains valid for heterogeneous materials, and that the upper bound in~\cref{eq:phi_bound_result} still bounds $\phi$ from above.

We consider the two-dimensional cross domain shown in~\cref{fig:eta_geometries}, divided along its vertical centerline into two regions of equal thickness, each corresponding to a different material, and apply the step function for $\etabar$ defined in~\cref{eq:step_function}. This setup is identical to that of~\cref{fig:homogeneous_eta_results_2}, except that the solid is now heterogeneous rather than homogeneous. The parameter $\sigma$ is assumed piecewise constant, taking values $\sigma_0$ and $2\sigma_0$ in the two layers. Imposing $\dashint_{\pOmega}\sigma=1$ gives $\sigma_0=2/3$, yielding a variance of $\dashint_{\pOmega}(\sigma-1)^2=1/9$. Note here that to compute the variance, the detailed spatial distribution of $\sigma$ is not required, and the variance can be computed with~\cref{eq:var_sigma_composite}. We solve the RHE$_a$ problem, using the same settings as before, for various $\Biot/\gamma$ values, and compute the actual error $E_l$ between the RHE$_a$ and LCM solutions. We plot the results in~\cref{fig:heterogeneous_eta_results}, where we also include the error bounds based on $\phi$, $\phi^{\text{ub}}$, and $\phi^{\text{ub,est}}$. We observe that the error bound based on $\phi$ is again asymptotically tight for small $\Biot$, and that both $\phi^{\text{ub}}$ and $\phi^{\text{ub,est}}$ bound $\phi$ from above.
\begin{figure}[ht]
    \centering
    \includegraphics[width=0.5\textwidth]{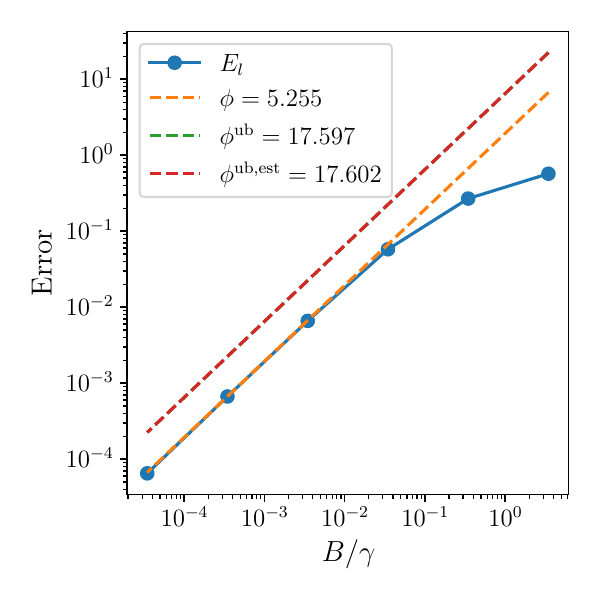}
    \caption{Comparison of $\phi$, $\phi^{\text{ub}}$, $\phi^{\text{ub,est}}$ and actual error $E_l$ for the cross specimen with two layers. The green curve ($\phi^{\text{ub}}$) and the red curve ($\phi^{\text{ub,est}}$) overlap, since the variance of the step-function $\etabar$ is close to 1.}
    \label{fig:heterogeneous_eta_results}
\end{figure}

Comparing~\cref{fig:heterogeneous_eta_results} with~\cref{fig:homogeneous_eta_results_2}, we observe that the actual error $E_l$, as well as $\phi$, $\phi^{\text{ub}}$, and $\phi^{\text{ub,est}}$ all increase for the layered material. This is expected, since $\dashint_{\Omega}(\sigma-1)^2=1/9$ here, whereas it was zero in~\cref{fig:homogeneous_eta_results_2}.

\subsection{Error Due to Biot Approximation}\label{subsec:lcm_biot}
In this section, we analyze the error arising from the Biot approximation in~\cref{eq:error_triangular_biot}, namely
$$\max_{\tnd\in[0,\tf]} \big| \uLump(\tnd;\Biot) - \uLump(\tnd;\Biapprox) \big|.$$
In general, we do not know the true Biot number $\Biot$ and must rely on an estimate $\Biapprox$.
Applying the exponential inequality in~\cref{eq:exp_inequality} to our setting gives
\begin{align}
    \max_{\tnd\in[0,\tf]} \big|\uLump(\tnd;\Biot) - \uLump(\tnd;\Biapprox)\big|
        &= \max_{\tnd\in[0,\tf]} \Big| \exp(-\Biot \gamma \tnd) - \exp(-\Biapprox \gamma \tnd) \Big| \nonumber\\
        &= \max_{\tnd\in[0,\tf]} \Big| \exp\!\left(-\Biapprox \gamma \tnd \left(1 - \frac{\Biot-\Biapprox}{\Biapprox}\right)\right) - \exp(-\Biapprox \gamma \tnd) \Big| \nonumber\\
        &\leq \frac{|\Biot - \Biapprox|}{\Biapprox \exp(1)} + \mathcal{O}\!\left(\left(\frac{\Biot - \Biapprox}{\Biapprox}\right)^2\right), \label{eq:biot_error_bound}
\end{align}
where we set $z=\Biapprox\gamma \tnd$ and $\epsilon=(\Biot-\Biapprox)/\Biapprox$. Hence, the error induced by the Biot approximation scales linearly with the relative error in the Biot number, $|\Biot - \Biapprox|/\Biapprox$.

\subsection{Example Revisited: Cross-Flow Around a Circular Cylinder}
We conclude this chapter by revisiting the cross-flow around a circular cylinder from~\cref{subsec:example_cht} one last time, now evaluating all error sources in~\cref{eq:error_triangular_time,eq:error_triangular_lcm,eq:error_triangular_biot}:
\begin{itemize}
    \item The \emph{temporal error} in~\cref{eq:error_triangular_time} has already been assessed in~\cref{subsubsec:time_numerics} as $E_t=\num{0.010307143115081074}$, see~\cref{fig:time_r1_study_err}.
    \item Using ISO, we obtained a Biot number of $0.0678$, compared to the actual value $0.0680$. The \emph{Biot approximation error} in~\cref{eq:biot_error_bound} is bounded by $$\frac{|\Biot - \Biapprox|}{\Biapprox \exp(1)} = \num{0.00108199835}.$$
    \item For the \emph{lumping error}, we use the time-averaged variation function $\etabar$ from the CHT solution and compute $\phi=\num{1.1053025668470475}$ with~\cref{eq:sensitivity_problem,eq:phi_definition}. The variance in this case is $\dashint_{\pOmega}(\etabar-1)^2 = 0.316$ (much smaller than 1), which gives the upper bound $\phi^{\text{ub}}=\num{2.256650480672011}$ with~\cref{eq:phi_bound_result}. The values of $\phi(1,1,1)$ and $\gamma/\Lambda$ have been provided for the disk (which is equivalent to the cylinder in 2D) in~\cref{tab:geometry_constants}. Evaluating our error bound in~\cref{eq:lcm_error_bound} with $\phi$ and $\phi^{\text{ub}}$ for the given $\Biot=0.0680$ gives \num{0.0069125} and \num{0.01411296}. We also solve the RHE$_a$ and find the actual error $E_l=\num{0.006823082401541114}$. Thus, the bound with $\phi$ is very tight, while the bound with $\phi^{\text{ub}}$ is about twice as large.

\end{itemize}
We can also evaluate the heuristic criterion, used in engineering practice and defined in~\cref{eq:biot_small}, to assess the validity of the LCM. We obtain
\begin{align*}
\frac{B}{\gamma} = 0.01695 < 0.1,
\end{align*}
with $\gamma = 4$ from~\cref{eq:gamma_cylD}, and, hence, the LCM would be deemed valid in this case. The criterion works well here, since the value of $\phi$ is small for this problem.

In this example, the total error is dominated by the temporal approximation error. This contribution is typically ignored in engineering practice, where time scale separation is assumed ($r_1$ does not appear in ISO nor LCM). However, in~\cref{subsubsec:time_numerics}, we showed numerically that the temporal errors can become very large even in small Biot cases. 

If $r_1$ is sufficiently small, then the temporal approximation error becomes negligible. Similarly, if the Biot number can be accurately estimated, the Biot approximation error becomes small. The dominant contribution would then be the lumping error, which can be bounded \emph{a priori} with~\cref{eq:lcm_error_bound,eq:phi_bound_result}, provided that good estimates of $\dashint_{\pOmega}(\etabar-1)^2$ and $\dashint_{\Omega}(\sigma-1)^2$ are available.

\section{About the Nusselt Number}\label{sec:nusselt}
A key practical challenge in applying the LCM is finding an approximate Biot number $\Biapprox$ that closely matches the true Biot number $\Bcht$ for a given problem. Traditionally, this is achieved by estimating the averaged Nusselt number $\Nustavgcht$ through \emph{empirical correlations}, from which $\Bcht$ is computed via~\cref{eq:cht_biot}. Note that, as mentioned in the introduction, we use the terms correlation and empirical correlation interchangeably, without any statistical interpretation.

In~\cref{subsec:iso}, we have introduced the ISO formulation and the corresponding Nusselt number $\Nustavgiso$, and demonstrated in a representative case—characterized by large time scale separation and a small Biot number—that $\Nustavgiso$ can accurately approximate $\Nustavgcht$. While the CHT Nusselt $\Nustavgcht$ generally depends on $r_1$, $r_2$, $\kappa$, $\sigma$, as well as $\Reynolds$ and $\Prandtl$, the ISO Nusselt number $\Nustavgiso$ depends only on $\Reynolds$ and $\Prandtl$.

Compared to the full CHT model, the ISO formulation is computationally more efficient. However, because it still requires solving the Navier-Stokes equations, ISO simulations remain expensive, particularly at high Reynolds numbers. To enable rapid estimation of the Nusselt number, approximations of the ISO-based Nusselt number are employed—so-called empirical correlations. These correlations are typically expressed as power-law functions of the Reynolds and Prandtl numbers and are calibrated for a specific geometry family (e.g., sphere, cylinder, flat plate). 

Historically, empirical correlations have been derived through experimental measurements covering wide ranges of Reynolds and Prandtl numbers. However, due to the high cost and complexity of such experiments, they are often limited to a few canonical geometries. In recent years, numerical simulations have also been used to generate data for correlation development, allowing more complex geometries to be considered. However, simulations are typically restricted to low to intermediate Reynolds numbers, as fully resolving turbulence is prohibitively expensive and turbulence models introduce additional uncertainties. For those reasons, empirical correlations are generally only available for a limited set of geometries and flow conditions.

Realistic geometries encountered in practice often deviate from these canonical shapes for which correlations exist. To still obtain reasonable estimates, a practitioner would first identify the geometry family that best matches the shape of interest and then use the corresponding correlation to estimate the Nusselt number. However, most correlations are only valid for a narrow range of geometries and flow conditions, making it challenging to apply them to complex shapes encountered in real-world applications.

In this section, we propose a framework to systematically extend the validity of existing empirical correlations to much broader classes of geometries by assigning suitably chosen length scales to each geometry, thereby making existing correlations more \emph{universal}. We begin in~\cref{subsec:length_scale_functions} by introducing the concept of \emph{length scale functions}, which formalizes the notion of a characteristic length for general geometries. In~\cref{subsec:nusselt_approximation}, we then review empirical correlations with a few well-known examples. Subsequently, in~\cref{subsubsec:universal_correlations}, we present a data-driven approach for learning length scale functions from data and apply this methodology to a large dataset of spheroids, demonstrating that a single correlation originally defined for a sphere can adequately approximate the Nusselt numbers across the entire spheroid family. Finally, in~\cref{subsec:general_materials}, we contextualize the numerical results by comparing CHT and ISO Nusselt numbers across different parameter regimes and examining realistic solid-fluid material property ratios.

In the following (up to~\cref{subsec:general_materials}), we assume that we are in the regime where the ISO Nusselt number $\Nustavgiso$ closely approximates the CHT Nusselt number $\Nustavgcht$, and therefore focus on ISO. Since we are primarily concerned with the space-time averaged Nusselt number, we denote $\Nusselt \equiv \Nustavgiso$ and, for brevity, refer to it simply as the Nusselt number throughout this section. For the time-dependent, spatially averaged Nusselt number, we will omit the superscript and simply write $\Nuavg(\tnd)$ instead of $\Nuavgiso(\tnd)$. We also note a minor abuse of notation: dimensional quantities such as the coordinates $\dm{x}$, $\dm{y}$, $\dm{z}$, or other dimensional parameters, are occasionally treated as elements of sets like $\mathbb{R}$ or compared using inequalities (e.g., $\dm{\alpha} > 0$) for notational convenience. 

For simplicity, we restrict our attention to flows over the class of convex domains $\mathcal{C}$ that are centered at the origin, symmetric with respect to the coordinate planes, and may be rotated about the $z$-axis. Since the flow direction is taken along the $x$-axis, such a rotation can equivalently be interpreted as a flow at an angle of attack within the $x$-$y$ plane on a non-rotated body.

\subsection{Length Scale Functions}\label{subsec:length_scale_functions}
We introduce the concept of a \emph{length scale function}, denoted by $\dm{\mathcal{D}}$, which maps a domain $\Omegadim\in\mathcal{C}$ to a characteristic length $\elldim = \dm{\mathcal{D}}(\Omegadim)$. The resulting $\elldim$ is a positive real number with the physical dimension of length. The function $\dm{\mathcal{D}}$ is required to satisfy the following properties:
\begin{enumerate}
    \item \textbf{Positivity:} 
    \(
        \dm{\mathcal{D}}(\Omegadim) > 0
    \)
    for all admissible domains $\Omegadim$;
    \item \textbf{Invariance:} 
    \(
        \dm{\mathcal{D}}(\mathbf{R}\Omegadim + \mathbf{t}) 
        = \dm{\mathcal{D}}(\Omegadim)
    \)
    for any rigid-body motion with rotation $\mathbf{R}$ and translation $\mathbf{t}$;
    \item \textbf{Linear scaling:} 
    \(
        \dm{\mathcal{D}}(c\,\Omegadim) = c\,\dm{\mathcal{D}}(\Omegadim)
    \)
    for any scalar $c > 0$.
    Here, $c\,\Omegadim$ denotes the isotropically scaled domain
    \(
        c\,\Omegadim 
        \coloneqq 
        \{\, \xdim_{\mathrm{c}} + c(\xdim - \xdim_{\mathrm{c}}) 
        \mid \xdim \in \Omegadim \,\},
    \)
    where $\xdim_{\mathrm{c}}$ is the centroid of $\Omegadim$. 
\end{enumerate}
There are many common examples for length scale functions. A few examples are provided here:
\begin{itemize}
    \item The diameter function $\dm{\mathcal{D}}_{\text{diam}}(\Omegadim) \coloneqq \max_{\xdim,\dm{\bm{y}}\in\Omegadim} |\xdim - \dm{\bm{y}}|$;
    \item The equivalent spherical diameter $\dm{\mathcal{D}}_{\text{sphere}}(\Omegadim) \coloneqq \left(6|\Omegadim|/\pi\right)^{1/3}$;
    \item The square root of the surface area $\dm{\mathcal{D}}_{\text{area}}(\Omegadim) \coloneqq \sqrt{|\pOmegadim|}$;
    \item The volume to surface area ratio $\dm{\mathcal{D}}_{\text{vs}}(\Omegadim) \coloneqq |\Omegadim| / |\pOmegadim| = \Elldim$.
\end{itemize}
Each length scale function may be more or less suitable for a given problem, depending on the geometry and flow conditions. For convection problems, the characteristic length should be related to the length of the flow path over the body. While this choice is straightforward for simple geometries such as spheres and cylinders (e.g., the diameter), it becomes nontrivial for complex or irregular shapes.

\subsection{Empirical Correlations}\label{subsec:nusselt_approximation}
As mentioned above, practitioners typically estimate the Biot number using empirical correlations, which approximate the ISO Nusselt number. They are generally defined for a parameterized family of geometries, as we now describe.

We consider convex parameterized geometries denoted by 
$\Omegadim(\dm{\bm{\alpha}}, \bm{\beta}) \in \mathcal{C}$, 
where 
$\dm{\bm{\alpha}} = (\dm{\alpha}_1, \dm{\alpha}_2, \ldots, \dm{\alpha}_M)$ 
represents a vector of dimensional geometric parameters 
(e.g., lengths, radii, or thicknesses), 
and 
$\bm{\beta} = (\beta_1, \beta_2, \ldots, \beta_N)$ 
collects other non-dimensional parameters 
(e.g., rotation angles, aspect ratios). 
Varying $(\dm{\bm{\alpha}}, \bm{\beta})$ defines a family of geometries
\[
    \mathcal{S} = 
    \{\, \Omegadim(\dm{\bm{\alpha}}, \bm{\beta}) 
    \mid 
    (\dm{\bm{\alpha}}, \bm{\beta}) \in \mathcal{A} 
    \subset 
    \mathbb{R}^M \times \mathbb{R}^N 
    \,\} \subset \mathcal{C},
\]
where $\mathcal{A}$ denotes the admissible parameter space. 
A specific geometry corresponding to a particular parameter set 
$(\dm{\bm{\alpha}}^{(i)}, \bm{\beta}^{(i)})$ 
is written as 
$\Omegadim^{(i)} = \Omegadim(\dm{\bm{\alpha}}^{(i)}, \bm{\beta}^{(i)})$.

The Nusselt number for the family of geometries $\mathcal{S}$ depends on the flow properties, as well as on the geometric parameters. From dimensional analysis, it follows that, for a chosen length scale function $\dm{\mathcal{D}}$ and a given $\Omegadim \in \mathcal{S}$, the Nusselt number can be expressed as a function of the non-dimensional arguments
\begin{align}
  \Nusselt[\dm{\mathcal{D}}(\Omegadim)]
  = 
  f^{\mathcal{S}}[\dm{\mathcal{D}}(\Omegadim)]
  \!\left(
    \Reynolds[\dm{\mathcal{D}}(\Omegadim)],\, 
    \Prandtl,\,
    \underbrace{\frac{\dm{\alpha}_2}{\dm{\alpha}_1},\, \ldots,\, 
    \frac{\dm{\alpha}_M}{\dm{\alpha}_1}}_{\text{$M-1$ terms}},\, 
    \beta_1,\, \ldots,\, \beta_N
  \right), \label{eq:nusselt_functional}
\end{align}
where the argument list follows directly from the Buckingham~$\pi$ theorem. Consequently, only $M-1$ independent ratios of the dimensional parameters appear in the expression, obtained by normalizing all dimensional quantities with $\dm{\alpha}_1$. In~\cref{eq:nusselt_functional}, we make explicit the dependence on the characteristic length $\dm{\mathcal{D}}(\Omegadim)$ in the definitions of the Reynolds and Nusselt numbers with square brackets. The ratios of the dimensional parameters do not depend on the length scale.

Empirical correlations approximate this functional dependence through an empirical function $\undertilde{f}^{\mathcal{S}}[\dm{\mathcal{D}}(\Omegadim)]$:
\begin{align}
    \Nusselt[\dm{\mathcal{D}}(\Omegadim)]
    \approx 
    \undertilde{\Nusselt}[\dm{\mathcal{D}}(\Omegadim)]
    \coloneqq 
    \undertilde{f}^{\mathcal{S}}[\dm{\mathcal{D}}(\Omegadim)]\!\left(
      \Reynolds[\dm{\mathcal{D}}(\Omegadim)],\, 
      \Prandtl,\,
      \underbrace{\frac{\dm{\alpha}_2}{\dm{\alpha}_1},\, \ldots,\,
      \frac{\dm{\alpha}_M}{\dm{\alpha}_1}}_{\text{$M-1$ terms}},\, 
      \beta_1,\, \ldots,\, \beta_N
    \right),
    \label{eq:nusselt_correlation}
\end{align}
where the undertilde in $\undertilde{\Nusselt}$ emphasizes that it represents an approximation to the actual Nusselt number. The correlation function $\undertilde{f}^{\mathcal{S}}[\dm{\mathcal{D}}(\Omegadim)]$ is typically obtained by fitting to experimental or computational data and informed by boundary layer theory, and will depend on the choice of length scale function $\dm{\mathcal{D}}$.

The construction and use of empirical correlations naturally follow an offline-online decomposition. The offline phase involves the expensive process of generating data---either experimentally or through high-fidelity simulations---and subsequently fitting the correlation function $\undertilde{f}^{\mathcal{S}}[\dm{\mathcal{D}}(\Omegadim)]$. The online phase then consists of rapidly evaluating this fitted correlation for new parameter values to estimate the Nusselt number.

Empirical correlations have been developed for many canonical configurations, including flow over flat plates, cylinders, and spheres. Their accuracy typically lies within a relative error of 10--20\%, depending on the specific correlation and flow regime~\cite{incropera1990fundamentals}. The limited accuracy arises from several factors, including measurement uncertainties, surface roughness, variations in flow conditions—where even small perturbations in the incoming flow can lead to markedly different behaviors in turbulent regimes—and the restricted datasets used for fitting. Moreover, each correlation is valid only within specific ranges of Reynolds and Prandtl numbers, and geometrical parameters. Nevertheless, these correlations often provide sufficiently accurate results for an engineering estimate at a low computational cost.

In the following, we briefly review a few well-known experimental correlations, and discuss current computational approaches.

\subsubsection{Experimental correlations}\label{subsubsec:experimental_correlations}
Due to the complexity and high computational costs of simulating the Navier-Stokes equations, experiments have historically been the primary source of data. Many groups have studied for example flat plates, cylinders and spheres in cross-flow, and developed empirical correlations for the Nusselt number based on experimental measurements. A comprehensive review of such correlations can be found in~\cite{incropera1990fundamentals}. We provide here a few well-known correlations for the flow over a flat plate, circular cylinder and a sphere.

\paragraph{Flat plate aligned with the flow}
For the specific case of an infinitely wide flat plate aligned with the flow, 
the geometry family can be expressed as
\[
    \mathcal{S}_{\mathrm{fp}}
    =
    \{\, \Omegadim_{\mathrm{fp}}(\dm{L}) 
    \mid \dm{L} > 0 \,\},
\]
where each member corresponds to an idealized plate of length $\dm{L}$ in the 
streamwise ($\dm{x}$) direction, zero thickness in the $\dm{y}$-direction, 
and infinite width in the $\dm{z}$-direction,
\[
    \Omegadim_{\mathrm{fp}}(\dm{L})
    = [-\dm{L}/2,\, \dm{L}/2] \times \{0\} \times \mathbb{R}.
\]
The corresponding length scale function
\(
    \dm{\mathcal{D}}_{\mathrm{fp}}
\)
is defined as
\[
    \dm{\mathcal{D}}_{\mathrm{fp}}(\Omegadim_{\mathrm{fp}})
    =
    \max_{(\dm{x},\,\dm{y},\,\dm{z}) \in \Omegadim_{\mathrm{fp}}} \dm{x}
    - 
    \min_{(\dm{x},\,\dm{y},\,\dm{z}) \in \Omegadim_{\mathrm{fp}}} \dm{x},
\]
which extracts the total extent of the domain in the streamwise direction.  
For any $\Omegadim_{\mathrm{fp}}(\dm{L}) \in \mathcal{S}_{\mathrm{fp}}$, this yields 
$\dm{\mathcal{D}}_{\mathrm{fp}}(\Omegadim_{\mathrm{fp}}(\dm{L})) = \dm{L}$.  
The corresponding empirical correlation (for $M=1$, $N=0$) then depends only on the Reynolds and Prandtl numbers (compare~\cref{eq:nusselt_correlation}):
\begin{align*}
    \undertilde{\Nusselt}[\dm{L}]
    =
    \undertilde{f}^{\mathcal{S}_{\mathrm{fp}}}[\dm{L}]
      \bigl(\Reynolds[\dm{L}],\, \Prandtl\bigr).
\end{align*}
Assuming laminar flows ($\Reynolds[\dm{L}] \leq 1 \times 10^5$), the Nusselt number for this configuration can be computed in closed form with boundary layer theory~\cite{schlichting2016boundary} as
\begin{align*}
    \undertilde{\Nusselt}[\dm{L}] = \undertilde{f}^{\mathcal{S}_{\mathrm{fp}}}_{\mathrm{lam}}[\dm{L}](\Reynolds[\dm{L}],\Prandtl) = 0.664 \Reynolds[\dm{L}]^{1/2} \Prandtl^{1/3},
\end{align*}
which is valid for $\Prandtl\geq0.6$, and where we introduce the subscript ``lam'' for laminar. The exponents of $1/2$ on the Reynolds number and $1/3$ on the Prandtl number are characteristic of laminar boundary layers, in which the boundary layer thickness grows proportionally to $\sqrt{x}$ with the distance $x$ from the leading edge~\cite{schlichting2016boundary}. 
For turbulent flows, a correlation for the averaged Nusselt number (neglecting the transitional region) based on experimental data is given by~\cite{ahtt6e}
\begin{align*}
    \undertilde{\Nusselt}[\dm{L}] = \undertilde{f}^{\mathcal{S}_\mathrm{fp}}_{\mathrm{turb}}[\dm{L}](\Reynolds[\dm{L}],\Prandtl)
    = 0.664 \, \Reynolds_{\mathrm{tr}}^{1/2} \Prandtl^{1/3} 
    + 0.037 \left( \Reynolds[\dm{L}]^{4/5} - \Reynolds_{\mathrm{tr}}^{4/5} \right) \Prandtl^{0.6},
\end{align*}
where $\Reynolds_{\mathrm{tr}}$ denotes the transitional Reynolds number, typically of order $\mathcal{O}(10^5)$, and the subscript ``turb'' stands for turbulent.
In the turbulent regime, the boundary layer growth follows a different scaling, and the corresponding Reynolds exponent in the correlation changes to 4/5.

\paragraph{Circular cylinder in cross-flow}
For the case of a circular cylinder of diameter $\dm{D}$ exposed to a uniform cross-flow, 
the geometry family can be defined as
\[
    \mathcal{S}_{\mathrm{cyl}}
    =
    \{\, \Omegadim_{\mathrm{cyl}}(\dm{D})
    \mid \dm{D} > 0 \,\},
\]
where each member represents an idealized cylinder of infinite length in the 
$z$-direction and diameter $\dm{D}$ in the $x-y$ plane,
\[
    \Omegadim_{\mathrm{cyl}}(\dm{D})
    =
    \{\, (\dm{x}, \dm{y}, \dm{z}) \in \mathbb{R}^3 
    \mid \dm{x}^2 + \dm{y}^2 \le (\dm{D}/2)^2,\ \dm{z} \in \mathbb{R} \,\}.
\]
We define the corresponding length scale function 
\(
    \dm{\mathcal{D}}_{\mathrm{cyl}}
\)
as
\[
    \dm{\mathcal{D}}_{\mathrm{cyl}}(\Omegadim_{\mathrm{cyl}})
    = 2\max_{(\dm{x},\,\dm{y},\,\dm{z}) \in \Omegadim_{\mathrm{cyl}}} 
        \sqrt{\dm{x}^2 + \dm{y}^2}
\]
which extracts the diameter of the domain in the cross-flow direction.
For any $\Omegadim_{\mathrm{cyl}}(\dm{D}) \in \mathcal{S}_{\mathrm{cyl}}$, this gives 
$\dm{\mathcal{D}}_{\mathrm{cyl}}(\Omegadim_{\mathrm{cyl}}(\dm{D})) = \dm{D}$.  
The corresponding empirical correlation (for $M=1$, $N=0$) then depends only on the Reynolds and Prandtl numbers (compare~\cref{eq:nusselt_correlation}):
\begin{align*}
    \undertilde{\Nusselt}[\dm{D}]
    =
    \undertilde{f}^{\mathcal{S}_{\mathrm{cyl}}}[\dm{D}]
      \bigl(\Reynolds[\dm{D}],\, \Prandtl\bigr).
\end{align*}
A widely used expression for this configuration is the Churchill-Bernstein correlation~\cite{churchill1977correlating}, given by
\begin{align*}
    \undertilde{\Nusselt}[\dm{D}] 
    = \undertilde{f}^{\mathcal{S}_{\mathrm{cyl}}}_{\mathrm{cb}}[\dm{D}](\Reynolds[\dm{D}], \Prandtl) 
    =
    0.3 
    + \frac{0.62\, \Reynolds[\dm{D}]^{1/2} \Prandtl^{1/3}}
    {\left[1 + (0.4/\Prandtl)^{2/3}\right]^{1/4}}
    \left[1 + \left(\frac{\Reynolds[\dm{D}]}{282000}\right)^{5/8}\right]^{4/5},
\end{align*}
valid for $\Reynolds[\dm{D}] < 10^7$ and $0.7 < \Prandtl < 500$. As many other correlations exist for cylinders, we use the subscript ``cb'' to specify the Churchill-Bernstein correlation. The bracketed terms in the expression become significant primarily at high Prandtl and Reynolds numbers. Note again the $1/2$ exponent on the Reynolds number and the $1/3$ exponent on the Prandtl number—these arise because the local flow around the cylinder resembles that over a flat plate. Furthermore, the $4/5$ Reynolds exponent reappears in the turbulent regime, reflecting the same scaling behavior as observed for the flat plate.

The data underlying the Churchill-Bernstein correlation were experimentally obtained for finite-length cylinders with large length-to-diameter ratios, but the correlation is commonly assumed to remain valid for the limiting case of an infinite cylinder. There also exist other correlations specifically developed for finite-length cylinders that account for end effects, see, e.g.,~\cite{incropera1990fundamentals}.

\paragraph{Sphere in uniform flow}
For the case of a sphere of diameter $\dm{D}$, 
the geometry family can be defined as
\[
    \mathcal{S}_{\mathrm{sph}}
    =
    \{\, \Omegadim_{\mathrm{sph}}(\dm{D})
    \mid \dm{D} > 0 \,\},
\]
where each member represents a sphere of diameter $\dm{D}$,
\[
    \Omegadim_{\mathrm{sph}}(\dm{D})
    =
    \{\, (\dm{x}, \dm{y}, \dm{z}) \in \mathbb{R}^3 
    \mid \dm{x}^2 + \dm{y}^2 + \dm{z}^2 \le (\dm{D}/2)^2 \,\}.
\]
We define the corresponding length scale function 
\(
    \dm{\mathcal{D}}_{\mathrm{sph}}
\)
as
\[
    \dm{\mathcal{D}}_{\mathrm{sph}}
    =
    \dm{\mathcal{D}}_{\mathrm{diam}},
\]
where $\dm{\mathcal{D}}_{\mathrm{diam}}$ denotes the diameter function introduced earlier in~\cref{subsec:length_scale_functions}, which extracts the diameter of the domain.  
For any $\Omegadim_{\mathrm{sph}}(\dm{D}) \in \mathcal{S}_{\mathrm{sph}}$, this yields 
$\dm{\mathcal{D}}_{\mathrm{sph}}(\Omegadim_{\mathrm{sph}}(\dm{D})) = \dm{D}$.
Since the geometry is fully characterized by a single dimensional parameter ($M=1$, $N=0$), the corresponding empirical correlation again only depends on the Reynolds and Prandtl numbers,
\begin{align*}
    \undertilde{\Nusselt}[\dm{D}]
    =
    \undertilde{f}^{\mathcal{S}_{\mathrm{sph}}}[\dm{D}]
      \bigl(\Reynolds[\dm{D}],\, \Prandtl\bigr).
\end{align*}
A well-known example for this case is the Ranz-Marshall correlation~\cite{ranz1952evaporation},
\begin{align}
    \undertilde{\Nusselt}[\dm{D}] 
    = \undertilde{f}_{\mathrm{rm}}^{\mathcal{S}_{\mathrm{sph}}}[\dm{D}]
      \bigl(\Reynolds[\dm{D}], \Prandtl\bigr)
    = 2 + 0.6\, \Reynolds[\dm{D}]^{1/2} \Prandtl^{1/3}.
    \label{eq:ranz_marshall}
\end{align}
which was originally developed for small droplets at low Reynolds numbers ($\Reynolds[\dm{D}] \leq 200$).  
Nevertheless, it provides reasonable estimates (within 10\% error) even up to $\Reynolds[\dm{D}] \approx 10^4$~\cite{hirasawa2012numerical}, as confirmed also in our own simulations; see~\cref{subsubsec:proof_of_concept_sphere}.

Here we again observe the $1/2$ exponent on the Reynolds number and the $1/3$ exponent on the Prandtl number, consistent with the boundary layer theory for a flat plate. The constant term of 2 arises from the purely diffusive contribution to the Nusselt number, which can be derived analytically for a sphere at zero Reynolds number~\cite{incropera1990fundamentals}.

\subsubsection{Computational correlations}
The experimental correlations presented in~\cref{subsubsec:experimental_correlations} are all defined for geometry families with only a single dimensional parameter (i.e., $M = 1$) and no non-dimensional parameters (i.e., $N = 0$). 
Performing experiments over broad ranges of geometries and flow conditions is often prohibitively expensive or even infeasible. 
Consequently, empirical correlations for more complex shapes with multiple dimensional or non-dimensional parameters remain scarce.

Recent advances in computational power and numerical methods have made it increasingly practical to compute Nusselt numbers for complex geometries and flow conditions using computational fluid dynamics (CFD). 
Such simulations enable systematic studies of configurations that are difficult to realize experimentally and can support the development of correlations that incorporate additional geometric descriptors. 
However, CFD-based studies are typically restricted to laminar or weakly turbulent regimes, as fully resolving turbulence remains computationally prohibitive. 
Turbulence models, while widely used, introduce additional uncertainties whose quantitative accuracy is often difficult to assess.

It is therefore crucial to reuse and systematically extend existing correlations wherever possible. The following sections develop a framework to achieve precisely this: by selecting a suitable length scale function, we extend the validity of existing correlations to a much broader class of geometries.

\subsection{Towards Universal Correlations}\label{subsubsec:universal_correlations}

In~\cite{lienhard1973commonality}, it was observed that many empirical correlations for different geometry families in natural convection share a common functional form, differing primarily in their choice of characteristic length scale. 
This observation suggests that by identifying a suitable length scale function, one could formulate a single, universal correlation applicable to a broad class of geometries. 
This also implies that correlations developed for simple canonical shapes could, in principle, be applied to more complex geometries through an appropriate length scale function.

We propose here a way to find such length scale functions in forced convection. The key concept is to utilize a length scale transformation and automatically learn the length scale function from data, obtained either experimentally or from numerical simulations.

\subsubsection{Transformation of length scales}\label{subsubsec:length_scale_transform}
Given a geometry family $\mathcal{S}\subset\mathcal{C}$ and domain $\Omegadim\in\mathcal{S}$, let $\dm{\mathcal{D}}_1$ and $\dm{\mathcal{D}}_2$ denote two length scale functions that map $\Omegadim$ to a length scale. A Reynolds number defined with respect to $\dm{\mathcal{D}}_1(\Omegadim)$, denoted $\Reynolds[\dm{\mathcal{D}}_1(\Omegadim)]$, can be transformed into a corresponding Reynolds number based on $\dm{\mathcal{D}}_2(\Omegadim)$ as
\begin{align}
    \Reynolds[\dm{\mathcal{D}}_2(\Omegadim)] 
    = \frac{\vinfdim \dm{\mathcal{D}}_2(\Omegadim)}{\nufdim}
    = \frac{\vinfdim \dm{\mathcal{D}}_2(\Omegadim)}{\nufdim} 
      \frac{\dm{\mathcal{D}}_1(\Omegadim)}{\dm{\mathcal{D}}_1(\Omegadim)}
    = q^{-1} \Reynolds[\dm{\mathcal{D}}_1(\Omegadim)],
    \label{eq:reynolds_transform}
\end{align}
where
\begin{align}
    q \coloneqq \frac{\dm{\mathcal{D}}_1(\Omegadim)}{\dm{\mathcal{D}}_2(\Omegadim)}
    \label{eq:length_scale_ratio}
\end{align}
is the ratio of the two length scales. As $q$ is a ratio of two length scales, and length scale functions scale linearly with the domain, $q$ is dimensionless and independent of $\Omegadim$. 

Analogously, an (approximate) Nusselt number defined with $\dm{\mathcal{D}}_1(\Omegadim)$ can be transformed into a Nusselt number defined with $\dm{\mathcal{D}}_2(\Omegadim)$:
\begin{align}
    \Nusselt[\dm{\mathcal{D}}_2(\Omegadim)] = q^{-1} \Nusselt[\dm{\mathcal{D}}_1(\Omegadim)]. \label{eq:nusselt_transform}
\end{align}
We can now use~\cref{eq:reynolds_transform,eq:length_scale_ratio,eq:nusselt_transform} 
to transform the input and output of a given correlation function 
$\undertilde{f}^{\mathcal{S}}[\dm{\mathcal{D}}_1(\Omegadim)](\Reynolds[\dm{\mathcal{D}}_1(\Omegadim)],\Prandtl,\ldots)$ 
from the characteristic length scale $\dm{\mathcal{D}}_1(\Omegadim)$ to $\dm{\mathcal{D}}_2(\Omegadim)$:
\begin{align}
    \undertilde{\Nusselt}[\dm{\mathcal{D}}_2(\Omegadim)]
    &=
    q^{-1}\,
    \undertilde{f}^{\mathcal{S}}[\dm{\mathcal{D}}_1(\Omegadim)]
    \!\left(
        q\,\Reynolds[\dm{\mathcal{D}}_2(\Omegadim)],\,
        \Prandtl,\,
        \frac{\dm{\alpha}_2}{\dm{\alpha}_1},\,\ldots,\,
        \frac{\dm{\alpha}_M}{\dm{\alpha}_1},\,
        \beta_1,\,\ldots,\,\beta_N
    \right).
    \label{eq:correlation_transform}
\end{align}
This transformation provides a direct means to convert correlations between different length scales. It serves as the key ingredient for learning suitable length scale functions from data to extend existing correlations, as we will discuss next.

\subsubsection{General idea to learn length scale functions}
Assume that we wish to estimate the Nusselt number for a geometry family 
$\mathcal{S} \subset \mathcal{C}$ 
characterized by $M$ dimensional parameters 
$\dm{\alpha}_1, \ldots, \dm{\alpha}_M$ 
and $N$ dimensionless parameters 
$\beta_1, \ldots, \beta_N$.
From~\cref{eq:nusselt_functional}, the Nusselt number for a given 
$\Omegadim \in \mathcal{S}$ can be expressed as
\begin{align}
    \Nusselt[\dm{\mathcal{D}}_2(\Omegadim)] 
    = f^{\mathcal{S}}[\dm{\mathcal{D}}_2(\Omegadim)]\!\left(
        \Reynolds[\dm{\mathcal{D}}_2(\Omegadim)],\,
        \Prandtl,\,
        \frac{\dm{\alpha}_2}{\dm{\alpha}_1}, \ldots, 
        \frac{\dm{\alpha}_M}{\dm{\alpha}_1},\,
        \beta_1, \ldots, \beta_N
    \right), \label{eq:nusselt_functional_S2}
\end{align}
where $\dm{\mathcal{D}}_2$ denotes a chosen length scale function.

Suppose now that we are given an ansatz correlation 
$\undertilde{f}^{\mathcal{S}^1}[\dm{\mathcal{D}}^{\mathrm{ref}}_1(\medbullet)](\Reynolds[\dm{\mathcal{D}}^{\mathrm{ref}}_1(\medbullet)],\Prandtl)$
that has been fitted for a geometry family 
$\mathcal{S}^1 \subset \mathcal{C}$ 
with a single geometrical parameter $\dm{\alpha}^1_1$ 
and a corresponding length scale function $\dm{\mathcal{D}}^{\mathrm{ref}}_1$,
such as those presented in~\cref{subsubsec:experimental_correlations} for simple canonical geometries. We use the superscript “ref'' to indicate that this length scale function was chosen when constructing the correlation for the geometry family $\mathcal{S}^1$. 

We would now like to reuse this correlation for the geometry family $\mathcal{S}$.
However, directly applying 
$\undertilde{f}^{\mathcal{S}^1}[\dm{\mathcal{D}}^{\mathrm{ref}}_1(\Omegadim)](\Reynolds[\dm{\mathcal{D}}^{\mathrm{ref}}_1(\Omegadim)],\Prandtl)$ 
with $\Omegadim \in \mathcal{S}$ 
will, in general, not yield accurate results, 
as the two families $\mathcal{S}^1$ and $\mathcal{S}$ 
may exhibit substantially different flow behaviors.
Instead, we seek to \emph{learn} a suitable length scale function 
$\dm{\mathcal{D}}_1$ such that the correlation 
$\undertilde{f}^{\mathcal{S}^1}[\dm{\mathcal{D}}_1(\Omegadim)](\Reynolds[\dm{\mathcal{D}}_1(\Omegadim)],\Prandtl)$
becomes applicable to $\Omegadim\in\mathcal{S}$ as well.

For $\Omegadim \in \mathcal{S}$, we express the \emph{unknown} length scale function $\dm{\mathcal{D}}_1$ with~\cref{eq:length_scale_ratio} in terms of a \emph{known} length scale function 
$\dm{\mathcal{D}}_2$ as
\begin{align*}
    \dm{\mathcal{D}}_1(\Omegadim) 
    \coloneqq q \, \dm{\mathcal{D}}_2(\Omegadim),
\end{align*} 
where $q$ is a function to be determined from data.
To identify $q$, we first employ the transformation introduced in~\cref{eq:correlation_transform}
to express the ansatz correlation 
$\undertilde{f}^{\mathcal{S}^1}[\dm{\mathcal{D}}_1(\Omegadim)](\Reynolds[\dm{\mathcal{D}}_1(\Omegadim)],\Prandtl)$ 
in terms of the characteristic length $\dm{\mathcal{D}}_2(\Omegadim)$, yielding
\begin{align}
    \undertilde{\Nusselt}[\dm{\mathcal{D}}_2(\Omegadim)]
    =
    q^{-1}\,
    \undertilde{f}^{\mathcal{S}^1}[\dm{\mathcal{D}}_1(\Omegadim)]
    \!\left(
        q\,\Reynolds[\dm{\mathcal{D}}_2(\Omegadim)],\,
        \Prandtl
    \right).
    \label{eq:transformed_correlation}
\end{align}
Using the transformed correlation~\cref{eq:transformed_correlation} to approximate the general Nusselt relation in~\cref{eq:nusselt_functional_S2} yields
\begin{align}
    f^{\mathcal{S}}[\dm{\mathcal{D}}_2(\Omegadim)]\!\left(
        \Reynolds[\dm{\mathcal{D}}_2(\Omegadim)],\, \Prandtl,\,
        \frac{\dm{\alpha}_2}{\dm{\alpha}_1},\, \ldots,\,
        \frac{\dm{\alpha}_M}{\dm{\alpha}_1},\,
        \beta_1,\, \ldots,\, \beta_N
    \right)
    \approx
    q^{-1}\,
    \undertilde{f}^{\mathcal{S}^1}[\dm{\mathcal{D}}_1(\Omegadim)]
    \!\left(
        q\,\Reynolds[\dm{\mathcal{D}}_2(\Omegadim)],\, \Prandtl
    \right).
    \label{eq:universal_correlation_approx}
\end{align}
Comparing the arguments, we observe that the left hand side of~\cref{eq:universal_correlation_approx} 
depends on additional geometric parameters 
$\dm{\alpha}_2/\dm{\alpha}_1, \ldots, \dm{\alpha}_M/\dm{\alpha}_1, 
\beta_1, \ldots, \beta_N$ 
that are not present on the right hand side.
To account for this discrepancy, the length scale ratio \( q \) must depend on these parameters,
\begin{align}
    q\!\left(
        \frac{\dm{\alpha}_2}{\dm{\alpha}_1}, \ldots,
        \frac{\dm{\alpha}_M}{\dm{\alpha}_1},
        \beta_1, \ldots, \beta_N
    \right),
    \label{eq:length_scale_ratio_parametric}
\end{align}
so that their influence on the Nusselt number is effectively absorbed into the definition of \( q \). 
For each geometry $\Omegadim \in \mathcal{S}$, 
the corresponding value of \( q \) automatically determines an appropriate length scale 
$\dm{\mathcal{D}}_1(\Omegadim)=q\,\dm{\mathcal{D}}_2(\Omegadim)$ 
with which to evaluate the correlation 
$\undertilde{f}^{\mathcal{S}^1}[\dm{\mathcal{D}}_1(\Omegadim)](\Reynolds[\dm{\mathcal{D}}_1(\Omegadim)],\, \Prandtl)$.

This represents the key idea of the proposed framework: 
for every \(\Omegadim \in \mathcal{S}\), we seek a parameter-dependent length scale ratio \( q \)---independent of the Reynolds and Prandtl numbers---that transforms a chosen and known length scale $\mathcal{\dm{D}}_2(\Omegadim)$ into a new one, \(\dm{\mathcal{D}}_1(\Omegadim)\),
to be used in the ansatz correlation \(\undertilde{f}^{\mathcal{S}^1}[\dm{\mathcal{D}}_1(\Omegadim)](\Reynolds[\dm{\mathcal{D}}_1(\Omegadim)],\, \Prandtl)\). 
Although the ansatz correlation may have been originally developed for a simple geometry family \(\mathcal{S}^1\) with length scale function \(\dm{\mathcal{D}}^{\mathrm{ref}}_1\), 
the learned length scale function \( \dm{\mathcal{D}}_1 \) extends its applicability to a broader class of geometries \(\mathcal{S}\), 
thereby making the correlation more \emph{universal}.
To determine \( q \), one must generate simulation or experimental data for the Nusselt number 
across the parameter space of $\mathcal{S}$, 
as will be demonstrated in the following section.

\begin{cmt}
    In the discussion above, we considered an ansatz correlation with a single geometric parameter for simplicity. 
    The same principle, however, extends naturally to more general correlations that involve multiple dimensional and non-dimensional parameters. 
    In such cases, the length scale ratio $q$ would be identified jointly with the additional parameters of the chosen ansatz correlation based on available data.
\end{cmt}

\subsubsection{Learning length scale functions from data}
Assume we are given a geometry family $\mathcal{S}$ and a dataset obtained for multiple instances of this family, corresponding to different parameter combinations $(\bm{\dm{\alpha}}, \bm{\beta})$. 
For each geometry, simulations are performed at several Reynolds numbers to compute the corresponding Nusselt numbers, while the Prandtl number $\Prandtl$ is held fixed. The resulting dataset can be represented as tuples
\[
  \bigl(\Omegadim^{(i)},\, \Reynolds^{(i,k)},\, \Nusselt^{(i,k)}\bigr),
  \qquad
  \Omegadim^{(i)} = \Omegadim(\bm{\dm{\alpha}}^{(i)}, \bm{\beta}^{(i)}) \in \mathcal{S},
\]
where
\[
  \Reynolds^{(i,k)} \coloneqq 
  \Reynolds\!\bigl[\dm{\mathcal{D}}_2(\Omegadim^{(i)})\bigr]^{(k)}, 
  \qquad
  \Nusselt^{(i,k)} \coloneqq 
  \Nusselt\!\bigl[\dm{\mathcal{D}}_2(\Omegadim^{(i)})\bigr]^{(k)}.
\]
Here, $\dm{\mathcal{D}}_2$ denotes a \emph{known} length scale function. The indices $i = 1, \ldots, N_\Omega$ enumerate the geometry instances, and $k = 1, \ldots, N_i$ enumerate Reynolds-Nusselt samples for each geometry.

For each tuple $(i,k)$, we solve for $q^{(i,k)}$ as the minimizer of
\begin{align}
    q^{(i,k)} 
    \coloneqq 
    \argmin_q 
    \left[
      \Nusselt^{(i,k)} 
      - q^{-1} 
        \undertilde{f}^{\mathcal{S}^1}[\dm{\mathcal{D}}_1(\Omegadim^{(i)})]\!\left(q\,\Reynolds^{(i,k)},\, \Prandtl\right)
    \right]^2, 
    \label{eq:regression_pointwise}
\end{align}
where $\undertilde{f}^{\mathcal{S}^1}[\dm{\mathcal{D}}_1(\medbullet)](\Reynolds[\dm{\mathcal{D}}_1(\medbullet)],\, \Prandtl)$ is the chosen ansatz correlation associated with the shape family $\mathcal{S}^1$ and \emph{unknown} length scale function $\dm{\mathcal{D}}_1$, which, for simplicity, is assumed to depend on only a single dimensional parameter (i.e., $M = 1$, $N = 0$) here.

Because the data will, in general, not fit the ansatz perfectly, the pointwise values $q^{(i,k)}$ will vary with $k$ (i.e., across different Reynolds numbers), even though, by construction, $q$ should be independent of Reynolds number, see~\cref{eq:length_scale_ratio_parametric}. To obtain a single representative $q^{(i)}$ for geometry $i$, we average the $q^{(i,k)}$ over the sampled Reynolds numbers in log-space using a trapezoidal rule. Sorting $\{\Reynolds^{(i,k)}\}_{k=1}^{N_i}$ in ascending order, we set
\begin{align}
  q^{(i)}
  = \frac{1}{\log \Reynolds^{(i,N_i)} - \log \Reynolds^{(i,1)}}
    \sum_{k=1}^{N_i-1}
    \frac{q^{(i,k)}+q^{(i,k+1)}}{2}\,
    \bigl(\log \Reynolds^{(i,k+1)} - \log \Reynolds^{(i,k)}\bigr),
  \label{eq:q_average_log}
\end{align}
and assign this value to the geometry $\Omegadim^{(i)}$.

Finally, the set 
\(
  \{(\bm{\dm{\alpha}}^{(i)},\bm{\beta}^{(i)},q^{(i)})\}_{i=1}^{N_\Omega}
\)
can be used in a second-stage regression to construct a continuous surrogate
\[
  \tilde{q}\!\left(
    \frac{\dm{\alpha}_2}{\dm{\alpha}_1}, \ldots, 
    \frac{\dm{\alpha}_M}{\dm{\alpha}_1}, 
    \bm{\beta}
  \right)
\]
over the parameter space.  
The learned length scale function is then given by
\begin{align}
  \dm{\mathcal{D}}_1(\Omegadim(\bm{\dm{\alpha}}, \bm{\beta}))
  = 
  \tilde{q}\!\left(
    \frac{\dm{\alpha}_2}{\dm{\alpha}_1}, \ldots, 
    \frac{\dm{\alpha}_M}{\dm{\alpha}_1}, 
    \bm{\beta}
  \right)
  \dm{\mathcal{D}}_2(\Omegadim(\bm{\dm{\alpha}}, \bm{\beta})). \label{eq:learned_length_scale}
\end{align}

\begin{cmt}
    Note, in this work, we consider a fixed Prandtl number $\Prandtl$, but there is nothing preventing us from collecting data for different $\Reynolds$ and $\Prandtl$ numbers and learning $q$ with~\cref{eq:universal_correlation_approx}. We would expect universality, i.e., no strong dependence of $q$ on $\Reynolds$ and $\Prandtl$, for most practical cases with $\Prandtl > 0.6$.
\end{cmt}

\subsubsection{Proof of concept: Sphere}\label{subsubsec:proof_of_concept_sphere}
To demonstrate the approach, we consider the Ranz-Marshall correlation defined in~\cref{eq:ranz_marshall} for a sphere of diameter $\dm{D}$ as the ansatz correlation function $\undertilde{f}_{\mathrm{rm}}^{\mathcal{S}_{\mathrm{sph}}}[\dm{\mathcal{D}}^{\mathrm{ref}}_1(\Omegadim)](\Reynolds[\dm{\mathcal{D}}^{\mathrm{ref}}_1(\Omegadim)],\, \Prandtl)$. Here, $\dm{\mathcal{D}}^{\mathrm{ref}}_1$ corresponds to the diameter function $\dm{\mathcal{D}}_{\mathrm{diam}}$ introduced in~\cref{subsec:length_scale_functions}, which returns the sphere diameter $\dm{D}$ for any spherical domain. We test whether the proposed framework can recover this reference length scale, $\dm{\mathcal{D}}_1(\Omegadim) \approx \dm{\mathcal{D}}^{\mathrm{ref}}_1(\Omegadim)$, given simulation data of a sphere generated using a length scale function $\dm{\mathcal{D}}_2$.

We select the length scale function $\dm{\mathcal{D}}_2 = \dm{\mathcal{D}}_{\mathrm{area}}$, defined in~\cref{subsec:length_scale_functions}, which returns the square root of the surface area of the domain, and we denote the resulting characteristic length scale as $\sqrt{A} \coloneqq \dm{\mathcal{D}}_{\mathrm{area}}(\medbullet)$. Our goal is therefore to test whether
$$
\frac{\dm{\mathcal{D}}_1(\Omegadim)}{\dm{\mathcal{D}}^{\mathrm{ref}}_1(\Omegadim)} = \frac{q \, \sqrt{\dm{A}}}{\dm{D}} \approx 1.
$$

To generate simulation data, we solve the flow over the sphere with the ISO formulation with $\Prandtl = 0.71$ and $\Reynolds[\sqrt{\dm{A}}]$ ranging from $100$ to $10^4$ using Nek5000~\cite{nek5000-web-page}. The simulations employ seventh-order polynomials and a BDF2 time integrator with adaptive time stepping, targeting a CFL number of 0.45. For high Reynolds numbers, a filter is applied to attenuate the highest mode by 5\% and the second-highest by 1.25\%, which is required to ensure long-time stability~\cite{fischer2001filter}. Each case is advanced until a statistically steady state is reached; the algorithm used to detect steady state is described in Appendix~\ref{sec:steady_state_algorithm}. A corresponding mesh convergence study is provided in Appendix~\ref{subsec:convergence_study_sph}.

For each Reynolds number, the spatially averaged Nusselt number is shown in~\cref{fig:nusselt_sphere_time}, where time is normalized by the final time $\tf$ determined from the steady-state criterion. The space-time averaged Nusselt number is indicated by a grey dashed line for each $\Reynolds$. As expected, the Nusselt number increases with increasing Reynolds number and exhibits random oscillations that become more pronounced at higher $\Reynolds$. However, these fluctuations remain small compared to the mean value.

\begin{figure}[ht]
    \centering
    \includegraphics[width=\textwidth]{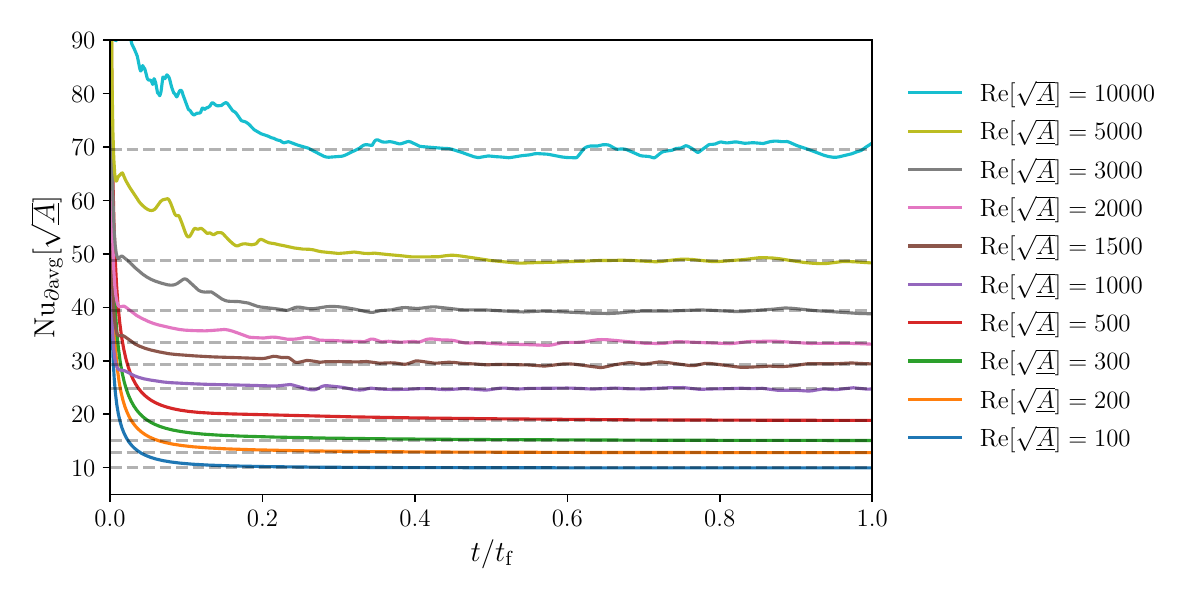}
    \caption{Spatially averaged Nusselt number for a sphere in uniform flow for various Reynolds numbers, using $\sqrt{\dm{A}}$ as length scale. The space-time averaged Nusselt numbers are indicated by the grey dashed lines.}
    \label{fig:nusselt_sphere_time}
\end{figure}

With the space-time averaged Nusselt numbers for each Reynolds number, we first compare them with the Ranz-Marshall correlation (transformed to length scale $\sqrt{\dm{A}}$ with~\cref{eq:correlation_transform}) in~\cref{fig:nu_re_sphere_1}. The correlation slightly overpredicts the Nusselt number, with average deviations of 7\%, consistent with the findings of~\cite{hirasawa2012numerical} where errors of up to 10\% were reported. Subsequently, we solve the minimization problem in~\cref{eq:regression_pointwise} to obtain $q^{(k)}$, where $k$ indexes the Reynolds number samples (we omit the geometry index $i$ as there is only one domain), and the resulting ratio $q^{(k)}\sqrt{\dm{A}} / \dm{D}$ is plotted in~\cref{fig:nu_re_sphere_2}. The average $q\sqrt{\dm{A}} / \dm{D} = 1.12$ is computed with~\cref{eq:q_average_log} and also plotted in~\cref{fig:nu_re_sphere_2} (indicated by the black dashed line). We observe that $q^{(k)}\sqrt{\dm{A}} / \dm{D}$ is consistently greater than one. The deviation arises from the intrinsic approximation error of the Ranz-Marshall correlation which overpredicts the Nusselt number, and thus the learned length scale $q\sqrt{\dm{A}}$ cannot perfectly recover the diameter.

Using the learned $q$, we can now evaluate the Ranz-Marshall correlation with $q\sqrt{\dm{A}}$ as the length scale. The modified correlation is also shown as a dashed line in~\cref{fig:nu_re_sphere_1}. The average deviations from the simulation data are reduced to 2\%, demonstrating that the learned length scale improves the correlation's accuracy.
\begin{figure}[ht]
    \centering
    \begin{subfigure}[b]{0.48\textwidth}
        \centering
        \includegraphics[width=\textwidth]{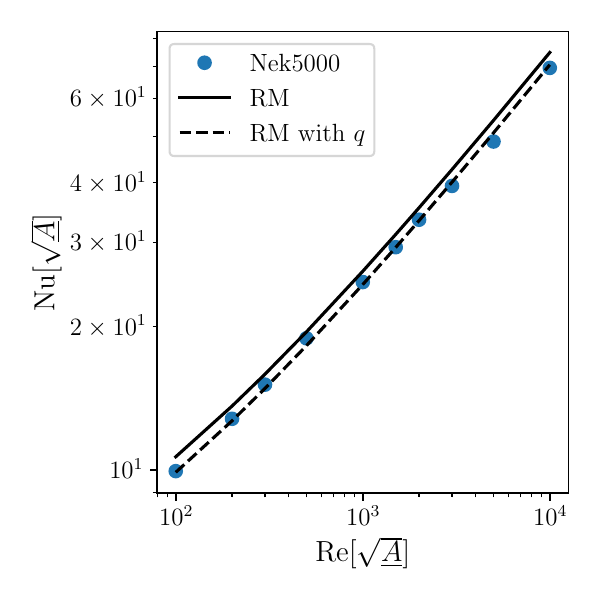}
        \caption{Nusselt number in $\sqrt{\dm{A}}$}
        \label{fig:nu_re_sphere_1}
    \end{subfigure}
    \hfill
    \begin{subfigure}[b]{0.48\textwidth}
        \centering
        \includegraphics[width=\textwidth]{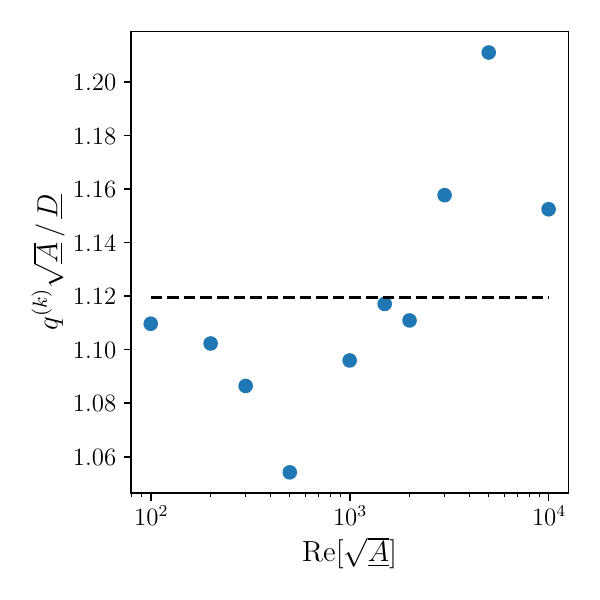}
        \caption{Learned length scales $q^{(k)}\sqrt{\dm{A}}/\dm{D}$}
        \label{fig:nu_re_sphere_2}
    \end{subfigure}
    \caption{(a) Comparison of simulation results in Nek5000 with the Ranz-Marshall correlation for a sphere. Ranz-Marshall correlation with learned $q$ is also shown. (b) Learned length scales $q^{(k)}\sqrt{\dm{A}}/\dm{D}$ plotted over Reynolds number.}
    \label{fig:nu_re_sphere}
\end{figure}

\subsubsection{Numerical study: Spheroids}\label{subsec:numerical_spheroids}
In~\cref{subsubsec:proof_of_concept_sphere}, we demonstrated that the Ranz-Marshall correlation describes our simulation data for the sphere well, and the accuracy is further improved with the learned length scale $q\sqrt{\dm{A}}$. We now extend this numerical study to a broader class of geometries—\emph{rotated spheroids}—to investigate whether the Ranz-Marshall correlation can describe this entire family when equipped with learned length scales.

We consider rotated spheroids: the geometry family can be defined as
\[
    \mathcal{S}_{\mathrm{spheroid}}
    =
    \{\, \Omegadim_{\mathrm{spheroid}}(\dm{a}, \dm{b}, \theta)
    \mid 
    \dm{a} > 0,\, \dm{b} > 0,\, \theta \in [0, \pi/2] \,\},
\]
where each member represents a spheroid with semi-major axis $\dm{a}$, 
semi-minor axis $\dm{b}$, and rotation angle $\theta$ about the $z$-axis, corresponding to the angle of attack relative to the incoming flow direction. 
The parameterized geometries, centered at the origin, are defined as
\begin{align}
    \Omegadim_{\mathrm{spheroid}}(\dm{a}, \dm{b}, \theta)
    =
    \left\{
        (\dm{x}, \dm{y}, \dm{z}) \in \mathbb{R}^3
        \ \bigg| \
        \frac{(\dm{x}\cos\theta + \dm{y}\sin\theta)^2}{\dm{a}^2}
        +
        \frac{(-\dm{x}\sin\theta + \dm{y}\cos\theta)^2}{\dm{b}^2}
        +
        \frac{\dm{z}^2}{\dm{b}^2}
        \le 1
    \right\}.
    \label{eq:spheroid}
\end{align}
Each geometry is characterized by two dimensional parameters ($M=2$) 
and one non-dimensional parameter ($N=1$). 
According to~\cref{eq:nusselt_functional}, only one non-dimensional combination 
of the dimensional parameters can influence the Nusselt number. We choose similarly as before
\[
    \dm{\mathcal{D}}_2 = \dm{\mathcal{D}}_{\mathrm{area}},
\]
and write $\sqrt{\dm{A}} \coloneqq \dm{\mathcal{D}}_{\mathrm{area}}(\medbullet)$.
Using this length scale and with~\cref{eq:nusselt_functional}, the Nusselt number can be expressed as
\begin{align}
    \Nusselt[\sqrt{\dm{A}}]
    =
    f^{\mathcal{S}_{\mathrm{spheroid}}}[\sqrt{\dm{A}}]
      \bigl(\Reynolds[\sqrt{\dm{A}}],\, \Prandtl,\, s,\, \theta\bigr).
    \label{eq:spheroid_correlation}
\end{align}
where $s \coloneqq \dm{a} / \dm{b}$ denotes the aspect ratio of the spheroid. Prolate spheroids correspond to $s > 1$, oblate spheroids to $s < 1$, and the special case $s = 1$ represents a sphere. 

We now want to approximate the Nusselt number for all $\Omegadim\in\mathcal{S}_{\mathrm{spheroid}}$ in~\cref{eq:spheroid_correlation} using the Ranz-Marshall correlation as the ansatz correlation $\undertilde{f}_{\mathrm{rm}}^{\mathcal{S}_{\mathrm{sph}}}[\dm{\mathcal{D}}^{\mathrm{ref}}_1(\Omegadim)]$, with $\dm{\mathcal{D}}^{\mathrm{ref}}_1 = \dm{\mathcal{D}}_{\mathrm{diam}}$, meaning $q$ depends on the aspect ratio and rotation angle. We consider aspect ratios in the range $s \in [0.1,\,10]$, spanning from extremely oblate spheroids (resembling flat plates) to highly prolate spheroids (resembling long cylinders), and rotation angles of $\theta \in [0^\circ,\,90^\circ]$. For data generation, we use the parameter grid
\begin{align*}
    \Reynolds[\sqrt{\dm{A}}] &= \{100,\,200,\,300,\,500,\,1000,\,1500,\,2000\}, \\
    s &= \{0.1,\,0.133,\,0.2,\,0.4,\,1,\,2.5,\,5,\,7.5,\,10\}, \\
    \theta &= \{0^\circ,\,15^\circ,\,30^\circ,\,45^\circ,\,60^\circ,\,75^\circ,\,90^\circ\},
\end{align*}
and fix $\Prandtl = 0.71$. This results in a total of $7 \times 9 \times 7 = 441$ simulations. For the sphere ($s = 1$), the angle is irrelevant, reducing the total to 399 simulations.

For each parameter sample, the ISO formulation is solved using Nek5000~\cite{nek5000-web-page}. The numerical setup is identical to that described in~\cref{subsubsec:proof_of_concept_sphere}, employing the same polynomial order, time-integration scheme, and steady-state detection algorithm (see Appendix~\ref{sec:steady_state_algorithm}). Mesh convergence studies for representative oblate and prolate cases are provided in Appendices~\ref{subsec:convergence_study_sph_pro} and~\ref{subsec:convergence_study_sph_obl}.

We plot the entire dataset of space-time averaged Nusselt numbers against Reynolds numbers in~\cref{fig:nusselt_spheroids_1}, using $\sqrt{\dm{A}}$ as the characteristic length scale. 
The Ranz-Marshall correlation is shown as a black dashed line for reference. 
The Nusselt numbers exhibit considerable scatter: the highest and lowest curves correspond to the prolate spheroid under cross-flow and axial-flow (flow perpendicular or aligned with the long axis of the spheroid, respectively) conditions, respectively, while the sphere and the Ranz-Marshall correlation lie roughly in the middle. If the Ranz-Marshall correlation were used to estimate the Nusselt numbers with $\sqrt{\dm{A}}$ for all cases, the maximum deviations reach up to 52\%.

We now proceed to learn the length scale function. We compute $q^{(i,k)}$ for each case with~\cref{eq:regression_pointwise}, and then average over Reynolds numbers using~\cref{eq:q_average_log} to obtain $q^{(i)}$ for each geometry defined by $(s,\theta)$. We plot in~\cref{fig:nusselt_spheroids_2} the averaged Nusselt numbers over Reynolds numbers again, but now using the learned length scale $q^{(i)}\sqrt{\dm{A}}$. We observe that the data points collapse very well onto the Ranz-Marshall correlation, with a maximum error of 15\% for the data.

\begin{figure}[ht]
    \centering
    \begin{subfigure}[b]{0.48\textwidth}
        \centering
        \includegraphics[width=\textwidth]{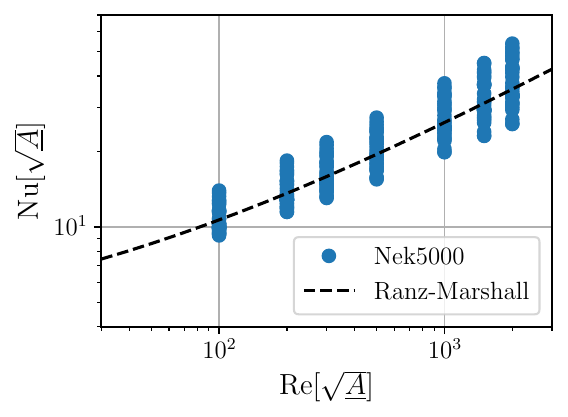}
        \caption{Length scale $\sqrt{\dm{A}}$}
        \label{fig:nusselt_spheroids_1}
    \end{subfigure}
    \hfill
    \begin{subfigure}[b]{0.48\textwidth}
        \centering
        \includegraphics[width=\textwidth]{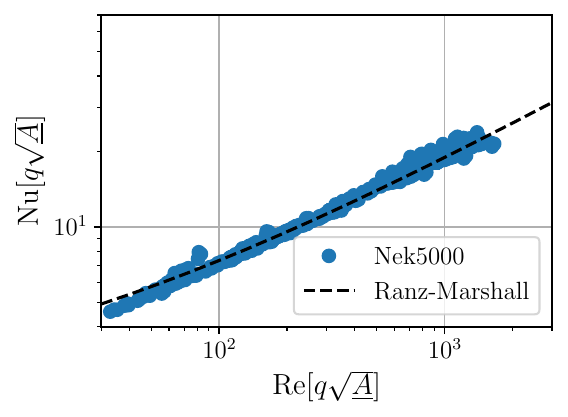}
        \caption{Length scale $q^{(i)}\sqrt{\dm{A}}$}
        \label{fig:nusselt_spheroids_2}
    \end{subfigure}
    \caption{Nusselt number as a function of Reynolds number from Nek5000 simulations, compared with the Ranz-Marshall correlation. (a) Using the length scale $\sqrt{\dm{A}}$; (b) using the learned length scale $q^{(i)}\sqrt{\dm{A}}$.}
    \label{fig:nusselt_spheroids}
\end{figure}

To assess the extrapolation capability of the learned length scale, we restrict the training data to $\Reynolds[\sqrt{\dm{A}}] = \{100, 200, 300, 500\}$, resulting in four data points per geometry. The resulting Nusselt-Reynolds relationship is shown in~\cref{fig:nusselt_spheroids_extra_1}, where the learned length scale $q^{(i)}\sqrt{\dm{A}}$ is used. Despite the reduced dataset, the Ranz-Marshall correlation still captures the overall trend very well, with a maximum deviation of 20\%. 
As compared to~\cref{fig:nusselt_spheroids_2}, the spread increases for $\Reynolds[q\sqrt{\dm{A}}]$ larger than $\approx500$, as data in this range were not included in the learning process.

To further test extrapolation performance, we examine selected high Reynolds number cases up to $\Reynolds[\sqrt{\dm{A}}]=10000$. Specifically, we consider the prolate spheroid with $s=10$ under both cross-flow and axial-flow conditions, since these configurations exhibited the largest deviations from the Ranz-Marshall correlation when $\sqrt{\dm{A}}$ was used as the length scale. Together with the sphere data from~\cref{subsubsec:proof_of_concept_sphere}, the results are plotted in~\cref{fig:nusselt_spheroids_extra_2}, where $q^{(i)}$ was learned only from data with $\Reynolds[\sqrt{\dm{A}}] \leq 500$. We can see that the Ranz-Marshall correlation with the learned length scale still predicts the Nusselt number well even at higher Reynolds numbers. For the prolate spheroid under axial-flow (orange line), the scaling deviates more noticeably from the Ranz-Marshall correlation. For $s=10$, the prolate spheroid is essentially like a long cylinder but in axial-flow, resembling neither a sphere nor a flat plate. The maximum error is 32\%, which is still a significant improvement over the 52\% error observed when using Ranz-Marshall with $\sqrt{\dm{A}}$ as the length scale. For much larger Reynolds numbers, the accuracy is expected to deteriorate, as the scaling behavior changes in the turbulent regime (the Reynolds exponent deviates from $1/2$), whereas the Ranz-Marshall correlation only contains the $1/2$ Reynolds exponent term.

\begin{figure}[ht]
    \centering
    \begin{subfigure}[b]{0.48\textwidth}
        \centering
        \includegraphics[width=\textwidth]{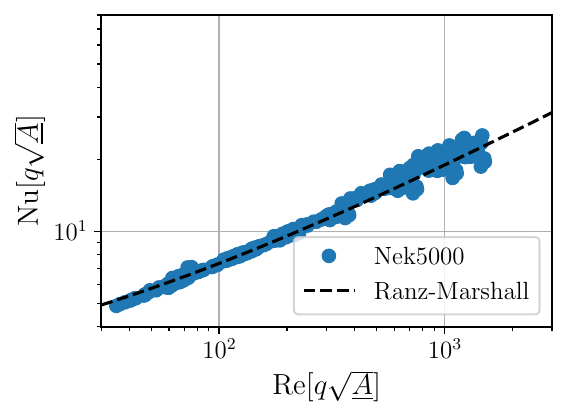}
        \caption{Nek5000 results up to $\Reynolds[\sqrt{\dm{A}}] \le 2000$.}
        \label{fig:nusselt_spheroids_extra_1}
    \end{subfigure}
    \hfill
    \begin{subfigure}[b]{0.48\textwidth}
        \centering
        \includegraphics[width=\textwidth]{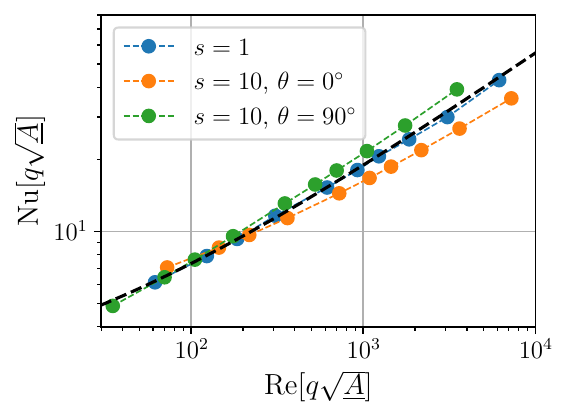}
        \caption{Nek5000 results up to $\Reynolds[\sqrt{\dm{A}}] \le 10^4$.}
        \label{fig:nusselt_spheroids_extra_2}
    \end{subfigure}
    \caption{Comparison of Nusselt number versus Reynolds number from Nek5000 simulations with the Ranz-Marshall correlation. (a) Data expressed in terms of $\Reynolds[q\sqrt{\dm{A}}]$, where $q$ was identified using data with $\Reynolds[\sqrt{\dm{A}}] \le 500$. (b) Higher Reynolds number results for the prolate spheroid under cross- and axial-flow, and for the sphere, compared with the Ranz-Marshall correlation using the learned length scale $q\sqrt{\dm{A}}$ (with $q$ identified from $\Reynolds[\sqrt{\dm{A}}] \le 500$).}
    \label{fig:nusselt_spheroids_extra}
\end{figure}

Using the available data $(s^{(i)}, \theta^{(i)}, q^{(i)})$, where $q^{(i)}$ was identified from cases with $\Reynolds[\sqrt{\dm{A}}] \le 500$, we construct a continuous surrogate $\tilde{q}(s, \theta)$ via piecewise linear interpolation. We plot $\tilde{q}$ against $\log_{10}(s)$ for various angles of attack in~\cref{fig:q_spheroids}. The ratio $\tilde{q}$ varies fairly smoothly with both the aspect ratio and the angle of attack. A clear trend is observed for prolate spheroids, whereas the behavior for oblate spheroids is less distinct. The overall spread of $\tilde{q}$ values is significantly larger for prolate than for oblate spheroids, consistent with the wider variation in Nusselt numbers observed in~\cref{fig:nusselt_spheroids_1}. This difference can be attributed to the flow topology: a prolate spheroid behaves essentially like a long cylinder, where the flow can vary drastically from $\theta=0^\circ$ to $\theta=90^\circ$. For $\theta = 0^\circ$, the flow impinges directly on the front surface and no longer exhibits local flat plate behavior, whereas for $\theta = 90^\circ$, the flow resembles a cross-flow over a cylinder, where the local boundary layer development is similar to that of a flat plate. In contrast, for oblate spheroids, for both $\theta=0^\circ$ to $\theta=90^\circ$ the flows exhibit regions where the flow remains locally similar to that over a flat plate, resulting in less variation in $\tilde{q}$.
\begin{figure}[ht]
    \centering
    \includegraphics[width=0.6\textwidth]{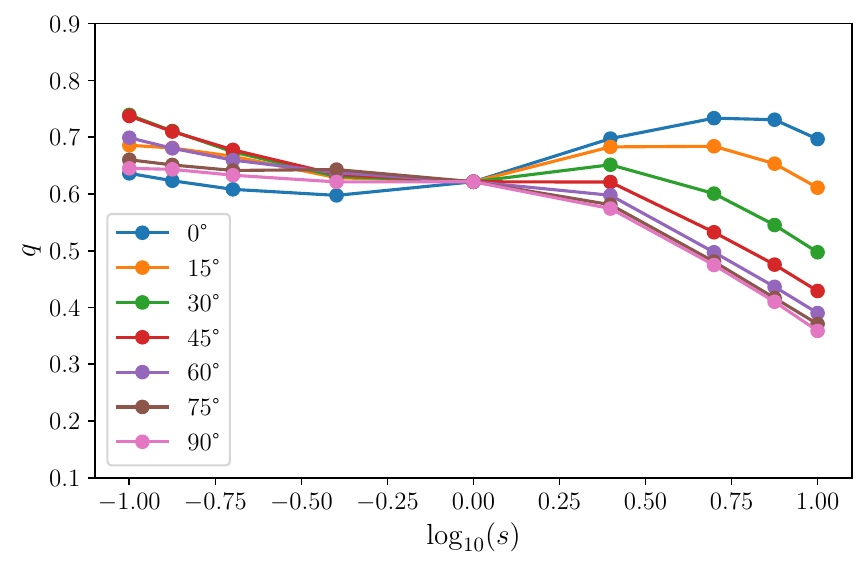}
    \caption{Learned $q$ (for $\Reynolds[\sqrt{\dm{A}}] \le 500$) over $\log_{10}(s)$ for various angles of attack $\theta$.}
    \label{fig:q_spheroids}
\end{figure}

\subsubsection{Test case: Rectangular cuboid}\label{subsubsec:nusselt_cuboid}
In~\cref{subsec:numerical_spheroids}, we demonstrated that the proposed framework can successfully learn length scale functions for the rotated spheroid family, thereby extending the applicability of the Ranz-Marshall correlation beyond spheres to highly prolate and oblate shapes. As a result, we have now effectively obtained a correlation for rotated spheroids that spans a large portion of the class of convex geometries $\mathcal{C}$.

As a final test, we consider a geometry $\Omegadim \in \mathcal{C}$ that lies outside the spheroid family, $\Omegadim \not\subset \mathcal{S}_{\mathrm{spheroid}}$: a rectangular cuboid with length $6.25\dm{L}$, width $\dm{L}$, and thickness $\dm{L}$, subjected to a flow at an angle of attack of $7.5^\circ$. This example demonstrates that the learned spheroid correlation can be used to estimate the Nusselt number for more general shapes that remain geometrically similar to spheroids.

The cuboid is solved with the ISO formulation using Nek5000 with $\Prandtl = 0.71$ and $\Reynolds[\sqrt{\dm{A}}] = 5000$, employing the same numerical setup as described in~\cref{subsubsec:proof_of_concept_sphere}. A mesh convergence study for this case is provided in Appendix~\ref{subsec:convergence_study_cuboid}. We show the temperature field in the $\dm{x}$-$\dm{y}$ plane ($\dm{z}=0$) at the final time step in~\cref{fig:flow_cuboid_xy}, where we can see significant differences on the top or bottom side due to the angle of attack. The spatially averaged Nusselt number as a function of time is shown in~\cref{fig:nusselt_cuboid_time}, where time is normalized by the final time $\tf$ determined from the steady-state criterion. Similar as the examples before, the Nusselt number exhibits small random oscillations of small amplitude around a stationary mean value. We obtain a space-time averaged Nusselt number of $\Nusselt[\sqrt{\dm{A}}] = 45.71$.

To estimate the Nusselt number for the cuboid, we directly use the Ranz-Marshall correlation together with the learned length scale ratio $\tilde{q}$ obtained in~\cref{subsec:numerical_spheroids} for spheroids with $\Reynolds[\sqrt{\dm{A}}] \le 500$---no additional data for the cuboid is required. To identify the most representative spheroid for the cuboid, we generate a point cloud of the cuboid surface and perform a principal component analysis (PCA) to fit an equivalent spheroid, which yields an aspect ratio of $s = 6.24$. The fit is shown in~\cref{fig:cuboid_fit}, where we used 500 randomly generated points. Evaluating the Ranz-Marshall correlation with the learned length scale $\tilde{q}(s = 6.24, \theta = 7.5^\circ)\sqrt{\dm{A}}$, we obtain a predicted Nusselt number of $\undertilde{\Nusselt}[\sqrt{\dm{A}}] = 47.94$, corresponding to a relative error of $5\%$.

This shows that the Ranz-Marshall correlation combined with the learned length scale function $\tilde{q}$---trained solely on spheroids---can accurately estimate the Nusselt number of a cuboid. The cuboid is geometrically close to its best-fit spheroid, and since geometrically similar bodies share comparable characteristic length scales, they consequently exhibit similar convective heat transfer behavior.

\begin{figure}[ht]
    \centering
    \begin{subfigure}[b]{0.3\textwidth}
        \centering
        \includegraphics[width=\textwidth]{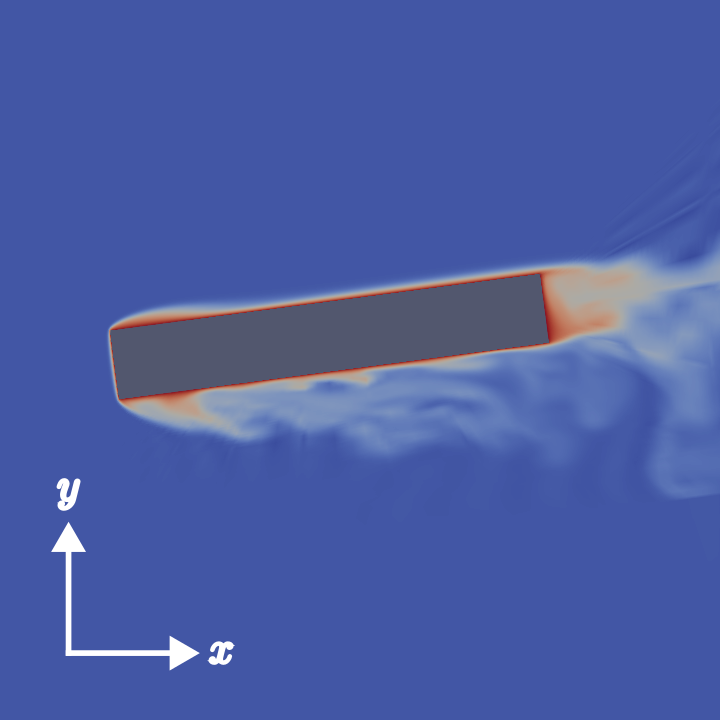}
        \caption{Temperature field in the $\dm{x}-\dm{y}$ plane ($\dm{z}=0$) at the final time step.}
        \label{fig:flow_cuboid_xy}
    \end{subfigure}
    \hfill
    \begin{subfigure}[b]{0.305\textwidth}
        \centering
        \includegraphics[width=\textwidth]{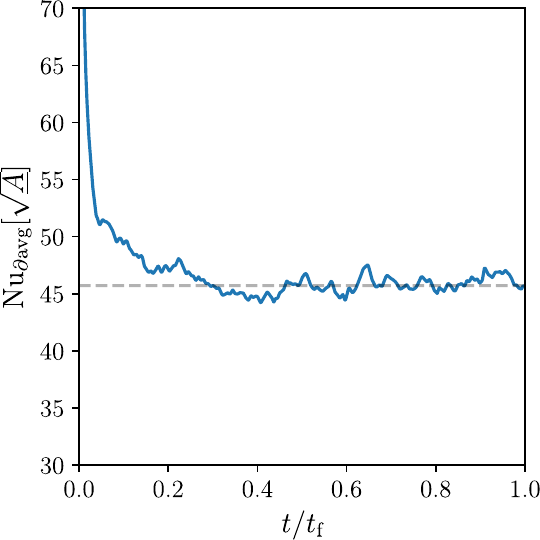}
        \caption{Spatially averaged Nusselt number over time.}
        \label{fig:nusselt_cuboid_time}
    \end{subfigure}
    \hfill
    \begin{subfigure}[b]{0.305\textwidth}
        \centering
        \includegraphics[width=\textwidth]{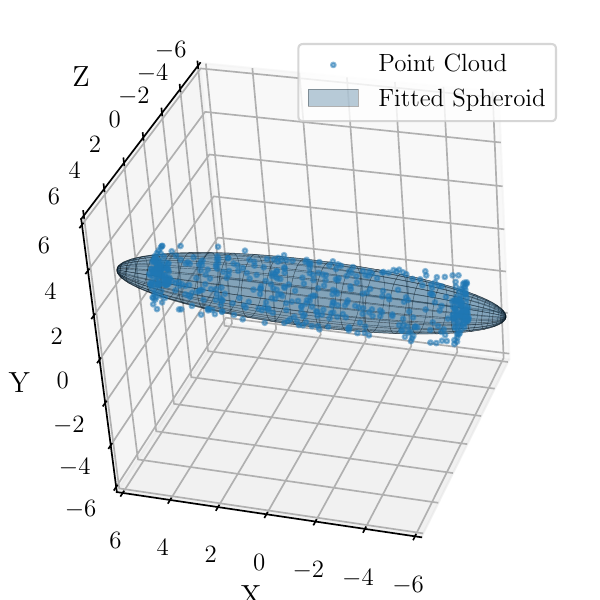}
        \caption{Best-fit spheroid with $s=6.24$ for the cuboid.}
        \label{fig:cuboid_fit}
    \end{subfigure}
    \caption{
    Results for the cuboid at $\Reynolds[\sqrt{\dm{A}}] = 5000$ and $\Prandtl = 0.71$.
    (a) Temperature field in the $\dm{x}-\dm{y}$ plane ($\dm{z}=0$) at the final time step.
    (b) Spatially averaged Nusselt number over time, normalized by the final time $\tf$.
    The space-time averaged value is indicated by the grey dashed line.
    (c) Best-fit spheroid obtained via PCA for the cuboid geometry.
    }
    \label{fig:flow_cuboid_combined}
\end{figure}

\subsubsection{Final remarks and outlook}\label{subsubsec:final_remarks_length_scale_learning}
We have demonstrated that the Nusselt number of the spheroid family $\mathcal{S}_{\mathrm{spheroid}}$, 
covering aspect ratios from 0.1 to 10 and a wide range of angles of attack, 
can be adequately approximated over a wide range of Reynolds numbers using the Ranz-Marshall correlation—originally defined for spheres—provided that an appropriate length scale function is employed. 
In other words, by learning a suitable length scale function from data, 
we effectively extended the applicability of the Ranz-Marshall correlation beyond spheres to the broader class of spheroids.

By performing a best-fit within $\mathcal{S}_{\mathrm{spheroid}}$, 
the derived length scale function for the spheroid family can also be applied to convex bodies outside the spheroid family, 
as demonstrated for the cuboid example, where a PCA-based geometric mapping was employed. 
This generalization works because geometrically similar shapes tend to share similar characteristic length scales. 
However, if a convex domain $\Omegadim \in \mathcal{C}$ is highly dissimilar from any member of $\mathcal{S}_{\mathrm{spheroid}}$—for example, a wedge-shaped body—the accuracy of the estimate will be limited.

It is worth noting that similar ideas have been explored in previous studies. For instance, Culham and Yovanovich~\cite{culham2001simplified} proposed a methodology based on flow path lengths to estimate Nusselt numbers for arbitrary cuboids, using the spheroid correlation they developed earlier in~\cite{yovanovich1988general}. However, their approach relies on physical intuition to define flow path lengths and cannot readily handle rotated or arbitrarily oriented geometries.

A natural next step is to apply this framework to additional geometry families for which empirical correlations already exist. 
In doing so, one could progressively extend the applicability of established correlations to a wider range of shapes, 
thereby building a library of generalized, data-driven correlations. 
By reusing existing correlations—which often already account for turbulence—this approach could also provide a viable pathway to describe turbulent regimes, where fully resolved simulations remain prohibitively expensive. 
Importantly, the learning process itself does not require high Reynolds number data, 
making it computationally efficient; 
high Reynolds number simulations would primarily serve for validation.

Once a sufficiently rich set of geometry classes is available, 
a classification framework for general convex geometries $\Omegadim\in\mathcal{C}$ could be envisioned. 
For example, one could represent geometries as point clouds 
and employ modern machine learning classifiers to automatically identify the closest member among all shape families. Although we restricted our study to convex geometries, the approach itself does not rely on convexity and can be applied to other domains, such as star-shaped bodies.

Together, these developments could transform the current ad hoc practice of selecting correlations and characteristic lengths into a systematic, data-driven, and automated framework for convective heat transfer modeling.

\subsection{Realistic Fluid-Solid Material Properties}\label{subsec:general_materials}
In the previous sections, we assumed that the ISO Nusselt number provides a good approximation of the CHT Nusselt number, and, hence, the Biot approximation error in~\cref{eq:error_triangular_biot} is small. This is generally true when both $r_1$ and $r_2$ are small, which implies a large time scale separation and a small Biot number, respectively.

However, an important practical question remains: \emph{how small is small?} And what are typical $r_1$ and $r_2$ values for realistic solid-fluid combinations? To address this, we first present a numerical study in~\cref{subsubsec:iso_versus_cht} that investigates the influence of $r_1$ and $r_2$ on the CHT Nusselt number and compares the results with the ISO model, followed by representative material property ratios in~\cref{subsubsec:material_ratios}.

\subsubsection{\texorpdfstring{Numerical study on the influence of $r_1$ and $r_2$ on the Nusselt number}{Numerical study on the influence of r1 and r2 on the Nusselt number}} \label{subsubsec:iso_versus_cht}
We perform a numerical study to investigate the influence of $r_1$ and $r_2$ on the space-time averaged Nusselt number in the CHT problem and compare the results with the ISO model. To enable a consistent comparison across different $r_1$ and $r_2$ values, we first run the ISO simulation for each parameter pair until the Nusselt number reaches a steady-state value, denoted $\Nustavgiso$. The corresponding end time is then used to run the CHT simulation, from which we compute $\Nustavgcht$.

To determine the steady-state, we use the following method: for the current end time $\tf$, we compute the time integrals
\begin{align*}
    \Nustavg^{\text{iso},(i)} = \dashint_0^{(0.5+0.1i)\tf} \Nuavgiso(t)
\end{align*}
for $i=1, 2, \dots, 5$. We check the relative error $$\frac{\left|\Nustavg^{\text{iso},(i+1)} - \Nustavg^{\text{iso},(i)}\right|}{\left|\Nustavg^{\text{iso},(i)}\right|}$$ for $i=1,\dots,4$ and compute the average relative error. If it is below 0.5\%, we consider the Nusselt number to have reached a steady state.

We consider three geometries: a sphere, a prolate spheroid, and an oblate spheroid, each defined by the equation in~\cref{eq:spheroid}. The aspect ratios and orientations are set as follows: $s = 5$ with $\theta = 0^\circ$ for the prolate spheroid, and $s = 0.2$ with $\theta = 90^\circ$ for the oblate spheroid. The parameter ranges $r_1 \in [0.00264,\, 0.16896]$ and $r_2 \in [0.0005425,\, 0.03472]$ are selected, with the characteristic length taken as the square root of the surface area. Simulations are performed at Reynolds numbers of 100 and 500, with a fixed Prandtl number of 0.71. All computations are carried out using Nek5000, following the same numerical setup as described in~\cref{subsubsec:proof_of_concept_sphere}.

The results are shown in~\cref{fig:nu_r1r2_Re100,fig:nu_r1r2_Re500}, where we plot, for fixed $r_2$, the error over $r_1$ and vice versa in columns, with different geometries in rows. As expected, for very small $r_1$ and $r_2$, the relative error between $\Nustavgiso$ and $\Nustavgcht$ is negligible (around 0.1--1\%). The trend with respect to $r_1$ and $r_2$ is not strictly monotonic: as $r_1$ and $r_2$ increase, the relative error generally grows, reaching up to 17\% for the prolate spheroid at $\Reynolds=100$, yet there exist ranges where the difference decreases with increasing $r_1$ or $r_2$. Overall, the errors for the prolate and oblate spheroids are larger than for the sphere, and the influence of $r_1$ appears to be more pronounced than that of $r_2$. Nevertheless, over the range of $r_1$ and $r_2$ considered, the relative error remains below 20\% for all cases, which is on the same order of magnitude as the empirical correlations commonly used in engineering practice.

To better understand why the trend is not monotonic, we consider the thermal diffusivity ratio, defined as $\sqrt{r_1 / r_2}$. This ratio characterizes the relative speed of thermal diffusion in the solid compared to the fluid. The relative errors for different $r_2$ values are plotted against this ratio in~\cref{fig:nu_r1r2}. Interestingly, all curves corresponding to different $r_2$ values exhibit a minimum. However, the location and sharpness of this minimum depend on the Reynolds number and geometry: for the sphere, the minimum is broader, whereas for the oblate and prolate spheroids it is much sharper. The minimum generally occurs at a thermal diffusivity ratio of $\mathcal{O}(1)$, indicating that the thermal diffusivities of the solid and fluid are comparable. There appears to be a systematic connection between the thermal diffusivity ratio and the relative error between the ISO and CHT Nusselt numbers, although a detailed investigation of this relationship is beyond the scope of the present work.

Our numerical study provides insight into the inherent CHT-to-ISO errors.  
In~\cref{eq:error_triangular_biot}, we expressed the Biot approximation error with the CHT Biot number $\Bcht$ and an approximation $\undertilde{B}$ as
\begin{align*}
    \max_{\tnd\in[0,\tf]} \big| \uLump(\tnd; \Bcht) - \uLump(\tnd; \Biapprox) \big|.
\end{align*}
This error can be further decomposed using the triangle inequality:
\begin{align*}
    \max_{\tnd\in[0,\tf]} \big| \uLump(\tnd; \Bcht) - \uLump(\tnd; \Biapprox) \big|
    &\leq
    \max_{\tnd\in[0,\tf]} \big| \uLump(\tnd; \Bcht) - \uLump(\tnd; \Biotiso) \big|
    + 
    \max_{\tnd\in[0,\tf]} \big| \uLump(\tnd; \Biotiso) - \uLump(\tnd; \Biapprox) \big|.
\end{align*}
The first term represents the inherent CHT-to-ISO error, which we have quantified numerically for several representative cases in this section.  
The second term accounts for the additional error introduced by further approximations of the ISO model, such as the use of empirical correlations.  
In principle, one could construct comprehensive tables of the inherent CHT-to-ISO error across different parameter combinations—$(r_1, r_2)$, Reynolds numbers, and geometries—and use them as a reference to assess whether the ISO model is suitable for a given application.

It is important to note that even if the relative Nusselt error is small for certain combinations of $\sqrt{r_1 / r_2}$ with large $r_1$ and $r_2$, implying a small Biot approximation error in~\cref{eq:error_triangular_biot}, both the lumping error in~\cref{eq:error_triangular_lcm} and the temporal approximation error in~\cref{eq:error_triangular_time} increase with $r_1$ and $r_2$. Consequently, a small relative error in Nusselt numbers does not necessarily imply a small total error in the ISO+LCM temperature solution.

\begin{figure}[p]
    \centering

    \begin{subfigure}[t]{0.48\textwidth}
        \centering
        \includegraphics[width=\textwidth]{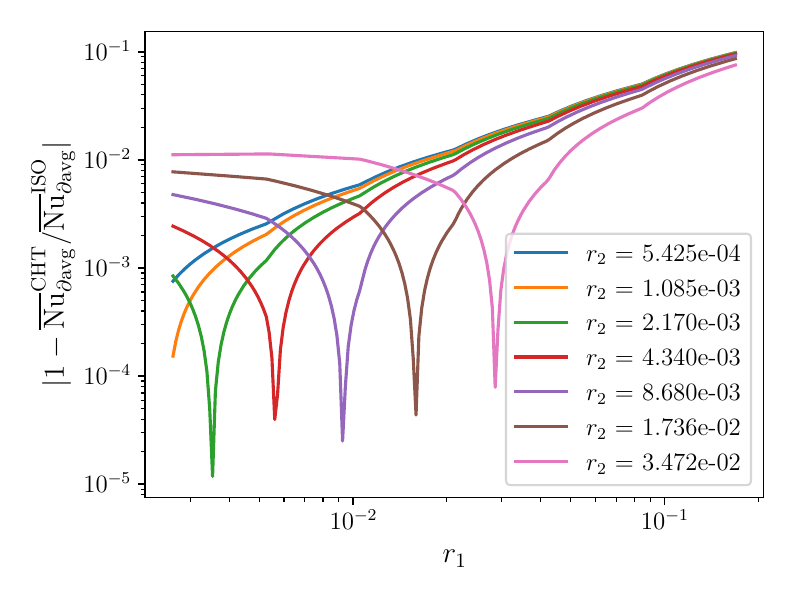}
        \caption{Sphere: rel.~Nusselt error vs.\ $r_1$ (curves: $r_2$).}
        \label{fig:sph_r1_Re100}
    \end{subfigure}\hfill
    \begin{subfigure}[t]{0.48\textwidth}
        \centering
        \includegraphics[width=\textwidth]{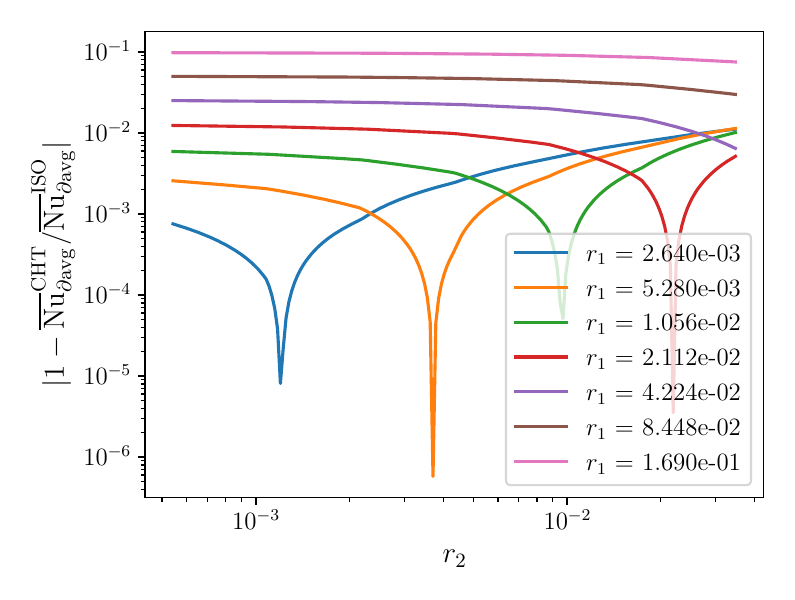}
        \caption{Sphere: rel.~Nusselt error vs.\ $r_2$ (curves: $r_1$).}
        \label{fig:sph_r2_Re100}
    \end{subfigure}

    \vspace{0.6em}
    \begin{subfigure}[t]{0.48\textwidth}
        \centering
        \includegraphics[width=\textwidth]{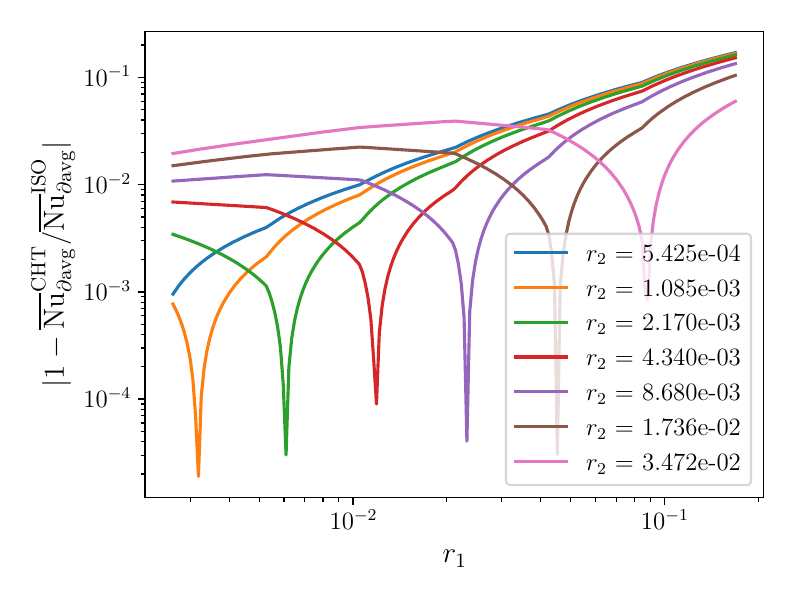}
        \caption{Prolate: rel.~Nusselt error vs.\ $r_1$ (curves: $r_2$).}
        \label{fig:prsph_r1_Re100}
    \end{subfigure}\hfill
    \begin{subfigure}[t]{0.48\textwidth}
        \centering
        \includegraphics[width=\textwidth]{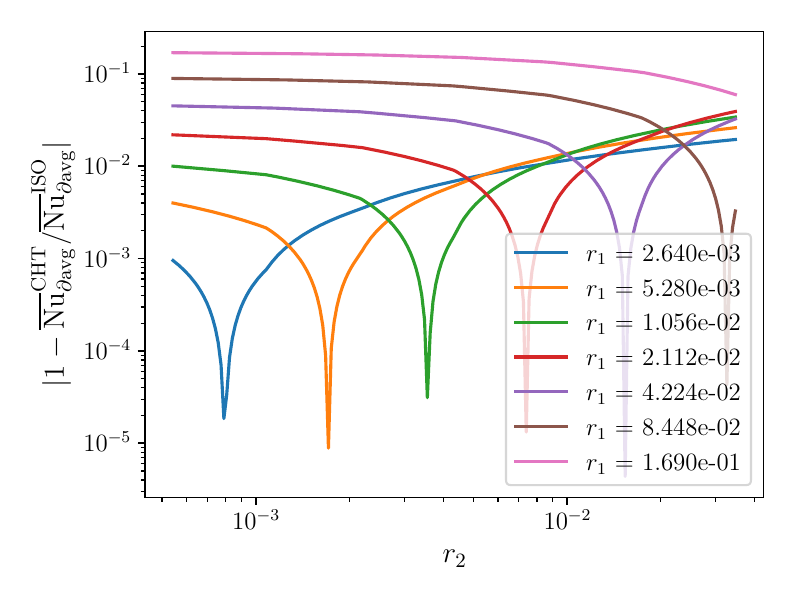}
        \caption{Prolate: rel.~Nusselt error vs.\ $r_2$ (curves: $r_1$).}
        \label{fig:prsph_r2_Re100}
    \end{subfigure}

    \vspace{0.6em}

    \begin{subfigure}[t]{0.48\textwidth}
        \centering
        \includegraphics[width=\textwidth]{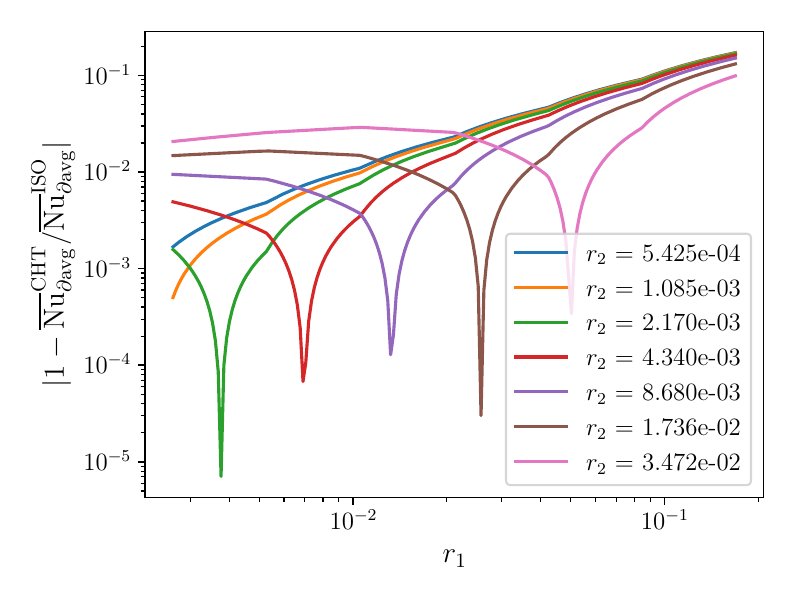}
        \caption{Oblate: rel.~Nusselt error vs.\ $r_1$ (curves: $r_2$).}
        \label{fig:obsph_r1_Re100}
    \end{subfigure}\hfill
    \begin{subfigure}[t]{0.48\textwidth}
        \centering
        \includegraphics[width=\textwidth]{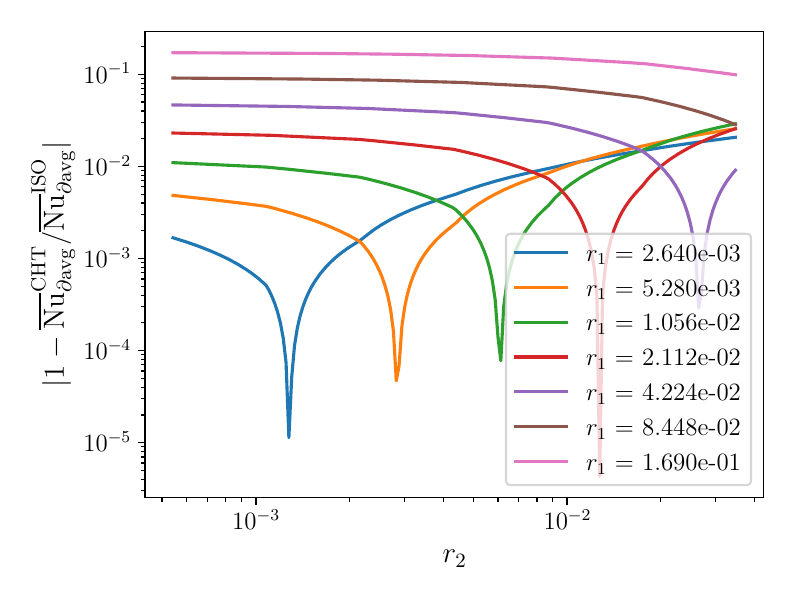}
        \caption{Oblate: rel.~Nusselt error vs.\ $r_2$ (curves: $r_1$).}
        \label{fig:obsph_r2_Re100}
    \end{subfigure}

    \caption{Nusselt-number comparisons at $\Reynolds=100$: left column plots relative difference between ISO and CHT Nusselt vs.\ $r_1$ (individual curves correspond to fixed $r_2$); right column plots relative difference between ISO and CHT Nusselt vs.\ $r_2$ (curves: fixed $r_1$) for sphere (top), prolate (middle), and oblate (bottom) spheroids.}
    \label{fig:nu_r1r2_Re100}
\end{figure}

\begin{figure}[p]
    \centering

    \begin{subfigure}[t]{0.48\textwidth}
        \centering
        \includegraphics[width=\textwidth]{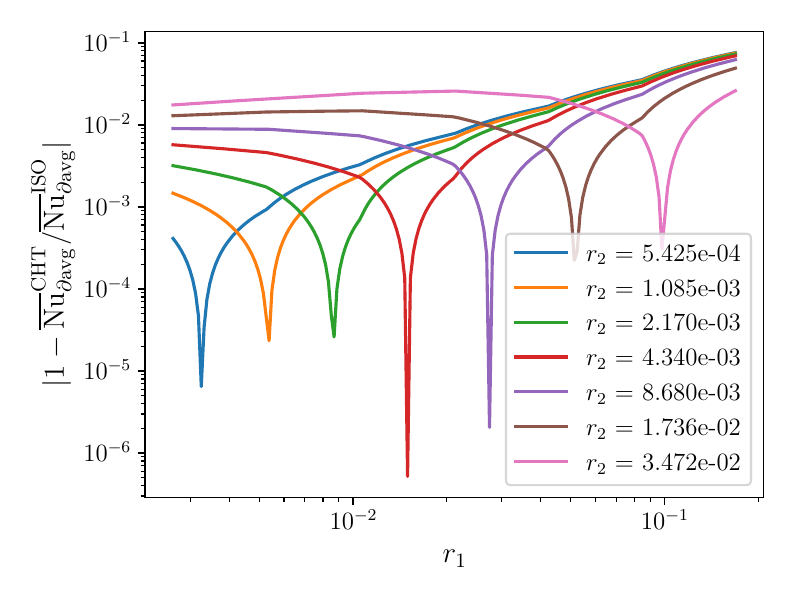}
        \caption{Sphere: rel.~Nusselt error vs.\ $r_1$ (curves: $r_2$).}
        \label{fig:sph_r1_Re500}
    \end{subfigure}\hfill
    \begin{subfigure}[t]{0.48\textwidth}
        \centering
        \includegraphics[width=\textwidth]{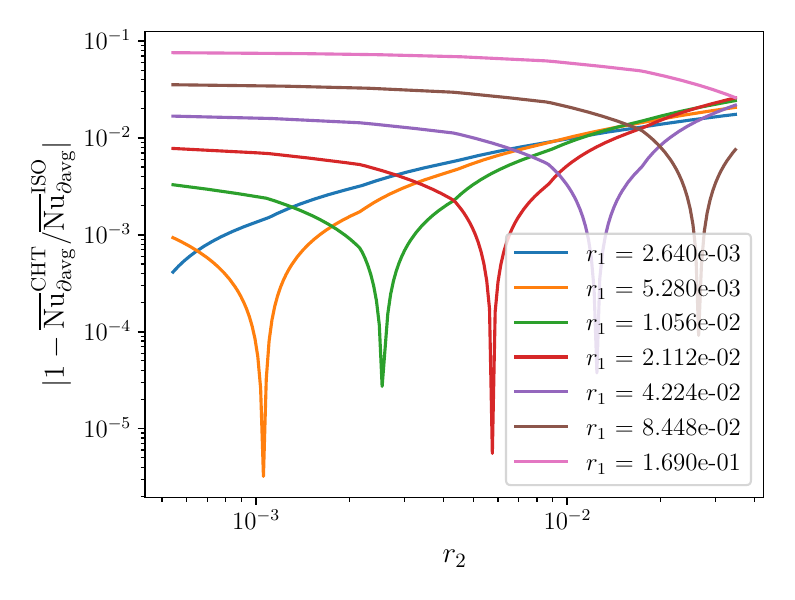}
        \caption{Sphere: rel.~Nusselt error vs.\ $r_2$ (curves: $r_1$).}
        \label{fig:sph_r2_Re500}
    \end{subfigure}

    \vspace{0.6em}

    \begin{subfigure}[t]{0.48\textwidth}
        \centering
        \includegraphics[width=\textwidth]{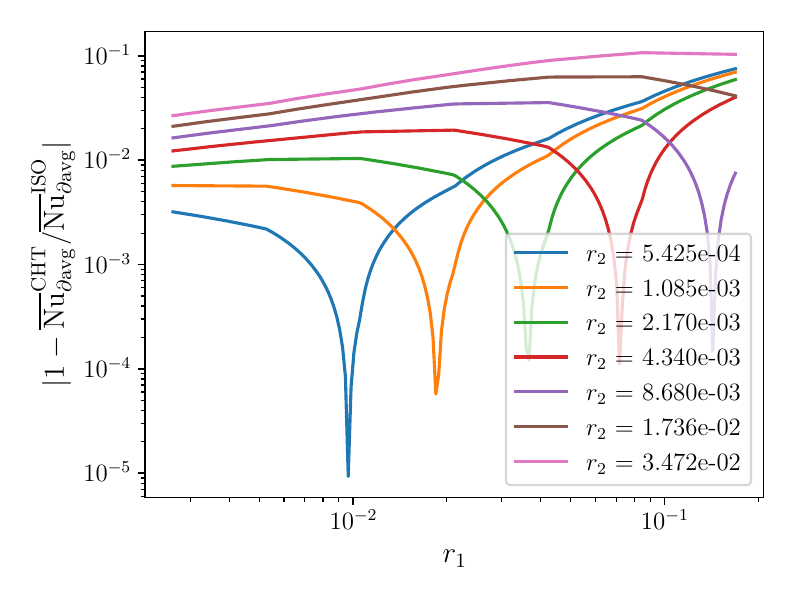}
        \caption{Prolate: rel.~Nusselt error vs.\ $r_1$ (curves: $r_2$).}
        \label{fig:prsph_r1_Re500}
    \end{subfigure}\hfill
    \begin{subfigure}[t]{0.48\textwidth}
        \centering
        \includegraphics[width=\textwidth]{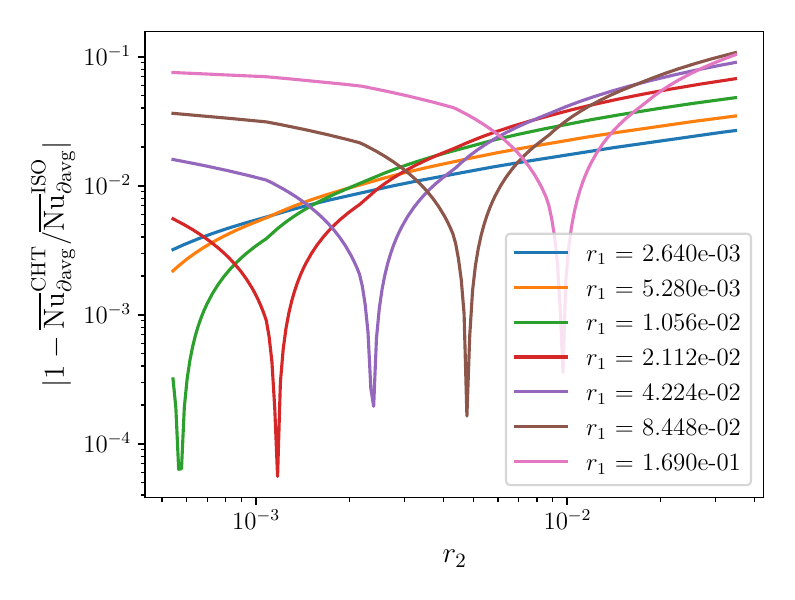}
        \caption{Prolate: rel.~Nusselt error vs.\ $r_2$ (curves: $r_1$).}
        \label{fig:prsph_r2_Re500}
    \end{subfigure}

    \vspace{0.6em}

    \begin{subfigure}[t]{0.48\textwidth}
        \centering
        \includegraphics[width=\textwidth]{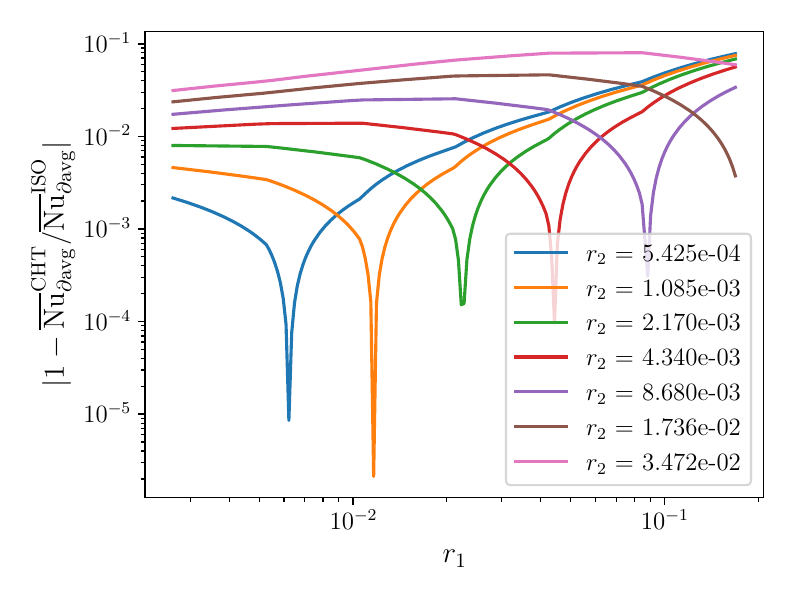}
        \caption{Oblate: rel.~Nusselt error vs.\ $r_1$ (curves: $r_2$).}
        \label{fig:obsph_r1_Re500}
    \end{subfigure}\hfill
    \begin{subfigure}[t]{0.48\textwidth}
        \centering
        \includegraphics[width=\textwidth]{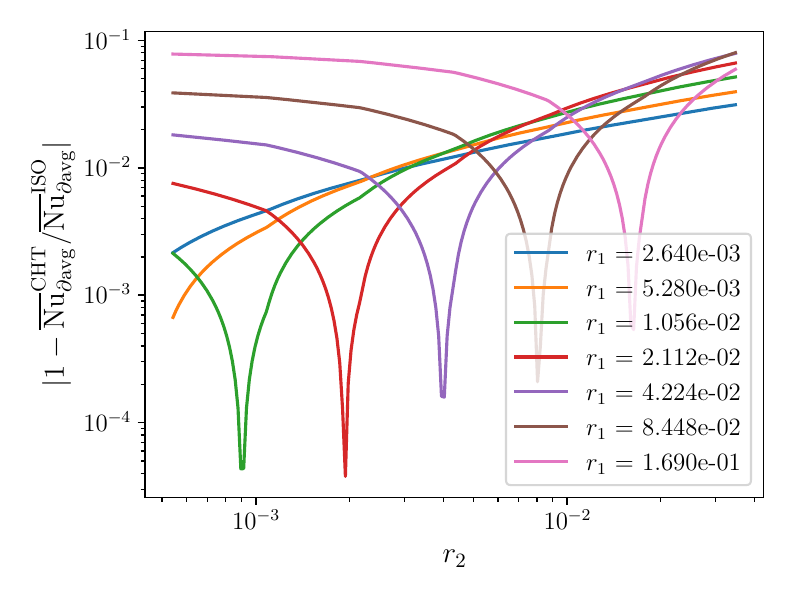}
        \caption{Oblate: rel.~Nusselt error vs.\ $r_2$ (curves: $r_1$).}
        \label{fig:obsph_r2_Re500}
    \end{subfigure}

    \caption{Nusselt-number comparisons at $\Reynolds=500$: left column plots relative difference between ISO and CHT Nusselt vs.\ $r_1$ (individual curves correspond to fixed $r_2$); right column plots relative difference between ISO and CHT Nusselt vs.\ $r_2$ (curves: fixed $r_1$) for sphere (top), prolate (middle), and oblate (bottom) spheroids.}
    \label{fig:nu_r1r2_Re500}
\end{figure}

\begin{figure}[p]
    \centering

    \begin{subfigure}[t]{0.48\textwidth}
        \centering
        \includegraphics[width=\textwidth]{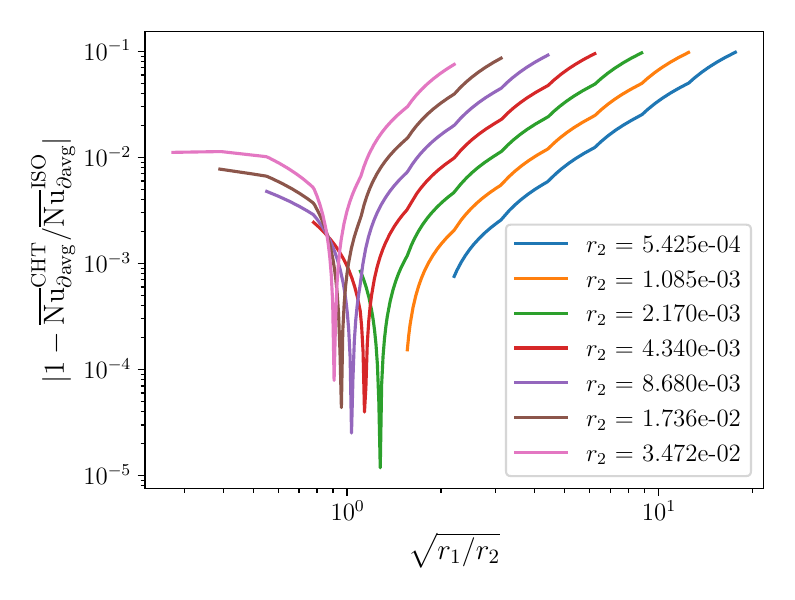}
        \caption{Sphere, $\Reynolds = 100$.}
        \label{fig:sph_Re100}
    \end{subfigure}\hfill
    \begin{subfigure}[t]{0.48\textwidth}
        \centering
        \includegraphics[width=\textwidth]{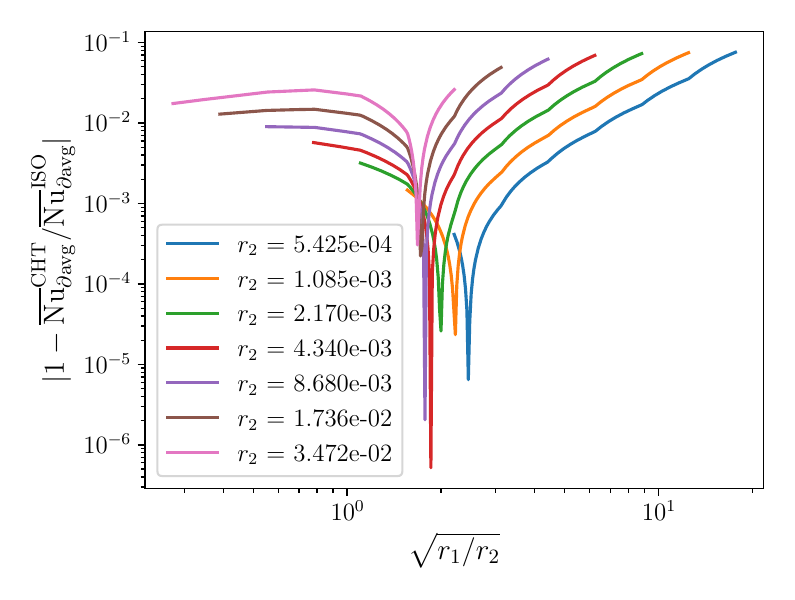}
        \caption{Sphere, $\Reynolds = 500$.}
        \label{fig:sph_Re500}
    \end{subfigure}

    \vspace{0.6em}

    \begin{subfigure}[t]{0.48\textwidth}
        \centering
        \includegraphics[width=\textwidth]{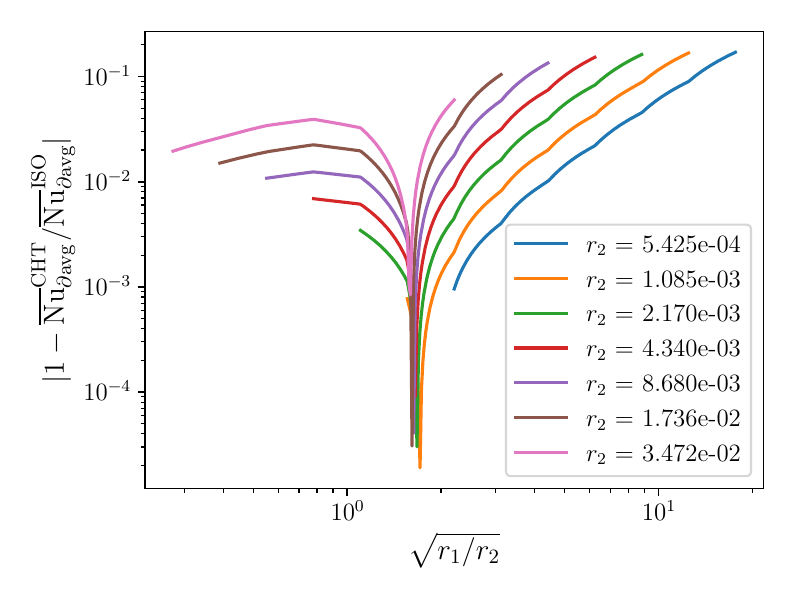}
        \caption{Prolate, $\Reynolds = 100$.}
        \label{fig:prsph_Re100}
    \end{subfigure}\hfill
    \begin{subfigure}[t]{0.48\textwidth}
        \centering
        \includegraphics[width=\textwidth]{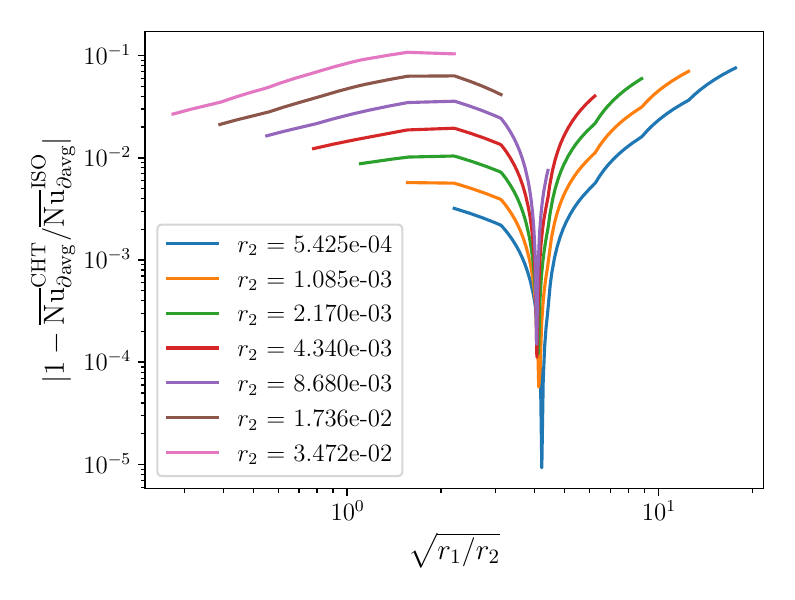}
        \caption{Prolate, $\Reynolds = 500$.}
        \label{fig:prsph_Re500}
    \end{subfigure}

    \vspace{0.6em}

    \begin{subfigure}[t]{0.48\textwidth}
        \centering
        \includegraphics[width=\textwidth]{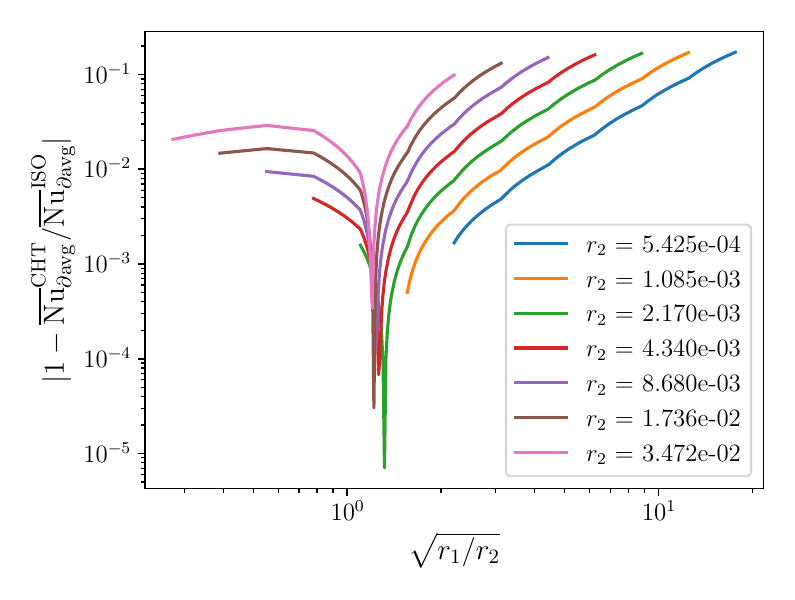}
        \caption{Oblate, $\Reynolds = 100$.}
        \label{fig:obsph_Re100}
    \end{subfigure}\hfill
    \begin{subfigure}[t]{0.48\textwidth}
        \centering
        \includegraphics[width=\textwidth]{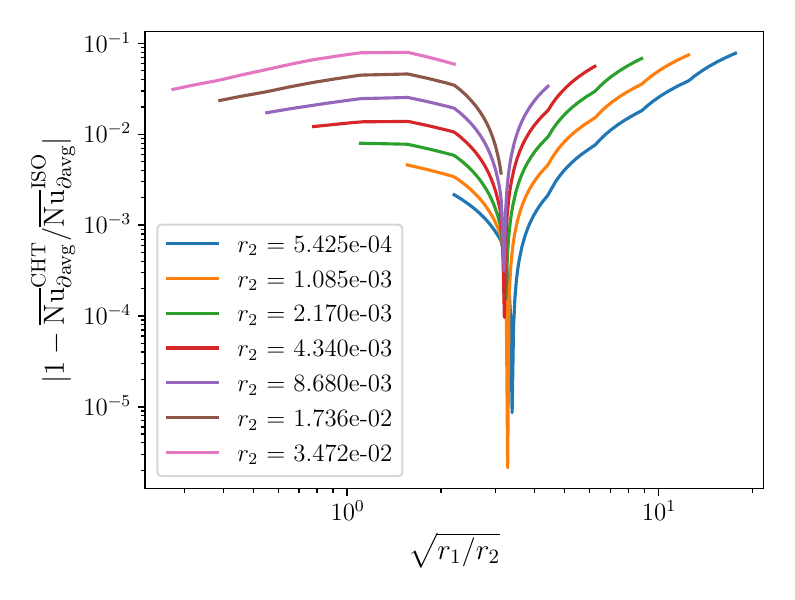}
        \caption{Oblate, $\Reynolds = 500$.}
        \label{fig:obsph_Re500}
    \end{subfigure}

    \caption{Relative Nusselt-number error $|1 - \overline{\mathrm{Nu}}_{\partial\mathrm{avg}}^{\mathrm{CHT}} / \overline{\mathrm{Nu}}_{\partial\mathrm{avg}}^{\mathrm{ISO}}|$ 
    as a function of $r_1/r_2$ for different $r_2$ for different geometries and Reynolds numbers.}
    \label{fig:nu_r1r2}
\end{figure}

\subsubsection{Material property ratios for realistic fluid-solid combinations}\label{subsubsec:material_ratios}

In many traditional engineering applications, such as the cooling of metals in air, the thermophysical property ratios between the working fluid and the solid are indeed small. However, in modern applications, solid and fluid properties can be of comparable magnitude—for example, plastic components in water, metal parts in liquid-metal coolants, or aerogels in air.

Representative thermophysical properties for selected solids and fluids are summarized in~\cref{tab:material_properties}. The data were collected from~\cite{ahtt6e} (unless otherwise noted) and correspond to different reference temperatures. These values are not intended for precise quantitative use but rather to indicate the order of magnitude of property ratios relevant to practical material combinations. We consider aluminum, stainless steel, polycarbonate, and silica aerogel as solid materials, and air, water, helium (used as coolant in pebble-bed reactors), and lead (also used as coolant in nuclear reactors) as fluids.

The resulting ratios $r_1$ and $r_2$ for all solid-fluid combinations are reported in~\cref{tab:r1,tab:r2}. For metallic solids in gaseous media (air or helium), both $r_1$ and $r_2$ are very small. In contrast, for certain combinations—such as aerogel in lead or polycarbonate in water—one or both ratios are on the order of one or higher. In such cases, neither the LCM nor ISO models can be expected to approximate the full CHT behavior accurately. The latter is typically not addressed in heat transfer textbooks such as~\cite{lienhard1973commonality}.

\begin{table}[ht]
    \centering
    \caption{Thermophysical properties of selected solid and fluid materials.}
    \label{tab:material_properties}
    \begin{tabular}{lccccc}
        \toprule
        \textbf{Material / Fluid} 
        & $\rho~[\si{kg/m^3}]$ 
        & $c_p~[\si{J/kg\cdot K}]$ 
        & $k~[\si{W/m\cdot K}]$ 
        & $\nu~[\si{m^2/s}]$ 
        & $\Prandtl$ \\
        \midrule
        Aluminum (20°C) & 2707 & 905 & 237 & -- & -- \\
        Stainless steel AISI 316 (20°C) & 8000 & 460 & 13.5 & -- & -- \\
        Polycarbonate (23°C) & 1200 & 1250 & 0.29 & -- & -- \\
        Silica aerogel \cite{scheuerpflug1992thermal} & 62 & 1000 & 0.01 & -- & -- \\
        Air (27°C) & 1.177 & 1006 & 0.0264 & $1.575\times10^{-5}$ & 0.71 \\
        Water (22°C) & 997.8 & 4183 & 0.6017 & $9.6\times10^{-7}$ & 6.66 \\
        Helium (1300 K) & 0.0375 & 5193 & 0.437 & $1.49\times10^{-3}$ & 0.664 \\
        Lead (600°C) & $1.0\times10^{4}$ & 130 & 20 & $1.3\times10^{-7}$ & $\approx 0.001$ \\
        \bottomrule
    \end{tabular}
\end{table}

\begin{table}[ht]
\centering
\caption{Volumetric heat capacity ratio \(r_1 = (\rho c)_f / (\rho c)_s\) for various fluid-solid combinations.}
\label{tab:r1}
\begin{tabular}{lcccc}
\toprule
               & Aluminum & Stainless Steel & Polycarbonate & Silica Aerogel \\
\midrule
Air            & \num{0.000483}    & \num{0.000322}    & \num{0.000789}    & \num{0.019098}     \\
Water          & \num{1.703706}    & \num{1.134184}    & \num{2.782532}    & \num{67.319313}    \\
Helium         & \num{0.000080}    & \num{0.000053}    & \num{0.000130}    & \num{0.003143}     \\
Lead           & \num{0.530648}    & \num{0.353261}    & \num{0.866667}    & \num{20.967742}    \\
\bottomrule
\end{tabular}
\end{table}

\begin{table}[H]
\centering
\caption{Thermal conductivity ratio \(r_2 = k_f / k_s\) for various fluid-solid combinations.}
\label{tab:r2}
\begin{tabular}{lcccc}
\toprule
               & Aluminum & Stainless Steel & Polycarbonate & Silica Aerogel \\
\midrule
Air            & \num{0.000111}    & \num{0.001956}    & \num{0.091034}    & \num{2.64}         \\
Water          & \num{0.002537}    & \num{0.044574}    & \num{2.075517}    & \num{60.17}        \\
Helium         & \num{0.001844}    & \num{0.03237}     & \num{1.506897}    & \num{43.7}         \\
Lead           & \num{0.084388}    & \num{1.481481}    & \num{68.965517}   & \num{2000.0}       \\
\bottomrule
\end{tabular}
\end{table}

\section{Perspectives}\label{sec:application}
In this manuscript, we have analyzed in detail the approximation error of the CHT problem when modeled using the ISO+LCM approach. While the CHT is the most accurate model for the thermal dunking problem, it is a coupled fluid-solid system with vastly different time scales, making it computationally expensive to solve. The ISO+LCM approximation significantly reduces the computational cost by decoupling the problem into a fluid-only system (ISO) and a solid-only system (LCM). However, this simplification introduces errors that must be carefully evaluated to ensure the accuracy of the results.

The total error can be decomposed into three main contributions: (i)~the temporal approximation error arising from time homogenization, (ii)~the lumping approximation error introduced by the uniform temperature assumption, and (iii)~the Biot approximation error resulting from using the ISO model or empirical correlations to estimate the true Biot number. Below, we summarize the main findings and provide a practical guideline for evaluating the validity of the ISO+LCM approximation in general engineering applications.

Traditionally, for a given problem setup, engineers estimate the Biot number to determine whether the LCM is applicable. If the Biot number falls below a predefined threshold such as~\cref{eq:biot_small}, the LCM is used. Even assuming that the Biot number could be estimated accurately, this criterion does not correctly account for the lumping error and the temporal approximation error.

Our analysis shows that the lumping error depends not only on the Biot number but also on the geometry, the thermophysical properties of the solid ($\kappa$, $\sigma$), and the flow field around the body ($\etabar$). These dependencies are summarized by the parameter $\phi$, which can be computed with~\cref{eq:sensitivity_problem,eq:phi_definition}. We also derive an upper bound for $\phi$ in~\cref{eq:phi_bound_result}, which can be evaluated even when $\kappa$, $\sigma$, and $\etabar$ are not known pointwise. This bound depends only on the variances $\dashint_\Omega (\sigma - 1)^2$ and $\dashint_{\pOmega} (\etabar - 1)^2$, allowing the estimation of the lumping error using only scalar quantities. These variances can often be reasonably estimated based on prior knowledge of the solid material properties, see~\cref{eq:var_sigma_composite}, and simplified flow simulations, see Appendix~\ref{sec:eta_study}. Importantly, the upper bound captures the effects of spatial variations in the heat transfer coefficient $\etabar$—a quantity that would otherwise need to be obtained by solving the full CHT problem. Although its evaluation involves computing $\phi(1,1,1)$ by solving the PDE in~\cref{eq:sensitivity_problem}, as well as two eigenvalue problems for $\mu$ and $\Lambda$ defined in~\cref{eq:stability_1,eq:stability_2}, the computational cost remains negligible compared to that of a full CHT simulation.

Even if the Biot number and lumping error are small, the temporal approximation error—which is typically neglected in practice—may still be significant. Our numerical study in~\cref{subsubsec:time_numerics} demonstrates that this error can grow large even for small Biot numbers, depending on the property ratio $r_1 = (\rhofdim \cfdim) / (\rhosdim \csdim)$. As $r_1$ increases, time scale separation becomes less pronounced, and the initial values of the Nusselt number during short times have a significant effect on the average temperature, resulting in large temporal errors, see~\cref{subsubsec:time_stability,propo:short_time_stiff}. To quantify this behavior, a ratio is introduced in~\cref{eq:time_ratio_2}, which measures the relative response times of the fluid and solid. However, it remains unclear how large this ratio must be to ensure a small temporal error. Ideally, one would like to derive a practical criterion like
\begin{align*}
    \frac{1}{r_1}\frac{\Reynolds\Prandtl}{\Nustavg\gamma} > f(\text{geometry}),
\end{align*}
where $f(\text{geometry})$ is a geometry-dependent factor. In the example considered in~\cref{subsubsec:time_numerics}, a value of 300 is found. Such a criterion would offer a quantitative condition for the validity of the time scale separation assumption, and ensure that the temporal approximation error is small.

Finally, it is important to recognize that the Biot number estimated using the ISO model may deviate from the true value, particularly when either $r_1$ or $r_2$ is large. Our numerical study in~\cref{subsubsec:iso_versus_cht} shows that even for moderate values of $r_1$ and $r_2$, the relative error between the ISO and CHT Nusselt numbers can reach up to 20\%. This level of deviation may be acceptable, as it is on the same order of magnitude as the uncertainty typically associated with empirical correlations (see~\cref{subsec:nusselt_approximation}). However, for many modern applications, where either $r_1$ or $r_2$ may take even larger values, one should exercise caution.

\appendix
\section{\texorpdfstring{Numerical Study of $\etabar$ for Real Flows}{Numerical Study of Uniformity for Real Flows}}\label{sec:eta_study}

To evaluate the upper bound for the lumping error in~\cref{eq:lcm_error_bound} without detailed knowledge of the spatial quantities $\kappa$, $\sigma$, and $\etabar$, we can use the upper bound for $\phi$ in~\cref{eq:phi_bound_result}. This bound requires only the geometry and the scalar-valued variances $\dashint_{\pOmega}(\etabar-1)^2$ and $\dashint_{\Omega}(\sigma-1)^2$. The variance of $\sigma$ can be computed whenever $\sigma$ is piecewise constant with known values and volume fractions, see~\cref{eq:var_sigma_composite}. By contrast, $\etabar$ depends on the flow field and is generally unknown in practice.

In this section, we present a numerical study for three representative cases of realistic flows: forced convection over a cylinder with circular cross section, a sphere, and a cylinder with square cross section. We investigate the variance of $\etabar$ across a range of Reynolds numbers, using both experimental data from the literature and simulation data obtained with Nek5000~\cite{nek5000-web-page}. Based on the results, we propose a strategy to estimate $\dashint_{\pOmega}(\etabar-1)^2$ for practical applications.

\subsection{Empirical results}
\subsubsection{Cylinder with circular cross section}
Several experimental studies have measured the time-averaged Nusselt number in space around a long cylinder with circular cross section in cross flow, including Giedt~\cite[Fig.~3]{giedt1951effect}, Schmidt~\cite[Figs.~4--6]{schmidt1943heat}, and Achenbach~\cite[Figs.~14 and 17]{achenbach1977effect}, all at high Reynolds numbers. Since no raw data were provided, we extracted points from the published figures and reconstructed the curves using radial basis function regression. As all data were for high Reynolds numbers, we additionally performed simulations with Nek5000 for low Reynolds number cases. The simulations were carried out in a two-dimensional domain, corresponding to an infinitely long cylinder, as also assumed in the experiments. This approximation is appropriate at low Reynolds numbers, where three-dimensional effects are negligible. 

The resulting $\etabar$ curves are shown in~\cref{fig:eta_disk}. Since the cylinder is assumed infinitely long, $\etabar$ is defined on the circular cross section and therefore depends only on the azimuthal angle $\theta$. Its shape varies markedly with Reynolds number (evaluated with diameter as the length scale): at low values, $\etabar$ is spread out with no pronounced peak; as Reynolds increases, the profile sharpens and becomes more localized; and at even higher Reynolds numbers, a second peak emerges around $100^\circ$, which overtakes the leading edge as the location of maximum heat transfer.

We compute the variance $\dashint_{\pOmega}(\etabar-1)^2$ for all data sets and summarize the results in~\cref{tab:variance_disk}. The maximum variance for each dataset is highlighted in boldface. At very low Reynolds numbers, the variance increases with Reynolds number up to the onset of vortex shedding (around $\Reynolds \approx 90$), reaching its maximum at $\Reynolds = 79.58$. Beyond this point, the variance decreases until around $\Reynolds = 1.7\times10^5$, as the flow becomes increasingly turbulent, and the heat transfer around the cylinder tends to become more uniform due to enhanced mixing. A secondary increase in variance is observed between $\Reynolds \approx 1.7\times 10^5$ and $1.27\times 10^6$, corresponding to the so-called drag crisis: the boundary layer transitions from laminar to turbulent, allowing it to remain attached over a larger portion of the surface and leading to a sudden change in both heat-transfer and drag characteristics. Nevertheless, even in this regime, the variance remains far below the maximum attained at low Reynolds numbers.

\begin{table}[ht]
    \centering
    \begin{minipage}{0.48\textwidth}
        \centering
        \begin{tabular}{l|c|c}
            Source & Re & $\dashint_{\pOmega}(\etabar-1)^2$ \\
            \hline
            Giedt~\cite{giedt1951effect}     & 70800    & \textbf{0.092} \\
                             & 101300   & 0.056 \\
                             & 140000   & 0.020 \\
                             & 170000   & 0.031 \\
                             & 219000   & 0.037 \\
            \hline
            Schmidt~\cite{schmidt1943heat}   & 8290     & \textbf{0.234} \\
                             & 15550    & 0.169 \\
                             & 21200    & 0.133 \\
                             & 52800    & 0.108 \\
                             & 102000   & 0.082 \\
                             & 170000   & 0.049 \\
                             & 257600   & 0.052 \\
                             & 426000   & 0.123 \\
        \end{tabular}
    \end{minipage}\hfill
    \begin{minipage}{0.48\textwidth}
        \centering
        \begin{tabular}{l|c|c}
            Source & Re & $\dashint_{\pOmega}(\etabar-1)^2$ \\
            \hline
            Achenbach~\cite{achenbach1977effect} & 1270000  & \textbf{0.183} \\
                             & 4000000  & 0.139 \\
            \hline
            Nek5000 sim. & 15.92   & 0.193 \\
                         & 31.83   & 0.262 \\
                         & 47.75   & 0.300 \\
                         & 63.66   & 0.320 \\
                         & 79.58   & \textbf{0.332} \\
                         & 95.49   & 0.313 \\
                         & 143.24  & 0.305 \\
                         & 190.99  & 0.301 \\
                         & 238.73  & 0.298 \\
                         & 286.48  & 0.297 \\
                         & {\color{white}0}  & {\color{white}0} \\
        \end{tabular}
    \end{minipage}
    \caption{Variance $\dashint_{\pOmega}(\etabar-1)^2$ from literature and simulations across Reynolds numbers for $\Prandtl=0.71$ for the cylinder with circular cross section. The maximum value for each dataset is highlighted in boldface.}
    \label{tab:variance_disk}
\end{table}

\begin{figure}[p]
    \centering
    \begin{subfigure}[b]{0.48\textwidth}
        \centering
        \includegraphics[width=\textwidth]{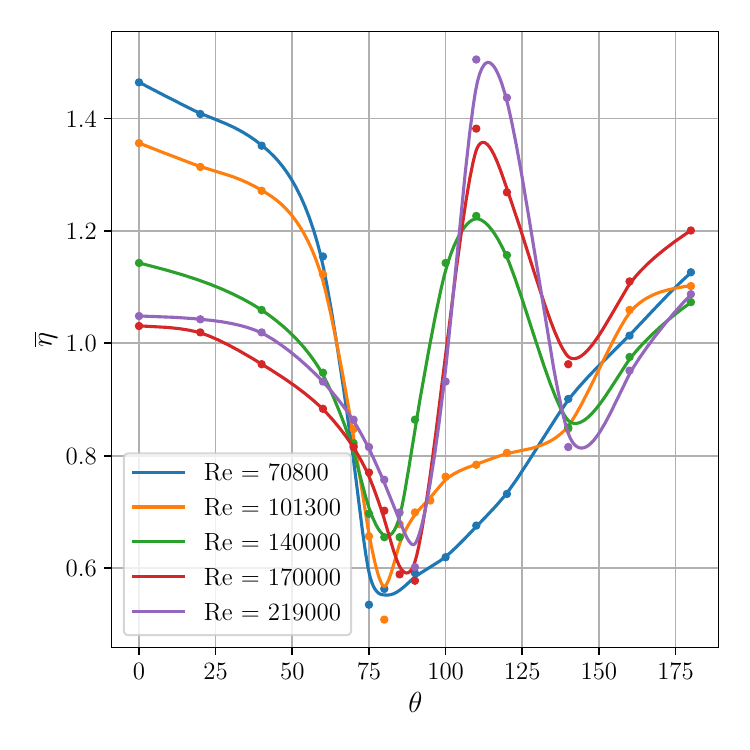}
        \caption{Giedt~\cite{giedt1951effect}}
        \label{fig:eta_giedt}
    \end{subfigure}
    \hfill
    \begin{subfigure}[b]{0.48\textwidth}
        \centering
        \includegraphics[width=\textwidth]{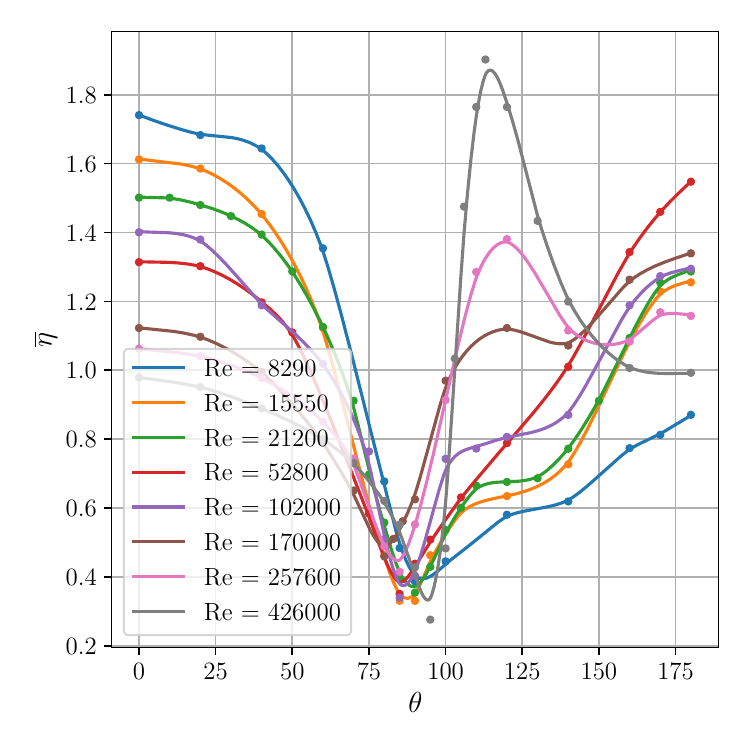}
        \caption{Schmidt~\cite{schmidt1943heat}}
        \label{fig:eta_schmidt}
    \end{subfigure}
    \begin{subfigure}[b]{0.48\textwidth}
        \centering
        \includegraphics[width=\textwidth]{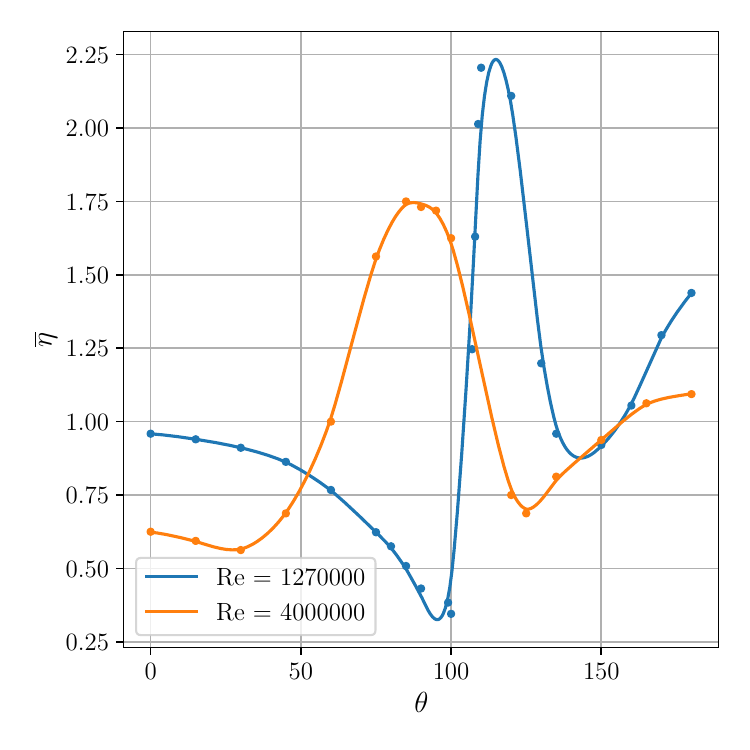}
        \caption{Achenbach~\cite{achenbach1977effect}}
        \label{fig:eta_achenbach}
    \end{subfigure}
    \hfill
    \begin{subfigure}[b]{0.48\textwidth}
        \centering
        \includegraphics[width=\textwidth]{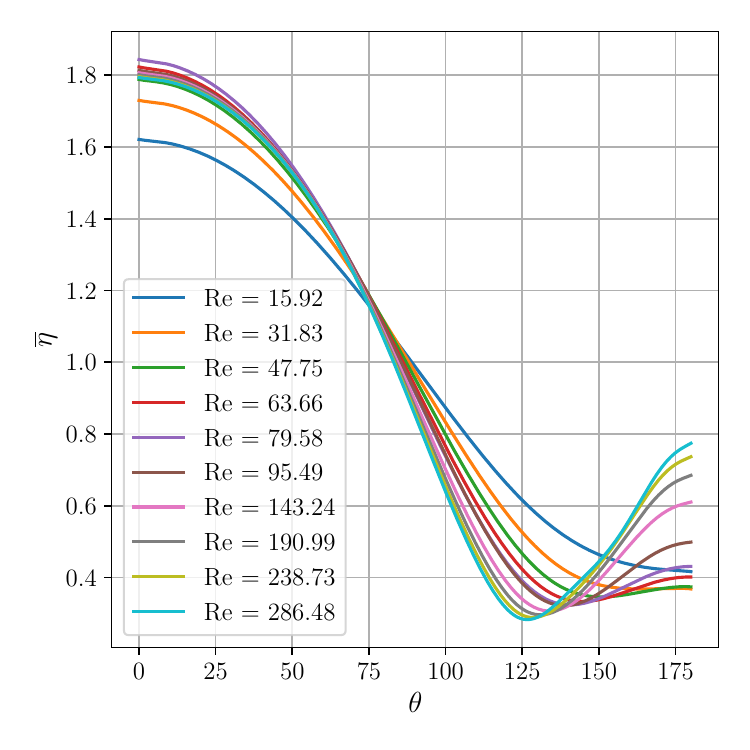}
        \caption{Nek5000}
        \label{fig:eta_nek_disk}
    \end{subfigure}
    \caption{$\etabar$ values for cross flow over a cylinder with circular cross section at various Reynolds numbers from different sources.}
    \label{fig:eta_disk}
\end{figure}

\subsubsection{Sphere}
For the available experimental data in Aufdermaur~\cite{aufdermaur1967wind} for the sphere, the trend of decreasing $\dashint_{\pOmega}(\etabar-1)^2$ with Reynolds number is again observed. The $\etabar$ values were measured along the equatorial plane, as a function of the polar angle from the forward stagnation point, corresponding to a great-circle cross section of the sphere. The resulting $\etabar$ curves are shown in~\cref{fig:eta_aufdermauer}, and the corresponding variances are reported in~\cref{tab:variance_sph}, with the maximum value highlighted in bold. 
As only values of $\etabar$ along a line on the equator are available, the variance is computed under the assumption of axisymmetry about the flow direction.
The order of magnitude of the variance is the same as for the cylinder, and the same qualitative trend emerges: after reaching a maximum at $\Reynolds=8200$, the variance decreases with increasing Reynolds number.

\begin{figure}[p]
    \centering
    \begin{subfigure}[c]{0.48\textwidth}
        \centering
        \includegraphics[width=\textwidth]{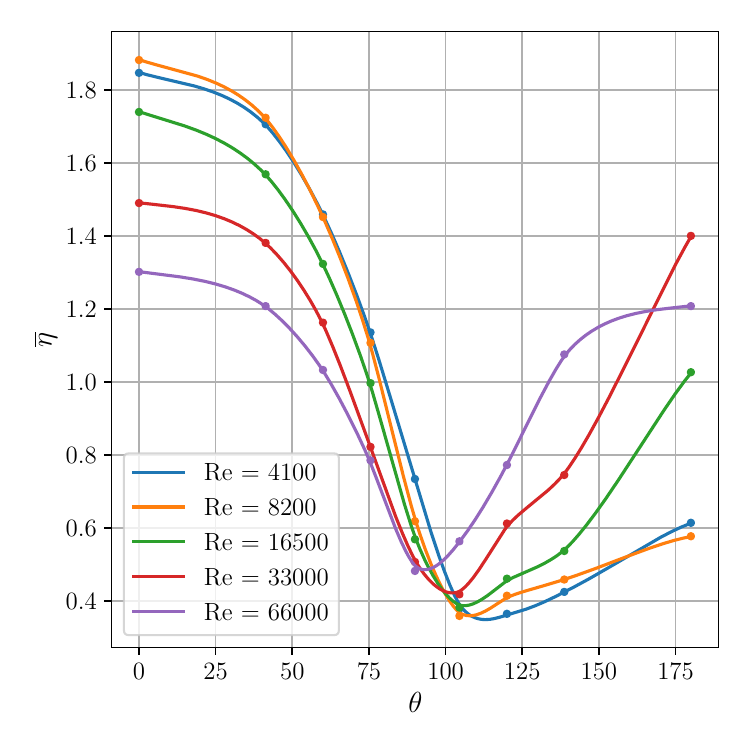}
        \caption{$\etabar$ curves.}
        \label{fig:eta_aufdermauer}
    \end{subfigure}
    \hfill
    \begin{subfigure}[c]{0.48\textwidth}
        \centering
        \begin{tabular}{l|c|c}
            Source & Re & $\dashint_{\pOmega}(\etabar-1)^2$ \\ \hline
            Aufdermaur~\cite{aufdermaur1967wind}  & 4100  & 0.338 \\
                                       & 8200  & \textbf{0.348} \\
                                       & 16500 & 0.232 \\
                                       & 33000 & 0.134 \\
                                       & 66000 & 0.070
        \end{tabular}
        \caption{Variance values.}
        \label{tab:variance_sph}
    \end{subfigure}
    \caption{Time-independent variation function $\etabar$ for cross-flow around a sphere at various Reynolds numbers from Aufdermaur~\cite{aufdermaur1967wind}. (a) $\etabar$ distributions. (b) Variance $\dashint_{\pOmega}(\etabar-1)^2$.}
    \label{fig:eta_sph}
\end{figure}

\subsubsection{Cylinder with square cross section}
We also examine the flow around a square cylinder. For this geometry, no experimental data on the local Nusselt number were found, so we conducted simulations with Nek5000. The Reynolds number is defined using a characteristic length corresponding to a perimeter of $\pi\dm{D}$, allowing direct comparison with the circular cylinder of diameter $\dm{D}$. As the Reynolds numbers are low, the simulations were performed in two dimensions, and $\etabar$ is defined as a function of the azimuthal angle measured from the forward stagnation point. The $\etabar$ curves are shown in~\cref{fig:eta_square} and the corresponding variances are reported in~\cref{tab:variance_square}. Within the limited Reynolds number range studied, the behavior is qualitatively similar to that of the disk: $\dashint_{\pOmega}(\eta-1)^2$ increases with Reynolds number up to the onset of vortex shedding at approximately $\Reynolds=90$, and decreases thereafter. The values, however, are consistently higher than for the disk, which can be attributed to stronger flow separation at the square’s corners.
\begin{figure}[p]
    \centering
    \begin{subfigure}[c]{0.48\textwidth}
        \centering
        \includegraphics[width=\textwidth]{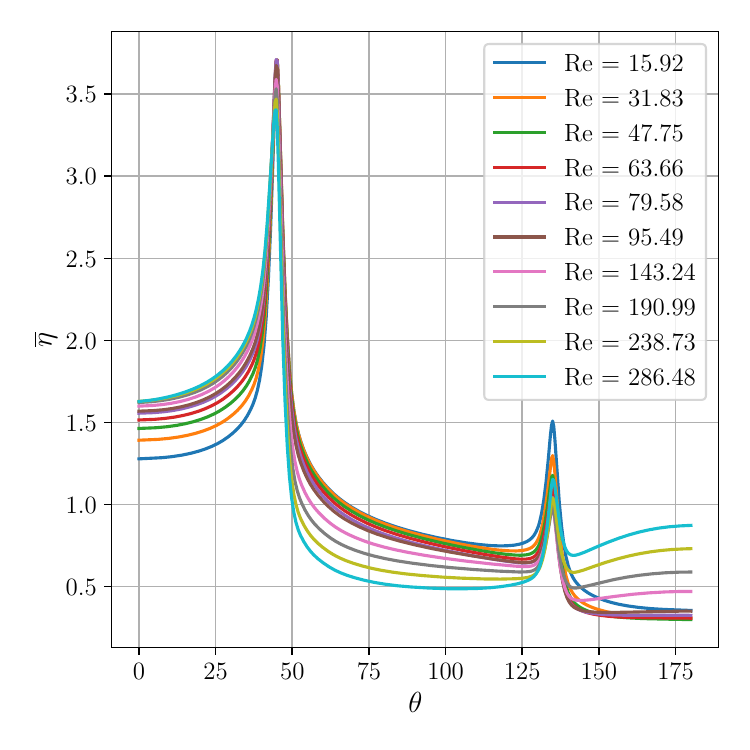}
        \caption{$\etabar$ curves}
        \label{fig:eta_square}
    \end{subfigure}
    \hfill
    \begin{subfigure}[c]{0.48\textwidth}
        \centering
        \begin{tabular}{l|c|c}
            Source & Re & $\dashint_{\pOmega}(\etabar-1)^2$ \\ \hline
            Nek5000 & 15.92   & 0.291 \\
                    & 31.83   & 0.341 \\
                    & 47.75   & 0.371 \\
                    & 63.66   & 0.391 \\
                    & 79.58   & \textbf{0.407} \\
                    & 95.49   & 0.404 \\
                    & 143.24  & 0.391 \\
                    & 190.99  & 0.387 \\
                    & 238.73  & 0.385 \\
                    & 286.48  & 0.392
        \end{tabular}
        \caption{Variance values}
        \label{tab:variance_square}
    \end{subfigure}
    \caption{Time-independent variation function $\etabar$ for cross flow over a cylinder with square cross section from Nek5000 simulations. (a)~Corresponding $\etabar$ curves. (b)~Variance $\dashint_{\pOmega}(\etabar-1)^2$ for different Reynolds numbers.}
    \label{fig:square_combined}
\end{figure}

\subsection{Interpretation}
The study, though anecdotal, suggests that the variance of heat exchange around a body first increases up to a critical value and then decreases as the flow becomes more turbulent. This indicates that a satisfactory upper bound for $\dashint_{\pOmega}(\etabar-1)^2$ can be obtained from low Reynolds number simulations. These cases are relatively simple to compute, and one could in principle construct a database of $\deltaeta$ for different geometries that could then be used to find $\phi^{\text{ub}}$ in~\cref{eq:phi_bound_result}.
\section{Numerical Treatment of Eigenvalue Problems}\label{sec:steklov_numerical}
To evaluate the upper bound for $\phi$ in~\cref{eq:phi_bound_result}, one linear PDE (\cref{eq:sensitivity_problem}) and two eigenvalue problems (\cref{eq:stability_1,eq:stability_2}) need to be solved. The numerical treatment of~\cref{eq:sensitivity_problem} has been discussed in~\cite{kaneko2024error} (for $\etabar=1$). Here, we focus on the numerical solution of the eigenvalue problems:
\begin{align*}
    \mu \coloneqq \inf_{w\in Z_0(1)} \frac{a_0(w,w)}{\int_\Omega w^2}, \quad \Lambda \coloneqq \inf_{w\in Z_0(1)} \frac{a_0(w,w)}{\int_{\pOmega} w^2},
\end{align*}
where $Z_0(1)\coloneqq \{ w\in H^1(\Omega) \ | \ \dashint_\Omega w=0 \}$, and for $u,v\in H^1(\Omega)$,
\begin{align*}
    a_0(u,v) \coloneqq \int_\Omega \delnd u \cdot \delnd v.
\end{align*}
Using a discretization $X_h \subset H^1(\Omega)$ with basis $\{\varphi_i\}_{i=1}^{N_h}$, we can approximate the eigenvalue problems by the following matrix problems:
\begin{align*}
    \mu_h \coloneqq \inf_{\substack{\mathbf{w}\in \mathbb{R}^{N_h} \\ \mathbf{c}^T\mathbf{w}=0}} \frac{\mathbf{w}^T\mathbf{A}\mathbf{w}}{\mathbf{w}^T\mathbf{M}\mathbf{w}}, \quad \Lambda_h \coloneqq \inf_{\substack{\mathbf{w}\in \mathbb{R}^{N_h} \\ \mathbf{c}^T\mathbf{w}=0}} \frac  {\mathbf{w}^T\mathbf{A}\mathbf{w}}{\mathbf{w}^T\mathbf{B}\mathbf{w}},
\end{align*}
where the matrices $\mathbf{A},\mathbf{M},\mathbf{B}\in\mathbb{R}^{N_h\times N_h}$ and vector $\mathbf{c}\in\mathbb{R}^{N_h}$ are defined as
\begin{align*}
    A_{ij} &\coloneqq a_0(\varphi_j,\varphi_i) = \int_\Omega \delnd \varphi_j \cdot \delnd \varphi_i, \\
    M_{ij} &\coloneqq \int_\Omega \varphi_j \varphi_i, \\
    B_{ij} &\coloneqq \int_{\pOmega} \varphi_j \varphi_i, \\
    c_i &\coloneqq \int_\Omega \varphi_i.
\end{align*}
To impose the constraint $\mathbf{c}^T\mathbf{w}=0$, we use a Lagrange multiplier $\lambda\in\mathbb{R}$ and obtain the following saddle point problems:
\begin{align*}
    \begin{bmatrix}
    \mathbf{A} & \mathbf{c} \\
    \mathbf{c}^T & 0
    \end{bmatrix}
    \begin{bmatrix}
    \mathbf{w} \\ \lambda
    \end{bmatrix}
    &= \mu_h
    \begin{bmatrix}
    \mathbf{M} & 0 \\
    0 & 0
    \end{bmatrix}
    \begin{bmatrix}
    \mathbf{w} \\ \lambda
    \end{bmatrix}, \\
    \begin{bmatrix}
    \mathbf{A} & \mathbf{c} \\
    \mathbf{c}^T & 0
    \end{bmatrix}
    \begin{bmatrix}
    \mathbf{w} \\ \lambda
    \end{bmatrix}
    &= \Lambda_h
    \begin{bmatrix}
    \mathbf{B} & 0 \\
    0 & 0
    \end{bmatrix}
    \begin{bmatrix}
    \mathbf{w} \\ \lambda
    \end{bmatrix}.
\end{align*}
Eigenvalues are computed with SLEPc~\cite{hernandez2005slepc} using shift-and-invert method.
\section{Convergence Study and Validation of Simulations}\label{sec:convergence_study}
For the three-dimensional simulations considered in this work, polynomial order $p=7$ is used, and a mild spectral filter is applied to maintain numerical stability over long simulation times. 
Following the recommendations of~\cite{fischer2001filter,nek5000-web-page}, the filter is applied to the highest two modes, attenuating them by 5\% and 1.25\%, respectively.
Because the filter amplitudes are small, they have a negligible influence on the solution and therefore do not affect the analysis of the boundary layer profiles discussed above, while ensuring numerical stability of the simulations.

To verify that the spatial resolution is sufficient, we examine the boundary layers near the fluid-solid interface. 
This is done by inspecting the temperature variation along lines normal to the interface within individual spectral elements, as shown in the following subsections.
A smooth variation of temperature within each element and an absence of oscillations indicate that the polynomial representation within each element is sufficient to resolve the boundary layer accurately.

\subsection{Flow Around Cylinder (CHT)}\label{subsec:convergence_study_cyl}
We show in~\cref{fig:disk_validation} the two-dimensional simulation mesh used for the cross-flow around a cylinder in~\cref{subsec:example_cht} (assumed infinitely long). Each quadrilateral represents a spectral element of polynomial order $p=7$. A very fine mesh is employed near the fluid-solid interface to resolve the boundary layers, as illustrated in~\cref{fig:disk_validation_2}. To assess the spatial resolution,~\cref{fig:disk_validation_3} shows the temperature variation along the red line at the leading edge, spanning two elements at the interface. The temperature profile along this line varies linearly, confirming that the mesh resolution is sufficient to capture the boundary layer behavior accurately. A similar behavior is observed at the trailing edge (not shown here).

\begin{figure}[ht]
    \centering
    \begin{subfigure}[b]{0.58\textwidth}
        \centering
        \includegraphics[width=\textwidth]{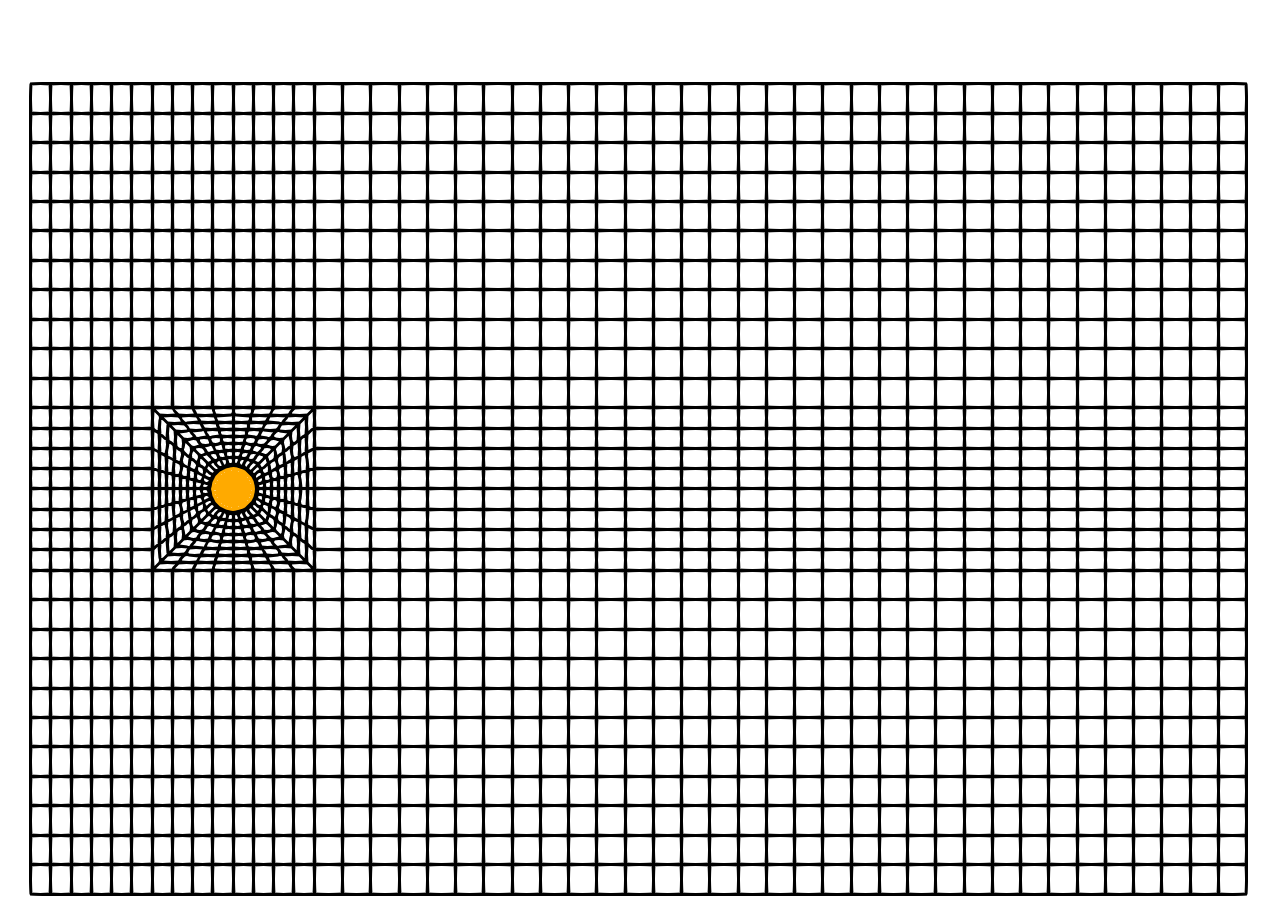}
        \caption{Entire mesh}
        \label{fig:disk_validation_1}
    \end{subfigure}
    \hfill
    \begin{subfigure}[b]{0.37\textwidth}
        \centering
        \includegraphics[width=\textwidth]{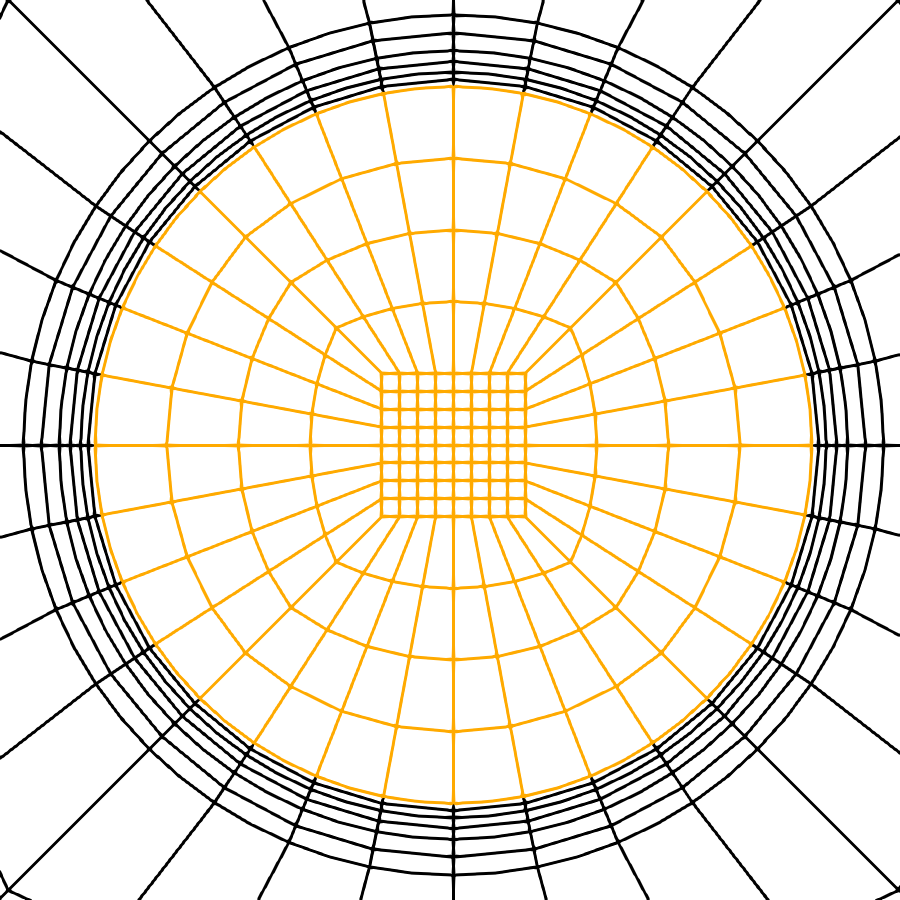}
        \caption{Mesh around cylinder}
        \label{fig:disk_validation_2}
    \end{subfigure}
    \caption{Simulation mesh for flow around a cylinder. The mesh of the fluid domain is in black and the solid domain in orange. The mesh at the interface is very fine to resolve the boundary layers.}
    \label{fig:disk_validation}
\end{figure}
\begin{figure}[ht]
    \centering
    \includegraphics[width=0.7\textwidth]{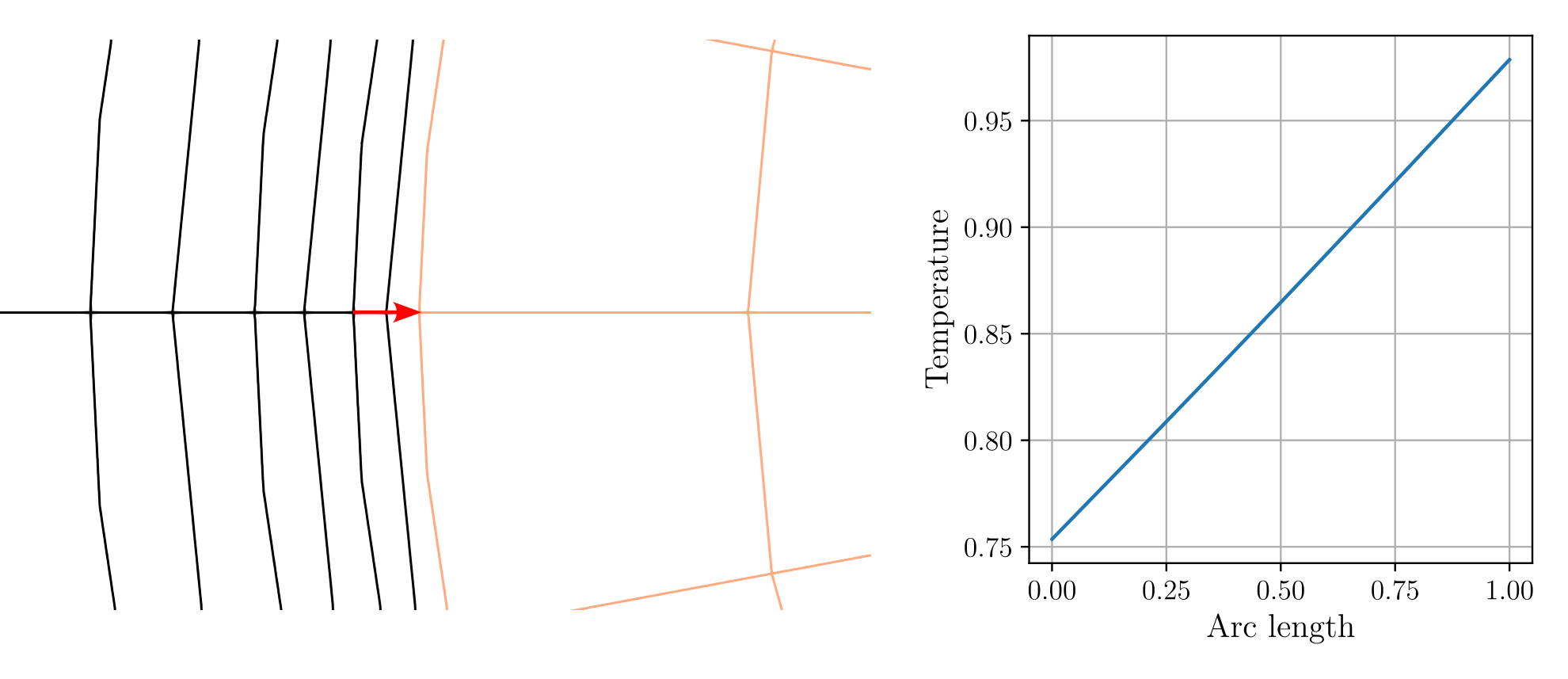}
    \caption{Temperature variation along the red line at the leading edge.}
    \label{fig:disk_validation_3}
\end{figure}

\subsection{Flow Around Sphere (ISO)}\label{subsec:convergence_study_sph}
We show in~\cref{fig:sphere_validation} the computational mesh used for the case $\Reynolds[\sqrt{\dm{A}}] = 10^4$ described in~\cref{subsubsec:proof_of_concept_sphere}, visualized in the plane $z = 0$ near the sphere surface. A close-up view of the stagnation point region is also provided. The triangular patterns visible in the slice are artefacts of the ParaView visualization and do not correspond to the actual spectral-element discretization. All grid points shown correspond to the Gauss–Lobatto integration points used within each spectral element. The temperature field along the red line within a single spectral element is plotted on the right. The temperature exhibits a slightly nonlinear variation across the element, which is accurately captured by the seventh-order polynomial representation, demonstrating that the thermal boundary layer is well resolved.

\begin{figure}[ht]
    \centering
    \includegraphics[width=\textwidth]{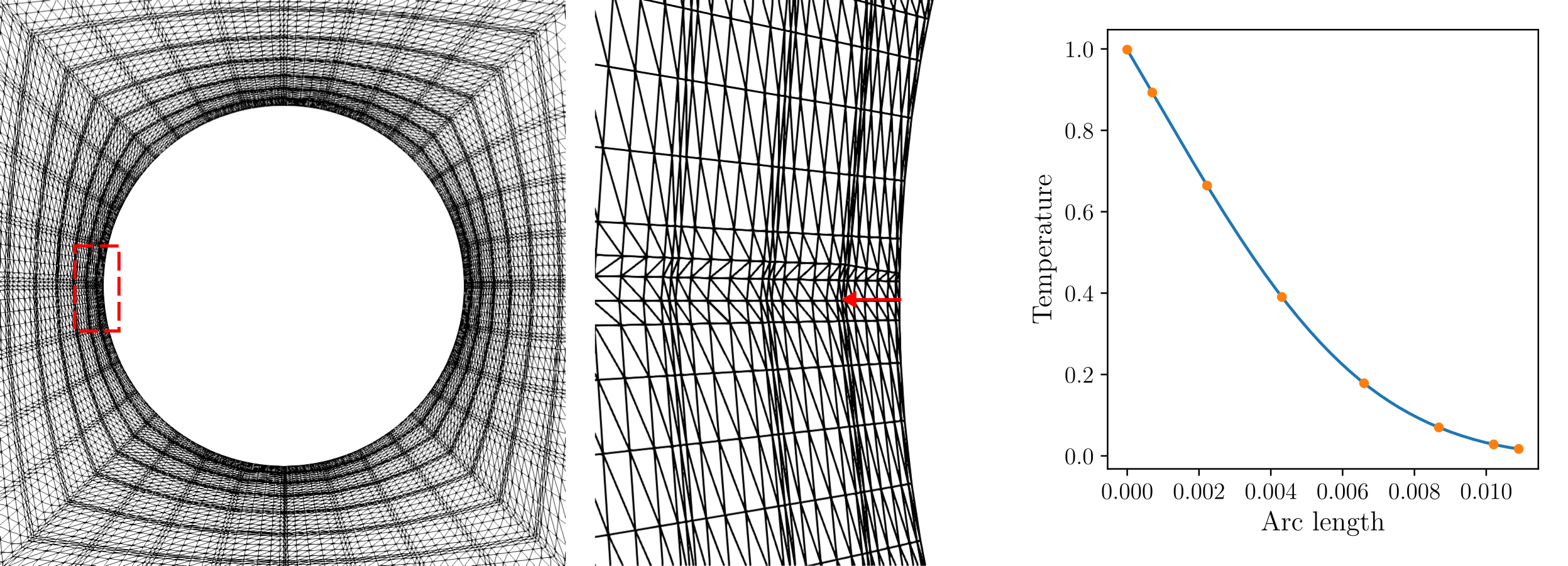}
    \caption{
        Visualization of the mesh and temperature field for the sphere case at $\Reynolds[\sqrt{\dm{A}}] = 10^4$. 
        Left: mesh around the sphere for a slice $z = 0$; center: zoomed view at the stagnation point; right: temperature variation along the red line at the leading edge.
    }
    \label{fig:sphere_validation}
\end{figure}

The same mesh was also used for a lower Reynolds number case of $\Reynolds[\sqrt{\dm{A}}] = 886$, which was independently computed in~\cite{bagchi2001direct}. In the length scale based on $\sqrt{\dm{A}}$, their reported Nusselt number is $\Nusselt[\sqrt{\dm{A}}] \approx 23.52$. Our simulation yields $\Nustavgiso = 23.60$, corresponding to a relative deviation of $0.3\%$. The spatially averaged Nusselt number as a function of time is shown in~\cref{fig:sphere_nusselt_validation}. A sliding window was applied to this time series to compute the running temporal average, and the simulation was advanced until statistical steady state was reached, both as described in~\cref{sec:steady_state_algorithm}.
\begin{figure}[ht]
    \centering
    \includegraphics[width=0.5\textwidth]{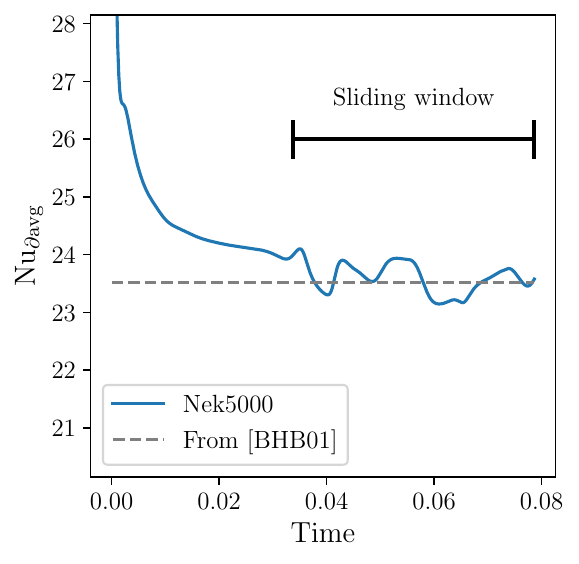}
    \caption{Spatially averaged Nusselt number for the sphere at $\Reynolds[\sqrt{\dm{A}}] = 886$. The horizontal dashed line denotes the reference value from~\cite{bagchi2001direct}, $\Nustavgiso[\sqrt{\dm{A}}] = 23.52$. The black bar indicates the sliding window used for time averaging, which yields $\Nustavgiso[\sqrt{\dm{A}}] = 23.60$. The implementation of the sliding window averaging procedure is described in~\cref{sec:steady_state_algorithm}.}
    \label{fig:sphere_nusselt_validation}
\end{figure}

\subsection{Flow Around Prolate Spheroid (ISO)}\label{subsec:convergence_study_sph_pro}
We show in~\cref{fig:prolate_spheroid_validation} the computational mesh (with Gauss--Lobatto points) used for the prolate spheroid $s=10$ in axial-flow for $\Reynolds[\sqrt{\dm{A}}] = 2000$, visualized in the plane $z = 0$ near the prolate spheroid surface. A close-up view of the stagnation point region is also provided. The temperature shows a highly nonlinear variation across the element, and it is unclear if it is converged. To validate the result, we therefore also consider a higher polynomial order $p=8$ for this case. The Nusselt numbers for both polynomial orders are found to be basically identical, with less than 0.1\% difference.
\begin{figure}[p]
    \centering
    \includegraphics[width=\textwidth]{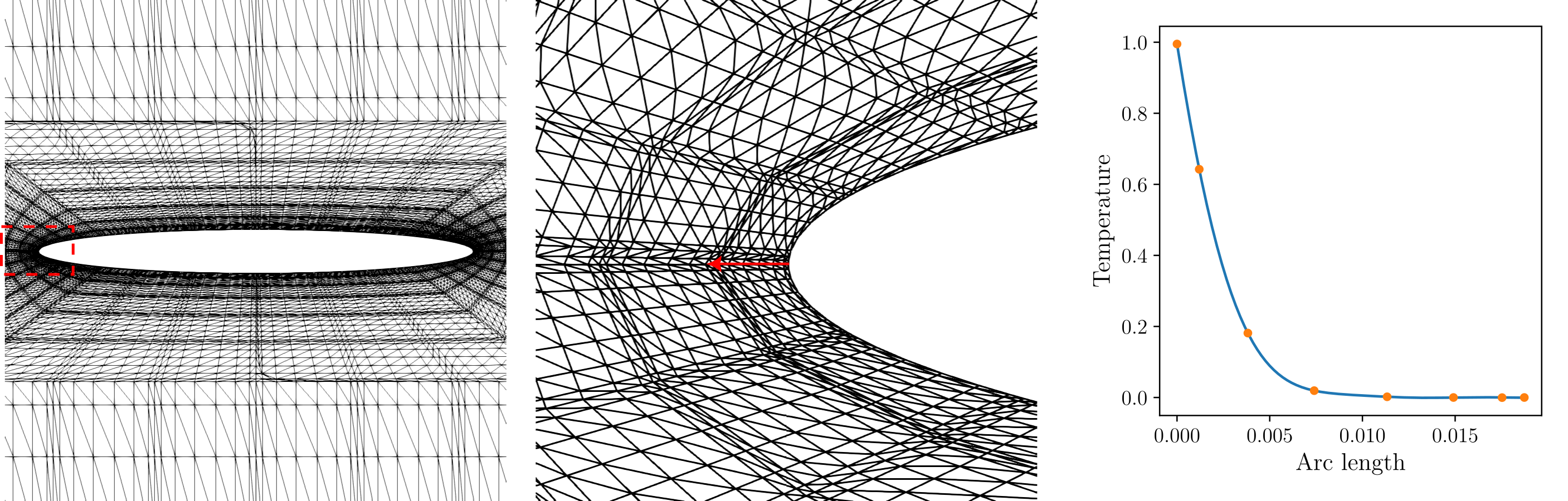}
    \caption{
        Visualization of the mesh and temperature field for the prolate spheroid case at $\Reynolds[\sqrt{\dm{A}}] = 2000$. 
        Left: mesh around the prolate spheroid; center: zoomed view at the stagnation point; right: temperature variation along the red line at the leading edge.
    }
    \label{fig:prolate_spheroid_validation}
\end{figure}

\subsection{Flow Around Oblate Spheroid (ISO)}\label{subsec:convergence_study_sph_obl}
We show in~\cref{fig:oblate_spheroid_validation} the computational mesh (with Gauss--Lobatto points) used for the oblate spheroid $s=0.1$ in cross-flow for $\Reynolds[\sqrt{\dm{A}}] = 2000$, visualized in the plane $z = 0$ near the oblate spheroid surface. A close-up view of the stagnation point region is also provided. The temperature field along the red line within a single spectral element is plotted on the right. The temperature shows a linear variation across the element, which is accurately captured by the seventh-order polynomial representation.
\begin{figure}[p]
    \centering
    \includegraphics[width=\textwidth]{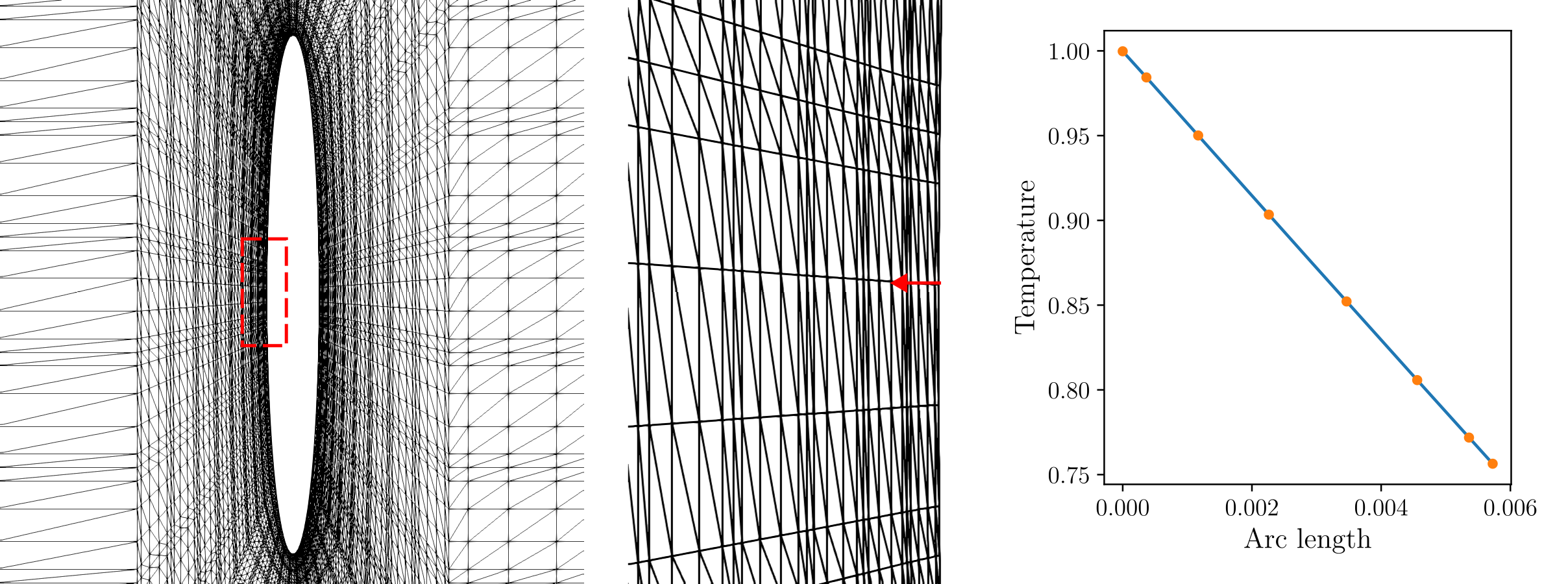}
    \caption{
        Visualization of the mesh and temperature field for the oblate spheroid case at $\Reynolds[\sqrt{\dm{A}}] = 2000$. 
        Left: mesh around the oblate spheroid; center: zoomed view at the stagnation point; right: temperature variation along the red line at the leading edge.
    }
    \label{fig:oblate_spheroid_validation}
\end{figure}

\subsection{Flow Around Cuboid (ISO)}\label{subsec:convergence_study_cuboid}
We show in~\cref{fig:cuboid_validation} the computational mesh (with Gauss--Lobatto points) used for the cuboid case as presented in~\cref{subsubsec:nusselt_cuboid}, for $\Reynolds[\sqrt{\dm{A}}] = 5000$, visualized in the plane $z = 0$ near the cuboid surface. A close-up at one of the corners is also provided. The temperature field along the red line within a single spectral element is plotted on the right. The temperature shows a slightly non-linear variation across the element, which is accurately captured by the seventh-order polynomial representation.
\begin{figure}[p]
    \centering
    \includegraphics[width=\textwidth]{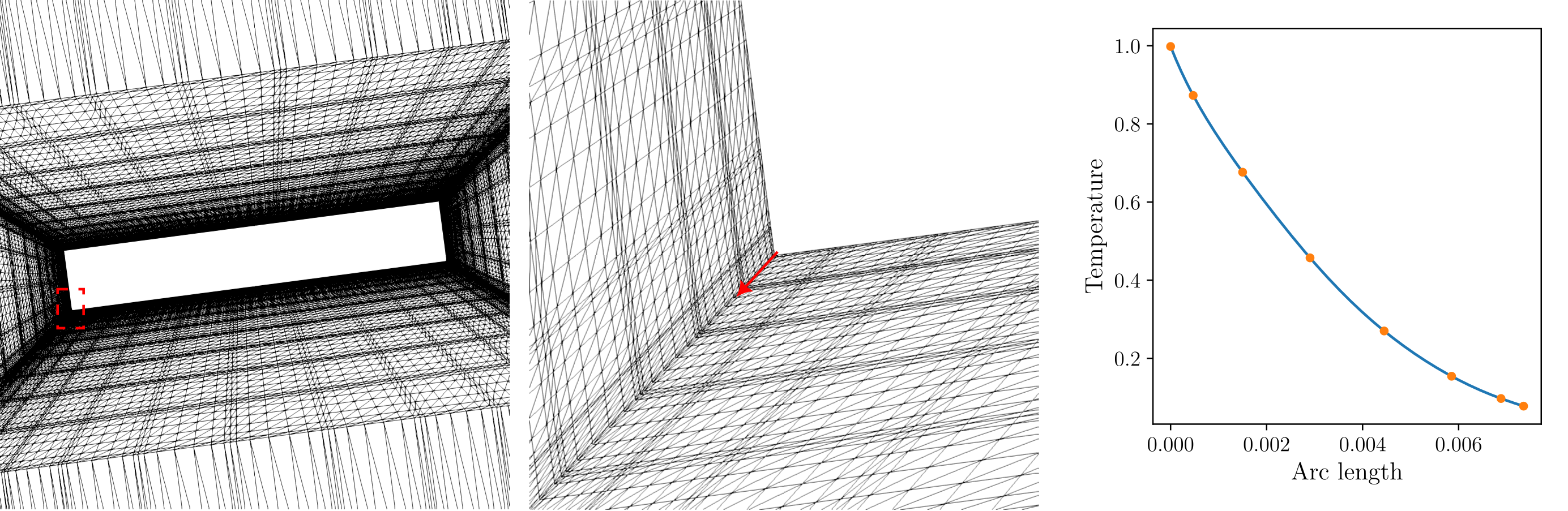}
    \caption{
        Visualization of the mesh and temperature field for the cuboid case at $\Reynolds[\sqrt{\dm{A}}] = 5000$. 
        Left: mesh around the cuboid; center: zoomed view at the bottom left corner; right: temperature variation along the red line at the leading edge.
    }
    \label{fig:cuboid_validation}
\end{figure}
\section{Time Homogenization For A Simplified Robin Heat Equation}\label{sec:time_homogenization_rhe}

In this section, we present a simplified version of the Robin heat equation from~\cref{eq:rhe_nondim_T} to illustrate the time homogenization procedure. We consider the $\varepsilon$-dependent problem over the time interval $[0,T]$:
\begin{subequations}\label{eq:rhe_simple_eps}
\begin{alignat}{3}
    \frac{\partial u_\varepsilon}{\partial t} - \Deltand u_\varepsilon &= 0 &\quad& \text{in } \Omega, \label{eq:rhe_simple_eps_1}\\
    \delnd u_\varepsilon \cdot \bm{n} + \eta\!\left(\frac{t}{\varepsilon}\right) u_\varepsilon &= 0 &\quad& \text{on } \pOmega, \label{eq:rhe_simple_eps_2}\\
    u_\varepsilon(t=0,\cdot) &= u^0(\xnd) &\quad& \text{in } \Omega, \label{eq:rhe_simple_eps_3}
\end{alignat}
\end{subequations}
where the subscript $\varepsilon$ indicates dependence on the small parameter $\varepsilon$ and $\eta$ is assumed to be periodic (of period $T_0$) in time, and $\eta\in L^\infty([0,T])$. Furthermore, $\eta$ is assumed non-negative for almost all times and not identically zero. $\Omega$ is a sufficiently regular domain. We claim, and it is a standard result for which, for the sake of consistency, we are going to provide a proof here, that the homogenized problem reads
\begin{subequations}\label{eq:rhe_simple_hom}
\begin{alignat}{3}
    \frac{\partial \overline{u}}{\partial t} - \Deltand \overline{u} &= 0 &\quad& \text{in } \Omega, \label{eq:rhe_simple_hom_1}\\
    \delnd \overline{u} \cdot \bm{n} + \etabar \,\overline{u} &= 0 &\quad& \text{on } \pOmega, \label{eq:rhe_simple_hom_2}\\
    \overline{u}(t=0,\cdot) &= u^0(\xnd) &\quad& \text{in } \Omega, \label{eq:rhe_simple_hom_3}
\end{alignat}
\end{subequations}
where $\etabar$ is constant and denotes the time average of $\eta$ over $[0,T_0]$,
\begin{align*}
    \etabar \coloneqq \dashint_0^{T_0} \eta,
\end{align*}
which is strictly positive, and the oscillatory coefficient $\eta$ converges weak-* in $L^\infty([0,T])$ to $\etabar$ as $\varepsilon \to 0$. Our goal is to show that the solution $u_\varepsilon$ of~\cref{eq:rhe_simple_eps} converges to the solution $\overline{u}$ of~\cref{eq:rhe_simple_hom} as $\varepsilon \to 0$.

Multiplying~\cref{eq:rhe_simple_eps_1} by $u_\varepsilon$ and integrating over $\Omega$ and in time over $[0,T]$ yields
\begin{align*}
    \frac{1}{2} \int_\Omega |u_\varepsilon(T,\cdot)|^2 
    + \int_0^T \int_\Omega |\delnd u_\varepsilon|^2 
    + \int_0^T \eta\!\left(\frac{t}{\varepsilon}\right) \int_{\pOmega} |u_\varepsilon|^2 
    = \frac{1}{2} \int_\Omega |u^0|^2.
\end{align*}
From this estimate we immediately obtain a bound uniform in $\varepsilon$ in space
\begin{align*}
    L^\infty\!\big([0,T],L^2(\Omega)\big) 
    \cap L^2\!\big([0,T],H^1(\Omega)\big).
\end{align*}
The same manipulations with~\cref{eq:rhe_simple_hom} show that:
\begin{align*}
    \overline{u} \in L^\infty\!\big([0,T],L^2(\Omega)\big) 
    \cap L^2\!\big([0,T],H^1(\Omega)\big).
\end{align*}
We now consider the difference
\begin{align*}
    v_\varepsilon \coloneqq u_\varepsilon - \overline{u},
\end{align*}
which is also a bounded sequence in $L^\infty\!\big([0,T],L^2(\Omega)\big) \cap L^2\!\big([0,T],H^1(\Omega)\big)$ and satisfies
\begin{subequations}\label{eq:rhe_simple_diff}
\begin{alignat}{3}
    \frac{\partial v_\varepsilon}{\partial t} - \Deltand v_\varepsilon &= 0 &\quad& \text{in } \Omega, \label{eq:rhe_simple_diff_1}\\
    \delnd v_\varepsilon \cdot \bm{n} + \eta\!\left(\frac{t}{\varepsilon}\right) v_\varepsilon &= -\Big(\eta\!\left(\frac{t}{\varepsilon}\right) - \etabar\Big)\,\overline{u} &\quad& \text{on } \pOmega, \label{eq:rhe_simple_diff_2}\\
    v_\varepsilon(t=0,\cdot) &= 0 &\quad& \text{in } \Omega. \label{eq:rhe_simple_diff_3}
\end{alignat}
\end{subequations}
Since $v_\varepsilon$ is bounded in $L^\infty([0,T],L^2(\Omega))\cap L^2([0,T],H^1(\Omega))$, there exists a function $\overline{v}$ and a subsequence (not relabeled) such that
$v_\varepsilon$ converges weak-* to $\overline{v}$ in $L^\infty([0,T],L^2(\Omega))$ and weakly to $\overline{v}$ in $L^2([0,T],H^1(\Omega))$ as $\varepsilon\to0$. We want to show that $\overline{v} = 0$, which implies that $u_\varepsilon$ converges to $\overline{u}$ in the same topology.

Because of~\cref{eq:rhe_simple_diff_1}, we infer that $\dfrac{\partial v_\varepsilon}{\partial t}$ is a bounded sequence in
\begin{align*}
    L^2([0,T],H^{-1}(\Omega)).
\end{align*}
Since $H^1(\Omega)$ is compactly embedded in $L^2(\Omega)$, and $L^2(\Omega)$ is at least continuously embedded in $H^{-1}(\Omega)$, we can apply the Aubin–Lions Lemma~\cite[Theorem II.5.16]{boyer2012mathematical},
\begin{align*}
    v_\varepsilon \to \overline{v} 
    \quad \text{in } L^2([0,T],L^2(\Omega)) 
    \quad \text{as } \varepsilon \to 0,
\end{align*}
that is, $v_\varepsilon$ converges \emph{strongly} to $\overline{v}$ in $L^2([0,T],L^2(\Omega))$.

Since $v_\varepsilon$ converges weakly to $\overline{v}$ in $L^2([0,T],H^1(\Omega))$, and strongly in $L^2([0,T],L^2(\Omega))$, an interpolation argument yields
\begin{align*}
    v_\varepsilon \to \overline{v} 
    \quad \text{in } L^2([0,T],H^{1/2}(\Omega)) 
    \quad \text{as } \varepsilon \to 0.
\end{align*}
Applying the trace theorem then gives the strong convergence
\begin{align*}
    v_\varepsilon \to \overline{v} 
    \quad \text{in } L^2([0,T],L^2(\pOmega)) 
    \quad \text{as } \varepsilon \to 0.
\end{align*}
We can now multiply~\cref{eq:rhe_simple_diff_1} by $v_\varepsilon$ and integrate over $\Omega$ and over $[0,T]$ to obtain
\begin{align*}
    \frac{1}{2} \int_\Omega |v_\varepsilon(T,\cdot)|^2 
    + \int_0^T \int_\Omega |\delnd v_\varepsilon|^2 
    + \int_0^T \eta\!\left(\frac{t}{\varepsilon}\right) \int_{\pOmega} |v_\varepsilon|^2
    = -\overline{u} \int_0^T \left(\eta\!\left(\frac{t}{\varepsilon}\right) - \etabar\right) \int_{\pOmega} \, v_\varepsilon.
\end{align*}
On the right hand side, $\eta\!\left(\tfrac{t}{\varepsilon}\right) - \etabar$ converges weak-* to $0$ in $L^\infty([0,T])$. Since $v_\varepsilon$ converges strongly to $\overline{v}$ in $L^2([0,T],L^2(\pOmega))$, it also converges strongly in $L^2([0,T],L^1(\pOmega))$, and hence the right hand side vanishes as $\varepsilon \to 0$. On the left hand side, $v_\varepsilon$ converges weakly to $\overline{v}$ in $L^2([0,T],H^1(\Omega))$, and strongly in $L^2([0,T],L^2(\pOmega))$. Observing that the leftmost term is non-negative, we may thus pass at the $\liminf$ in the left hand side and obtain:
\begin{align*}
    \int_0^T \int_\Omega |\delnd \overline{v}|^2 
    + \int_0^T \etabar \int_{\pOmega} |\overline{v}|^2
    = 0.
\end{align*}
As all terms on the left are non-negative and $\etabar>0$, this implies $\overline{v} \equiv 0$. Using uniqueness of the limit of subsequences and compactness, we conclude that the full sequence converges.
\section{Algorithm for Determination of Steady State of ISO Simulation}\label{sec:steady_state_algorithm}

To determine the space-time averaged Nusselt number $\Nustavgiso$ from the ISO simulation, it is essential to ensure that the flow has reached a statistically stationary state. To minimize computational cost, the simulation should be terminated as soon as $\Nustavgiso$ no longer varies significantly with time.

The dominant vortex shedding oscillation frequency, denoted by $\dm{f}_{\mathrm{vs}}$, is estimated using the Strouhal number. The Strouhal number is defined as
\begin{align}
    \Strouhal = \frac{\dm{f}_{\mathrm{vs}}\, \elldim}{\vinfdim},
    \label{eq:strouhal_definition}
\end{align}
where $\dm{f}_{\mathrm{vs}}$ is the dimensional vortex shedding frequency, $\elldim$ the characteristic length, and $\vinfdim$ the far-field velocity. 
In our nondimensionalization based on the solid diffusive time scale (see~\cref{eq:tdiff}), the corresponding nondimensional frequency becomes
\begin{align}
    f_{\mathrm{vs}}
    = \dm{f}_{\mathrm{vs}}\, \tdiffdim
    = \Strouhal\, \frac{r_2}{r_1}\, \Reynolds\, \Prandtl.
    \label{eq:freq_estimate}
\end{align}
A representative value of $\Strouhal = 0.2$ is adopted for all cases, consistent with experimental observations for bluff bodies such as spheres~\cite{kim1988observations} and cylinders~\cite{lienhard1973commonality} over a broad range of Reynolds numbers (based on diameter).

We use a sliding window to average in time. The sliding window must be sufficiently wide to capture several oscillations of the instantaneous Nusselt number. From~\cref{eq:freq_estimate}, the oscillation period follows as $$t_{\mathrm{vs}} = 1 / f_{\mathrm{vs}}.$$ 
The initial width of the sliding window is set to $5\,t_{\mathrm{vs}}$, ensuring that at least five oscillation cycles are included in each average. 
A new average is computed after shifting the window forward by a “step size” of $0.5\,t_{\mathrm{vs}}$; this shift defines one “step”. At each step, we compute $\Nustavgiso$ over the current window and monitor its temporal convergence. The simulation is considered statistically stationary once the average relative change in $\Nustavgiso$ over the last five consecutive steps falls below $0.1\%$, after which it is terminated. 

To exclude the influence of the initial transient phase, the convergence criterion is applied only after the simulation time exceeds $7.5\,t_{\mathrm{vs}}$. 
To further ensure convergence of the sliding-window averages, beginning at $t = 5\,t_{\mathrm{vs}}$, the window width is gradually increased by $0.05\,t_{\mathrm{vs}}$ every time the simulation time advances by $0.5\,t_{\mathrm{vs}}$, while the step size remains fixed at $0.5\,t_{\mathrm{vs}}$. 
Because the flow and thermal responses vary across geometries, the precise values of the initial window width and step size were occasionally adjusted empirically.

For validation, we compared our implementation against the results of~\cite{bagchi2001direct} for flow over a sphere at a Reynolds number (based on diameter) of $\Reynolds = 500$. The reference value reported there is $\Nusselt \approx 13.27$, while the present algorithm yielded $\Nustavgiso = 13.31$, corresponding to a relative deviation of only $0.3\%$. The temporal evolution of the spatially averaged Nusselt number, together with the final window width, were shown in~\cref{fig:sphere_nusselt_validation}. The difference in the Reynolds and Nusselt values arises from the different choices of characteristic length scale—here the diameter $\dm{D}$, and there $\sqrt{\dm{A}}$. The two length scales are related by $\sqrt{\dm{A}} = \dm{D}\sqrt{\pi}$.

\printbibliography

\end{document}